\documentclass[aps,twocolumn,superscriptaddress,showkeys,letterpaper]{revtex4-2}%
\usepackage{amsfonts}
\usepackage{amsmath}
\usepackage{amssymb}
\usepackage{subfigure}
\usepackage{mathtools}
\usepackage{graphicx}
\usepackage{multirow}
\usepackage{array}
\usepackage{xcolor}
\usepackage[letterpaper,margin=.6in,tmargin=.8in,bmargin=.9in]{geometry}
\usepackage[colorlinks = true, linkcolor = blue, urlcolor=blue, citecolor =
red, anchorcolor = green,]{hyperref}%
\usepackage{adjustbox}
\usepackage{float}
\providecommand{\U}[1]{\protect\rule{.1in}{.1in}}
\newtheorem{theorem}{Theorem}

\newtheorem{algorithm}[theorem]{Algorithm}

\newtheorem{corollary}[theorem]{Corollary}

\newtheorem{definition}[theorem]{Definition}

\newtheorem{lemma}[theorem]{Lemma}

\newtheorem{proposition}[theorem]{Proposition}

\newenvironment{proof}[1][Proof]{\noindent\textbf{#1.} }{\ \rule{0.5em}{0.5em}}

\newcolumntype{P}[1]{>{\centering\arraybackslash}p{#1}}
\newcolumntype{S}[1]{>{\centering\arraybackslash}m{#1}}

\newcommand{\tn}[1]{{\color{blue} TN: #1}}

\newcommand\tdpsandwich{\adjustbox{valign=m, vspace=0.0pt}{\includegraphics[width=.30\linewidth]{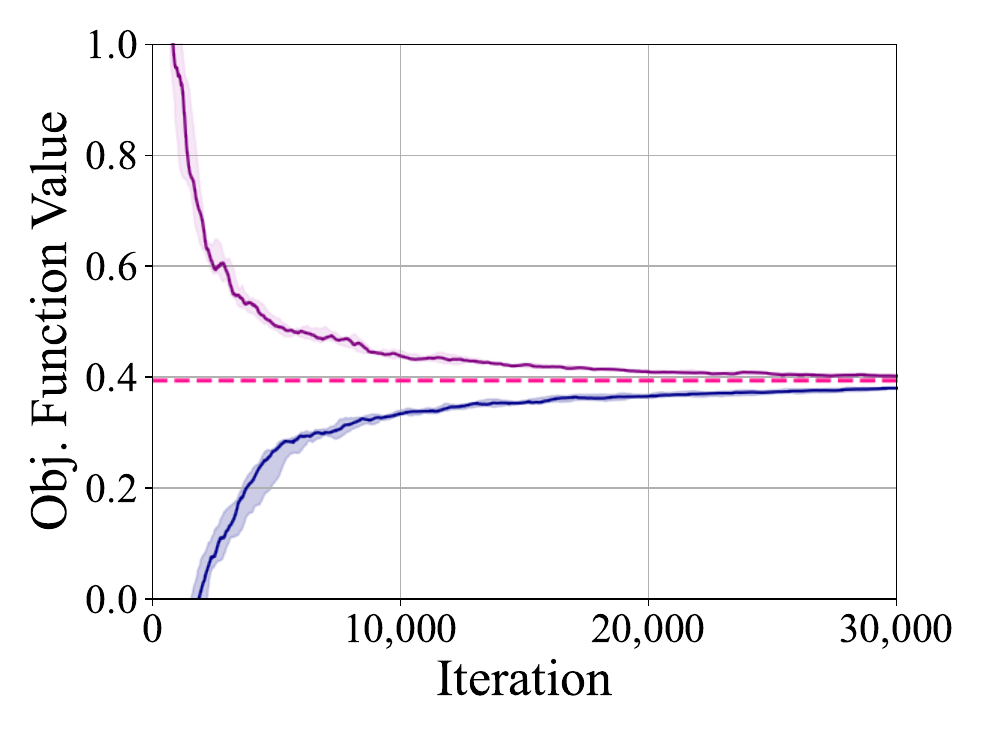}}}
\newcommand\tdperror{\adjustbox{valign=m, vspace=0.0pt}{\includegraphics[width=.40\linewidth]{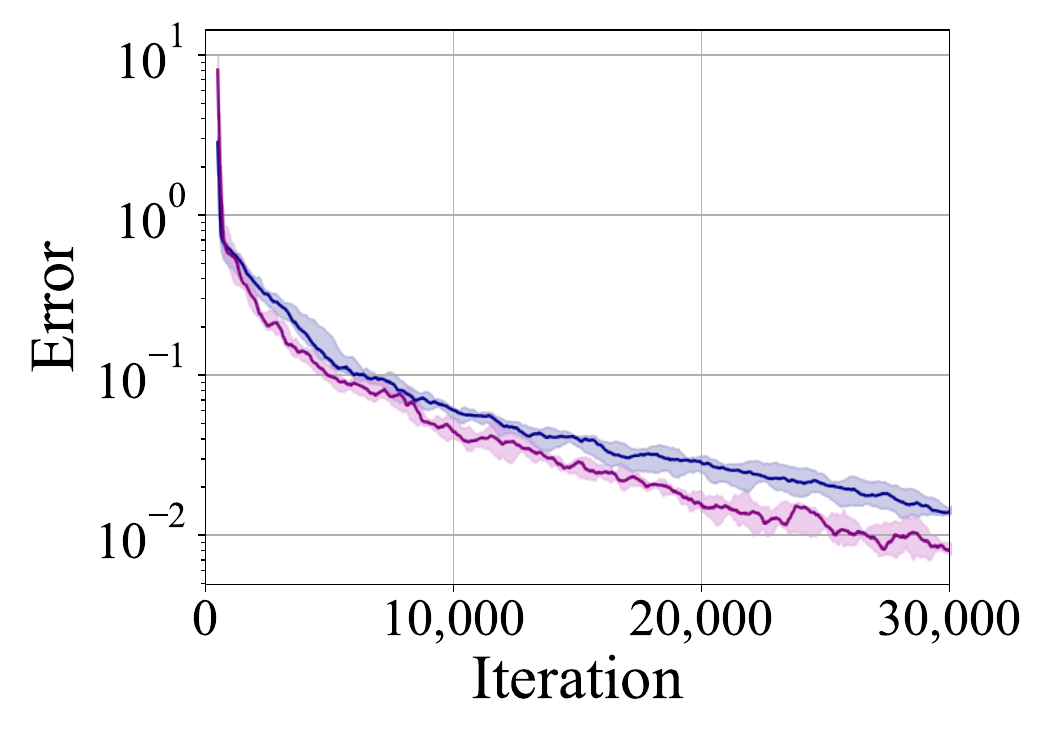}}}
\newcommand\tdppenalty{\adjustbox{valign=m, vspace=0.0pt}{\includegraphics[width=.40\linewidth]{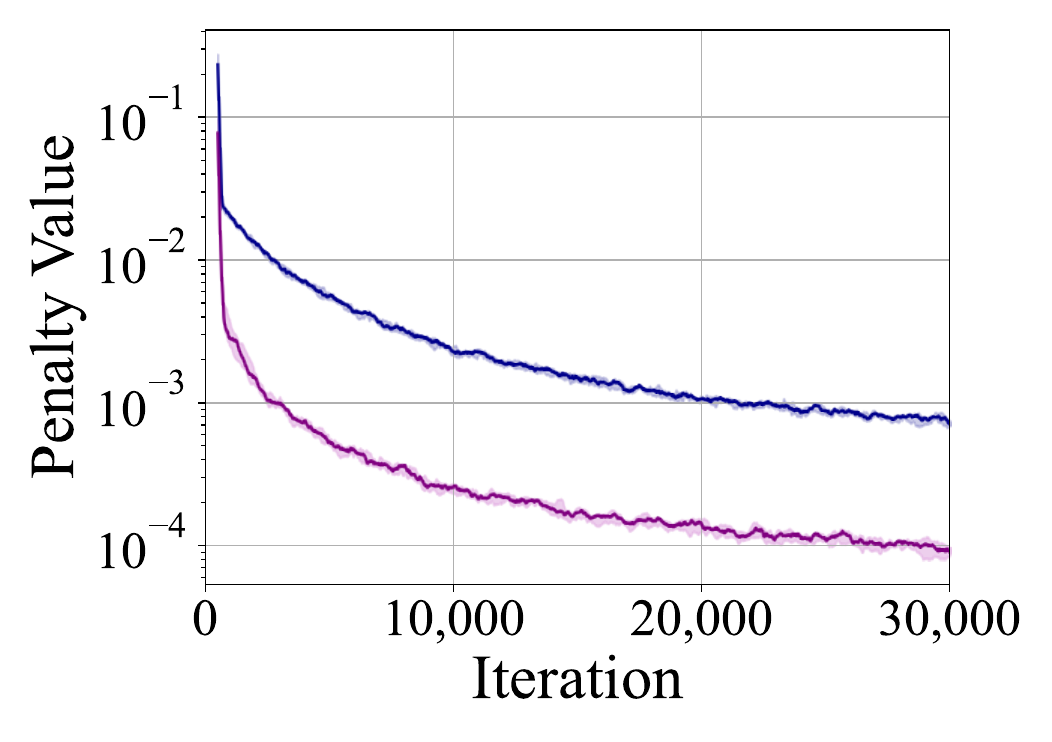}}}

\newcommand\tdccasandwich{\adjustbox{valign=m, vspace=0.0pt}{\includegraphics[width=.30\linewidth]{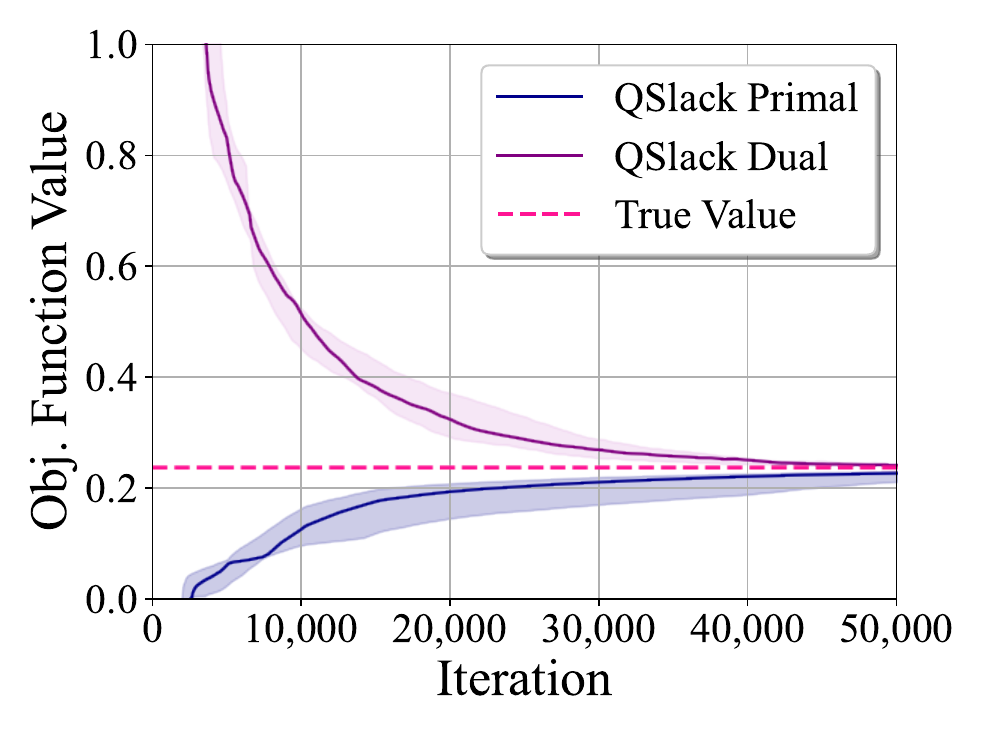}}}
\newcommand\tdccaerror{\adjustbox{valign=m, vspace=0.0pt}{\includegraphics[width=.40\linewidth]{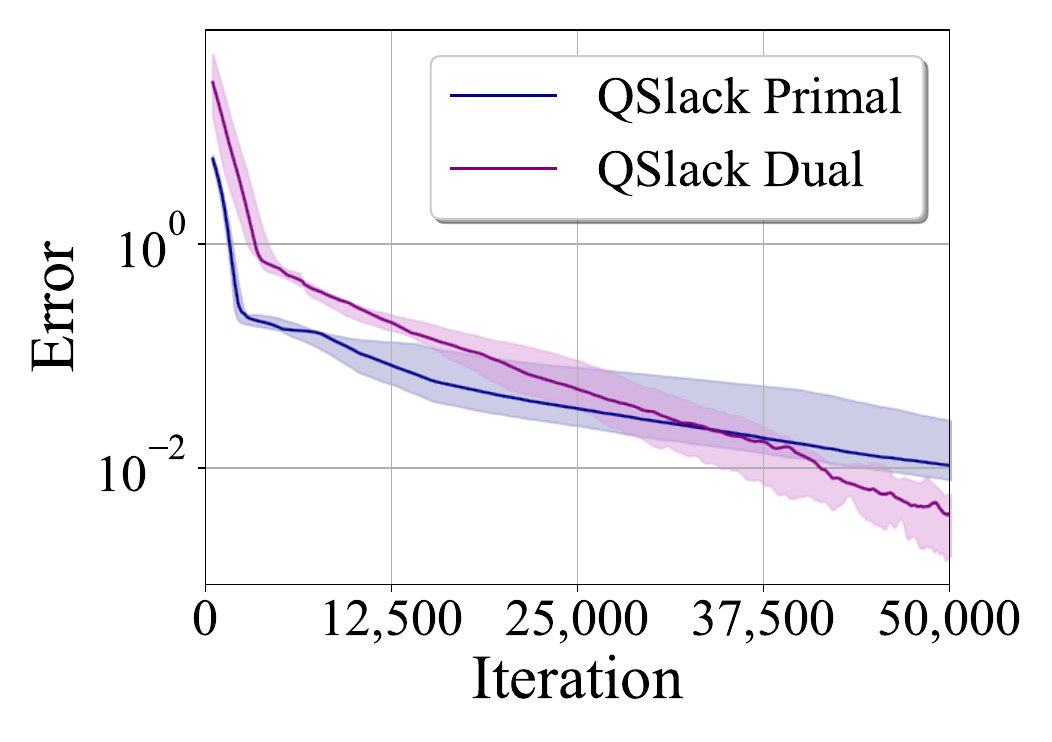}}}
\newcommand\tdccapenalty{\adjustbox{valign=m, vspace=0.0pt}{\includegraphics[width=.40\linewidth]{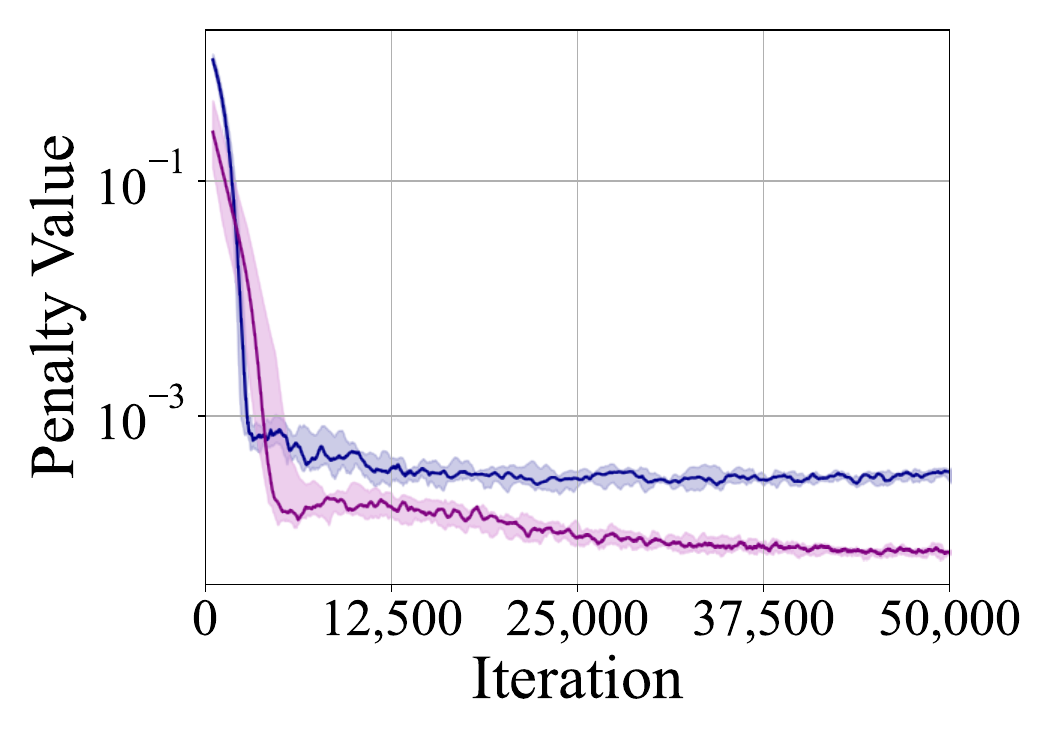}}}

\newcommand\fpsandwich{\adjustbox{valign=m, vspace=0.0pt}{\includegraphics[width=.30\linewidth]{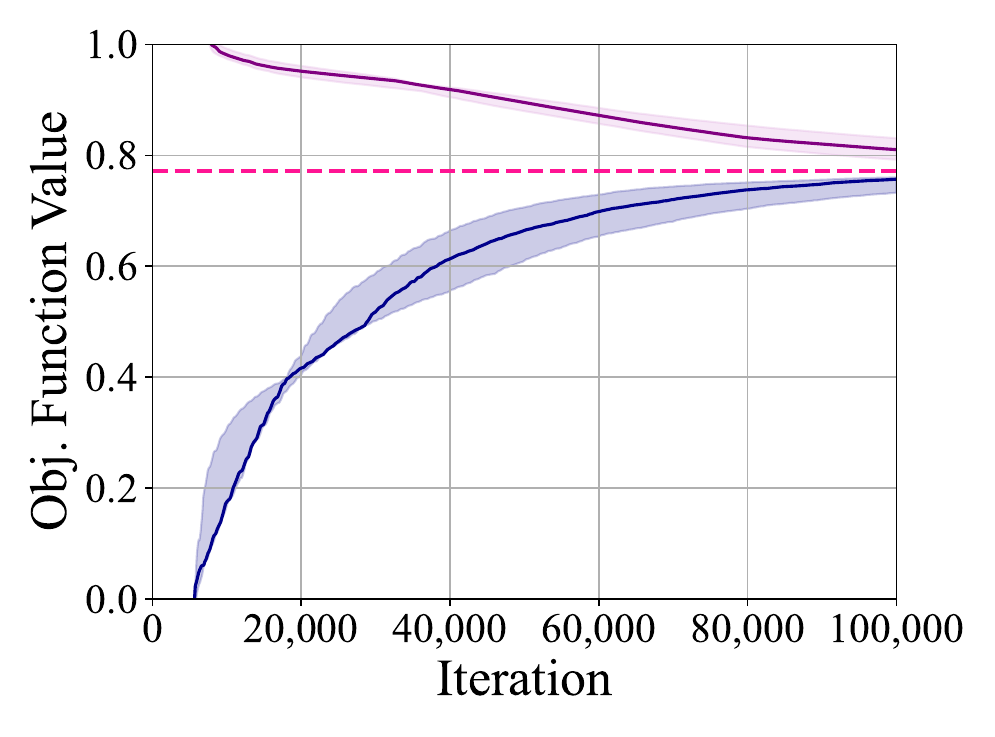}}}
\newcommand\fperror{\adjustbox{valign=m, vspace=0.0pt}{\includegraphics[width=.40\linewidth]{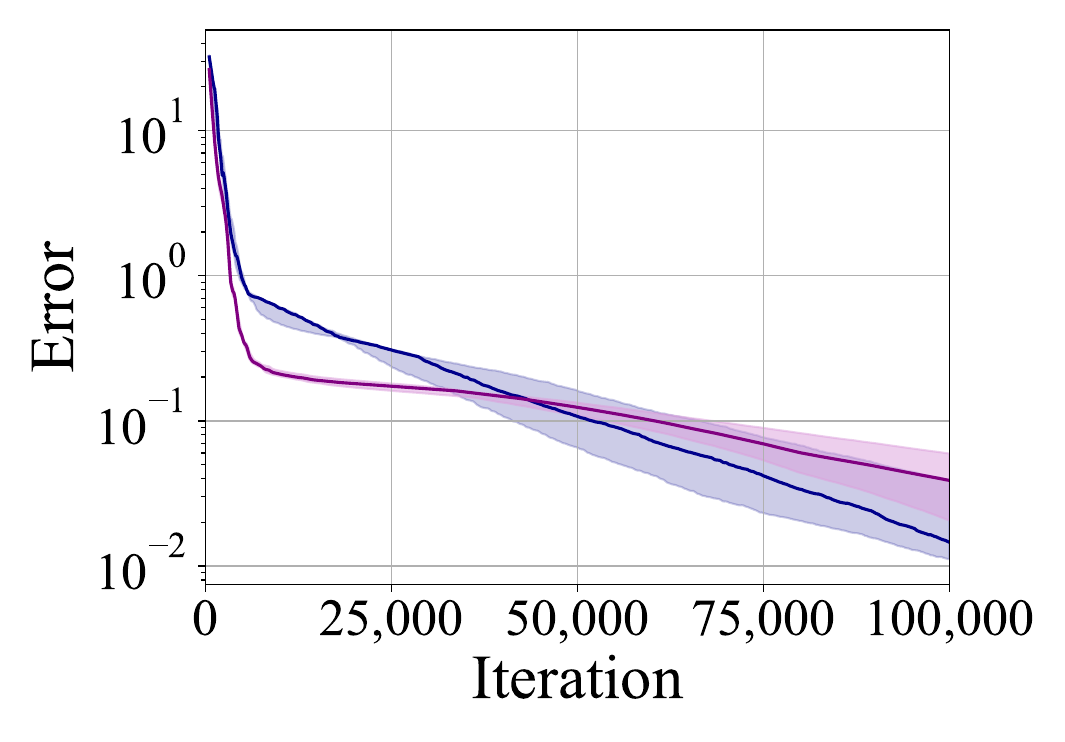}}}
\newcommand\fppenalty{\adjustbox{valign=m, vspace=0.0pt}{\includegraphics[width=.40\linewidth]{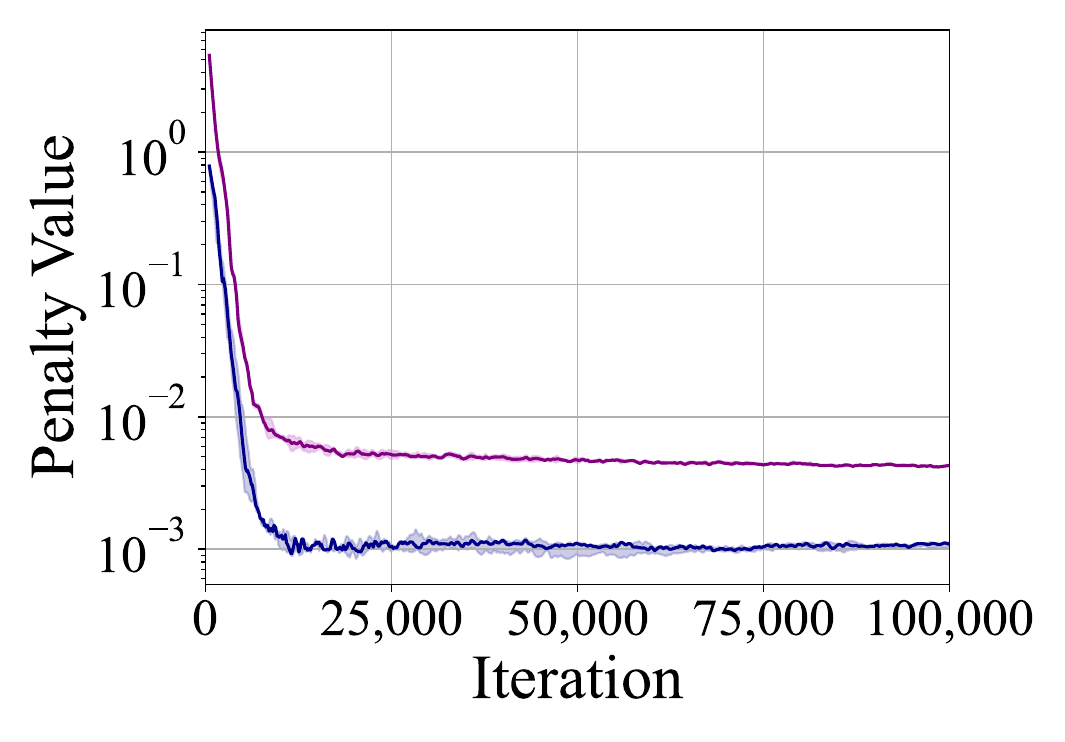}}}

\newcommand\fccasandwich{\adjustbox{valign=m, vspace=0.0pt}{\includegraphics[width=.30\linewidth]{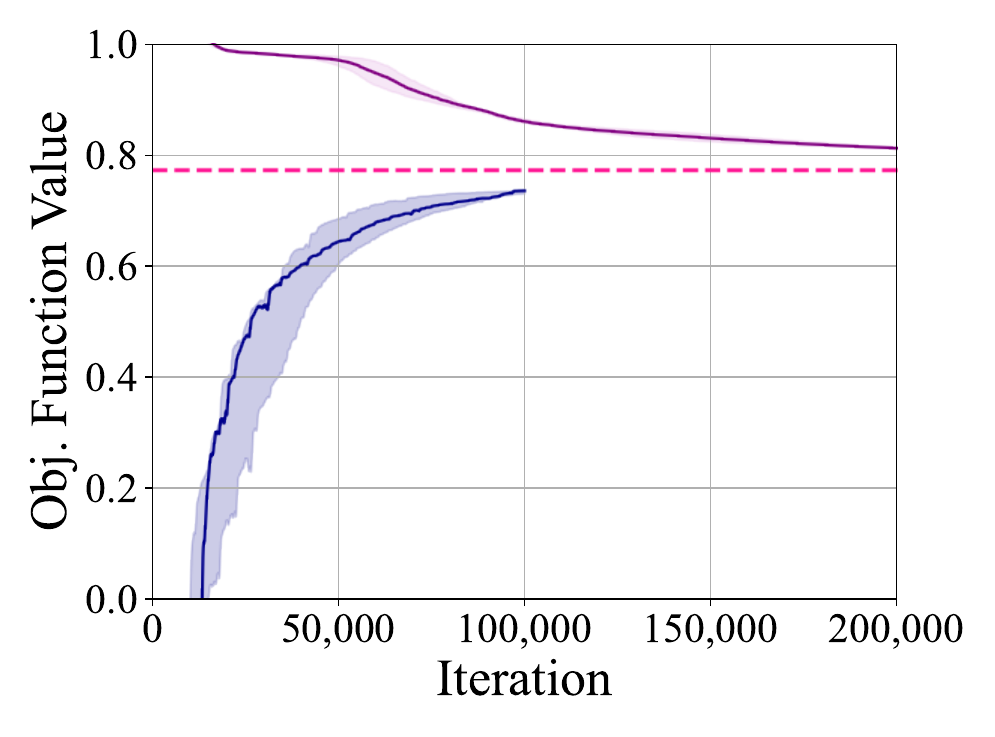}}}
\newcommand\fccaerror{\adjustbox{valign=m, vspace=0.0pt}{\includegraphics[width=.40\linewidth]{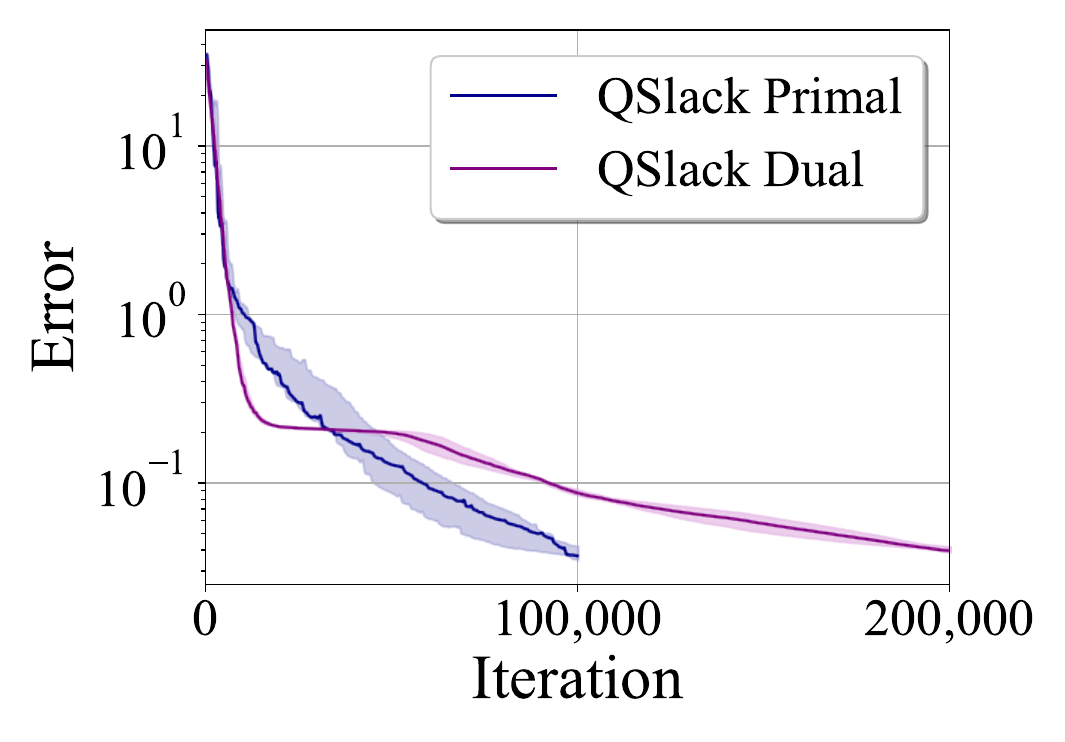}}}
\newcommand\fccapenalty{\adjustbox{valign=m, vspace=0.0pt}{\includegraphics[width=.40\linewidth]{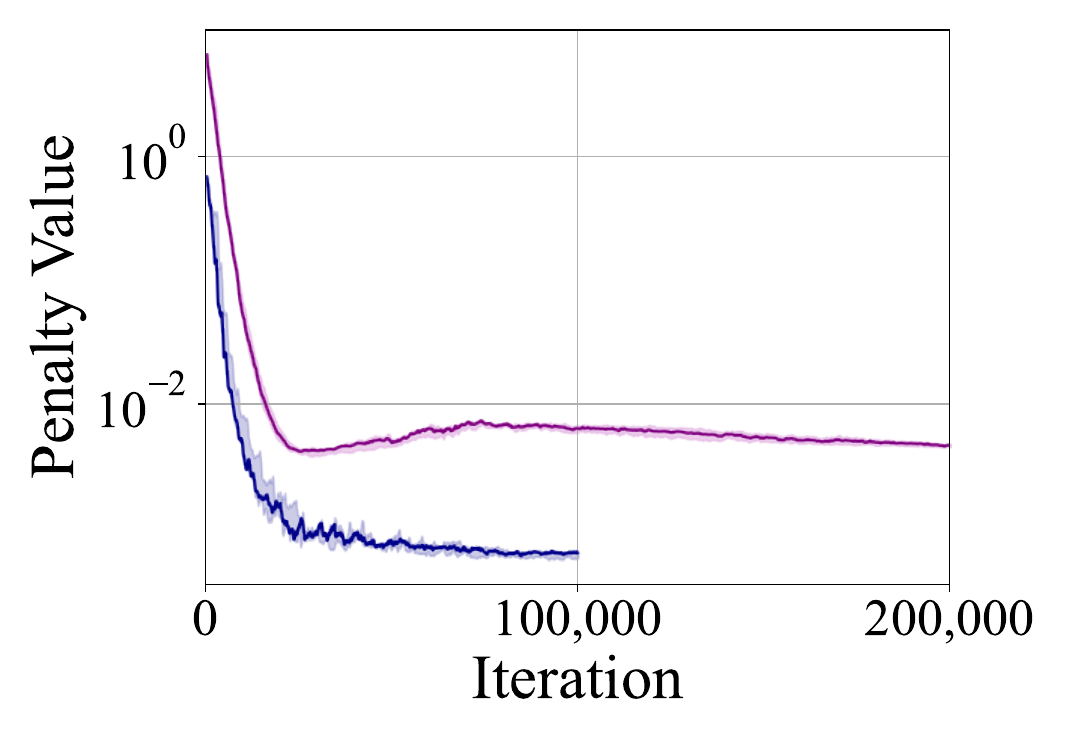}}}

\newcommand\enpsandwich{\adjustbox{valign=m, vspace=0.0pt}{\includegraphics[width=.30\linewidth]{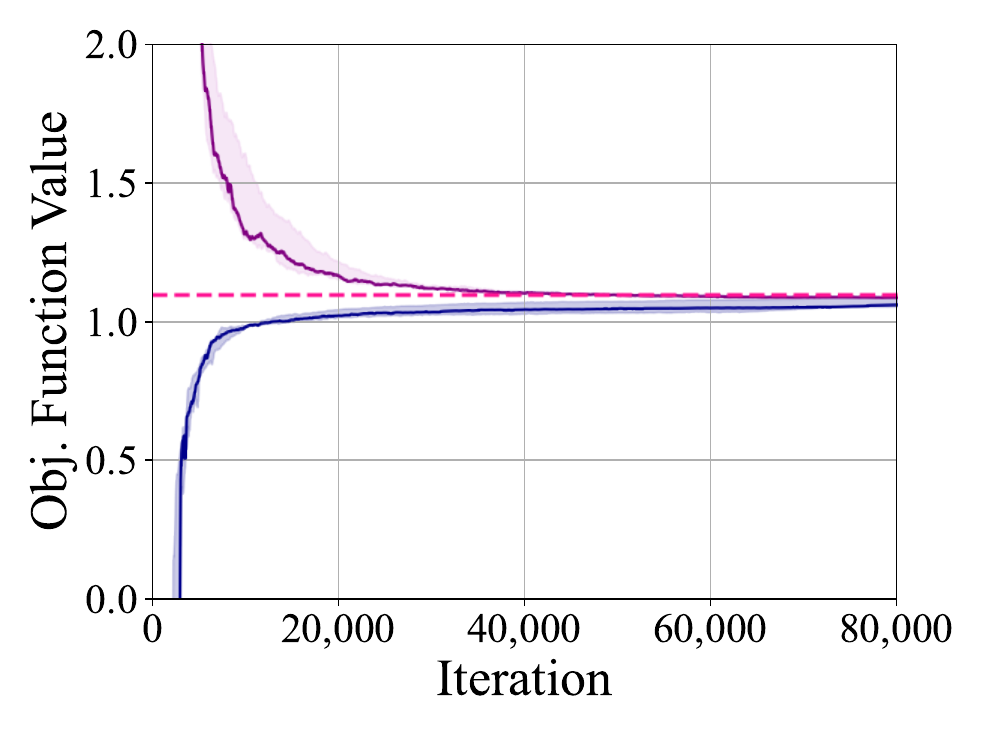}}}
\newcommand\enperror{\adjustbox{valign=m, vspace=0.0pt}{\includegraphics[width=.40\linewidth]{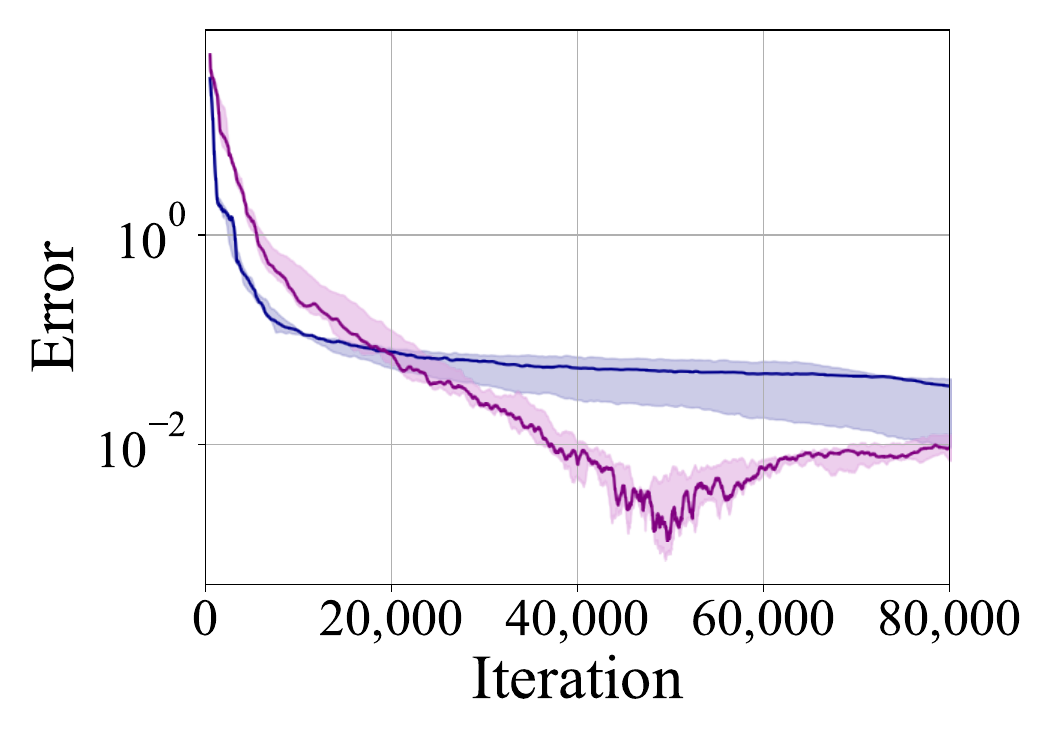}}}
\newcommand\enppenalty{\adjustbox{valign=m, vspace=0.0pt}{\includegraphics[width=.40\linewidth]{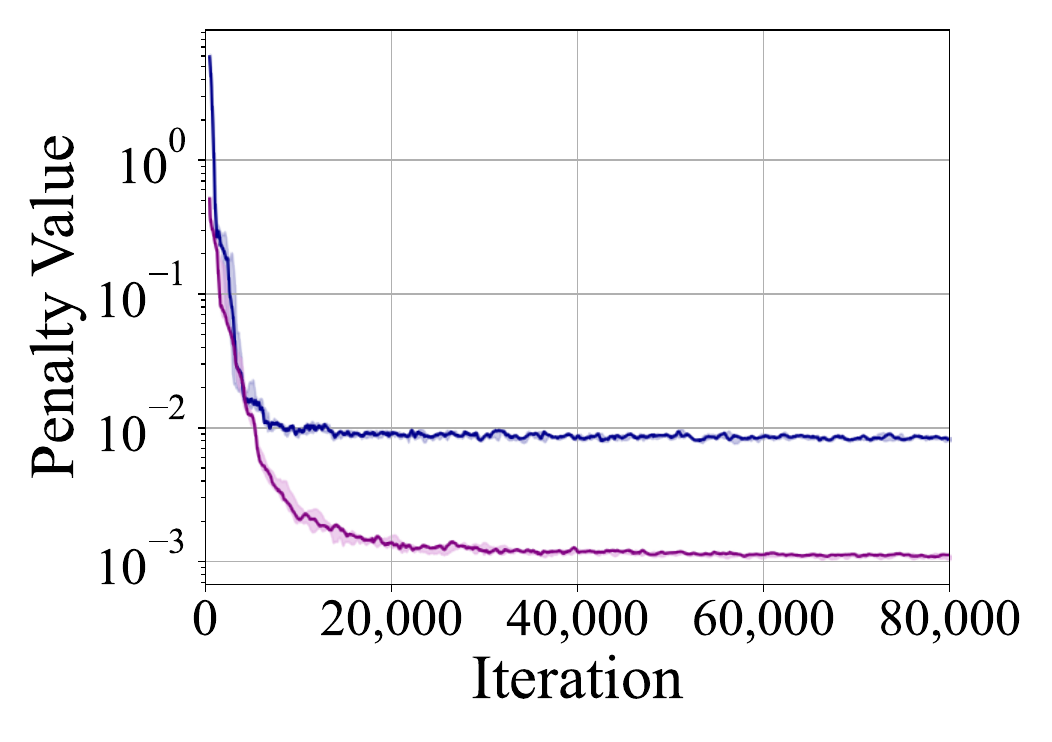}}}

\newcommand\enccasandwich{\adjustbox{valign=m, vspace=0.0pt}{\includegraphics[width=.30\linewidth]{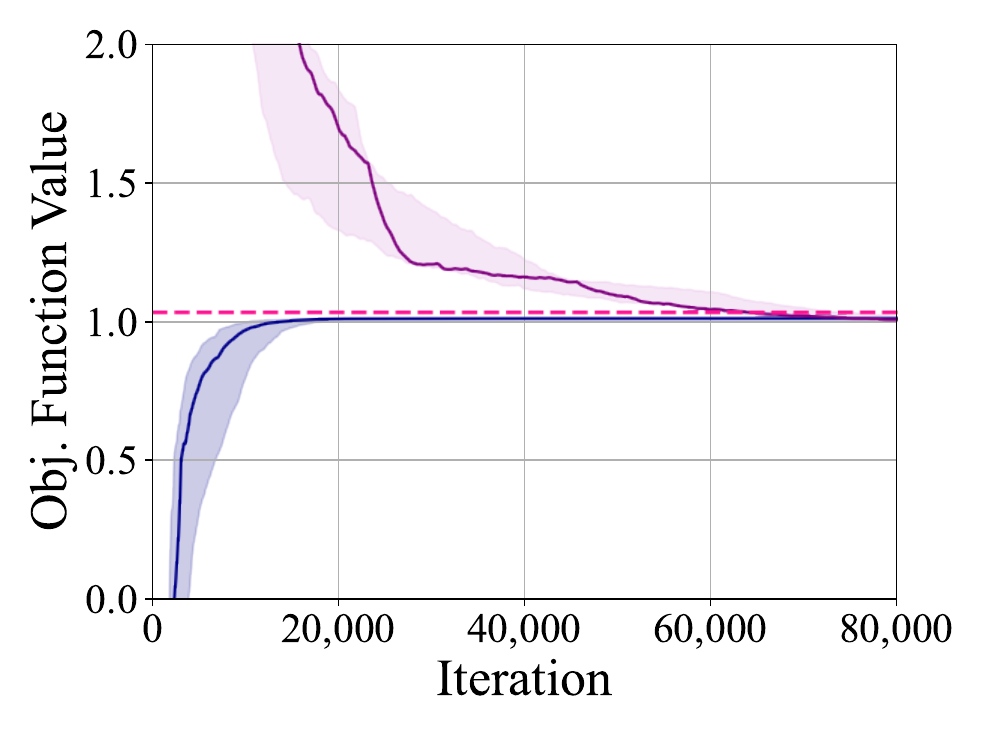}}}
\newcommand\enccaerror{\adjustbox{valign=m, vspace=0.0pt}{\includegraphics[width=.40\linewidth]{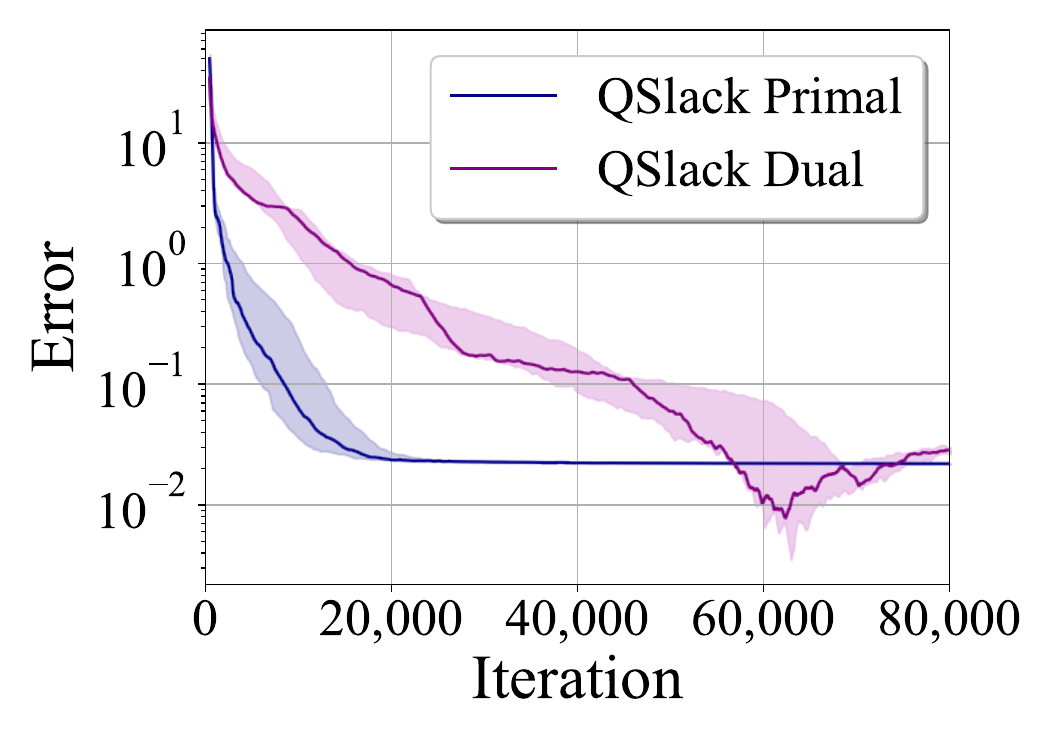}}}
\newcommand\enccapenalty{\adjustbox{valign=m, vspace=0.0pt}{\includegraphics[width=.40\linewidth]{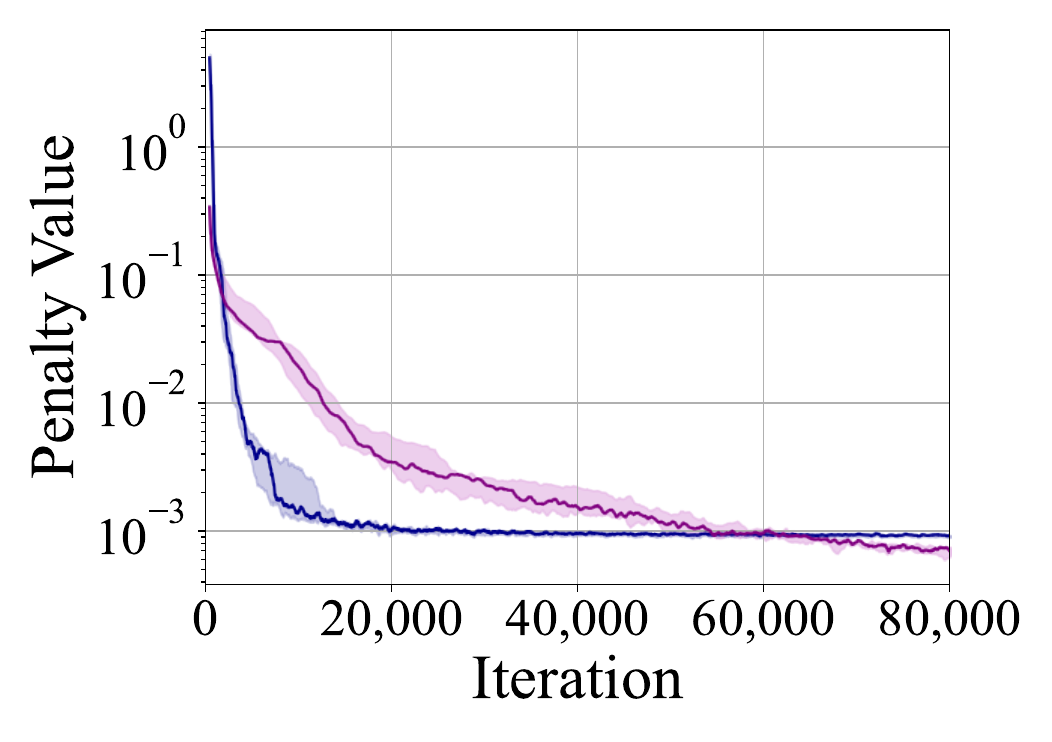}}}

\newcommand\hpsandwich{\adjustbox{valign=m, vspace=0.0pt}{\includegraphics[width=.30\linewidth]{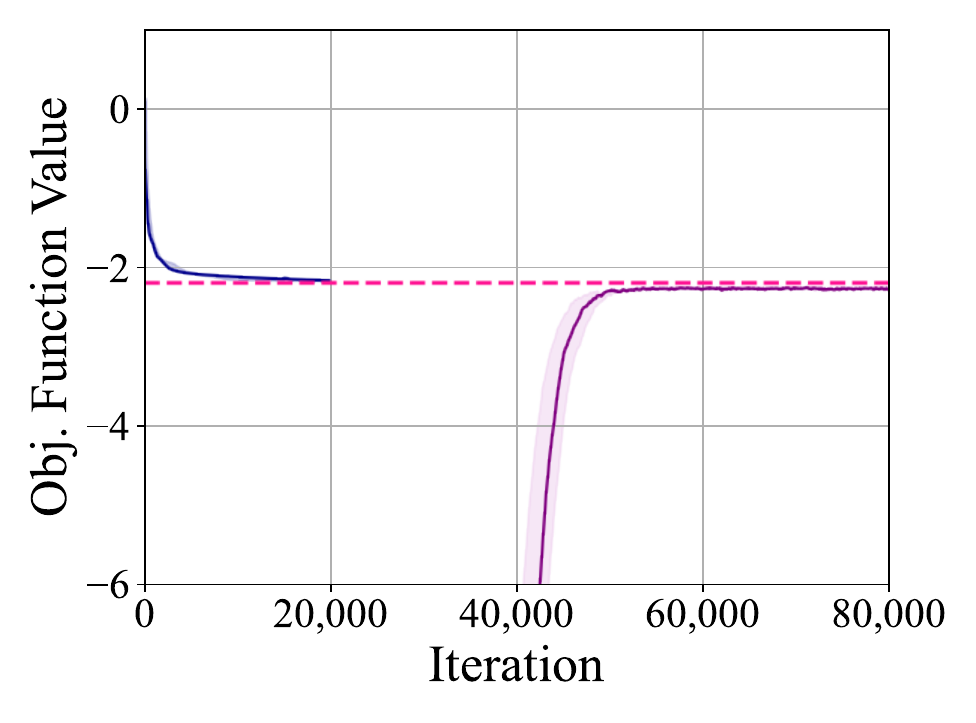}}}
\newcommand\hperror{\adjustbox{valign=m, vspace=0.0pt}{\includegraphics[width=.40\linewidth]{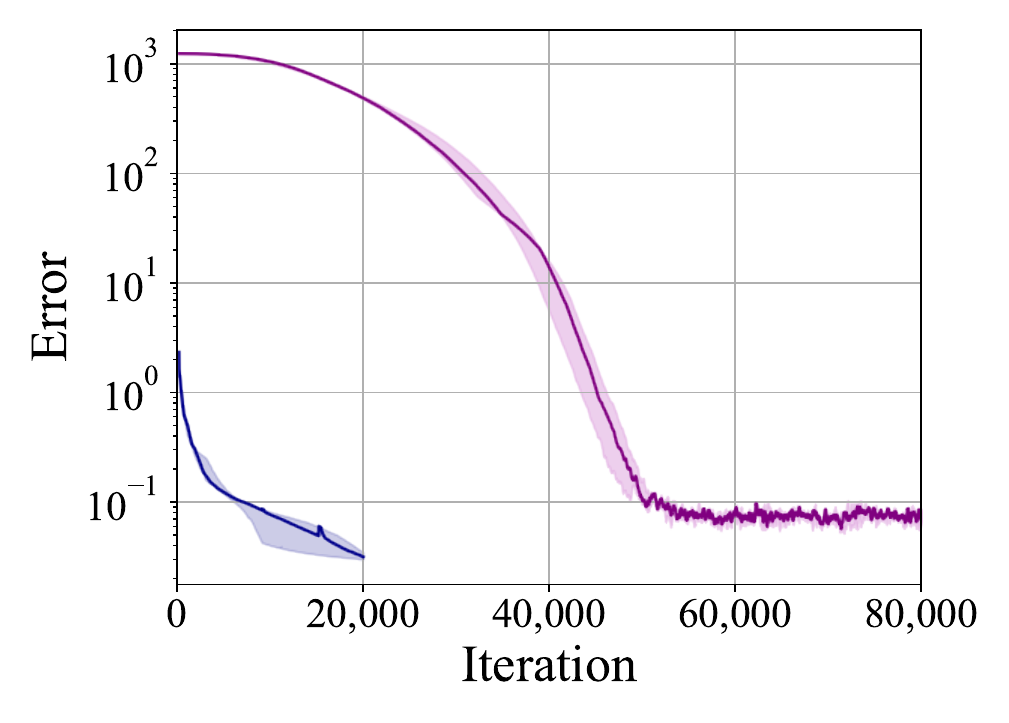}}}
\newcommand\hppenalty{\adjustbox{valign=m, vspace=0.0pt}{\includegraphics[width=.40\linewidth]{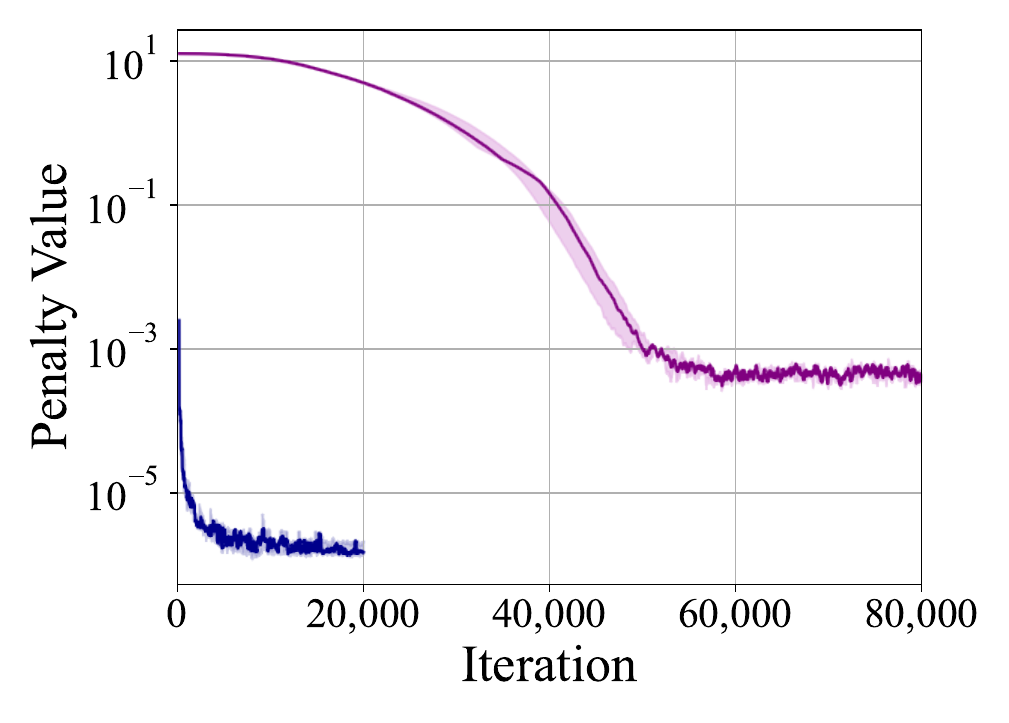}}}

\newcommand\hccasandwich{\adjustbox{valign=m, vspace=0.0pt}{\includegraphics[width=.30\linewidth]{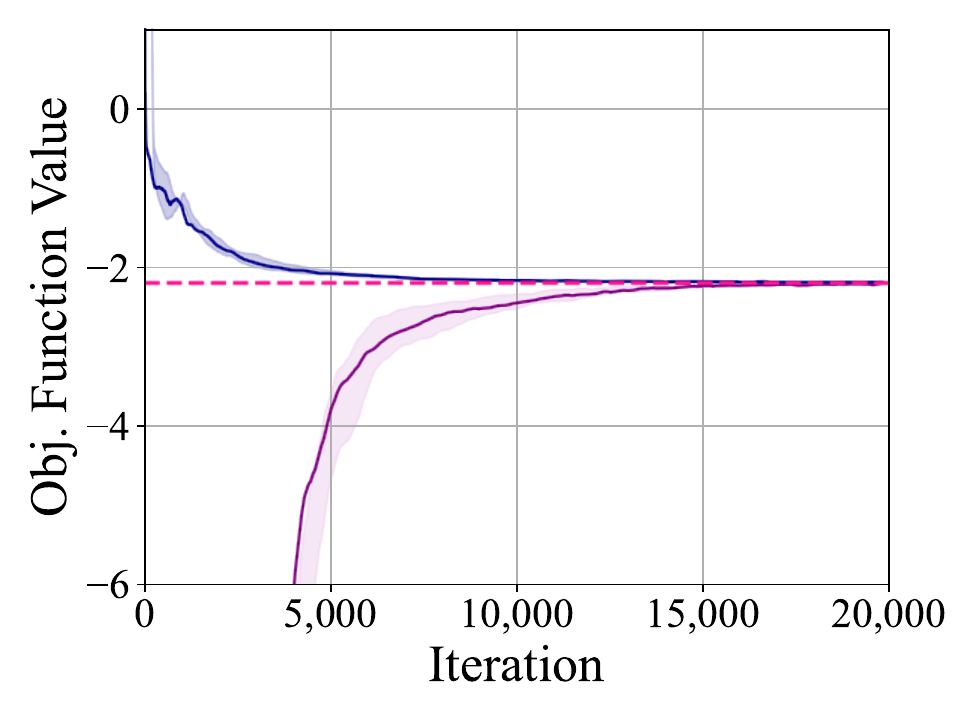}}}
\newcommand\hccaerror{\adjustbox{valign=m, vspace=0.0pt}{\includegraphics[width=.40\linewidth]{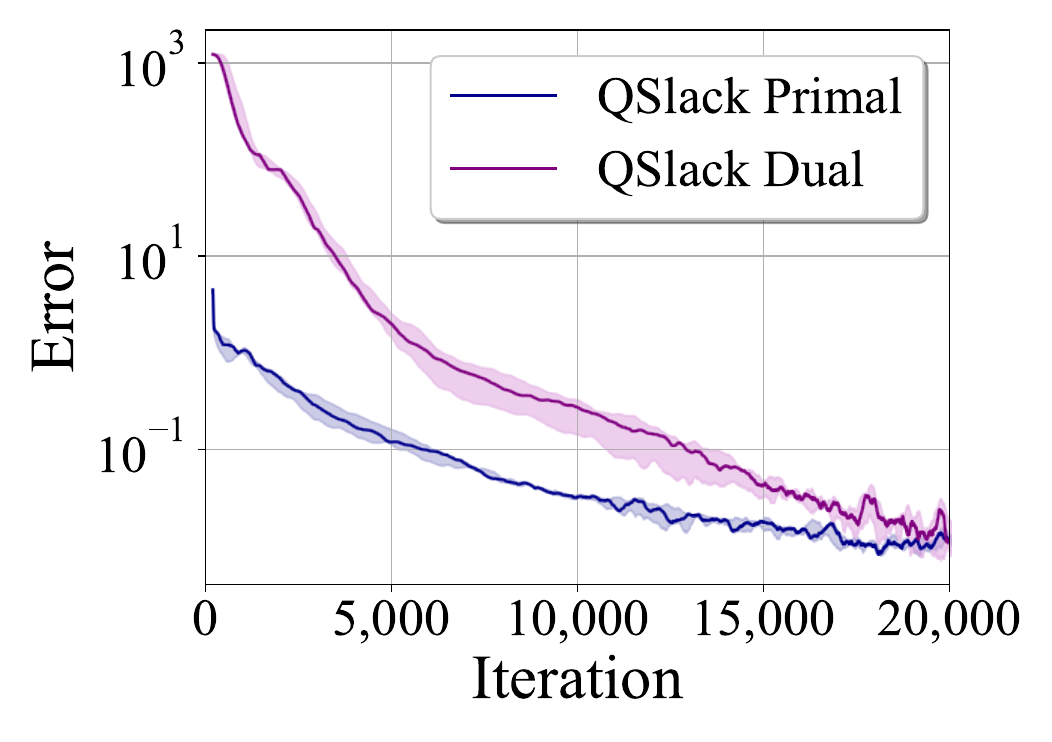}}}
\newcommand\hccapenalty{\adjustbox{valign=m, vspace=0.0pt}{\includegraphics[width=.40\linewidth]{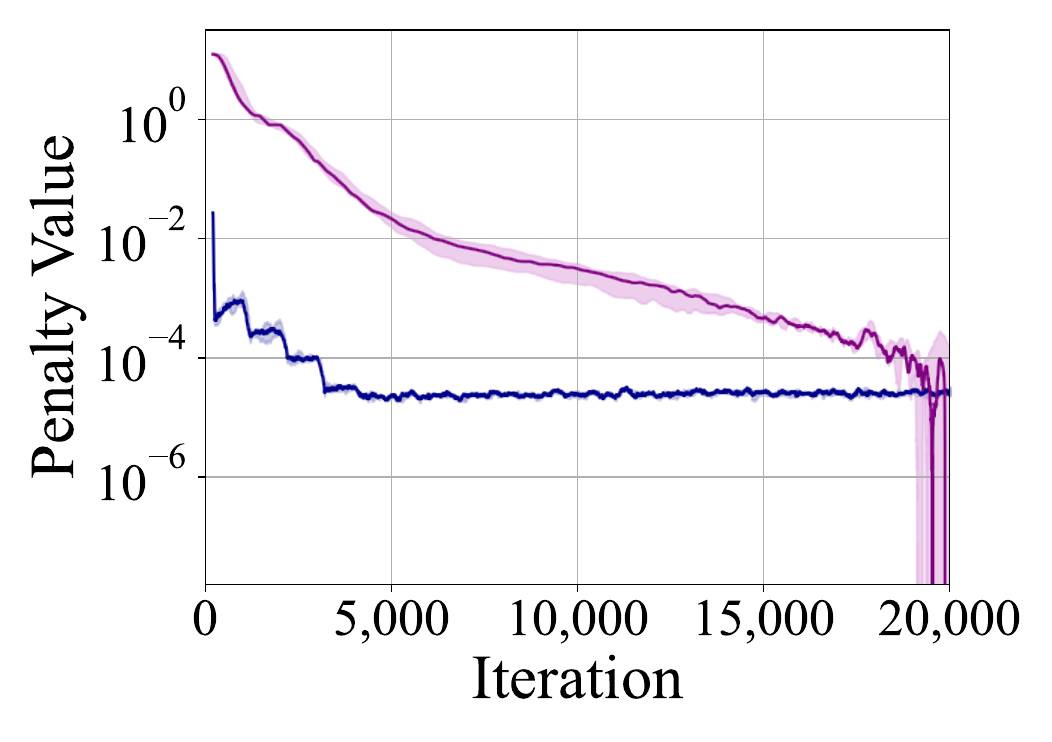}}}

\newcommand\chsandwich{\includegraphics[width=\columnwidth]{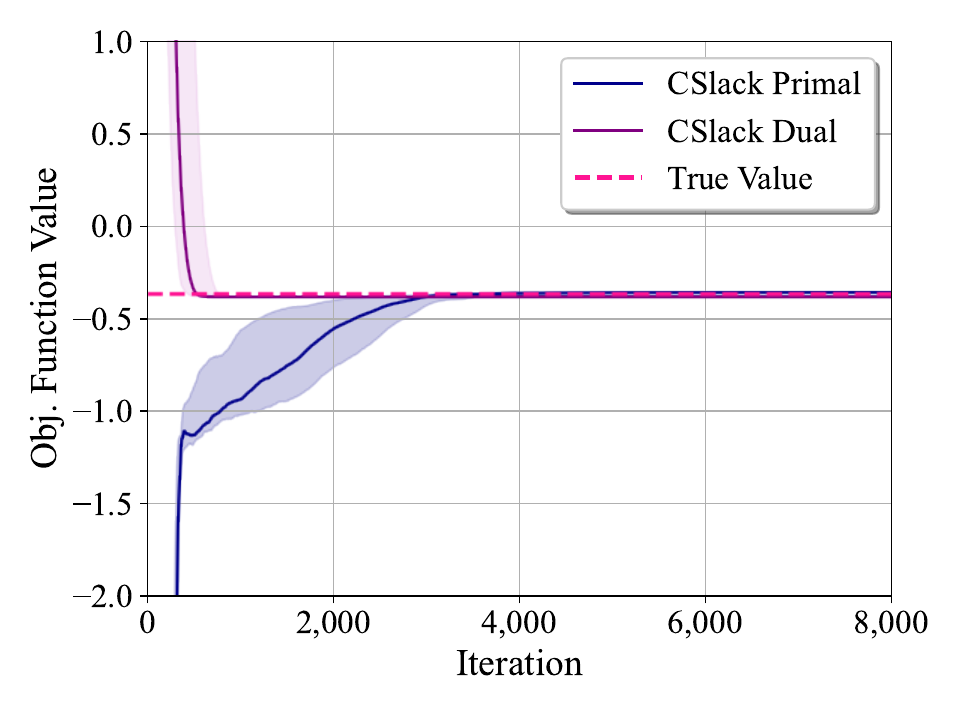}}
\newcommand\cherror{\adjustbox{valign=m, vspace=0.0pt}{\includegraphics[width=.40\linewidth]{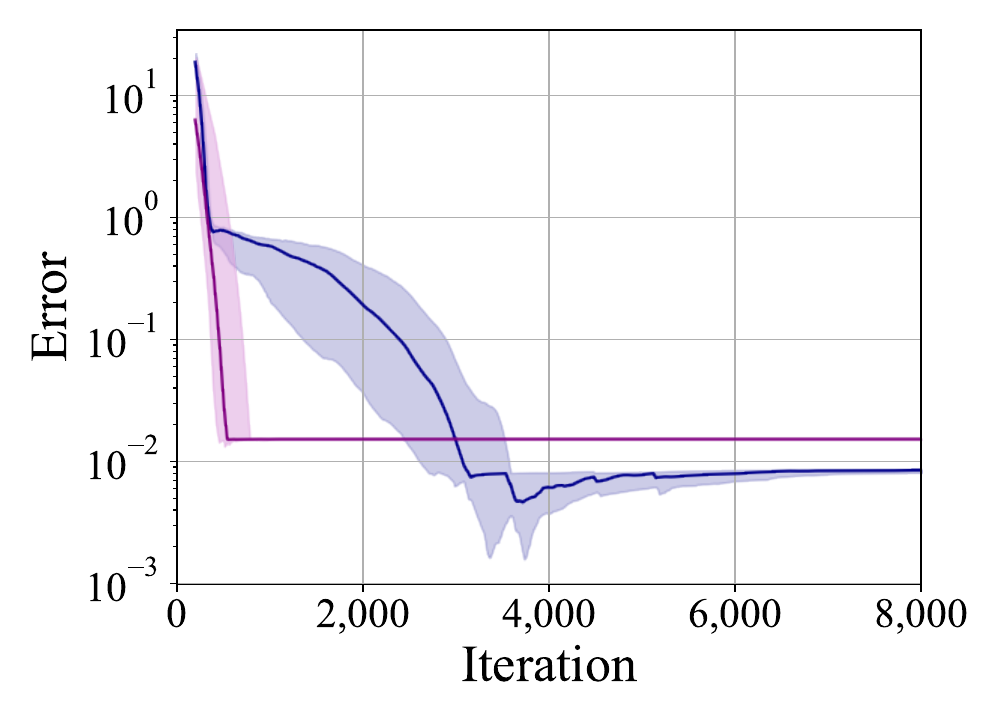}}}
\newcommand\chpenalty{\adjustbox{valign=m, vspace=0.0pt}{\includegraphics[width=.40\linewidth]{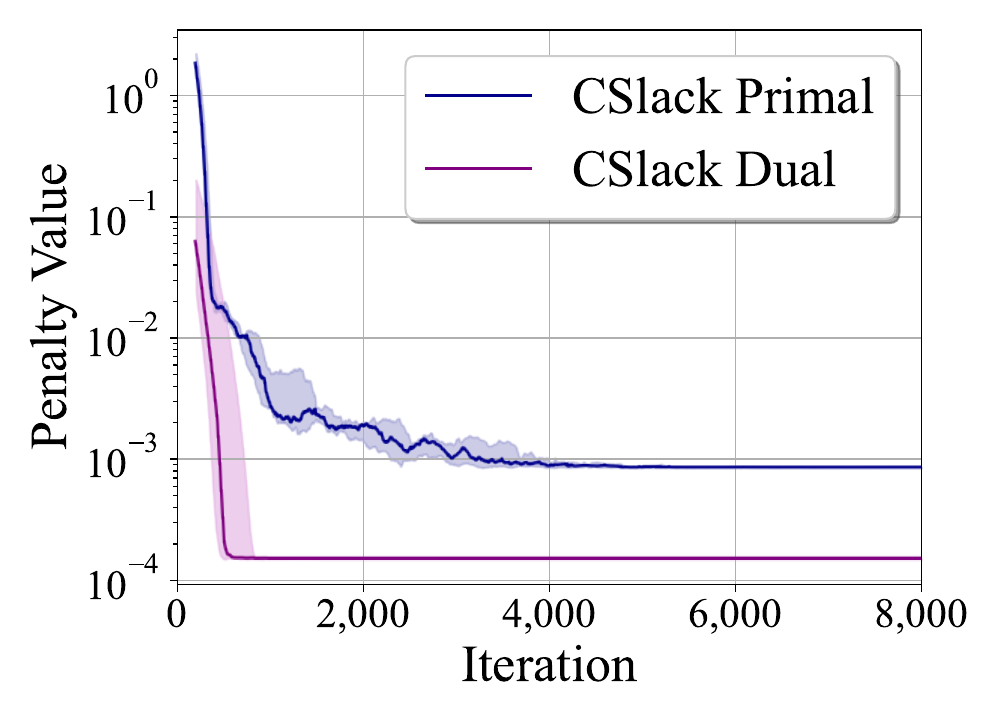}}}

\newcommand\tvdsandwich{\includegraphics[width=\columnwidth]{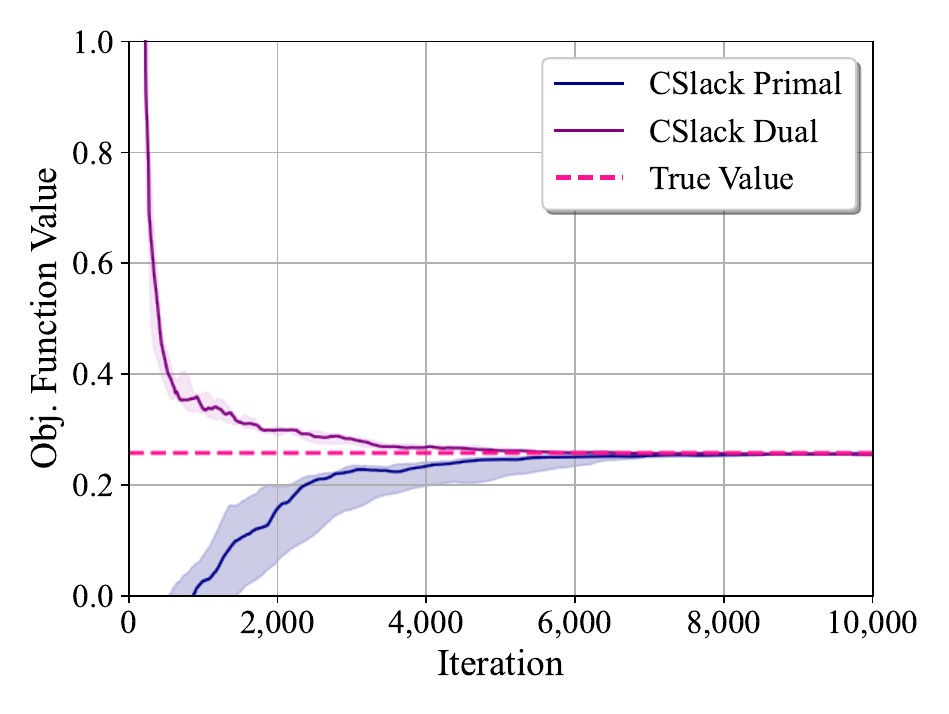}}
\newcommand\tvderror{\adjustbox{valign=m, vspace=0.0pt}{\includegraphics[width=.40\linewidth]{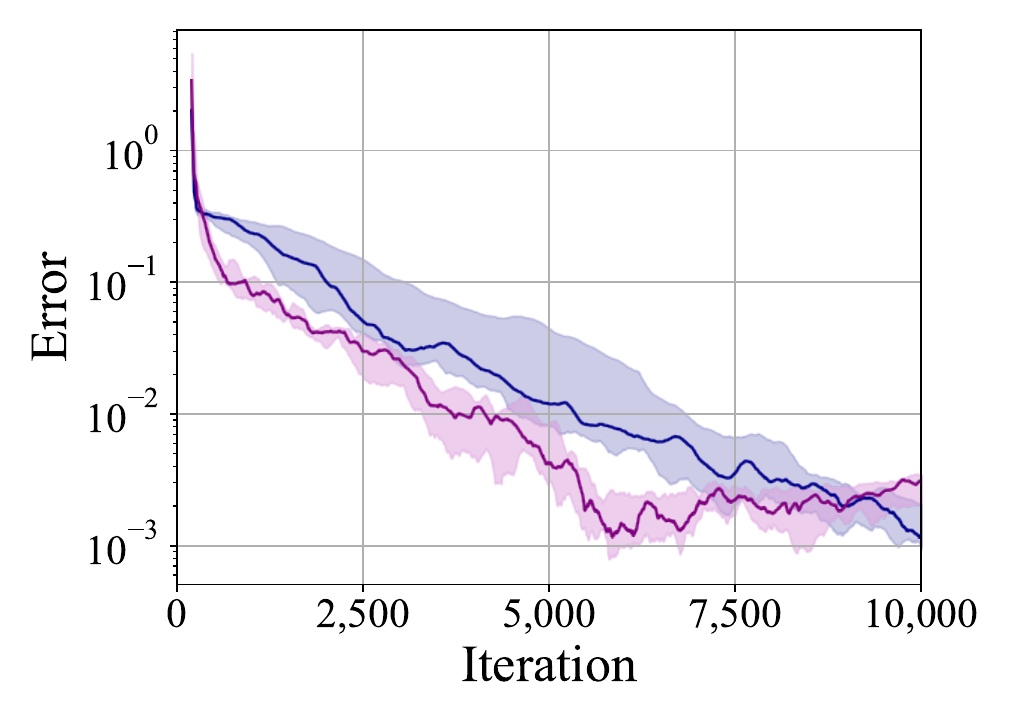}}}
\newcommand\tvdpenalty{\adjustbox{valign=m, vspace=0.0pt}{\includegraphics[width=.40\linewidth]{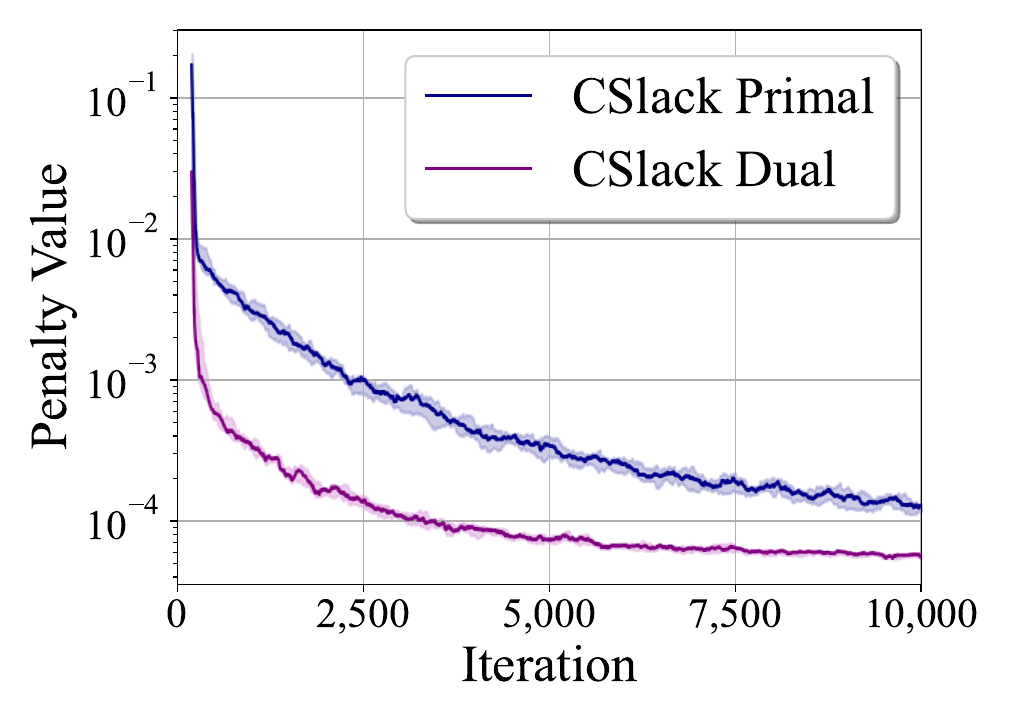}}}

\allowdisplaybreaks
\begin{document}
\preprint{ }
\title[QSlack]{QSlack: A slack-variable approach for variational quantum semi-definite programming}
\author{Jingxuan Chen}
\affiliation{Department of Computer Science, Cornell University, Ithaca, New York 14850, USA}
\affiliation{School of Electrical and Computer Engineering, Cornell University, Ithaca, New
York 14850, USA}
\author{Hanna Westerheim}
\affiliation{School of Applied and Engineering Physics, Cornell University, Ithaca, New
York 14850, USA}
\affiliation{School of Electrical and Computer Engineering, Cornell University, Ithaca, New
York 14850, USA}
\author{Zo\"e Holmes}
\affiliation{Institute of Physics, Ecole Polytechnique F\'ed\'erale de Lausanne (EPFL),
CH-1015 Lausanne, Switzerland}
\author{Ivy Luo}
\affiliation{School of Electrical and Computer Engineering, Cornell University, Ithaca, New
York 14850, USA}
\author{Theshani Nuradha}
\affiliation{School of Electrical and Computer Engineering, Cornell University, Ithaca, New
York 14850, USA}
\author{Dhrumil Patel}
\affiliation{Department of Computer Science, Cornell University, Ithaca, New York 14850, USA}
\author{Soorya Rethinasamy}
\affiliation{School of Applied and Engineering Physics, Cornell University, Ithaca, New
York 14850, USA}
\author{Kathie Wang}
\affiliation{School of Electrical and Computer Engineering, Cornell University, Ithaca, New
York 14850, USA}
\author{Mark M.~Wilde}
\affiliation{School of Electrical and Computer Engineering, Cornell University, Ithaca, New
York 14850, USA}
\begin{abstract}
Solving optimization problems is a key task for which quantum computers could
possibly provide a speedup over the best known classical algorithms. Particular
 classes of optimization problems including semi-definite programming (SDP)\ and
linear programming (LP)\ have wide applicability in many domains of computer
science, engineering, mathematics, and physics. Here we focus on semi-definite
and linear programs for which the dimensions of the variables involved are
exponentially large, so that standard classical SDP and LP solvers
are not helpful for such large-scale problems. We propose the QSlack and CSlack methods for
estimating their optimal values, respectively, which work by 1)~introducing
slack variables to transform inequality constraints to equality constraints,
2)~transforming a constrained optimization to an unconstrained one via the
penalty method, and 3)~replacing the optimizations over all possible
non-negative variables by optimizations over parameterized quantum states and
parameterized probability distributions. Under the assumption that the SDP and
LP inputs are efficiently measurable observables, it follows that all terms in
the resulting objective functions are efficiently estimable by either a
quantum computer in the SDP\ case or a quantum or probabilistic computer in the LP case. Furthermore, by making use of SDP\ and LP duality theory, we prove that these methods provide a theoretical guarantee that, if one could find global optima of the objective functions, then the resulting values sandwich the true optimal values from both above and below. Finally, we showcase the QSlack and CSlack methods on a variety of example optimization problems and discuss details of our implementation, as well as the resulting performance. We find that our implementations of both the primal and dual for these problems approach the ground truth, typically achieving errors of the order $10^{-2}$.
\end{abstract}
\date{\today}
\startpage{1}
\endpage{102}
\maketitle
\tableofcontents

\section{Introduction}

Quantum computation has arisen in the past few decades as a viable model of
computing, understood quite well by now from a theoretical perspective
\cite{dalzell2023quantum} and undergoing rapid development from the
experimental viewpoint
\cite{Burkard2020,Romaszko2020,Schaefer2020,Chatterjee2021,Pelucchi2022}. It
has garnered so much interest from a wide variety of scientists, engineers,
and mathematicians due to the promise of speedups over the best known
classical algorithms, for tasks such as factoring products of prime integers
\cite{Shor:1994:124,shor97}, unstructured search \cite{grover96,grover97},
simulation of physics
\cite{Lloyd1996UniversalSimulators,Berry2015HamiltonianParameters,Low2016OptimalProcessing}%
\ and chemistry \cite{Cao2019,McArdle2020},  learning features of solutions
to linear systems of equations \cite{Harrow2009,dervovic2018quantum}, and machine learning \cite{liu2021rigorous, gyurik2023exponential}.

Another target application of quantum computing is in solving optimization problems
\cite{Moll2018,blekos2023review}, again with the hope of achieving a speedup
over the best known classical algorithms. There have long been various
investigations of applying quantum simulated annealing to quadratic
unconstrained binary optimization problems \cite{Yarkoni_2022}; here instead,
we focus on different subclasses of optimization problems, known as
semi-definite and linear programming.

Semi-definite programming refers to a class of optimization problems involving
a linear objective function optimized over the cone of positive semi-definite
operators intersected with an affine space \cite{BV04,wolkowicz2012handbook}.
Solving semi-definite programs (SDPs) has wide applications across many
different fields, including combinatorial optimization
\cite{Goemans1997,Rendl1999,tunccel2016polyhedral}, finance \cite{Gepp2020},
job scheduling in operations research \cite{Skutella2001}, machine learning
\cite{dAspremont2007,hall2018optimization,majumdar2020recent}, physics
\cite{Mazziotti2004,Barthel2012,SimmonsDuffin2015,Berenstein2023,fawzi2023entropy}, and quantum
information science \cite{W18thesis,Siddhu2022,Skrzypczyk2023}. Furthermore,
many algorithms for solving SDPs efficiently on classical computers are now
known and in extensive use \cite{PW2000,Arora2005,Arora2012,Lee2015} (by
``efficiently,'' here we mean that the algorithms have runtime polynomial in
the dimension of the matrices involved).

The wide applications of SDPs and the tantalizing possibility of achieving
quantum speedup for them have motivated a number of researchers to investigate
quantum algorithms for solving SDPs
\cite{BS17,vanapeldoorn_et_al:LIPIcs:2019:10675,brando_et_al:LIPIcs:2019:10603,vanApeldoorn2020quantumsdpsolvers,KP20,patel2021variational,Augustino2023quantuminterior,BHVK22,watts2023quantum,Patti2023quantumgoemans}%
. Researchers have previously considered several approaches for this problem.
A number of proposals have guaranteed runtimes and are intended to be executed
on fault-tolerant quantum computers
\cite{BS17,vanapeldoorn_et_al:LIPIcs:2019:10675,brando_et_al:LIPIcs:2019:10603,vanApeldoorn2020quantumsdpsolvers,KP20,Augustino2023quantuminterior,watts2023quantum}; as such, it is not expected that the promised speedups of these algorithms will be realized in the near term. Another approach consists of hybrid quantum-classical computations
\cite{BHVK22}, in which a given SDP involves terms expressed as expectations
of observables, so that these terms can be efficiently estimated on a quantum
computer and fed into a classical SDP solver that produces an approximate
solution to the SDP. The main regime of interest for all of these algorithms
is when the matrices involved in the optimization are large (i.e., of
dimension $2^{n}\times2^{n}$, where $n$ is a parameter) and such that
expressions like $\operatorname{Tr}[AX]$ for Hermitian~$A$ and positive
semi-definite $X$ can be efficiently estimated on quantum computers but not so
using the best known classical algorithms.

In this paper, we focus on the variational approach to solving SDPs on quantum
computers, which we call variational quantum semi-definite programming. While
this approach has been explored in some recent papers
\cite{patel2021variational,Patti2023quantumgoemans}, our main theoretical
contribution here is a method that theoretically guarantees lower and upper
bounds on the optimal value of a given SDP. We do so by introducing slack
variables into both the primal and dual SDPs, which are physically realized as
either parameterized states or observables. Since slack variables transform
inequality constraints to equality constraints, we can then move all of the
resulting equality constraints into penalty terms in the objective function of
the SDP, leaving an unconstrained optimization problem. We express the
aforementioned penalty terms in terms of the Hilbert--Schmidt distance, and as
such, each of the penalty terms can be evaluated on a quantum
computer as expectations of observables, by means of the destructive swap test
\cite{GC13}, which is also known more recently as Bell sampling
\cite{montanaro2017learning,hangleiter2023bell}, or by means of a mixed-state
Loschmidt echo test that we introduce here (see Appendix~\ref{app:Loschmidt-echo-alg}).


\begin{figure*}[t]
  \centering
  \includegraphics[width=\textwidth]{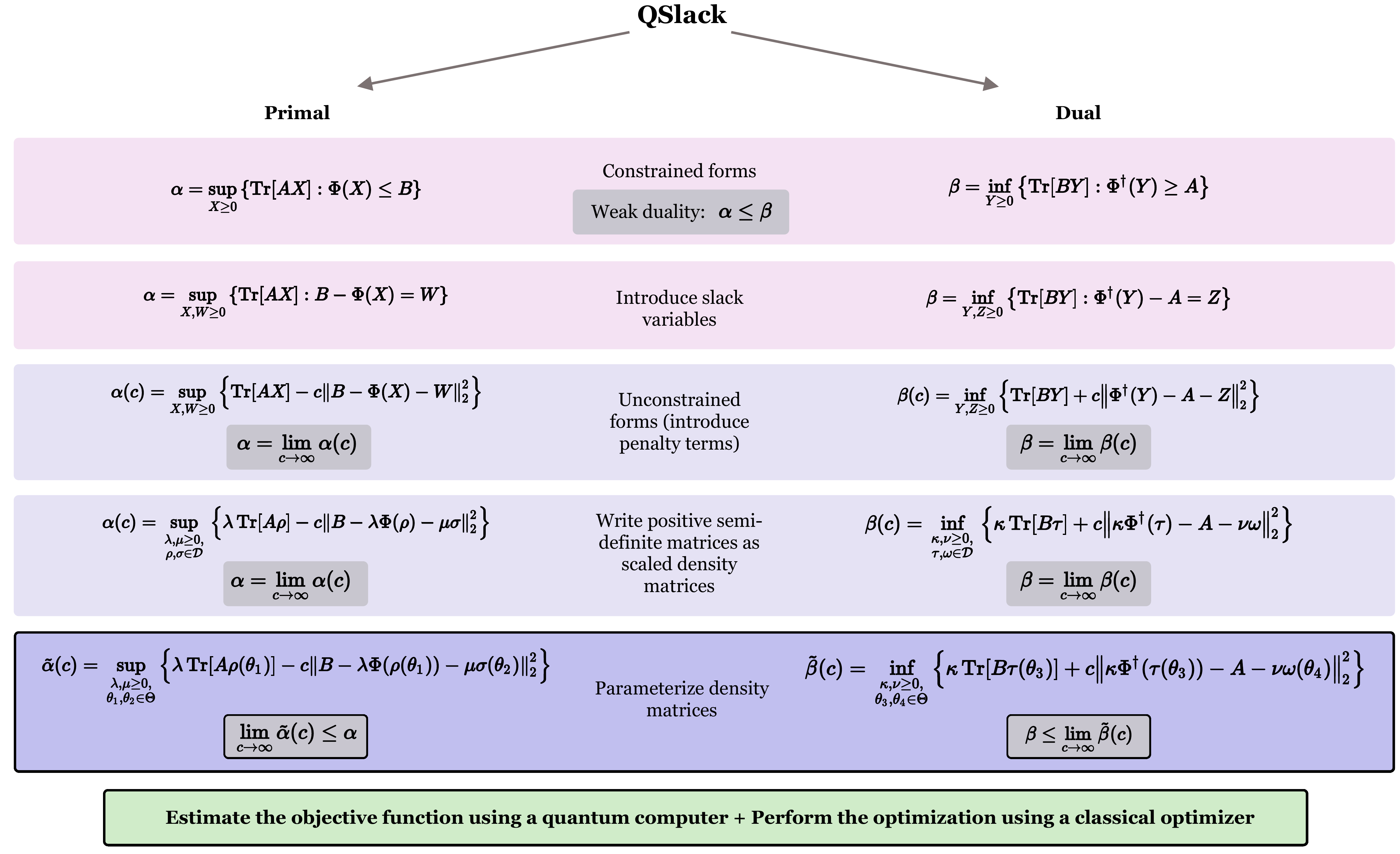}
  \caption{Overview of the QSlack method introduced in our paper. From top to bottom, we modify the expressions of the primal and dual semi-definite programs to a form that can be evaluated on a quantum computer. Section~\ref{sec:SDP-background} provides details for the first three rows, and Section~\ref{sec:qslack-alg} provides details for the last two rows. The same ideas lead to the CSlack method for solving linear programs, as discussed in Sections~\ref{sec:LP-background} and \ref{sec:cslack-alg}.}
  \label{fig:qslack-method-overview}
\end{figure*}

Our approach is thus a hybrid
quantum--classical optimization, only using a quantum computer to evaluate
expectations of observables or to perform the destructive swap test, and
leaving all other computations to a classical computer. As such, our algorithm
falls within the large and increasingly studied class of variational quantum
algorithms \cite{Cerezo2021,bharti2021noisy}. 
Using this approach for both the primal and dual SDPs, we can
guarantee lower and upper bounds on the optimal value of the optimization
problem, thus solving a key question left open in \cite{patel2021variational}.
We describe our method in far more detail in Section~\ref{sec:qslack-alg},
after recalling some background on semi-definite programming in
Section~\ref{sec:SDP-background}. Since a key feature of our algorithm is the
introduction of slack variables physically realized as parameterized states or
observables (i.e., \textquotedblleft quantum slack variables\textquotedblright%
), we call it the QSlack method. Figure~\ref{fig:qslack-method-overview} summarizes the main ideas behind QSlack.

We also delineate a classical variant of our method for approximately solving
linear programs. Indeed, a linear program (LP) involves the optimization of a
linear objective function over the positive orthant intersected with an affine
space. As such, our approach here again involves introducing slack variables
into both the primal and dual LPs, which are realized as either parameterized
probability distributions or parameterized observable vectors (these are the
classical reductions of parameterized states and observables, respectively).
Our approach can thus be called a variational linear programming algorithm,
and we refer to it as the CSlack method.

Similar to QSlack, the CSlack method
transforms inequality constraints to equality constraints by means of
parameterized slack variables, and all equality constraints are then moved
into the objective function as penalty terms, so that the resulting problem is
an unconstrained optimization. This method guarantees lower and upper bounds
on the optimal value of the~LP. Let us note that a similar method for solving
linear programs on quantum computers was recently introduced
\cite{le2023variational}, with a key difference between our approach and
theirs being that they do not make use of slack variables.

Another key contribution of our paper is to show how the QSlack method can be
applied to a wide variety of examples, focusing on problems of interest in
physics and quantum information science. In particular, we apply the QSlack
method to estimating the trace distance, fidelity, entanglement negativity,
and constrained minimum energy of Hamiltonians.  We have conducted extensive
numerical simulations of the QSlack method in order to understand how it would
perform in practice on quantum computers. For the examples considered, we find
that the QSlack method works well, leading to lower and upper bounds on the
optimal value of a given optimization problem. Although we made use of parameterized quantum circuits in our simulations (via the purification ansatz and the convex-combination ansatz), let us stress here that the general QSlack method is compatible with other parameterizations of quantum states, such as quantum Boltzmann machines \cite{Amin2018,verdon2019quantum}, which can be explored in future work. We also apply the CSlack method
to example LPs of interest and showcase its performance.

Our numerical simulations provide practical evidence of the QSlack and CSlack methods' ability to bound optimal values from above and below. Due to the number of parameters and terms in the modified objective functions, training CSlack and QSlack objectives is observed to be relatively slow and noisy, as compared to a standard variational quantum eigensolver. However, with our choices of hyperparameters, we observe consistent convergence to the true value, typically achieving an error on the order of $10^{-2}$, giving preliminary evidence of the potential versatility and practicality of our method.

The remainder of our paper proceeds as follows. We first provide details of the QSlack method in Section~\ref{sec:QSlack-all}, beginning by recalling basic aspects
of semi-definite programming (Section~\ref{sec:SDP-background}) and following with
an overview of the
QSlack  algorithm (Section~\ref{sec:qslack-alg}). Then we formulate several example problems of interest in quantum information and physics in the QSlack framework (Section~\ref{sec:qslack-examples}), after which we provide details of how we simulated the QSlack algorithms and then we discuss the results
of our numerical experiments (Section~\ref{sec:qslack-sims}).
We then mirror this structure for CSlack, giving background on linear programming (Section~\ref{sec:LP-background}), the CSlack algorithm (Section~\ref{sec:cslack-alg}), examples (Section~\ref{sec:CSlack-examples}), and details of simulations (Section~\ref{sec:cslack-sims}).
Finally, we conclude in
Section~\ref{sec:conclusion} with a summary and directions for future
research. All of the appendices provide even greater details of the methods
underlying our approach.

\section{QSlack background, algorithm, examples, and simulations}

\label{sec:QSlack-all}

\subsection{Background on semi-definite programming}

\label{sec:SDP-background}

Let us begin by recalling the basics of semi-definite programming, following
the formulation of \cite[page 223]{Watrous2009} (see also \cite[Section~2.4]%
{khatri2020principles}). Fix $d_{1},d_{2}\in\mathbb{N}$. Let $A$ be a
$d_{1}\times d_{1}$ Hermitian matrix, $B$ a $d_{2}\times d_{2}$ Hermitian
matrix, and $\Phi$ a Hermiticity-preserving superoperator that takes
$d_{1}\times d_{1}$ Hermitian matrices to $d_{2}\times d_{2}$ Hermitian
matrices. A semi-definite program is specified by the triple $(A,B,\Phi)$ and
is defined as the following optimization problem:%
\begin{equation}
\alpha\coloneqq\sup_{X\geq0}\left\{  \operatorname{Tr}[AX]:\Phi(X)\leq
B\right\}  , \label{eq:primal-SDP}%
\end{equation}
where the supremum optimization is over every $d_{1}\times d_{1}$ positive
semi-definite matrix $X$ (i.e., $X\geq0$ is a shorthand notation for $X$ being
positive semi-definite) and the inequality constraint $\Phi(X)\leq B$ is
equivalent to $B-\Phi(X)$ being a positive semi-definite matrix. We call the
optimization in \eqref{eq:primal-SDP} the primal SDP. A matrix $X$ is feasible
for the optimization in \eqref{eq:primal-SDP} if both constraints $X\geq0$ and
$\Phi(X)\leq B$ are satisfied; such an $X$ is also said to be primal feasible.

The dual optimization problem is as follows:%
\begin{equation}
\beta\coloneqq\inf_{Y\geq0}\left\{  \operatorname{Tr}[BY]:\Phi^{\dag}(Y)\geq
A\right\}  , \label{eq:dual-SDP}%
\end{equation}
where the infimum optimization is over every $d_{2}\times d_{2}$ positive
semi-definite matrix $Y$ and again the inequality constraint $\Phi^{\dag
}(Y)\geq A$ is equivalent to $\Phi^{\dag}(Y)-A$ being a positive semi-definite
matrix. Additionally, $\Phi^{\dag}$ is a Hermiticity preserving superoperator
that is the adjoint of~$\Phi$; i.e., it satisfies%
\begin{equation}
\left\langle Y,\Phi(X)\right\rangle =\left\langle \Phi^{\dag}%
(Y),X\right\rangle \label{eq:adjoint-superop-def}%
\end{equation}
for every $d_{1}\times d_{1}$ matrix $X$ and every $d_{2}\times d_{2}$ matrix
$Y$, where the Hilbert--Schmidt inner product is defined for square matrices
$C$ and $D$ as%
\begin{equation}
\left\langle C,D\right\rangle \coloneqq\operatorname{Tr}[C^{\dag}D].
\end{equation}
A matrix $Y$ is dual feasible if both constraints $Y\geq0$ and $\Phi^{\dag
}(Y)\geq A$ are satisfied.

Semi-definite programming is equipped with a duality theory, which is helpful
for both numerical and analytical purposes. Weak duality corresponds to the
inequality%
\begin{equation}
\alpha\leq\beta. \label{eq:weak-duality-SDP}%
\end{equation}
As a consequence of weak duality, an arbitrary primal feasible $X$ leads to a
lower bound on the optimal value because $\operatorname{Tr}[AX]\leq\alpha$ for
all such $X$, and an arbitrary dual feasible~$Y$ leads to an upper bound on
the optimal value because $\beta\leq\operatorname{Tr}[BY]$ for such a $Y$. By
producing a primal feasible $X$ and a dual feasible $Y$, we can thus sandwich
the optimal values of the primal and dual SDPs with guaranteed lower and upper
bounds. Strong duality corresponds to the equality%
\begin{equation}
\alpha=\beta,
\end{equation}
and it holds whenever Slater's condition is satisfied. Slater's condition
often holds in practice and corresponds to there existing a primal feasible
$X$ and a strictly dual feasible $Y$ (i.e., the strict inequalities $Y>0$ and
$\Phi^{\dag}(Y)>A$ hold).
Alternatively, Slater's condition holds if
there exists a strictly primal feasible $X$ and a dual feasible $Y$.

One can introduce slack variables that transform the inequality constraints in
\eqref{eq:primal-SDP} and \eqref{eq:dual-SDP} to equality constraints. To see
how this works, recall that the inequality $\Phi(X)\leq B$ is a shorthand for
$B-\Phi(X)$ being a positive semi-definite matrix. As such, this latter
condition is equivalent to the existence of a $d_{2}\times d_{2}$ positive
semi-definite matrix $W$ such that $B-\Phi(X)=W$. We can then rewrite the
optimization in \eqref{eq:primal-SDP} as follows:%
\begin{equation}
\alpha=\sup_{X,W\geq0}\left\{  \operatorname{Tr}[AX]:B-\Phi(X)=W\right\}  .
\label{eq:primal-SDP-slack}%
\end{equation}
By similar reasoning, we can rewrite the dual SDP\ in \eqref{eq:dual-SDP} as
follows:%
\begin{equation}
\beta=\inf_{Y,Z\geq0}\left\{  \operatorname{Tr}[BY]:\Phi^{\dag}%
(Y)-A=Z\right\}  , \label{eq:dual-SDP-slack}%
\end{equation}
where $Z$ is a $d_{1}\times d_{1}$ matrix.

The final observation that we recall in this review is that it is possible to
transform the constrained optimizations in~\eqref{eq:primal-SDP-slack} and
\eqref{eq:dual-SDP-slack} to unconstrained optimizations by introducing
penalty terms in the objective functions. Let us define the Hilbert--Schmidt
norm of an operator~$C$ as
\begin{equation}
\left\Vert C\right\Vert _{2}\coloneqq\sqrt{\left\langle C,C\right\rangle },
\end{equation}
which has the key property of being faithful: $C=0$ if and only if $\left\Vert
C\right\Vert _{2}=0$. As such, we can modify the optimizations as follows:%
\begin{align}
\alpha(c)  &  \coloneqq\sup_{X,W\geq0}\left\{  \operatorname{Tr}%
[AX]-c\left\Vert B-\Phi(X)-W\right\Vert _{2}^{2}\right\}
,\label{eq:primal-SDP-unconstrained}\\
\beta(c)  &  \coloneqq\inf_{Y,Z\geq0}\left\{  \operatorname{Tr}%
[BY]+c\left\Vert \Phi^{\dag}(Y)-A-Z\right\Vert _{2}^{2}\right\}  ,
\label{eq:dual-SDP-unconstrained}%
\end{align}
where $c>0$ is a penalty constant. The following equalities hold from standard
reasoning regarding the penalty method (specifically, see
\cite[Proposition~5.2.1]{Bertsekas2016} with $\lambda_{k}=0$ for all $k$):%
\begin{equation}
\alpha=\lim_{c\rightarrow\infty}\alpha(c),\qquad\beta=\lim_{c\rightarrow
\infty}\beta(c), \label{eq:convergence-penalties}%
\end{equation}
and they give us a way to approximate the optimal values~$\alpha$ and $\beta$,
respectively, by means of a sequence of unconstrained optimizations. The
reductions from \eqref{eq:primal-SDP} to \eqref{eq:primal-SDP-unconstrained}
and from \eqref{eq:dual-SDP} to~\eqref{eq:dual-SDP-unconstrained} are well
known, but they constitute some of the core preliminary observations behind
our QSlack method. See the first three rows of Figure~\ref{fig:qslack-method-overview} for a summary of these steps.

In Appendix~\ref{app:alt-penalty-methods}, we discuss some
variants of $\alpha(c)$ and $\beta(c)$ that have the additional feature of
weak duality holding for every $c>0$.
It is also of interest to consider SDPs and LPs with both equality and
inequality constraints. We discuss these cases in
Appendix~\ref{app:QCSlack-ineq-eq-constraints}, along with the QSlack and
CSlack methods for them.

\subsection{QSlack algorithm for variational quantum semi-definite
programming}

\label{sec:qslack-alg}As indicated previously, one of the basic ideas of the
QSlack algorithm is to solve the primal and dual SDPs
in~\eqref{eq:primal-SDP}--\eqref{eq:dual-SDP} by employing their reductions to
\eqref{eq:primal-SDP-unconstrained}--\eqref{eq:dual-SDP-unconstrained},
respectively. Additionally, some of the core assumptions are the same as those
made in \cite{patel2021variational}: we assume that the Hermitian matrices $A$
and $B$ are observables that can be efficiently measured on a quantum
computer, and we assume that the Hermiticity-preserving superoperator$~\Phi$ corresponds to
one also. Particular examples of efficiently measurable observables and input
models for $A$, $B$, and $\Phi$ are detailed in
Appendix~\ref{sec:eff-meas-obs-input-models}. As such, now we assume that
$d_{1}=2^{n}$ and $d_{2}=2^{m}$ for $n,m\in\mathbb{N}$, so that $A$ is an
$n$-qubit observable and $B$ is an $m$-qubit observable, and the superoperator~$\Phi$ maps $n$-qubit observables to $m$-qubit observables. Furthermore, we
expect these SDPs to be difficult to solve on classical computers, given that
standard algorithms for solving SDPs have runtime polynomial in the
dimensions of the inputs~$A$, $B$, and $\Phi$, which are now exponential in $n$,
$m$, or both.

We employ another basic observation also used in~\cite{patel2021variational}:
the positive semi-definite matrices $X$, $W$, $Y$, and $Z$ appearing in the
optimizations in \eqref{eq:primal-SDP-unconstrained} and
\eqref{eq:dual-SDP-unconstrained} can be represented as scaled density
matrices. That is, whenever $X\neq0$, it can be written as $X=\lambda\rho$,
where $\lambda\coloneqq\operatorname{Tr}[X]$ and $\rho\coloneqq X/\lambda$, so
that $\lambda>0$ and $\rho$ is a density matrix, satisfying $\rho\geq0$ and
$\operatorname{Tr}[\rho]=1$. This observation implies that the optimizations
in \eqref{eq:primal-SDP-unconstrained}--\eqref{eq:dual-SDP-unconstrained} can
be rewritten as follows, respectively:%
\begin{align}
\alpha(c)  &  =\sup_{\substack{\lambda,\mu\geq0,\\\rho,\sigma\in\mathcal{D}%
}}\left\{  \lambda\operatorname{Tr}[A\rho]-c\left\Vert B-\lambda\Phi(\rho
)-\mu\sigma\right\Vert _{2}^{2}\right\}  ,\label{eq:primal-q-states}\\
\beta(c)  &  =\inf_{\substack{\kappa,\nu\geq0,\\\tau,\omega\in\mathcal{D}%
}}\left\{  \kappa\operatorname{Tr}[B\tau]+c\left\Vert \kappa\Phi^{\dag}%
(\tau)-A-\nu\omega\right\Vert _{2}^{2}\right\}  , \label{eq:dual-q-states}%
\end{align}
where we have simply made the substitutions $X=\lambda\rho$, $W=\mu\sigma$,
$Y=\kappa\tau$, and $Z=\nu\omega$, and $\mathcal{D}$ denotes the set of all
density matrices. 

One can interpret what we are doing here as taking advantage
of the mathematical structure of quantum mechanics, namely, that states are
constrained to be positive semi-definite matrices, in order to impose the
positive semi-definite constraints on $X$, $W$, $Y$, and $Z$. Alternatively,
as happens in some of the examples that we explore in
Section~\ref{sec:qslack-examples}, if any of the matrices in the optimization are
not constrained to be positive semi-definite and are general or Hermitian,
then we can write them as a linear combination of Pauli matrices with complex
or real coefficients, respectively.

Expressions like $\operatorname{Tr}[A\rho]$ and $\operatorname{Tr}[B\tau]$ in
\eqref{eq:primal-q-states}--\eqref{eq:dual-q-states} are interpreted in
quantum mechanics in terms of the Born rule and are understood as expectations
of the observables $A$ and $B$ with respect to the states $\rho$ and $\tau$,
respectively. As such, in principle, the quantities $\operatorname{Tr}[A\rho]$
and $\operatorname{Tr}[B\tau]$ can be estimated through repetition, by
repeatedly preparing the states $\rho$ and $\tau$, performing the procedures
to measure the observables $A$ and $B$, and finally calculating sample means
as estimates of $\operatorname{Tr}[A\rho]$ and $\operatorname{Tr}[B\tau]$. In
Appendix~\ref{sec:estimating-objective-functions}, we discuss how the other
terms $\left\Vert B-\lambda\Phi(\rho)-\mu\sigma\right\Vert _{2}^{2}$ and
$\left\Vert \kappa\Phi^{\dag}(\tau)-A-\nu\omega\right\Vert _{2}^{2}$ can be
estimated, for which the destructive swap test \cite{GC13} or the mixed-state Loschmidt echo test (Appendix~\ref{app:Loschmidt-echo-alg}) are  helpful.

While the modifications in \eqref{eq:primal-q-states}--\eqref{eq:dual-q-states}
allow for rewriting the objective functions in
\eqref{eq:primal-SDP-unconstrained}--\eqref{eq:dual-SDP-unconstrained} in
terms of expectations of observables, evaluating the objective functions in
\eqref{eq:primal-q-states}--\eqref{eq:dual-q-states} is still too
difficult:\ the optimizations are over all possible states and not all states
are efficiently preparable. That is, even if the observables are efficiently
measurable, the overall procedure will not be efficient if there is not an
efficient method to prepare the states $\rho$ and$~\tau$. Thus, our next
modification is to replace the optimizations over all possible states in
\eqref{eq:primal-q-states}--\eqref{eq:dual-q-states} with optimizations over
parameterized states that are efficiently preparable, leading to the
following:
\begin{align}
\tilde{\alpha}(c)  &  \coloneqq\sup_{\lambda,\mu\geq0,\theta_{1},\theta_{2}%
\in\Theta}\Big\{\lambda\operatorname{Tr}[A\rho(\theta_{1})]\nonumber\\
&  \qquad\qquad-c\left\Vert B-\lambda\Phi(\rho(\theta_{1}))-\mu\sigma
(\theta_{2})\right\Vert _{2}^{2}\Big\},\label{eq:param-primal-SDP}\\
\tilde{\beta}(c)  &  \coloneqq\inf_{\kappa,\nu\geq0,\theta_{3},\theta_{4}%
\in\Theta}\Big\{\kappa\operatorname{Tr}[B\tau(\theta_{3})]\nonumber\\
&  \qquad\qquad+c\left\Vert \kappa\Phi^{\dag}(\tau(\theta_{3}))-A-\nu
\omega(\theta_{4})\right\Vert _{2}^{2}\Big\}, \label{eq:param-dual-SDP}%
\end{align}
where $\Theta$ is a general set of parameter values, $\theta_{1}$, $\theta
_{2}$, $\theta_{3}$, and $\theta_{4}$ are vectors of parameter values, and
$\rho(\theta_{1})$, $\sigma(\theta_{2})$, $\tau(\theta_{3})$, and
$\omega(\theta_{4})$ denote the corresponding parameterized states.
See the last two rows of Figure~\ref{fig:qslack-method-overview} for a summary of this final step and the previous one.

Let us
note that this modification, replacing an optimization over all possible
states with parameterized states, is common to all variational quantum
algorithms~\cite{Cerezo2021,bharti2021noisy}.
Appendix~\ref{sec:purification-CC-ansatze}\ discusses two methods for
parameterizing the set of density matrices, which are called the purification
ansatz and convex-combination ansatz, explored previously in related and
different contexts \cite{verdon2019quantum,CSZW20,Liu2021,patel2021variational,sbahi2022provably,Ezzell2023}.
Let us emphasize again that other quantum computational parameterizations
of density matrices, such as quantum Boltzmann machines or yet to be discovered ones, can be incorporated into QSlack. 

Related to what was
mentioned above, if any of the matrices in the optimization are general or
Hermitian, then we can employ optimizations over linear combinations of Pauli
matrices with complex or real coefficients, respectively, where the
coefficients play the role of the parameters. In these cases, we restrict the
number of non-zero coefficients to be polynomial in the number of qubits, so
that we can evaluate the objective functions involving them efficiently.

The modified optimization problems in
\eqref{eq:param-primal-SDP}--\eqref{eq:param-dual-SDP} have the benefit that
all expressions in the objective functions in
\eqref{eq:param-primal-SDP}--\eqref{eq:param-dual-SDP} are efficiently
measurable on quantum computers. That is, through repeated preparation of the
states $\rho(\theta_{1})$, $\sigma(\theta_{2})$, $\tau(\theta_{3})$, and
$\omega(\theta_{4})$, all quantities in
\eqref{eq:param-primal-SDP}--\eqref{eq:param-dual-SDP} can be estimated
efficiently. Since they are expectations of observables, one can additionally
employ the techniques of error mitigation \cite{cai2023quantum} to reduce the
effects of errors corrupting the estimates. 

Let us further observe that the
following inequalities hold:%
\begin{align}
\alpha(c)  &  \geq\tilde{\alpha}(c),\label{eq:primal-smaller}\\
\beta(c)  &  \leq\tilde{\beta}(c) \label{eq:dual-larger} \, .
\end{align}
This is a consequence of the fact that the set of parameterized states is a 
subset of the set of all possible states. As such, the ability to optimize
$\tilde{\alpha}(c)$ and $\tilde{\beta}(c)$ for each $c>0$ implies the
following theorem:
\begin{theorem}
The following inequalities hold:
\begin{equation}
\tilde{\alpha}\leq\alpha\leq\beta\leq\tilde{\beta},
\label{eq:main-theory-result}%
\end{equation}
where $\alpha$ and $\beta$ are defined in \eqref{eq:primal-SDP}--\eqref{eq:dual-SDP}, 
\begin{align}
\tilde{\alpha}  &  \coloneqq\sup_{\substack{\lambda,\mu\geq0,\\\theta
_{1},\theta_{2}\in\Theta}}\left\{
\begin{array}
[c]{c}%
\lambda\operatorname{Tr}[A\rho(\theta_{1})]:\\
B-\lambda\Phi(\rho(\theta_{1}))=\mu\sigma(\theta_{2})
\end{array}
\right\}  ,\label{eq:param-opt-constrained-primal}\\
\tilde{\beta}  &  \coloneqq\inf_{\substack{\kappa,\nu\geq0,\\\theta_{3}%
,\theta_{4}\in\Theta}}\left\{
\begin{array}
[c]{c}%
\kappa\operatorname{Tr}[B\tau(\theta_{3})]:\\
\kappa\Phi^{\dag}(\tau(\theta_{3}))-A=\nu\omega(\theta_{4})
\end{array}
\right\}  , \label{eq:param-opt-constrained-dual}%
\end{align}
$A$ is a $2^n\times 2^n$ Hermitian matrix, $B$ is a $2^m \times 2^m$ Hermitian matrix, $\Phi$ is a Hermiticity-preserving superoperator taking $2^n\times 2^n$ Hermitian matrices to $2^m \times 2^m$ Hermitian matrices, and $\Theta$ is a general set of parameter vectors.
\end{theorem}

\begin{proof}
The conclusion in \eqref{eq:main-theory-result} follows from
\eqref{eq:weak-duality-SDP}, \eqref{eq:convergence-penalties},
\eqref{eq:primal-smaller}--\eqref{eq:dual-larger}, and the limits
$\lim_{c\rightarrow\infty}\tilde{\alpha}(c)=\tilde{\alpha}$ and $\lim
_{c\rightarrow\infty}\tilde{\beta}(c)=\tilde{\beta}$, which both follow from
\cite[Proposition~5.2.1]{Bertsekas2016}.
\end{proof}

\bigskip

Eq.~\eqref{eq:main-theory-result}\ is one of the key theoretical insights of
our paper: under the assumptions that $A$, $B$, and $\Phi$ correspond to
efficiently measurable observables, in principle, it is possible to sandwich
the optimal value of the SDPs in \eqref{eq:primal-SDP} and \eqref{eq:dual-SDP}
by optimization problems whose objective functions are efficiently measurable
on quantum computers, but not clearly so on classical computers. Thus, by
following the standard approach of the penalty method, one could attempt to
optimize $\tilde{\alpha}(c_{k})$ and $\tilde{\beta}(c_{k})$ for a monotone
increasing sequence $\left(  c_{k}\right)  _{k}$ of penalty parameters, which
ultimately would yield the bounds in \eqref{eq:main-theory-result}. Let us
remark here that, in principle, this solves one of the problems with the
approach from \cite{patel2021variational}, which did not lead to guaranteed
bounds on the optimal value of primal and dual SDPs that have the form of
\eqref{eq:primal-SDP} and \eqref{eq:dual-SDP}. Furthermore, let us note that the equalities
$\tilde{\alpha}(c)=\alpha(c)$ and $\tilde{\beta}(c)=\beta(c)$ hold if the
optimal solution is contained in the parameterized sets of states.

It is also worth stressing at this point that QSlack only provides strict bounds in the limit that $c \rightarrow \infty$ and in the absence of shot noise and hardware noise. For finite $c$ values, QSlack provides instead an estimate of the upper or lower bound. As will be discussed in Section~\ref{sec:qslack-sims}, we do occasionally see the upper (lower) bound provided by QSlack dipping below (above) the true bound for low $c$ values, but this can be resolved by using a larger value of~$c$.

In order to optimize the objective functions in
\eqref{eq:param-primal-SDP}--\eqref{eq:param-dual-SDP}, we employ a hybrid
quantum--classical approach common to all variational quantum algorithms.
Focusing on \eqref{eq:param-dual-SDP}, such an algorithm involves using a
quantum computer to evaluate expressions like $\operatorname{Tr}[B\tau
(\theta_{3})]$ (expectations of observables), as well as other terms in the
objective function like $\operatorname{Tr}[A\omega(\theta_{4})]$, which arises
as part of one of the six terms after expanding $\left\Vert \kappa\Phi^{\dag}%
(\tau(\theta_{3}))-A-\nu\omega(\theta_{4})\right\Vert _{2}^{2}$. After doing
so, we can employ gradient descent or other related algorithms to determine a
new choice of the parameter vectors~$\theta_{3}$ and $\theta_{4}$; then we
iterate this process until either convergence occurs or after a specified
maximum number of iterations.

In order to execute this hybrid quantum--classical approach, gradient-based
methods require estimates of the gradient vectors of the objective functions
in \eqref{eq:param-primal-SDP}--\eqref{eq:param-dual-SDP}. There are known
approaches for doing so when using ans\"atze based on parameterized quantum circuits, such as the parameter-shift rule
\cite{Li2017,Mitarai2018,Schuld2019} or the gradient Hadamard test
\cite{Li2017a,guerreschi2017practical,Romero_2019}, whenever the objective
function can be written as $\langle0|U^{\dag}(\theta)HU(\theta)|0\rangle$,
where $H$ is a Hamiltonian, $U(\theta)$ is a parameterized unitary, and
$|0\rangle$ is shorthand for a tensor product of zero states. Given that it is
not obvious from \eqref{eq:param-primal-SDP} how these objective functions can
be written in the form $\langle0|U^{\dag}(\theta)HU(\theta)|0\rangle$, in
Appendix~\ref{sec:estimating-gradients} we discuss the specifics of how to
estimate the gradients in our case. One could alternatively employ the quantum
natural gradient method \cite{Stokes2020quantumnatural} in order to take
advantage of the geometry of quantum states. The parameter-shift rule, the
gradient Hadamard test, and quantum natural gradient are all called analytic
gradient estimation methods. As argued in \cite{Harrow2021}, there are
theoretical advantages of using analytic gradient estimates over
finite-difference estimates; at the same time, the computational complexity of
estimating each element of the gradient vector analytically is essentially the same as that of evaluating
the objective function.


A key drawback of QSlack, yet common to all variational
quantum algorithms, is that the optimizations in
\eqref{eq:primal-q-states}--\eqref{eq:dual-q-states} over the convex set of
density matrices are now replaced with the optimizations in \eqref{eq:param-primal-SDP}--\eqref{eq:param-dual-SDP}
 over the non-convex set
of parameterized states. Furthermore,  landscapes for optimization problems involving randomly initialized parameterized quantum circuits  are known to suffer from the barren-plateau problem \cite{McClean_2018}.  As such, even though the optimizations are over fewer
parameters, their landscapes become marked with local optima~\cite{Anschuetz2022Quantum} and barren
plateaus that can make convergence to global optima
difficult. Thus, it is difficult in practice to guarantee that the optimal
values $\tilde{\alpha}(c)$ and $\tilde{\beta}(c)$ can be estimated for every
$c>0$. Regardless, the dimension of the optimization task involved in QSlack is
significantly smaller than the original problem of optimizing over matrices of
exponential dimension (i.e., $d_{1}=2^{n}$ and $d_{2}=2^{m}$). Thus, QSlack
gives an approach to attempt solving the optimization task, while standard SDP  solvers cannot solve the original task in
time faster than exponential in $n$ and~$m$.

In summary, the main idea of the QSlack algorithm is to replace the original
SDPs in \eqref{eq:primal-SDP} and \eqref{eq:dual-SDP} with a sequence of
optimizations of the form in \eqref{eq:param-primal-SDP} and
\eqref{eq:param-dual-SDP}, respectively. The main advantage of QSlack is that
every term in the objective functions of \eqref{eq:param-primal-SDP} and
\eqref{eq:param-dual-SDP}, as well as their gradients, can be estimated
efficiently on quantum computers. Thus, one uses a quantum computer only to
evaluate these quantities, and the rest of the optimization is performed by
standard classical approaches like gradient descent and its variants. Since
the terms being estimated by a quantum computer involve only expectations of
observables, error mitigation techniques \cite{cai2023quantum} can be employed for reducing the
effects of errors. Furthermore, the inequalities in
\eqref{eq:main-theory-result} provide theoretical guarantees that, if one can
calculate the optimal values of \eqref{eq:param-primal-SDP} and
\eqref{eq:param-dual-SDP} for each $c>0$, then the resulting quantities
$\tilde{\alpha}$ and $\tilde{\beta}$ are certified bounds on the true optimal
values $\alpha$ and $\beta$, thus sandwiching them (recall
that actually $\alpha=\beta$ in the case that strong duality holds).

\subsection{QSlack example problems}

\label{sec:qslack-examples}

In this section, we present a variety of QSlack example problems relevant for quantum information and physics,
including the following:

\begin{enumerate}
\item the normalized trace distance, a measure of distinguishability for two
quantum states,

\item the fidelity, a similarity measure of two quantum states,

\item entanglement negativity, an entanglement measure of a bipartite state,

\item and constrained Hamiltonian optimization, for which the goal is to minimize
the energy of a Hamiltonian subject to constraints on the state,

\end{enumerate}

\noindent We have selected each of these problems to illustrate the
versatility of the QSlack  method, i.e., how it can handle a
diverse set of problems. For problems 1, 2, and 3,  we assume that one has
sample access to the states, so that the input model here is the linear
combination of states input model, as described in
Appendix~\ref{app:linear-combo-states}. That is, there is some procedure, whether it be by a quantum circuit or
other means, that prepares the states and can be repeated in such a way that
the same state is prepared each time. For example, if the state to be prepared
is $\rho$, we assume that there is a procedure to prepare $\rho^{\otimes n}$
for $n\in\mathbb{N}$ arbitrarily large. For problem~4, the Pauli input model is quite natural given that Hamiltonians specified in terms of Pauli strings are ubiquitous in physics (see Appendix~\ref{app:Pauli-input-model}  for more details of the Pauli input model).

\subsubsection{Normalized trace distance}
\label{sec:normalized-TD-example}

Let us begin with the normalized trace distance. For $n$-qubit states $\rho$
and $\sigma$, it is defined as follows \cite{Helstrom1967,Helstrom1969} and
has equivalent characterizations in terms of the following primal and dual
semi-definite programs (see, e.g., \cite[Proposition~3.51]{khatri2020principles}):
\begin{align}
\frac{1}{2}\left\Vert \rho-\sigma\right\Vert _{1}  &  =\sup_{\Lambda\geq
0}\left\{  \operatorname{Tr}[\Lambda(\rho-\sigma)]:\Lambda\leq I\right\}
\label{eq:primal-TD}\\
&  =\inf_{Y\geq0}\left\{  \operatorname{Tr}[Y]:Y\geq\rho-\sigma\right\}  .
\label{eq:dual-TD}%
\end{align}
As such, since $\rho-\sigma$ appears in both the primal and dual
optimizations, the input model for this problem is the linear combination of
states model, as outlined in Appendix~\ref{app:linear-combo-states}.
Additionally, the optimizations above are over $2^{n}\times2^{n}$ matrices
$\Lambda$ and $Y$, subject to the constraints stated above. It is thus not
possible to estimate $\frac{1}{2}\left\Vert \rho-\sigma\right\Vert _{1}$
efficiently using a standard SDP solver.

If $\rho$ and $\sigma$ are prepared by quantum circuits, it is known that a
decision problem related to estimating their normalized trace distance is a
QSZK-complete problem \cite{Watrous2002,Watrous2006}, which indicates that the
worst-case complexity of this problem is considered intractable for a quantum computer.
While worst-case instances are computationally hard to solve, this does not rule out the possibility of solving other instances.
In this spirit, one can attempt to estimate this quantity by
means of a variational quantum algorithm, and prior papers have done so
\cite{CSZW20,RASW23}, focusing on the primal optimization given in
\eqref{eq:primal-TD}. These prior papers thus provide lower bounds on
$\frac{1}{2}\left\Vert \rho-\sigma\right\Vert _{1}$, since the primal
optimization approaches $\frac{1}{2}\left\Vert \rho-\sigma\right\Vert _{1}$
from below. Let us note that, for the primal optimization in
\eqref{eq:primal-TD}, it is actually simpler here not to use the QSlack method
and instead it is more sensible to employ a parameterized measurement circuit,
as done in \cite{CSZW20,RASW23}. This is because the constraint $0\leq
\Lambda\leq I$ implies that $\Lambda$ is a measurement operator (see
\cite[Section~III-A]{RASW23} for details). Nevertheless, we show in Proposition~\ref{prop:TD-primal-qslack} of Appendix~\ref{app:norm-TD} how to rewrite the primal SDP in \eqref{eq:primal-TD} using the QSlack method.

Here we focus on the dual optimization in \eqref{eq:dual-TD}. Following the
QSlack method, we can rewrite \eqref{eq:dual-TD} exactly as follows:%
\begin{multline}
\inf_{Y\geq0}\left\{  \operatorname{Tr}[Y]:Y\geq\rho-\sigma\right\}
\label{eq:TD-dual-rewrite-QSlack}\\
=\lim_{c\rightarrow\infty}\inf_{\substack{\lambda,\mu\geq0,\\\omega,\tau
\in\mathcal{D}}}\left\{  \lambda+c\left\Vert \lambda\omega-\rho+\sigma-\mu
\tau\right\Vert _{2}^{2}\right\}  ,
\end{multline}
where $c>0$ is the penalty parameter. See
Proposition~\ref{prop:TD-dual-QSlack}\ in Appendix~\ref{app:norm-TD} for
details. As outlined in
\eqref{eq:param-primal-SDP}--\eqref{eq:param-dual-SDP}, we then replace the
optimizations over $\omega,\tau\in\mathcal{D}$ with optimizations over
parameterized states, using either the purification or convex-combination
ans\"{a}tze. The Hilbert--Schmidt norm in~\eqref{eq:TD-dual-rewrite-QSlack}
can be expanded into ten different trace overlap terms, and each of them can
be estimated efficiently using either the destructive swap test or the
mixed-state Loschmidt echo test given in Appendix~\ref{app:Loschmidt-echo-alg}.

\subsubsection{Root fidelity}
\label{sec:fidelity-main-text}

Next we consider estimating the root fidelity between $n$-qubit states $\rho$
and $\sigma$. Given quantum states $\rho$ and $\sigma$, the fidelity is
defined as $F(\rho,\sigma)\coloneqq\left\Vert \sqrt{\rho}\sqrt{\sigma
}\right\Vert _{1}^{2}$ \cite{Uhlmann1976}, and the root fidelity has
equivalent characterizations in terms of the following primal and dual
semi-definite programs \cite[Section~2.1]{Watrous2013}:%
\begin{align}
&  \sqrt{F}(\rho,\sigma)\nonumber\\
&  =\sup_{X\in\mathcal{L}}\left\{  \operatorname{Re}[\operatorname{Tr}[X]]:%
\begin{bmatrix}
\rho & X^{\dag}\\
X & \sigma
\end{bmatrix}
\geq0\right\} \label{eq:fidelity-SDP-primal}\\
&  =\frac{1}{2}\inf_{Y,Z\geq0}\left\{  \operatorname{Tr}[Y\rho
]+\operatorname{Tr}[Z\sigma]:%
\begin{bmatrix}
Y & I\\
I & Z
\end{bmatrix}
\geq0\right\}  , \label{eq:fidelity-SDP-dual}%
\end{align}
where $\mathcal{L}$ denotes the set of all $2^{n}\times2^{n}$ matrices. The
input model for this problem is again the linear combination of states input
model. Since the optimizations above are over $2^{n}\times2^{n}$ matrices $X$,
$Y$, and $Z$, it is not possible to estimate $\sqrt{F}(\rho,\sigma)$
efficiently using standard SDP solvers.

Similar to the case of normalized trace distance mentioned above, in the case
that $\rho$ and $\sigma$ are prepared by quantum circuits, it is known that a
decision problem related to estimating their fidelity is a QSZK-complete
problem \cite{Watrous2002,Watrous2006}, which again indicates that the
worst-case complexity of this problem is considered intractable for a quantum
computer. Regardless, prior papers have provided variational algorithms for estimating the fidelity by employing
other variational expressions \cite{CSZW20,RASW23,goldfeld2023quantum},
based on Uhlmann's theorem \cite{Uhlmann1976} to estimate it from below and
the Fuchs--Caves measurement \cite{FC95,F96} and Alberti's theorem
\cite{alberti1983note} to estimate it from above. The approach based on Uhlmann's theorem requires having access to purifications of $\rho$ and $\sigma$, while the QSlack approach and those based on the Fuchs--Caves measurement and Alberti's theorem only require sample access to $\rho$ and $\sigma$.

In spite of the prior work above, we view estimating the fidelity as a
fundamental task in quantum information and additionally as a way of testing
the QSlack method, as well as to demonstrate its versatility. Furthermore, as mentioned above, QSlack does not require access to purifications, and so in this sense, the primal SDP below can be viewed as an improvement on prior lower-bound variational approaches \cite{CSZW20,RASW23} based on Uhlmann's theorem. Let us first
focus on solving the primal optimization in \eqref{eq:fidelity-SDP-primal}.
Again following the QSlack method, but this time parameterizing the matrix $X$
in terms of the Pauli basis as%
\begin{equation}
X=\sum_{\overrightarrow{x}}\alpha_{\overrightarrow{x}}\sigma_{\overrightarrow
{x}},
\end{equation}
where $\overrightarrow{x} \equiv (x_1, \ldots, x_n)$, $\alpha_{\overrightarrow{x}}\in\mathbb{C}$, and $\sigma_{\overrightarrow
{x}}\equiv\sigma_{x_{1}}\otimes\cdots\otimes\sigma_{x_{n}}$ is a Pauli string,
we can write%
\begin{multline}
\sup_{X\in\mathcal{L}}\left\{  \operatorname{Re}[\operatorname{Tr}[X]]:%
\begin{bmatrix}
\rho & X^{\dag}\\
X & \sigma
\end{bmatrix}
\geq0\right\} \label{eq:qslack-fid-primal}\\
=\lim_{c\rightarrow\infty}\sup_{\substack{\left(  \alpha_{\overrightarrow{x}%
}\right)  _{\overrightarrow{x}},\\\lambda\geq0,\omega\in\mathcal{D}}}\left\{
\begin{array}
[c]{c}%
2^{n}\operatorname{Re}[\alpha_{\overrightarrow{0}}]\\
-c\cdot f(\rho,\sigma,\lambda,\omega,\left(  \alpha_{\overrightarrow{x}%
}\right)  _{\overrightarrow{x}})
\end{array}
\right\}  ,
\end{multline}
with
\begin{multline}
f(\rho,\sigma,\lambda,\omega,\left(  \alpha_{\overrightarrow{x}}\right)
_{\overrightarrow{x}})\coloneqq\operatorname{Tr}[\rho^{2}]+\operatorname{Tr}%
[\sigma^{2}]+\lambda^{2}\operatorname{Tr}[\omega^{2}%
]\label{eq:fid-function-q-slack}\\
+2^{n+1}\left\Vert \overrightarrow{\alpha}\right\Vert _{2}^{2}-2\lambda
\operatorname{Tr}[\left(  |0\rangle\!\langle0|\otimes\rho\right)  \omega]\\
-2\lambda\operatorname{Tr}[\left(  |1\rangle\!\langle1|\otimes\sigma\right)
\omega]\\
-2\lambda\sum_{\overrightarrow{x}}\operatorname{Re}\!\left[  \alpha
_{\overrightarrow{x}}\operatorname{Tr}\!\left[  \left(  \left(  \sigma
_{X}-i\sigma_{Y}\right)  \otimes\sigma_{\overrightarrow{x}}\right)
\omega\right]  \right]  .
\end{multline}
See Proposition~\ref{prop:fid-primal-SDP}\ in
Appendix~\ref{app:fidelity-qslack-details}\ for details. The terms
$2^{n}\operatorname{Re}[\alpha_{\overrightarrow{0}}]$ and $2^{n+1}\left\Vert
\overrightarrow{\alpha}\right\Vert _{2}^{2}$ do not require sampling and can
be evaluated exactly as a function of the Pauli coefficients in the tuple
$\left(  \alpha_{\overrightarrow{x}}\right)  _{\overrightarrow{x}}$. The last
term in \eqref{eq:fid-function-q-slack} can be evaluated as an expectation of
the Pauli observables $\sigma_{X}\otimes\sigma_{\overrightarrow{x}}$ and
$\sigma_{Y}\otimes\sigma_{\overrightarrow{x}}$ and, for every non-zero
coefficient in $\left(  \alpha_{\overrightarrow{x}}\right)  _{\overrightarrow
{x}}$, combined linearly. All other terms can be evaluated by means of the
destructive swap test or the mixed-state Loschmidt echo test.

As stated above, there is an exact equality in \eqref{eq:qslack-fid-primal}.
However, it is not yet possible to evaluate the objective function in
\eqref{eq:qslack-fid-primal} efficiently because there are $4^{n}$
coefficients in the tuple $\left(  \alpha_{\overrightarrow{x}}\right)
_{\overrightarrow{x}}$. As such, our next modification of the problem is to
restrict the tuple $\left(  \alpha_{\overrightarrow{x}}\right)
_{\overrightarrow{x}}$ to include only poly$(n)$ non-zero coefficients. With
this restriction, it is then possible to evaluate the objective function in
\eqref{eq:qslack-fid-primal} efficiently.

Let us now consider the dual optimization in \eqref{eq:fidelity-SDP-dual}.
Following the QSlack method, we can write%
\begin{multline}
\frac{1}{2}\inf_{Y,Z\geq0}\left\{  \operatorname{Tr}[Y\rho]+\operatorname{Tr}%
[Z\sigma]:%
\begin{bmatrix}
Y & I\\
I & Z
\end{bmatrix}
\geq0\right\} \label{eq:dual-opt-fidelity}\\
=\lim_{c\rightarrow\infty}\inf_{\substack{\lambda,\mu,\nu
\geq0,\\\omega,\tau,\xi\in\mathcal{D}}}\left\{
\begin{array}
[c]{c}%
\frac{1}{2}\lambda\operatorname{Tr}[\omega\rho]+\frac{1}{2}\mu\operatorname{Tr}[\tau\sigma]\\
+c\cdot g(\lambda,\mu,\nu,\omega,\tau,\xi)
\end{array}
\right\}  ,
\end{multline}
where%
\begin{multline}
g(\lambda,\mu,\nu,\omega,\tau,\xi)\coloneqq\lambda^{2}\operatorname{Tr}%
[\omega^{2}]+\mu^{2}\operatorname{Tr}[\tau^{2}]\\
+2^{n+1}+\nu^{2}\operatorname{Tr}[\xi^{2}]-2\lambda\nu\operatorname{Tr}%
[\left(  |0\rangle\!\langle0|\otimes\omega\right)  \xi]\\
-2\mu\nu\operatorname{Tr}[\left(  |1\rangle\!\langle1|\otimes\tau\right)
\xi]-2\nu\operatorname{Tr}[\left(  \sigma_{X}\otimes I\right)  \xi].
\end{multline}
See Proposition~\ref{prop:fid-dual-SDP}\ in
Appendix~\ref{app:fidelity-qslack-details} for details. As outlined in
\eqref{eq:param-primal-SDP}--\eqref{eq:param-dual-SDP}, we replace the
optimizations over $\omega,\tau,\xi\in\mathcal{D}$ with optimizations over
parameterized states, using either the purification or convex-combination
ans\"{a}tze. Then every term in the objective function on the right-hand side
of \eqref{eq:dual-opt-fidelity} can be evaluated efficiently by means of the
destructive swap test or the mixed-state Loschmidt echo test.


\subsubsection{Entanglement negativity}

\label{sec:entanglement-negativity}

The entanglement negativity \cite{ZHSL98,Vidal2002}\ is an entanglement
measure that has been considered extensively in entanglement theory
\cite{H42007}. Let $\rho_{AB}$ be a bipartite state of $n=n_{A}+n_{B}$ total
qubits, where system$~A$ consists of $n_{A}$ qubits and system$~B$ consists of
$n_{B}$ qubits. The entanglement negativity is defined as follows:%
\begin{equation}
E_{N}(\rho_{AB})\coloneqq\left\Vert T_{B}(\rho_{AB})\right\Vert _{1},
\end{equation}
where $T_{B}(\cdot)\coloneqq\sum_{i,j}|i\rangle\!\langle j|_{B}(\cdot
)|i\rangle\!\langle j|_{B}$ is the transpose map acting on the $B$ system. The
entanglement negativity does not increase under the action of local operations
and classical communication, and it is equal to its minimum value of one if
$\rho_{AB}$ is a separable, unentangled state~\cite{Vidal2002}. It is known to
have the following primal and dual semi-definite programming characterizations (see, e.g., 
\cite[Eqs.~(5.1.101)--(5.1.102)]{khatri2020principles}):
\begin{align}
E_{N}(\rho_{AB})  &  =\sup_{H_{AB}\in\operatorname{Herm}}\left\{
\begin{array}
[c]{c}%
\operatorname{Tr}[T_{B}(H_{AB})\rho_{AB}]:\\
-I_{AB}\leq H_{AB}\leq I_{AB}%
\end{array}
\right\} \label{eq:neg-primal}\\
&  =\inf_{K,L\geq0}\left\{
\begin{array}
[c]{c}%
\operatorname{Tr}[K_{AB}+L_{AB}]:\\
T_{B}(K_{AB}-L_{AB})=\rho_{AB}%
\end{array}
\right\}  . \label{eq:neg-dual}%
\end{align}
The input model for this problem is again the linear combination of states
input model, since here the input is sample access to the state $\rho_{AB}$.
Since the optimizations above are over $2^{n}\times2^{n}$ matrices $H_{AB}$,
$K_{AB}$, and $L_{AB}$, it is not possible to estimate $E_{N}(\rho_{AB})$
efficiently using standard SDP solvers.

To the best of our knowledge, the quantum computational complexity of
estimating $E_{N}(\rho_{AB})$ has not been investigated, whenever one has
access to a quantum circuit that prepares the state $\rho_{AB}$; we consider it an
interesting open problem to determine the worst-case complexity of estimating
$E_{N}(\rho_{AB})$. However, estimating the negativity on quantum computers
has previously been considered in \cite{carteret2017estimating} by using
low-order moments of the partially transposed state, each of which can be
estimated by means of the Hadamard-test circuits from \cite{Carteret2005}.
Variational quantum algorithms were also proposed recently in \cite{Chen2023}
for the case of pure states and in \cite{Wang2022}\ for the general case.

Estimating the entanglement negativity presents an interesting case for
QSlack, given the presence of the partial transpose. To tackle this problem,
we make use of the fact that every Hermitian matrix can be represented as a
linear combination of Pauli strings with real coefficients, so that we can
write $H_{AB}$ in \eqref{eq:neg-primal} as follows:%
\begin{equation}
H_{AB}=\sum_{\overrightarrow{x_{A}},\overrightarrow{x_{B}}}\alpha
_{\overrightarrow{x_{A}},\overrightarrow{x_{B}}}\sigma_{\overrightarrow{x_{A}%
}}\otimes\sigma_{\overrightarrow{x_{B}}}, \label{eq:pauli-rep-H-neg}%
\end{equation}
where $\alpha_{\overrightarrow{x_{A}},\overrightarrow{x_{B}}}\in\mathbb{R}$,
$\overrightarrow{x_{A}}\in\left\{  0,1,2,3\right\}  ^{n_{A}}$, and
$\overrightarrow{x_{B}}\in\left\{  0,1,2,3\right\}  ^{n_{B}}$. We can also use
the facts that%
\begin{align}
T(\sigma_{0})  &  =\sigma_{0},\quad T(\sigma_{1})=\sigma_{1}%
,\label{eq:partial-transpose-paulis-main}\\
T(\sigma_{2})  &  =-\sigma_{2},\quad T(\sigma_{3})=\sigma_{3},
\end{align}
where we have ordered the Pauli matrices in the conventional way as
$\sigma_{0}\equiv I$, $\sigma_{1}\equiv\sigma_{X}$, $\sigma_{2}\equiv
\sigma_{Y}$, and $\sigma_{3}\equiv\sigma_{Z}$. As shown in
Proposition~\ref{prop:neg-primal-SDP}\ in Appendix~\ref{app:negativity}, we
can then rewrite the optimization in \eqref{eq:neg-primal} as follows:%
\begin{multline}
\sup_{H_{AB}\in\operatorname{Herm}}\left\{
\begin{array}
[c]{c}%
\operatorname{Tr}[T_{B}(H_{AB})\rho_{AB}]:\\
-I_{AB}\leq H_{AB}\leq I_{AB}%
\end{array}
\right\}  =\label{eq:neg-primal-equality-to-pauli}\\
\lim_{c\rightarrow\infty}\sup_{\substack{\overrightarrow{\alpha},\lambda
,\mu\geq0,\\\sigma_{AB},\tau_{AB}\in\mathcal{D}}}\left\{
\begin{array}
[c]{c}%
g_{1}\!\left(  \overrightarrow{\alpha},\rho_{AB}\right) \\
-c\cdot g_{2}\!\left(  \overrightarrow{\alpha},\lambda,\mu,\sigma_{AB}%
,\tau_{AB}\right)
\end{array}
\right\}  ,
\end{multline}
where $\overrightarrow{\alpha}\equiv\left(  \alpha_{\overrightarrow{x_{A}%
},\overrightarrow{x_{B}}}\in\mathbb{R}\right)  _{\overrightarrow{x_{A}%
},\overrightarrow{x_{B}}}$,
\begin{multline}
g_{1}\!\left(  \overrightarrow{\alpha},\rho_{AB}\right)  \coloneqq\\
\sum_{\overrightarrow{x_{A}},\overrightarrow{x_{B}}}\left(  -1\right)
^{f(\overrightarrow{x_{B}})}\alpha_{\overrightarrow{x_{A}},\overrightarrow
{x_{B}}}\operatorname{Tr}\!\left[  \left(  \sigma_{\overrightarrow{x_{A}}%
}\otimes\sigma_{\overrightarrow{x_{B}}}\right)  \rho_{AB}\right]  ,
\end{multline}%
\begin{multline}
g_{2}\!\left(  \overrightarrow{\alpha},\lambda,\mu,\sigma_{AB},\tau
_{AB}\right)  \coloneqq2^{n_{A}+n_{B}+1}+2\left\Vert \overrightarrow{\alpha
}\right\Vert _{2}^{2}\\
+\lambda^{2}\operatorname{Tr}[\sigma_{AB}^{2}]+\mu^{2}\operatorname{Tr}%
[\tau_{AB}^{2}]-2\lambda-2\mu\\
+2\lambda\sum_{\overrightarrow{x_{A}},\overrightarrow{x_{B}}}\alpha
_{\overrightarrow{x_{A}},\overrightarrow{x_{B}}}\operatorname{Tr}\!\left[
\left(  \sigma_{\overrightarrow{x_{A}}}\otimes\sigma_{\overrightarrow{x_{B}}%
}\right)  \sigma_{AB}\right] \\
-2\mu\sum_{\overrightarrow{x_{A}},\overrightarrow{x_{B}}}\alpha
_{\overrightarrow{x_{A}},\overrightarrow{x_{B}}}\operatorname{Tr}\!\left[
\left(  \sigma_{\overrightarrow{x_{A}}}\otimes\sigma_{\overrightarrow{x_{B}}%
}\right)  \tau_{AB}\right]  ,
\end{multline}
and $f(\overrightarrow{x_{B}})\coloneqq\sum_{i=1}^{n_{B}}\delta_{x_{B}^{i},2}$
counts the number of $\sigma_{Y}$ terms in the sequence $\overrightarrow
{x_{B}}$, with $x_{B}^{i}$ denoting the $i$th entry in $\overrightarrow{x_{B}%
}$. The terms $2^{n_{A}+n_{B}+1}$, $2\left\Vert \overrightarrow{\alpha
}\right\Vert _{2}^{2}$, $2\lambda$, and $2\mu$ do not require sampling and can
be calculated exactly. The sole term in $g_{1}$ and the last two terms of
$g_{2}$ can be estimated as expectations of Pauli strings and linearly
combined, while all other terms in $g_{2}$ can be estimated by means of the
destructive swap test or the mixed-state Loschmidt echo test.

The main advantage of using the Pauli representation in~\eqref{eq:pauli-rep-H-neg} for $H_{AB}$ is that it provides a clear route for
handling the partial transpose operation, by means of the equalities in
\eqref{eq:partial-transpose-paulis-main}. We take a similar approach in the
dual SDP below. Since the partial transpose is an unphysical operation, it is
unclear how to handle it if we had instead represented~$H_{AB}$ as a linear
combination of scaled density matrices.

Let us now consider the dual optimization in \eqref{eq:neg-dual}. In
Proposition~\ref{prop:neg-dual-SDP}\ in Appendix~\ref{app:negativity}, we show
that%
\begin{multline}
\inf_{K_{AB},L_{AB}\geq0}\left\{
\begin{array}
[c]{c}%
\operatorname{Tr}[K_{AB}+L_{AB}]:\\
T_{B}(K_{AB}-L_{AB})=\rho_{AB}%
\end{array}
\right\}  =\label{eq:neg-dual-paulis}\\
\lim_{c\rightarrow\infty}\inf_{\substack{\overrightarrow{\alpha}%
,\overrightarrow{\beta},\\\lambda,\mu\geq0,\\\sigma_{AB},\tau_{AB}%
\in\mathcal{D}}}\left\{
\begin{array}
[c]{c}%
2^{n}\left(  \alpha_{\overrightarrow{0},\overrightarrow{0}}+\beta
_{\overrightarrow{0},\overrightarrow{0}}\right) \\
+c\cdot g_{3}\!\left(  \overrightarrow{\alpha},\overrightarrow{\beta}%
,\lambda,\mu,\sigma_{AB},\tau_{AB}\right)
\end{array}
\right\}  ,
\end{multline}
where%
\begin{align}
\overrightarrow{\alpha}  &  \equiv\left(  \alpha_{\overrightarrow{x_{A}%
},\overrightarrow{x_{B}}}\in\mathbb{R}\right)  _{\overrightarrow{x_{A}%
},\overrightarrow{x_{B}}},\\
\overrightarrow{\beta}  &  \equiv\left(  \beta_{\overrightarrow{x_{A}%
},\overrightarrow{x_{B}}}\in\mathbb{R}\right)  _{\overrightarrow{x_{A}%
},\overrightarrow{x_{B}}},
\end{align}%
\begin{multline}
g_{3}\!\left(  \overrightarrow{\alpha},\overrightarrow{\beta},\lambda
,\mu,\sigma_{AB},\tau_{AB}\right)  \coloneqq\\
2^{n+1}\left(  \left\Vert \overrightarrow{\alpha}\right\Vert _{2}%
^{2}+\left\Vert \overrightarrow{\beta}\right\Vert _{2}^{2}\right)
+\operatorname{Tr}[\rho_{AB}^{2}]+\mu^{2}\operatorname{Tr}[\tau_{AB}^{2}]\\
-2\sum_{\overrightarrow{x_{A}},\overrightarrow{x_{B}}}\left(  -1\right)
^{f(\overrightarrow{x_{B}})}\left(  \alpha_{\overrightarrow{x_{A}%
},\overrightarrow{x_{B}}}-\beta_{\overrightarrow{x_{A}},\overrightarrow{x_{B}%
}}\right)  \operatorname{Tr}\!\left[  \left(  \sigma_{\overrightarrow{x_{A}}%
}\otimes\sigma_{\overrightarrow{x_{B}}}\right)  \rho_{AB}\right] \\
-2^{n+1}\sum_{\overrightarrow{x_{A}},\overrightarrow{x_{B}}}\alpha
_{\overrightarrow{x_{A}},\overrightarrow{x_{B}}}\beta_{\overrightarrow{x_{A}%
},\overrightarrow{x_{B}}}+\lambda^{2}\operatorname{Tr}[\sigma_{AB}^{2}]\\
-2\lambda\sum_{\overrightarrow{x_{A}},\overrightarrow{x_{B}}}\alpha
_{\overrightarrow{x_{A}},\overrightarrow{x_{B}}}\operatorname{Tr}\!\left[
\left(  \sigma_{\overrightarrow{x_{A}}}\otimes\sigma_{\overrightarrow{x_{B}}%
}\right)  \sigma_{AB}\right] \\
-2\mu\sum_{\overrightarrow{x_{A}},\overrightarrow{x_{B}}}\beta
_{\overrightarrow{x_{A}},\overrightarrow{x_{B}}}\operatorname{Tr}\!\left[
\left(  \sigma_{\overrightarrow{x_{A}}}\otimes\sigma_{\overrightarrow{x_{B}}%
}\right)  \tau_{AB}\right]
\end{multline}

As was the case with fidelity, there are exact equalities in
\eqref{eq:neg-primal-equality-to-pauli} and \eqref{eq:neg-dual-paulis}, but it
is not possible to evaluate their objective functions efficiently because
there are $4^{n}$ coefficients in the vector $\overrightarrow{\alpha}$ in
\eqref{eq:neg-primal-equality-to-pauli} and similarly for the vectors
$\overrightarrow{\alpha}$ and $\overrightarrow{\beta}$ in
\eqref{eq:neg-dual-paulis}. As such, we modify these problems to restrict the
vectors to include only poly$(n)$ non-zero coefficients. Additionally, we
restrict the optimizations over $\sigma_{AB},\tau_{AB}\in\mathcal{D}$ in both
\eqref{eq:neg-primal-equality-to-pauli} and \eqref{eq:neg-dual-paulis} to be
over parameterized states. With these restrictions, we can then estimate the
objective functions on the right-hand sides of
\eqref{eq:neg-primal-equality-to-pauli} and \eqref{eq:neg-dual-paulis} efficiently.

\subsubsection{Constrained Hamiltonian optimization}

\label{sec:constrained-Ham-opt}

The goal of the constrained Hamiltonian optimization problem is to minimize
the energy of a given Hamiltonian subject to a list of constraints. As such,
it is a generalization of the standard ground-state energy problem, and it can also
be viewed as a variation of the standard form of SDPs considered in most works
on quantum semi-definite programming (for example, see \cite[Section~3.3]%
{patel2021variational}). This kind of problem was considered recently in
\cite{le2023variational}, but the optimization therein was restricted to be
over pure states, even though the optimal solution generally is a mixed state
for such problems.

A variation of constrained Hamiltonian optimization using semi-definite programming arises in
the context of the quantum marginal problem
\cite{Mazziotti2004,Mazziotti2004a,Mazziotti2006,Barthel2012,fawzi2023entropy} (see also
\cite[Chapter~3]{Skrzypczyk2023}), in which the goal is to calculate the
ground-state energy subject to a list of constraints on the reduced density
matrices of a global state. As such, the problem considered there has applications in
quantum chemistry and condensed matter physics. The problem we consider is complementary, being a variation of the ground-state energy problem in which there are further constraints on the state being optimized. 

The inputs to the constrained Hamiltonian optimization problem are the
Hamiltonian $H$ and $\ell$ Hermitian constraint operators $A_{1}, \ldots,
A_{\ell}$, as well as real constraint numbers $b_{1}, \ldots, b_{\ell}$.
Suppose that $H$, $A_{1}, \ldots, A_{\ell}$ are efficiently measurable
observables. Then the constrained Hamiltonian optimization corresponds to the
following optimization problem:%
\begin{align}
&  \mathcal{L}(H,A_{1},\ldots,A_{\ell})\nonumber\\
&  \coloneqq\inf_{\rho\in\mathcal{D}}\left\{  \operatorname{Tr}[H\rho
]:\operatorname{Tr}[A_{i}\rho]\geq b_{i}\text{ }\forall i\in\left[
\ell\right]  \right\} \label{eq:constrained-Ham-primal}\\
&  =\sup_{\substack{y_{1},\ldots,y_{\ell}\geq0,\\\mu\in\mathbb{R}}}\left\{
\sum_{i=1}^{\ell}b_{i}y_{i}+\mu:\sum_{i=1}^{\ell}y_{i}A_{i}+\mu I\leq
H\right\}  , \label{eq:constrained-Ham-dual}%
\end{align}
where we have written the dual SDP\ in \eqref{eq:constrained-Ham-dual}
(derived in Appendix~\ref{app:constrained-Ham-opt}). Indeed, the ground-state
energy problem is a special case with $A_{i}=0$ and $b_{i}=0$ for all
$i\in\left[  \ell\right]  $, leading to%
\begin{align}
\mathcal{L}(H)  &  \coloneqq\inf_{\rho\in\mathcal{D}}\left\{
\operatorname{Tr}[H\rho]\right\} \\
&  =\sup_{\mu\in\mathbb{R}}\left\{  \mu:\mu I\leq H\right\}  ,
\end{align}
for which we explore the QSlack approach in our companion paper
\cite{West2023dualVQE}.

A natural input model for this problem is the Pauli input model, described
further in Appendix~\ref{app:Pauli-input-model}, such that%
\begin{align}
H  &  =\sum_{\overrightarrow{x}}h_{\overrightarrow{x}}\sigma_{\overrightarrow
{x}},\\
A_{i}  &  =\sum_{\overrightarrow{x}}a_{\overrightarrow{x}}^{i}\sigma
_{\overrightarrow{x}}\qquad\forall i\in\left[  \ell\right]  ,
\end{align}
where $h_{\overrightarrow{x}},a_{\overrightarrow{x}}^{1},\ldots
,a_{\overrightarrow{x}}^{\ell}\in\mathbb{R}$ for all $\overrightarrow{x}$.
This problem is again an interesting case study for QSlack, different from the
previous examples. As shown in
Proposition~\ref{prop:constrained-Ham-opt-primal-SDP}\ in
Appendix~\ref{app:constrained-Ham-opt}, we can rewrite the primal optimization
in \eqref{eq:constrained-Ham-primal} as follows:%
\begin{multline}
\inf_{\rho\in\mathcal{D}}\left\{  \operatorname{Tr}[H\rho]:\operatorname{Tr}%
[A_{i}\rho]\geq b_{i}\text{ }\forall i\in\left[  \ell\right]  \right\}
=\label{eq:obj-func-primal-constr-Ham}\\
\lim_{c\rightarrow\infty}\inf_{\substack{\rho\in\mathcal{D},\\z_{1}%
,\ldots,z_{\ell}\geq0}}\left\{
\begin{array}
[c]{c}%
\sum_{\overrightarrow{x}}h_{\overrightarrow{x}}\operatorname{Tr}%
[\sigma_{\overrightarrow{x}}\rho]\\
+c\cdot f\!\left(  \left(  \overrightarrow{a}^{i}\right)  _{i=1}^{\ell
},\overrightarrow{b},\overrightarrow{z}\right)
\end{array}
\right\}  ,
\end{multline}
where%
\begin{align}
f\!\left(  \left(  \overrightarrow{a}^{i}\right)  _{i=1}^{\ell}%
,\overrightarrow{b},\overrightarrow{z}\right)   &  \coloneqq\sum_{i=1}^{\ell
}\left(  \sum_{\overrightarrow{x}}a_{\overrightarrow{x}}^{i}\operatorname{Tr}%
[\sigma_{\overrightarrow{x}}\rho]-b_{i}-z_{i}\right)  ^{2},\\
\overrightarrow{a}^{i}  &  \coloneqq\left(  a_{\overrightarrow{x}}^{i}\right)
_{\overrightarrow{x}},\\
\overrightarrow{b}  &  \coloneqq\left(  b_{1},\ldots,b_{\ell}\right)  ,\\
\overrightarrow{z}  &  \coloneqq\left(  z_{1},\ldots,z_{\ell}\right)  .
\end{align}
The objective function in \eqref{eq:obj-func-primal-constr-Ham} involves
expectations of the form $\operatorname{Tr}[\sigma_{\overrightarrow{x}}\rho]$,
which can be estimated by sampling using a quantum computer, and then linearly
combined to estimate $\sum_{\overrightarrow{x}}h_{\overrightarrow{x}%
}\operatorname{Tr}[\sigma_{\overrightarrow{x}}\rho]$ and $\sum
_{\overrightarrow{x}}a_{\overrightarrow{x}}^{i}\operatorname{Tr}%
[\sigma_{\overrightarrow{x}}\rho]$.\ All other terms, like $b_{i}$ and $z_{i}%
$, can be evaluated exactly.

Let us also note here, that since \eqref{eq:constrained-Ham-primal} involves only scalar expressions in the objective function and constraints, we can also solve it by means of the interior-point method \cite[Section~5.1]{Bertsekas2016}. In short, this means rewriting it as the following unconstrained optimization with a barrier function (negative logarithm) and barrier parameter $\eta >0 $:
\begin{equation}
\inf_{\rho\in\mathcal{D}}\left\{  \operatorname{Tr}[H\rho] - \eta \sum_{i=1}^{\ell} \ln (\operatorname{Tr}[A_{i}\rho]- b_{i})  \right\}.
\label{eq:constr-opt-inter-point}
\end{equation}
If this optimization can be solved for each $\eta > 0$, then the solution to~\eqref{eq:constr-opt-inter-point} is guaranteed to converge to the solution to~\eqref{eq:constrained-Ham-primal} in the limit as $\eta \to 0$ \cite[Proposition~5.1.1]{Bertsekas2016}. For this method to work, it is necessary to find an initial state~$\rho$ such that the constraints in~\eqref{eq:constrained-Ham-primal} are satisfied strictly. 

We also show in Proposition~\ref{prop:constrained-Ham-opt-dual-SDP}\ in
Appendix~\ref{app:constrained-Ham-opt} that the dual optimization in
\eqref{eq:constrained-Ham-dual}\ can be rewritten as%
\begin{multline}
\sup_{\substack{y_{1},\ldots,y_{\ell}\geq0,\\\mu\in\mathbb{R}}}\left\{
\sum_{i=1}^{\ell}b_{i}y_{i}+\mu:\sum_{i=1}^{\ell}y_{i}A_{i}+\mu I\leq
H\right\}  =\label{eq:obj-func-dual-constr-Ham}\\
\lim_{c\rightarrow\infty}\sup_{\substack{y_{1},\ldots,y_{\ell}\geq0,\\\mu
\in\mathbb{R},\nu\geq0,\\\omega\in\mathcal{D}}}\left\{
\begin{array}
[c]{c}%
\sum_{i=1}^{\ell}b_{i}y_{i}+\mu\\
-c\cdot f\!\left(  \overrightarrow{h},\left(  \overrightarrow{a}^{i}\right)
_{i=1}^{\ell},\overrightarrow{y},\overrightarrow{A},\mu,\nu,\omega\right)
\end{array}
\right\}  , 
\end{multline}
where%
\begin{multline}
f\!\left(  \overrightarrow{h},\left(  \overrightarrow{a}^{i}\right)
_{i=1}^{\ell},\overrightarrow{y},\overrightarrow{A},\mu,\nu,\omega\right)
\coloneqq2^{n}\left\Vert \overrightarrow{h}\right\Vert _{2}^{2}+\mu^{2}2^{n}\\
+2^{n}\sum_{i,j=1}^{\ell}y_{i}y_{j}\left(  \overrightarrow{a}^{i}%
\cdot\overrightarrow{a}^{j}\right)  +\nu^{2}\operatorname{Tr}[\omega^{2}]\\
-2^{n+1}\sum_{i=1}^{\ell}y_{i}\overrightarrow{h}\cdot\overrightarrow{a}%
^{i}-2^{n+1}\mu h_{\overrightarrow{0}}\\
-2\nu\sum_{\overrightarrow{x}}h_{\overrightarrow{x}}\operatorname{Tr}%
[\sigma_{\overrightarrow{x}}\omega]+2^{n+1}\mu\sum_{i=1}^{\ell}y_{i}%
a_{\overrightarrow{0}}^{i}\\
+2\nu\sum_{i=1}^{\ell}y_{i}\sum_{\overrightarrow{x}}a_{\overrightarrow{x}}%
^{i}\operatorname{Tr}\!\left[  \sigma_{\overrightarrow{x}}\omega\right]
+2\mu\nu.
\end{multline}
The terms $2^{n}\left\Vert \overrightarrow{h}\right\Vert _{2}^{2}$, $\mu
^{2}2^{n}$, $2^{n}\sum_{i,j=1}^{\ell}y_{i}y_{j}\left(  \overrightarrow{a}%
^{i}\cdot\overrightarrow{a}^{j}\right)  $, $-2^{n+1}\sum_{i=1}^{\ell}%
y_{i}\overrightarrow{h}\cdot\overrightarrow{a}^{i}$, $2^{n+1}\mu
h_{\overrightarrow{0}}$, $2^{n+1}\mu\sum_{i=1}^{\ell}y_{i}a_{\overrightarrow
{0}}^{i}$, and $2\mu\nu$ can be evaluated exactly, while the term $\nu
^{2}\operatorname{Tr}[\omega^{2}]$ can be estimated by means of the
destructive swap test or the mixed-state Loschmidt echo test, and the terms $2\nu\sum_{\overrightarrow{x}%
}h_{\overrightarrow{x}}\operatorname{Tr}[\sigma_{\overrightarrow{x}}\omega]$
and $2\nu\sum_{i=1}^{\ell}y_{i}\sum_{\overrightarrow{x}}a_{\overrightarrow{x}%
}^{i}\operatorname{Tr}\!\left[  \sigma_{\overrightarrow{x}}\omega\right]  $
can be estimated as expectations of observables.

As with the other examples we have considered, there are exact equalities in
\eqref{eq:obj-func-primal-constr-Ham} and \eqref{eq:obj-func-dual-constr-Ham}.
In the general case, there are $4^{n}$ coefficients in the vector
$\overrightarrow{h}$ and tuple of vectors, $\left(  \overrightarrow{a}%
^{i}\right)  _{i=1}^{\ell}$. However, for problems of physical interest (in
which the Hamiltonians and constraints consist of few-body interactions),
these vectors include only poly$(n)$ non-zero coefficients. Additionally, in
order for the objective functions in \eqref{eq:obj-func-primal-constr-Ham} and
\eqref{eq:obj-func-dual-constr-Ham} to be efficiently estimated, we restrict
the optimizations over $\rho$ and $\omega$ to be over parameterized states.

\subsection{QSlack simulations}

\label{sec:qslack-sims}

\begin{figure}[t]
\centering
    \subfigure[\ Parameterized unitary]{\includegraphics[width=\linewidth]{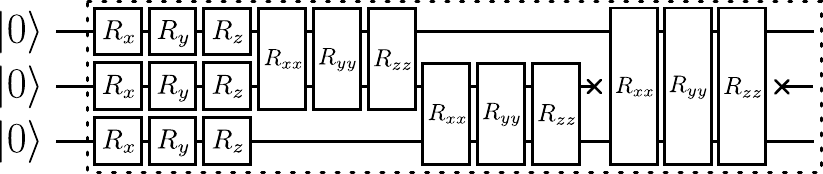}\label{fig:param_unitary}}
    \hfill
    \subfigure[\ Parameterized quantum circuit Born machine]{\includegraphics[width=0.55\linewidth]{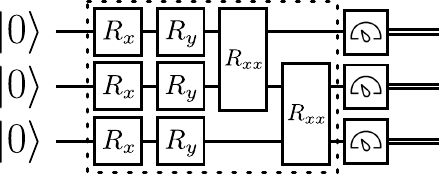}\label{fig:param_qcbm}}
    \caption{(a) A three-qubit example of the parameterized unitary used in the purification ansatz and the eigenvector creation in the convex-combination ansatz. The two-qubit gates for a chain act on qubits $i$ and $i+1$. The two-qubit gates also act on first and last qubit within a layer. (b) A three-qubit example of the parameterized quantum circuit Born machine that generates the probability distribution used in the convex-combination ansatz. The circuit elements within the dotted box form a layer, and each layer is repeated multiple times, depending on the problem instance. $\{R_i\}_{i \in \{x, y, z\}}$ denotes a single-qubit parameterized rotation about the axis $i$. Similarly $\{R_i\}_{i \in \{xx, yy, zz\}}$ are two-qubit parameterized rotations that can be used to generate entanglement.}
\end{figure}

In this section, we discuss the results of simulations of the QSlack algorithm for the example problems from Section~\ref{sec:qslack-examples}. We first discuss the common features to all the experiments and then delve into specifics for each example.

\subsubsection{Input models}

We begin by discussing the input models to the QSlack example problems. For the normalized trace distance, fidelity, and entanglement negativity, the inputs to the problem are quantum states. We generate mixed states as inputs for these problems using either the purification ansatz or the convex-combination ansatz (see Appendix~\ref{sec:purification-CC-ansatze}).

For the purification ansatz, following Appendix~\ref{app:purification-ansatz}, we pick the parameterized unitary to be of the form shown in Figure~\ref{fig:param_unitary}. The size of the purifying subsystem $R$, is chosen to be equal to that of the system of interest $S$, leading to a full-rank state on subsystem $S$. The parameters, i.e., the rotation angles, of the unitary operator are chosen at random. For all the parameterized unitaries, we fix the number of layers to be two. 

For the convex-combination ansatz, following Appendix~\ref{app:convex-combo-ansatz}, we use the same structure as the purification ansatz, as depicted in  Figure~\ref{fig:param_unitary}, for the parameterized unitary $U(\gamma)$. 
To generate probabilities $p_{\varphi}$, we use a quantum circuit Born machine with the structure shown in Figure~\ref{fig:param_qcbm}. The size of the quantum circuit Born machine is chosen to be equal to that of system $S$. The number of layers for all parameterized unitaries can be chosen arbitrarily, and in this work, we pick the number of layers to be two.

Lastly, for the constrained Hamiltonian optimization, the input to the problem is a Hermitian operator. As discussed in Section~\ref{sec:constrained-Ham-opt}, we use the Pauli input model from Appendix~\ref{app:Pauli-input-model}. 

Throughout training, we have sample access to the density matrices and probability distributions, through the purifying parameterized unitaries and measurement outcomes of the relevant quantum circuit Born machines.

\subsubsection{Training}

Similar to the inputs just discussed, we use either the purification or convex-combination ansatz to parameterize the density matrices being optimized. These density matrices are indeed optimization variables and unlike the input density matrices, the parameters are not held fixed. Training involves attempting to find optimal parameters that extremize the objective function value. In this work, we pick the form of the parameterized quantum states to be the same as the form of the input. For example, if the input to the trace-distance problem is two density matrices specified using a purification or convex-combination ansatz, then we pick the optimization variable to be of the same form. 

At each iteration, the training process crucially depends on gradient estimation to pick the next set of parameters. While several algorithms can be employed for gradient estimation, we rely on the Simultaneous Perturbation Stochastic Approximation (SPSA) method \cite{Spall1992} to estimate gradients. SPSA produces an unbiased estimator of the gradient with runtime constant in the number of parameters. Details of the method can be found in Appendix~\ref{sec:gradient-CCA}. We use a maximum number of iterations as the stopping condition for all simulations.

The hyperparameters, like the learning rate and perturbation parameter, are tailored to each problem instance. We pick the smallest possible penalty parameter $c$ that suffices to enforce the constraints of the problem. This leads to faster convergence to the optimal value in practice. 

\begin{figure*}
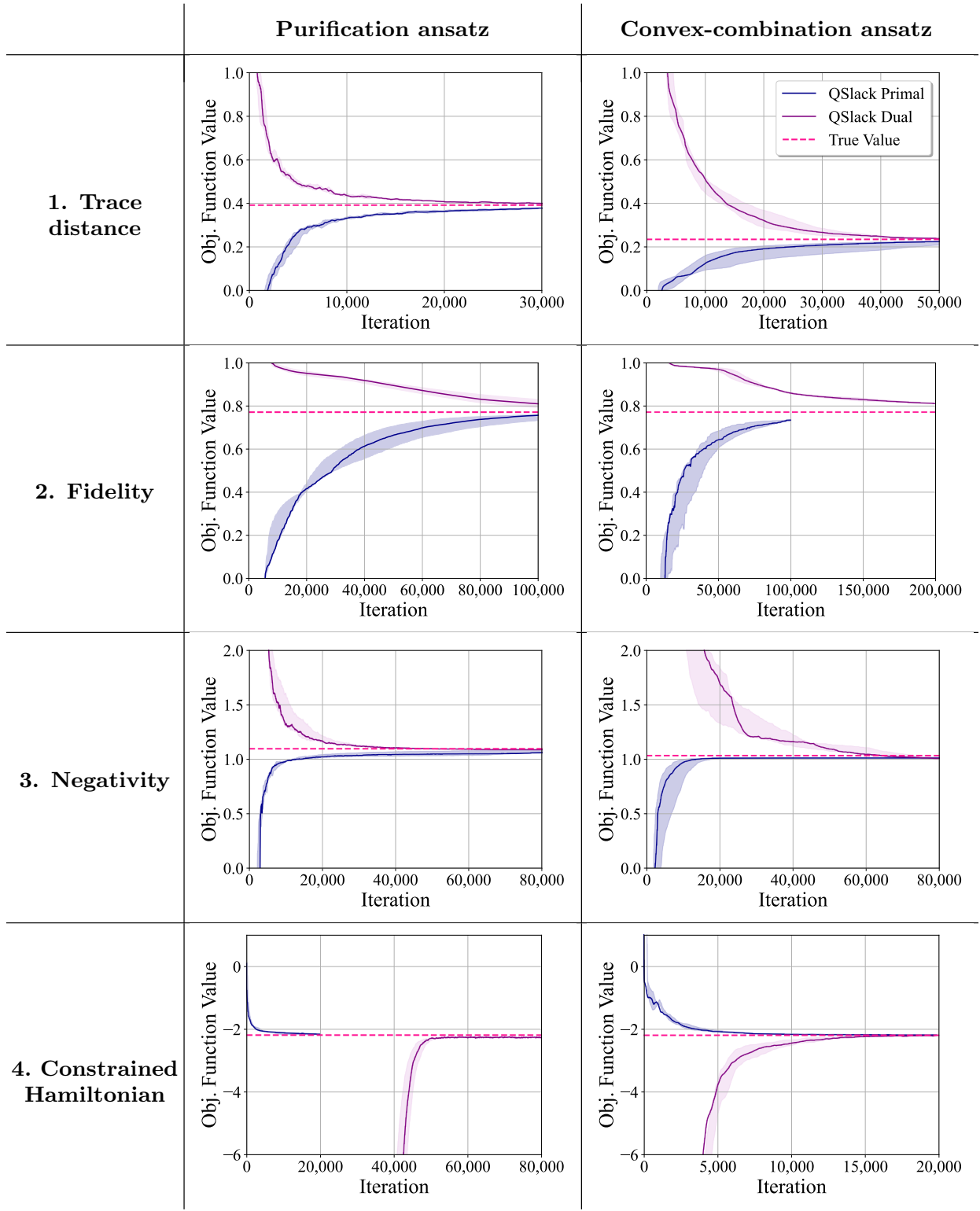

\renewcommand{\arraystretch}{1.8}
\centering
\resizebox{.9\linewidth}{!}{

\begin{tabular}{P{0.13\linewidth}|c|c}
& \textbf{Purification ansatz} & \textbf{Convex-combination ansatz}
\vspace{-0.4cm} \\
\hline
\centering \textbf{1. Trace distance}  & \tdpsandwich & \tdccasandwich \\
\hline
\centering\textbf{2. Fidelity}  & \fpsandwich & \fccasandwich \\
\hline
\centering\textbf{3. Negativity} & \enpsandwich & \enccasandwich \\
\hline
\centering\textbf{4.~Constrained Hamiltonian}  & \hpsandwich & \hccasandwich \\
\end{tabular}
}
\caption{Convergence of both the primal and dual optimizations to their optimal values for various QSlack example problems. We used both the purification and convex-combination ans\"atze for each problem. The solid line shows the median value of the estimate, the shaded region represents the interquartile range, while the ground truth is marked by the dashed line. The number of runs over which the median and interquartile range are computed is ten for trace distance, four for root fidelity, and five for all others. Specific details regarding the runs, including the number of qubits, layers, gradient method, etc., can be found in Appendix~\ref{sec:details-simulations}. }
\label{fig:sandwich-plots}
\end{figure*}

Prior to training, some of the parameters are initialized to specific values while others are set randomly. The initial parameter values of all parameterized quantum circuits are chosen uniformly at random from $[0,2\pi]$. For the normalized trace distance problem, in both the dual and the primal (see Propositions~\ref{prop:TD-dual-QSlack} and \ref{prop:TD-primal-qslack}), parameters $\lambda$ and $\mu$ are initialized to one. For the root fidelity simulations, we initialize $\lambda$ to one and all components of vector $\left( \alpha_{\overrightarrow{x}%
}\right)_{\overrightarrow{x}}$ to zero in \eqref{eq:qslack-fid-primal}. Since the system size in our examples is relatively small, we include all components of the vectors in our training, instead of truncating to a polynomial number of them. We initialize to $\lambda = \mu = \nu = 1$ in \eqref{eq:dual-opt-fidelity}. For the entanglement negativity problem, parameters $\lambda$ and $\mu$ in \eqref{eq:neg-primal-equality-to-pauli} and \eqref{eq:neg-dual-paulis} are initialized to one. The components of vectors $\alpha
_{\overrightarrow{x_{A}},\overrightarrow{x_{B}}}$ and $\beta_{\overrightarrow{x_{A}},\overrightarrow{x_{B}}}$ are all initialized to zero. Similar to root fidelity, we include all components of these vectors in our training. Lastly for the constrained Hamiltonian optimization, we initialize to $y_1= y_2= \nu=0.001, \  \mu= -0.005$ in \eqref{eq:constrained-Ham-opt-HS-norm}, and $z_1= 0.1, \ z_2= 0.5$ in \eqref{eq:primal_constrained_initializations}.

\subsubsection{Results}

We used a noiseless simulator for the experiments in this work. However, we used the destructive swap test to estimate terms of the form $\operatorname{Tr}[\rho \sigma]$, which involves Bell measurements. Thus, estimates of the different terms involve sampling noise (also called shot noise). We leave the simulations of the different problems with other noise models to future work.

The plots for the various examples using both ansatz types are given in Figure~\ref{fig:sandwich-plots}. In these plots, we show the median (in solid lines) and interquartile range (in shading) of the objective function values over the course of training. The variation in objective function values across different runs corresponds to three sources of randomness: the randomized initializations of all parametrized quantum circuits, the inherent randomness of the SPSA method for gradient estimation, and shot noise in quantum circuit measurements. For root fidelity, we performed simulations without shot noise. For all other problems, we took $10^{12}$ shots, and the error achieved was typically on the order of $10^{-2}$.  Details and specifics are discussed in Appendix~\ref{sec:details-simulations}, and there we also give plots for the error and penalty values across training. 

The simulations were performed using both the Qiskit quantum software development kit from IBM \cite{Qiskit} and the Qulacs library \cite{Suzuki2021qulacsfast}.

\section{CSlack background, algorithm, examples, and simulations}

\subsection{Background on linear programming}

\label{sec:LP-background}

Linear programming can be understood as a special case of semi-definite
programming in which the Hermitian matrices $A$ and $B$ are diagonal and the
Hermiticity preserving superoperator $\Phi$ takes diagonal matrices to
diagonal matrices. This is quite similar to how classical information theory
can be understood as a special case of quantum information theory in which all
density matrices are diagonal and all quantum channels take diagonal density
matrices to diagonal density matrices. This observation is what we use to
understand the connection between the QSlack and CSlack methods for
variational quantum semi-definite programming and variational linear
programming, respectively.

Given this observation, we provide a quick review of linear programming,
keeping in mind that it is the special case in which \textquotedblleft
everything from Section~\ref{sec:SDP-background}\ is
diagonal.\textquotedblright\ As such, Hilbert--Schmidt inner products reduce
to standard vector inner products, and the action of a superoperator on a
matrix reduces to usual matrix-vector multiplication. Fix $d_{1},d_{2}%
\in\mathbb{N}$. Let $a$ be a $d_{1}\times1$ real vector, $b$ a $d_{2}\times1$
real vector, and $\phi$ a $d_{2}\times d_{1}$ real matrix. A linear program is
specified by the triple $(a,b,\phi)$ and is defined as the following
optimization problem:%
\begin{equation}
\alpha_{L}\coloneqq\sup_{x\geq0}\left\{  a^{T}x:\phi x\leq b\right\}  ,
\label{eq:primal-LP}%
\end{equation}
where the supremum optimization is over every $d_{1}\times1$ real vector $x$
with non-negative entries (i.e., the notation $x\geq0$ is shorthand for every
entry of $x$ being non-negative). Furthermore, the inequality constraint $\phi
x\leq b$ is shorthand for every entry of the vector $b-\phi x$ being
non-negative. The optimization in \eqref{eq:primal-LP}\ is called the primal
LP, and a vector $x$ is primal feasible if $x\geq0$ and $\phi x\leq b$. The
dual optimization problem is as follows:%
\begin{equation}
\beta_{L}\coloneqq\inf_{y\geq0}\left\{  b^{T}y:\phi^{T}y\geq a\right\}  ,
\label{eq:dual-LP}%
\end{equation}
where the infimum optimization is over every $d_{2}\times1$ real vector $y$
with non-negative entries and $\phi^{T}$ is the matrix transpose of $\phi$. A
vector $y$ is dual feasible if both $y\geq0$ and $\phi^{T}y\geq a$.

Weak duality corresponds to the inequality%
\begin{equation}
\alpha_{L}\leq\beta_{L}, \label{eq:weak-duality-LP}%
\end{equation}
and strong duality corresponds to the equality $\alpha_{L}=\beta_{L}$. Strong
duality holds whenever Slater's condition is satisfied, i.e., if there exists
a primal feasible $x$ and a strictly dual feasible $y$ or if there exists a
strictly primal feasible~$x$ and a dual feasible $y$.

As before, and for similar reasons, we can introduce slack variables to
transform the inequality constraints in~\eqref{eq:primal-LP} and
\eqref{eq:dual-LP} to equality constraints. That is, we have the following:%
\begin{align}
\alpha_{L}  &  =\sup_{x,w\geq0}\left\{  a^{T}x:b-\phi x=w\right\}  ,\\
\beta_{L}  &  =\inf_{y,z\geq0}\left\{  b^{T}y:\phi^{T}y-a=z\right\}  ,
\end{align}
where $w$ is a $d_{2}\times1$ real vector and $z$ is a $d_{1}\times1$ real
vector. Finally, it is possible to transform the constrained optimizations to
unconstrained optimizations by introducing penalty terms in the objective
functions, as follows:%
\begin{align}
&  \alpha_{L}(c)\coloneqq\sup_{x,w\geq0}\left\{  a^{T}x-c\left\Vert b-\phi
x-w\right\Vert _{2}^{2}\right\}  ,\label{eq:primal-LP-unconstrained}\\
&  \beta_{L}(c)\coloneqq\inf_{y,z\geq0}\left\{  b^{T}y+c\left\Vert \phi
^{T}y-a-z\right\Vert _{2}^{2}\right\}  , \label{eq:dual-LP-unconstrained}%
\end{align}
where $c>0$ is a penalty constant and $\left\Vert s\right\Vert _{2}$ is the
Euclidean norm of a vector $s$. Again, by standard reasoning
\cite[Proposition~5.2.1]{Bertsekas2016}, we have that%
\begin{equation}
\alpha_{L}=\lim_{c\rightarrow\infty}\alpha_{L}(c),\qquad\beta_{L}%
=\lim_{c\rightarrow\infty}\beta_{L}(c), \label{eq:limits-penalty-LP}%
\end{equation}
and we thus arrive at a method for approximating the optimal values
$\alpha_{L}$ and $\beta_{L}$. The reductions from \eqref{eq:primal-LP} to
\eqref{eq:primal-LP-unconstrained}\ and from \eqref{eq:dual-LP} to
\eqref{eq:dual-LP-unconstrained} are indeed standard, but they constitute some
of the core preliminary observations behind our CSlack method.

\subsection{CSlack algorithm for variational linear programming}

\label{sec:cslack-alg}

Now we turn to describing the CSlack algorithm for variational linear
programming. The idea is conceptually similar to QSlack and can be thought of
as the classical version of it. Indeed, as mentioned before, linear programs
can be thought of as the classical version of semi-definite programs in which
every object is a diagonal matrix. As such, it follows that the density
matrices, parameterized states, and observables from the previous section can
be replaced with diagonal density matrices, diagonal parameterized states, and
diagonal observables, which are equivalent to probability distributions,
parameterized probability distributions, and observable vectors, respectively.
In what follows, our discussion of CSlack mirrors that of QSlack, and as such,
it can be quickly skimmed if the ideas behind QSlack are already clear.

A key assumption of the CSlack algorithm is that the vectors $a$ and $b$ and
the matrix $\phi$ correspond to efficiently measurable observable vectors,
meaning that they can be efficiently estimated on quantum or probabilistic
computers (we describe precisely what we mean here later on and we provide
particular examples of efficiently measurable observable vectors and input
models for $a$, $b$, and $\phi$ in
Appendices~\ref{app:linear-combo-dist-model} and
\ref{app:Walsh-Hadamard-input-model}). We assume that $d_{1}=2^{n}$ and
$d_{2}=2^{m}$ for $n,m\in\mathbb{N}$, so that the entries of $a$ are indexed
by $n$-bit strings and the entries of $b$ are indexed by $m$-bit strings.
Thus, the vectors $a$ and $b$ and the matrix $\phi$ are exponentially large in
the parameters $n$ and $m$.

A basic observation is that the vectors $x$, $w$, $y$, and $z$ appearing in
the optimizations in
\eqref{eq:primal-LP-unconstrained}--\eqref{eq:dual-LP-unconstrained} can be
represented as scaled probability distributions. Whenever $x\neq0$, it can be
written as $x=\lambda r$, where $\lambda\coloneqq\boldsymbol{1}^{T}x$ and
$r\coloneqq x/\lambda$, with $\boldsymbol{1}$ the vector of all ones, so
that $\lambda>0$ and $r$ is a probability distribution on $2^{n}$ elements. We
can thus write the optimizations in
\eqref{eq:primal-LP-unconstrained}--\eqref{eq:dual-LP-unconstrained} as
follows:%
\begin{align}
\alpha_{L}(c)  &  =\sup_{\substack{\lambda,\mu\geq0,\\r,s\in\mathcal{P}%
}}\left\{  \lambda a^{T}r-c\left\Vert b-\lambda\phi r-\mu s\right\Vert
_{2}^{2}\right\}  ,\label{eq:primal-prob-dist-opt}\\
\beta_{L}(c)  &  =\inf_{\substack{\kappa,\nu\geq0,\\t,w\in\mathcal{P}%
}}\left\{  \kappa b^{T}t+c\left\Vert \kappa\phi^{T}t-a-\nu w\right\Vert
_{2}^{2}\right\}  , \label{eq:dual-prob-dist-opt}%
\end{align}
where we have made the substitutions $x=\lambda r$, $w=\mu s$, $y=\kappa t$,
and $z=\nu w$, and $\mathcal{P}$ denotes the set of all probability
distributions. Here we are taking advantage of the structure of probability
theory, namely, that probability distributions are described by vectors with
non-negative entries, in order to impose the entrywise non-negative
constraints on $x$, $w$, $y$, and $z$.

Expressions like $a^{T}r$ and $b^{T}t$ in
\eqref{eq:primal-prob-dist-opt}--\eqref{eq:dual-prob-dist-opt} can be
understood as expectations of the observable vectors $a$ and $b$ with respect
to the probability distributions $r$ and $t$, respectively. That is, by
labeling the $i$th entry of $a$ and $r$ as $a_{i}$ and $r_{i}$, respectively,
we see that
\begin{equation}
a^{T}r=\sum_{i}r_{i}a_{i},
\end{equation}
which is the equation for the expectation of a random variable taking the value
$a_{i}$ with probability $r_{i}$. As such, the quantities $a^{T}r$ and
$b^{T}t$ can be estimated through repetition, by repeatedly sampling from the
probability distributions $r$ and $t$, performing the procedures to evaluate
the entries of $a$ and $b$, and finally calculating sample means as estimates
of $a^{T}r$ and $b^{T}t$. In Appendices~\ref{app:linear-combo-dist-model} and
\ref{app:Walsh-Hadamard-input-model}, we discuss how the other terms
$\left\Vert b-\lambda\phi r-\mu s\right\Vert _{2}^{2}$ and $\left\Vert
\kappa\phi^{T}t-a-\nu w\right\Vert _{2}^{2}$ can be estimated, for which the
collision test is helpful (i.e., the classical version of the swap test).

We have rewritten the objective functions in
\eqref{eq:primal-LP-unconstrained}--\eqref{eq:dual-LP-unconstrained} as in
\eqref{eq:primal-prob-dist-opt}--\eqref{eq:dual-prob-dist-opt}, in terms of
expectations of observable vectors. However, evaluating these is too difficult
because the optimizations are over all possible probability distributions and
not all probability distributions are efficiently preparable. Similar to
QSlack, our next modification is to replace the optimizations over all
possible probability distributions in
\eqref{eq:primal-prob-dist-opt}--\eqref{eq:dual-prob-dist-opt} with
optimizations over parameterized probability distributions that are
efficiently preparable, leading to%
\begin{align}
\tilde{\alpha}_{L}(c)  &  \coloneqq\Big\{\sup_{\lambda,\mu\geq0,\theta
_{1},\theta_{2}\in\Theta}\lambda a^{T}r(\theta_{1})\nonumber\\
&  \qquad\qquad-c\left\Vert b-\lambda\phi r(\theta_{1})-\mu s(\theta
_{2})\right\Vert _{2}^{2}\Big\},\label{eq:param-primal-LP}\\
\tilde{\beta}_{L}(c)  &  \coloneqq\Big\{\inf_{\kappa,\nu\geq0,\theta
_{3},\theta_{4}\in\Theta}\kappa b^{T}t(\theta_{3})\nonumber\\
&  \qquad\qquad+c\left\Vert \kappa\phi^{T}t(\theta_{3})-a-\nu w(\theta
_{4})\right\Vert _{2}^{2}\Big\}, \label{eq:param-dual-LP}%
\end{align}
where $\Theta$ is a general set of all possible parameter vectors, $\theta
_{1}$, $\theta_{2}$, $\theta_{3}$, and $\theta_{4}$ are vectors of parameter
values, and $r(\theta_{1})$, $s(\theta_{2})$, $t(\theta_{3})$, and
$w(\theta_{4})$ are parameterized probability distributions.
Appendix~\ref{app:convex-combo-ansatz} discusses two methods for
parameterizing the set of probability distributions, one based on
neural-network generative models \cite{Ackley1985,bengio2013estimating,mohamed2020monte} and another based on Born machines
\cite{Han2018,Cheng2018,Benedetti2019}, specifically, quantum circuit Born
machines \cite{Benedetti2019}. In the latter case, i.e., if a quantum generative model is used, it is worth noting that CSlack is still a quantum algorithm.  A key aspect of all these methods is that
one can efficiently sample from these parameterized probability distributions. 

One could also consider using an explicit model to train the model probability distribution directly. Examples of such models include auto-regressive models~\cite{Pixel_RNN}, RNNs~\cite{RNN}, tensor networks without loops (which includes tensor network Born machines)~\cite{Han2018,Cheng2019TTNBM}. While these models are in general less expressive they can be easier to train~\cite{rudolph2023trainability}. For the sake of brevity, we will not explore this possibility further here.

A benefit of the optimization problems in
\eqref{eq:param-primal-LP}--\eqref{eq:param-dual-LP} is that all expressions
in the objective functions in
\eqref{eq:param-primal-LP}--\eqref{eq:param-dual-LP} are efficiently estimable
by sampling. Additionally, the following inequalities hold:%
\begin{align}
\alpha_{L}(c)  &  \geq\tilde{\alpha}_{L}(c),\label{eq:param-optimal-primal-LP}%
\\
\beta(c)  &  \leq\tilde{\beta}_{L}(c), \label{eq:param-optimal-dual-LP}%
\end{align}
because the set of parameterized probability distributions is a strict subset
of all possible probability distributions. Then the ability to optimize
$\tilde{\alpha}_{L}(c)$ and $\tilde{\beta}_{L}(c)$ for each $c>0$ implies the
following theorem:
\begin{theorem}
The following inequalities hold
    \begin{equation}
\tilde{\alpha}_{L}\leq\alpha_{L}\leq\beta_{L}\leq\tilde{\beta}_{L},
\label{eq:LP-lower-upper-bnds-param}%
\end{equation}
where $\alpha_{L}$ and $\beta_{L}$ are defined in \eqref{eq:primal-LP}--\eqref{eq:dual-LP},
\begin{align}
\tilde{\alpha}_{L}  &  \coloneqq\sup_{\substack{\lambda,\mu\geq0,\\\theta
_{1},\theta_{2}\in\Theta}}\left\{  \lambda a^{T}r(\theta_{1}):b-\lambda\phi
r(\theta_{1})=\mu s(\theta_{2})\right\}  ,\\
\tilde{\beta}_{L}  &  \coloneqq\inf_{\substack{\kappa,\nu\geq0,\\\theta
_{3},\theta_{4}\in\Theta}}\left\{  \lambda b^{T}t(\theta_{3}):\kappa\phi
^{T}t(\theta_{3})-a=\nu w(\theta_{4})\right\}  ,
\end{align}
$a$ is a $2^n \times 1$ real vector, $b$ is a $2^m \times 1 $ real vector, $\phi$ is a $2^m \times 2^n$ real matrix, and $\Theta$ is a general set of parameter vectors.
\end{theorem}

\begin{proof}
    The conclusion in \eqref{eq:LP-lower-upper-bnds-param} follows from
\eqref{eq:weak-duality-LP}, \eqref{eq:limits-penalty-LP},
\eqref{eq:param-optimal-primal-LP}--\eqref{eq:param-optimal-dual-LP}, and
$\lim_{c\rightarrow\infty}\tilde{\alpha}_{L}(c)=\tilde{\alpha}_{L}$ and
$\lim_{c\rightarrow\infty}\tilde{\beta}_{L}(c)=\tilde{\beta}_{L}$, which both
follow from \cite[Proposition~5.2.1]{Bertsekas2016}.
\end{proof}

\bigskip

Mirroring Eq.~\eqref{eq:main-theory-result}\ for QSlack,
Eq.~\eqref{eq:LP-lower-upper-bnds-param} is a key theoretical insight for the
CSlack method. Under the assumptions that $a$, $b$, and $\phi$ correspond to
efficiently measurable observable vectors, in principle, we can sandwich the
optimal values of the LPs in \eqref{eq:primal-LP}--\eqref{eq:dual-LP} by
optimization problems whose objective functions are efficiently estimable. One
could then attempt to optimize $\tilde{\alpha}_{L}(c)$ and $\tilde{\beta}%
_{L}(c)$ for a sequence of penalties, following the standard approach of the
penalty method.

As with QSlack and other optimization problems involving generative models,
the main drawback of CSlack is that the optimizations in
\eqref{eq:primal-prob-dist-opt}--\eqref{eq:dual-prob-dist-opt} over the convex
set of probability distributions are now replaced with the optimizations in
\eqref{eq:param-primal-LP}--\eqref{eq:param-dual-LP} over the non-convex set
of parameterized probability distributions. Given this, even though the
optimizations are over fewer parameters, their landscapes feature local minima that can make convergence to global optima difficult.

To optimize the objective function in \eqref{eq:param-dual-LP}, we employ an
approach that involves using a probabilistic or quantum computer to evaluate quantities
like $b^{T}t(\theta_{3})$, by means of sampling and repetition, as well as
other terms in the objective function like $a^{T}w(\theta_{4})$, which arises
as one of the six terms after expanding $\left\Vert \kappa\phi^{T}t(\theta
_{3})-a-\nu w(\theta_{4})\right\Vert _{2}^{2}$. After doing so, we can employ
gradient descent or other related algorithms to determine a new choice of the
parameters $\theta_{3}$ and $\theta_{4}$; then we iterate this process until
either convergence occurs or after a specified maximum number of iterations.
One can follow a similar route for the optimization problem in \eqref{eq:param-primal-LP}.

In order to execute the CSlack method, we also require estimating the gradient
vectors of the objective functions in
\eqref{eq:param-primal-LP}--\eqref{eq:param-dual-LP}. For a quantum circuit
Born machine, one can again do so by means of a parameter-shift rule
\cite{Liu2018}. For generative models, one can use the simultaneous perturbation stochastic approximation \cite{Spall1992} to estimate
gradients.

In summary, the main idea of CSlack is similar to that of QSlack:\ replace the
original LPs in \eqref{eq:primal-LP}--\eqref{eq:dual-LP} with a sequence of
optimizations of the form in \eqref{eq:param-primal-LP} and
\eqref{eq:param-dual-LP}, respectively. The main advantage of CSlack is that
every term in the objective functions of \eqref{eq:param-primal-LP} and
\eqref{eq:param-dual-LP}, as well as their gradients, are efficiently
estimable on quantum or probabilistic computers. Furthermore, the inequalities
in~\eqref{eq:LP-lower-upper-bnds-param} provide theoretical guarantees that,
if one can calculate the optimal values in \eqref{eq:param-primal-LP} and
\eqref{eq:param-dual-LP} for each $c>0$, then the resulting quantities
$\tilde{\alpha}_{L}$ and $\tilde{\beta}_{L}$ are certified bounds on the true
optimal values $\alpha_{L}$ and $\beta_{L}$, thus sandwiching the true optimal
values. The main drawback of CSlack is that the reformulation involves an
optimization over the non-convex set of parameterized probability
distributions, but this aspect is common to all classical variational or
neural network methods.

\subsection{CSlack example problems}

\label{sec:CSlack-examples}

In this section, we present a variety of CSlack example problems relevant for information theory and physics,
including the following:

\begin{enumerate}
    \item total variation distance, the classical version of the normalized trace distance,

    \item and constrained classical Hamiltonian optimization, in which the Hamiltonian and constraints are all diagonal in the computational basis.
\end{enumerate}

\noindent For problem 1, we assume that one has
sample access to the probability distribution, so that the input model here is the linear
combination of distributions input model, as described in Appendix~\ref{app:linear-combo-dist-model}. That is, there is some procedure, whether it be by a probabilistic circuit or
other means, that prepares the distributions and can be repeated in such a way that
the same state is prepared each time. For example, if the distribution to be prepared
is $p$, we assume that there is a procedure to prepare $p^{\otimes n}$
for $n\in\mathbb{N}$ arbitrarily large. For problem 2, the  Walsh--Hadamard input model is quite natural and so we consider this in our example below (see Appendix~\ref{app:Walsh-Hadamard-input-model} for more details).

\subsubsection{Total variation distance}

We now move on to some example problems for CSlack, beginning with the total
variation distance of two probability distributions. This is the classical
version of the problem studied in Section~\ref{sec:normalized-TD-example}, and
as such, it is obtained easily from it by requiring that all of the density
matrices therein are diagonal. Nevertheless, we briefly summarize the problem
here for completeness.

For probability distributions $p$ and $q$ over $n$~bits (i.e., $2^{n}%
$-dimensional probability vectors), the total variation distance is defined as
follows:%
\begin{align}
\frac{1}{2}\left\Vert p-q\right\Vert _{1}  &  =\sup_{t\geq0}\left\{
t^{T}\left(  p-q\right)  :t\leq\boldsymbol{1}\right\}  , \label{eq:primal-TVD}%
\\
&  =\inf_{y\geq0}\left\{  \boldsymbol{1}^{T}y:y\geq p-q\right\}  ,
\label{eq:dual-TVD}%
\end{align}
where $\left\Vert \cdot\right\Vert _{1}$ denotes the $\ell_{1}$ vector norm,
$t$ and $y$ are $2^{n}$-dimensional vectors, $\boldsymbol{1}$ is a $2^{n}%
$-dimensional vector of all ones, and the inequalities have the meaning
discussed after~\eqref{eq:primal-LP}. The primal and dual SDPs in
\eqref{eq:primal-TVD}--\eqref{eq:dual-TVD} follow directly by plugging in
diagonal density matrices into \eqref{eq:primal-TD}--\eqref{eq:dual-TD}. Since
$p-q$ appears in the primal and dual optimizations, the input model for this
problem is the linear combination of distributions model, as discussed in
Appendix~\ref{app:linear-combo-dist-model}. Since the optimizations above are
over $2^{n}$-dimensional vectors $t$ and $y$, subject to the constraints
above, it is not possible to estimate $\frac{1}{2}\left\Vert p-q\right\Vert
_{1}$ efficiently using standard LP solvers.

If $p$ and $q$ are prepared by circuits with access to randomness, it is known
that a decision problem related to estimating their total variation distance
is SZK-complete \cite{Sahai1997}, which indicates that the worst-case
complexity of this problem is considered intractable for a probabilistic
classical computer (and also for a quantum computer). Regardless, we can
attempt to estimate this quantity by means of the CSlack method.

Focusing on the dual optimization in \eqref{eq:dual-TVD}, we can rewrite it
exactly as follows:%
\begin{multline}
\inf_{y\geq0}\left\{  \boldsymbol{1}^{T}y:y\geq p-q\right\}
\label{eq:TVD-dual-constr-l2}\\
=\lim_{c\rightarrow\infty}\inf_{\substack{\lambda,\mu\geq0,\\r,s\in
\mathcal{P}}}\left\{  \lambda+c\left\Vert \lambda r-p+q-\mu s\right\Vert
_{2}^{2}\right\}  .
\end{multline}
This follows directly from plugging diagonal density matrices into
\eqref{eq:TD-dual-rewrite-QSlack}. As outlined in
\eqref{eq:param-primal-LP}--\eqref{eq:param-dual-LP}, we replace the
optimizations over all probability distributions $r,s\in\mathcal{P}$ with
optimizations over parameterized distributions, using either a generative
model or a quantum circuit Born machine. The $\ell_{2}$ norm in
\eqref{eq:TVD-dual-constr-l2} can be expanded into ten different overlap
terms, and each of them can be estimated efficiently using the collision test
(i.e., classical version of swap test), reviewed in
Appendix~\ref{app:linear-combo-dist-model}.

\subsubsection{Constrained classical Hamiltonian optimization}

The constrained classical Hamiltonian optimization problem is similar to that
detailed in Section~\ref{sec:constrained-Ham-opt}, with the exception that all
matrices involved are diagonal, and as such, they can be represented as
classical vectors. For completeness, we go through the problem briefly here.

The inputs to this problem are a real Hamiltonian vector~$h$ and $\ell$ real
constraint vectors $a_{1}, \ldots, a_{\ell}$, as well as real constraint
numbers $b_{1}, \ldots, b_{\ell}$. We further suppose that $h$, $a_{1},
\ldots, a_{\ell}$ are efficiently measurable vectors. Then the classical
constrained Hamiltonian optimization problem is as follows:%
\begin{align}
&  \mathcal{L}(h,a_{1},\ldots,a_{\ell})\nonumber\\
&  \coloneqq\inf_{p\in\mathcal{P}}\left\{  h^{T}p:a_{i}^{T}p\geq b_{i}\text{
}\forall i\in\left[  \ell\right]  \right\} \\
&  =\sup_{\substack{y_{1},\ldots,y_{\ell}\geq0,\\\mu\in\mathbb{R}}}\left\{
\sum_{i=1}^{\ell}b_{i}y_{i}+\mu:\sum_{i=1}^{\ell}y_{i}a_{i}+\mu\boldsymbol{1}%
\leq h\right\}  , \label{eq:constrained-Ham-opt-dual-LP}%
\end{align}
where we have written the dual LP in \eqref{eq:constrained-Ham-opt-dual-LP}.
Furthermore, the classical ground-state energy problem is a special case with
$a_{i}=0$ and $b_{i}=0$ for all $i\in\left[  \ell\right]  $, leading to%
\begin{align}
\mathcal{L}(h)  &  \coloneqq\inf_{p\in\mathcal{P}}\left\{  h^{T}p\right\}
\nonumber\\
&  =\sup_{\mu\in\mathbb{R}}\left\{  \mu:\mu\boldsymbol{1}\leq h\right\}  .
\end{align}

A natural input model for this problem is the Walsh--Hadamard input model,
described further in Appendix~\ref{app:Walsh-Hadamard-input-model}, such that%
\begin{align}
h  &  =\sum_{\overrightarrow{x}}h_{\overrightarrow{x}}s_{\overrightarrow{x}%
},\\
a_{i}  &  =\sum_{\overrightarrow{x}}a_{\overrightarrow{x}}^{i}%
s_{\overrightarrow{x}}\qquad\forall i\in\left[  \ell\right]  ,
\end{align}
where $h_{\overrightarrow{x}},a_{\overrightarrow{x}}^{1},\ldots
,a_{\overrightarrow{x}}^{\ell}\in\mathbb{R}$ for all $\overrightarrow{x}$ and
$s_{\overrightarrow{x}}$ is a Walsh--Hadamard vector, i.e.,%
\begin{align}
&  s_{\overrightarrow{x}}\coloneqq s_{x_{1}}\otimes\cdots\otimes s_{x_{n}},\\
&  s_{0}\coloneqq%
\begin{bmatrix}
1\\
1
\end{bmatrix}
,\qquad s_{1}\coloneqq%
\begin{bmatrix}
1\\
-1
\end{bmatrix}
.
\end{align}
Observe that $s_{0}$ and $s_{1}$ are the diagonal entries of the diagonal
Pauli matrices $\sigma_{I}$ and $\sigma_{Z}$. Using this approach, we can
again rewrite the primal and dual optimization problems as in
\eqref{eq:obj-func-primal-constr-Ham} and \eqref{eq:obj-func-dual-constr-Ham},
respectively, with the only substitutions being as follows:%
\begin{align}
\operatorname{Tr}[\sigma_{\overrightarrow{x}}\rho]  &  \rightarrow
s_{\overrightarrow{x}}^{T}p,\\
\operatorname{Tr}[\omega^{2}]  &  \rightarrow w^{T}w,\\
\operatorname{Tr}[\sigma_{\overrightarrow{x}}\omega]  &  \rightarrow
s_{\overrightarrow{x}}^{T}w,
\end{align}
where $w$ is a probability distribution. Each of the inner products above can
be estimated by sampling, as discussed in
Appendix~\ref{app:Walsh-Hadamard-input-model}.

In general, the vectors $\left(  h_{\overrightarrow{x}}\right)
_{\overrightarrow{x}},\left(  a_{\overrightarrow{x}}^{1}\right)
_{\overrightarrow{x}},\ldots,\left(  a_{\overrightarrow{x}}^{\ell}\right)
_{\overrightarrow{x}}$ each have $2^{n}$ coefficients. However, for problems
of physical interest, in which the Hamiltonians and constraints consist of
few-body interactions, there are only poly$(n)$ non-zero coefficients in these
vectors. In order to be able to evaluate the modified objective functions
efficiently, we restrict the optimizations over the probability distributions
$p$ and $w$ to be over parameterized probability distributions.

\subsection{CSlack simulations}

\label{sec:cslack-sims}

In this section, we report the results of simulations of the CSlack algorithm for the example problems from Section~\ref{sec:CSlack-examples}. We first discuss the  features common to all experiments and then delve into specifics for each example.

\subsubsection{Input models}

Let us begin by discussing the input models to the CSlack example problems. For total variation distance, the inputs to the problem are two probability distributions. These distributions are generated using distinct two-layer quantum circuit Born machines with randomly set parameters (see Appendix~\ref{app:convex-combo-ansatz}). The input to the constrained classical Hamiltonian optimization is a vector operator, decomposed in the Walsh--Hadamard basis (see Appendix~\ref{app:Walsh-Hadamard-input-model}).

\subsubsection{Training}

Similar to QSlack, at each iteration, the training process crucially depends on gradient estimation to pick the next set of parameters. We use the SPSA method to produce an unbiased estimator of the gradient with runtime constant in the number of parameters. Details of the method can be found in Appendix~\ref{sec:gradient-CCA}. We use a maximum number of iterations as the stopping condition for all simulations.

Similarly, the hyperparameters, like the learning rate and perturbation parameter, are tailored to each problem instance. We pick the smallest possible penalty parameter $c$ that suffices to enforce the constraints of the problem. This leads to faster convergence to the optimal value in practice. Details and specifics are discussed in Appendix~\ref{sec:details-simulations}.

Similar to the QSlack simulations, the initial parameters of all quantum circuit Born machines are chosen uniformly at random from $[0,2\pi]$. For the total variation distance problem, the parameters $\lambda$ and $\mu$ in  \eqref{eq:TVD-dual-constr-l2} are initialized to one. For the constrained classical Hamiltonian optimization, parameters $y_{1},\ldots,y_{\ell}$ and $\mu$ in \eqref{eq:constrained-Ham-opt-dual-LP} are all initialized to zero.

\subsubsection{Results}

We used a noiseless simulator for the experiments in this work. The collision test involves sampling, leading to the estimate being affected by shot noise. We leave the simulations of the different problems with other noise models to future work.

The plots for the various examples using both ansatz types are given in Figures~\ref{fig:tvd-sandwich} and \ref{fig:cch-sandwich}. In these plots, we show the median (in solid lines) and interquartile range (in shading) of the objective function values over the course of training. The variation in objective function values across different runs corresponds to three sources of randomness: the randomized initializations of all parametrized quantum circuits, the inherent randomness of the SPSA method for gradient estimation, and shot noise in quantum circuit measurements. For all problems, we took $10^{12}$ shots, and the error achieved was typically on the order of $10^{-2}$. 

\begin{figure}[t]
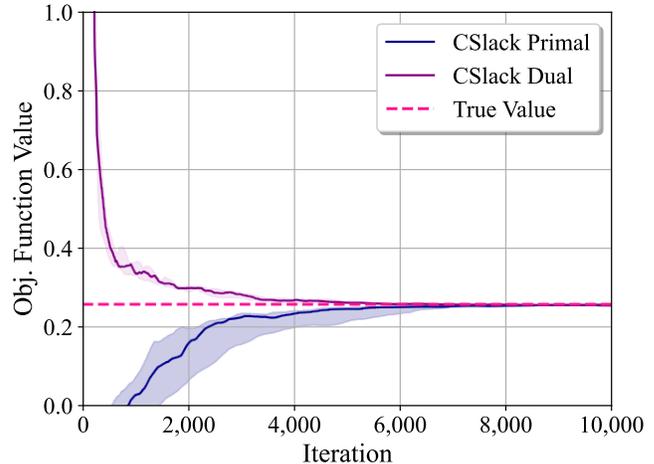

  \centering
  \tvdsandwich
  \caption{Convergence of the CSlack method for total variation distance, from both above and below across ten runs. Specific details regarding the runs, including the number of qubits, layers, gradient method, etc., can be found in Appendix~\ref{sec:details-simulations}.}
  \label{fig:tvd-sandwich}
\end{figure}

\begin{figure}[t]
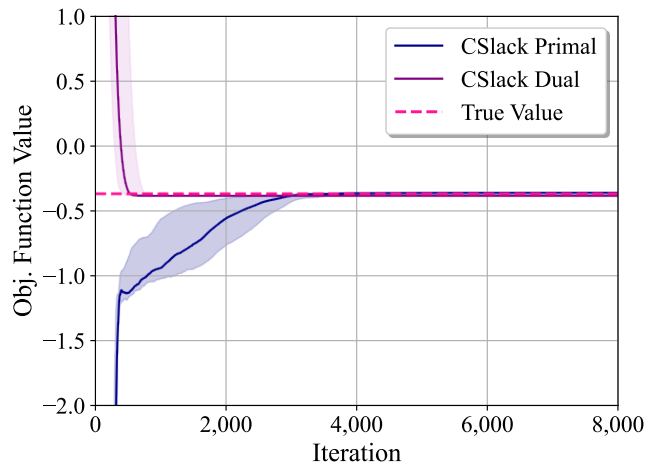

  \centering
  \chsandwich
  \caption{Convergence of the CSlack method for the constrained classical Hamiltonian optimization, from both above and below across ten runs. Specific details regarding the runs, including the particular Hamiltonian choice, gradient method, etc., can be found in Appendix~\ref{sec:details-simulations}.}
  \label{fig:cch-sandwich}
\end{figure}

\section{Conclusion}

\label{sec:conclusion}

In summary, we proposed the QSlack and CSlack methods as general approaches
for estimating upper and lower bounds on the optimal values of semi-definite
and linear programs, respectively. The methods consist of the following
steps: 1)\ replace inequality constraints with equality constraints via the
introduction of slack variables, 2)\ replace a constrained optimization with
an unconstrained one by means of the penalty method, using the
Hilbert--Schmidt distance for SDPs and the Euclidean distance for LPs, 3) replace optimizations
over non-negative variables with either scaled density matrices or scaled
probability distributions, and 4)\ replace optimizations over density matrices
or probability distributions with parameterized ones. Then we estimate all
terms in the objective functions by means of sampling, using a quantum
computer in the SDP\ case and a probabilistic or quantum computer for the LP
case. If it is possible to optimize the objective functions, then it follows
that the QSlack and CSlack methods give certified upper and lower bounds on
the true optimal values in the limit $c \rightarrow \infty$. 

We considered a variety of example problems in Section~\ref{sec:qslack-examples},
which showcase the variety of problems and input models that the QSlack and
CSlack methods can handle. Examples for QSlack include estimating normalized
trace distance, root fidelity, and entanglement negativity, problems of
interest in quantum information science, as well as constrained Hamiltonian
optimization, a problem of interest in physics and chemistry. Examples for
CSlack include estimating total variation distance and classical constrained
Hamiltonian optimization. In our numerical simulations, we found that both the primal and dual optimizations for each problem approach the true value, achieving errors typically on the order of $10^{-2}$.

Going forward from here, there a number of open directions to consider. First, the behavior of QSlack for finite $c$ values warrants a more detailed study. Our numerics suggest that in general QSlack provides genuine bounds even for finite $c$; however, as expected, occasional deviations are observed. It would be valuable to investigate whether it is possible to obtain bounds on the degree to which deviations are possible for finite $c$. Such bounds would allow one to certify the degree of confidence one may have in the outputs of QSlack and hence guide one's choice for the value of $c$.

Longer term,
we would like to scale up our simulations and examples to include many more
qubits, in order to have a sense of the scaling of QSlack and CSlack. Indeed,
these methods are intended to apply to large-scale optimization problems,
beyond the reach of what is possible classically. In scaling up, it will be
necessary to incorporate error mitigation \cite{cai2023quantum}, in order to
reduce the impact of noise. 
Additionally, when scaling up, it is inevitable
that the issue of barren plateaus \cite{McClean_2018}\ will come into play, in
which higher depth quantum circuits are expected to have vanishing gradients and vanishing cost differences~\cite{arrasmith2022equivalence}.
This occurs for highly expressive~\cite{holmes2022connecting} or highly entangling~\cite{marrero2021entanglement} ans\"atze and even at low depths for global costs~\cite{cerezo2021cost}. In Appendix~\ref{appendix:BarrenPlateaus}, we discuss the cases for which QSlack can or cannot potentially avoid barren plateaus.


We note that while we have here largely presented QSlack as a variational quantum algorithm, the parametrized unitary could in theory take other forms. For example, for low entangling models, one could use a tensor-network based optimization. Or, more generally, one could envision complementing classical simulation methods supported with data obtained from a quantum computer~\cite{basheer2022alternating, jerbi2023power, rudolph2023classical}. This would be an interesting avenue to explore should the barrier presented by barren plateaus prove insurmountable.

As
another open direction, we also wonder whether QSlack could be modified generally to
become an interior-point method, which is the standard approach to solving
SDPs and LPs on classical computers (we observed one case in~\eqref{eq:constr-opt-inter-point} in which it can be). The main question here is whether it is
possible to rewrite the penalty terms in
\eqref{eq:primal-q-states}--\eqref{eq:dual-q-states} as self-concordant
barrier functions \cite[Section~5.4.6]{nesterov2018lectures}\ that could be
efficiently estimated on quantum computers. The main advantage of QSlack is
that all terms in the objective functions in
\eqref{eq:primal-q-states}--\eqref{eq:dual-q-states} can be evaluated
efficiently on quantum computers, but the Hilbert--Schmidt norm is not a
self-concordant barrier function. The logarithm of the determinant of a matrix is known to be a self-concordant barrier function for the cone of positive semi-definite matrices. As such, finding a way to generalize the classical sampling approach of \cite{Boutsidis2017} to the quantum case would be helpful in modifying QSlack to become an interior-point method.

\textit{Data availability statement}---All software codes used to run the simulations and generate the figures in this paper and in our companion paper \cite{West2023dualVQE} are available as arXiv ancillary files with the arXiv posting of this paper.

\begin{acknowledgments}
We thank Paul Alsing, Ziv Goldfeld, Daniel Koch, Saahil Patel, and Manuel
S.~Rudolph for helpful discussions. JC and HW acknowledge support from the
Engineering Learning Initiative in Cornell University's College of
Engineering. ZH acknowledges support from the Sandoz Family Foundation Monique
de Meuron program for Academic Promotion. IL, TN, DP, SR, and MMW acknowledge
support from the School of Electrical and Computer Engineering at Cornell
University. TN, DP, SR, and MMW acknowledge support from the National Science
Foundation under Grant No.~2315398. DP, SR, and MMW acknowledge support from
AFRL under agreement no.~FA8750-23-2-0031.

This material is based on research
sponsored by Air Force Research Laboratory under agreement number
FA8750-23-2-0031. The U.S.~Government is authorized to reproduced and
distribute reprints for Governmental purposes notwithstanding any copyright
notation thereon. The views and conclusions contained herein are those of the
authors and should not be interpreted as necessarily representing the official
policies or endorsements, either expressed or implied, of Air Force Research
Laboratory or the U.S.~Government.

This research was conducted with support from the Cornell University Center for Advanced Computing, which receives funding from Cornell University, the National Science Foundation, and members of its Partner Program.
\end{acknowledgments}

\bibliographystyle{halpha}
\bibliography{Ref}

\onecolumngrid

\newpage\newgeometry{margin=1.7in,tmargin=.8in,bmargin=.9in}

\appendix

\section{Alternative penalty method optimizations}

\label{app:alt-penalty-methods}In
\eqref{eq:primal-SDP-unconstrained}--\eqref{eq:dual-SDP-unconstrained}, we
proposed unconstrained optimization problems $\alpha(c)$ and $\beta(c)$ that
each converge to $\alpha$ and $\beta$, respectively, in the limit as
$c\rightarrow\infty$. While these quantities are the main focus of our paper,
here we propose some variants of these quantities that still have the property
that they converge to $\alpha$ and $\beta$ in the limit as $c\rightarrow
\infty$, while they are also concave and convex optimization problems,
respectively, and obey a weak-duality inequality. They are defined as follows
for $c>0$, $p_{1},p_{2}\geq1$, with $q_{i}$ chosen such that $\frac{1}{p_{i}%
}+\frac{1}{q_{i}}=1$, for $i\in\left\{  1,2\right\}  $:%
\begin{align}
\alpha_{p_{1},q_{2}}^{\prime}(c)  &  \coloneqq\sup_{X,W\geq0}\left\{
\operatorname{Tr}[AX]-c\left\Vert B-\Phi(X)-W\right\Vert _{p_{1}}:\left\Vert
X\right\Vert _{q_{2}}\leq c\right\}  ,\label{eq:alpha-prime-SDP}\\
\beta_{p_{2},q_{1}}^{\prime}(c)  &  \coloneqq\inf_{Y,Z\geq0}\left\{
\operatorname{Tr}\!\left[  BY\right]  +c\left\Vert \Phi^{\dag}%
(Y)-A-Z\right\Vert _{p_{2}}:\left\Vert Y\right\Vert _{q_{1}}\leq c\right\}  ,
\label{eq:beta-prime-SDP}%
\end{align}
where the Schatten $p$-norm of an operator $M$ is defined for $p\geq1$ as%
\begin{equation}
\left\Vert M\right\Vert _{p}\coloneqq\left(  \operatorname{Tr}\!\left[
\left\vert M\right\vert ^{p}\right]  \right)  ^{1/p},
\end{equation}
with $\left\vert M\right\vert \coloneqq\sqrt{M^{\dag}M}$. The main difference
between these quantities and those in
\eqref{eq:primal-SDP-unconstrained}--\eqref{eq:dual-SDP-unconstrained} are
that their penalty terms do not feature the square and they have the
additional constraints $\left\Vert X\right\Vert _{q_{2}}\leq c$ and
$\left\Vert Y\right\Vert _{q_{1}}\leq c$. As an example, we could choose
$p_{i}=1$ and $q_{i}=\infty$ for $i\in\left\{  1,2\right\}  $, leading to the
optimizations:%
\begin{align}
\alpha_{1,\infty}^{\prime}(c)  &  \coloneqq\sup_{X,W\geq0}\left\{
\operatorname{Tr}[AX]-c\left\Vert B-\Phi(X)-W\right\Vert _{1}:\left\Vert
X\right\Vert _{\infty}\leq c\right\}  ,\\
\beta_{1,\infty}^{\prime}(c)  &  \coloneqq\inf_{Y,Z\geq0}\left\{
\operatorname{Tr}\!\left[  BY\right]  +c\left\Vert \Phi^{\dag}%
(Y)-A-Z\right\Vert _{1}:\left\Vert Y\right\Vert _{\infty}\leq c\right\}  ,
\end{align}
which are rather natural choices if $X$ and $Y$ are observables and $W$ and
$Z$ are states, along with the linear combination of states model from
Appendix~\ref{app:linear-combo-states}. Another natural choice, strongly
related to the choices made in the main text in
\eqref{eq:primal-SDP-unconstrained}--\eqref{eq:dual-SDP-unconstrained}, is
when $p_{i}=q_{i}=2$ for $i\in\left\{  1,2\right\}  $, leading to%
\begin{align}
\alpha_{2,2}^{\prime}(c)  &  \coloneqq\sup_{X,W\geq0}\left\{
\operatorname{Tr}[AX]-c\left\Vert B-\Phi(X)-W\right\Vert _{2}:\left\Vert
X\right\Vert _{2}\leq c\right\}  ,\\
\beta_{2,2}^{\prime}(c)  &  \coloneqq\inf_{Y,Z\geq0}\left\{  \operatorname{Tr}%
\!\left[  BY\right]  +c\left\Vert \Phi^{\dag}(Y)-A-Z\right\Vert _{2}%
:\left\Vert Y\right\Vert _{2}\leq c\right\}  ,
\end{align}
such that every operator and constraint are measured with respect to the
Hilbert--Schmidt norm or distance.

\begin{proposition}
Let $p_{1},p_{2}\geq1$, with $q_{i}$ chosen such that $\frac{1}{p_{i}}%
+\frac{1}{q_{i}}=1$, for $i\in\left\{  1,2\right\}  $. Then the optimization
problem $\alpha_{p_{1},q_{2}}^{\prime}(c)$ is concave, and the optimization
problem $\beta_{p_{2},q_{1}}^{\prime}(c)$ is convex.
\end{proposition}

\begin{proof}
Let $X_{i},W_{i}\geq0$ be such that $\left\Vert X_{i}\right\Vert _{q_{2}}\leq
c$ for $i\in\left\{  0,1\right\}  $, and let%
\begin{align}
\lambda &  \in\left[  0,1\right]  ,\\
X_{\lambda}  &  \coloneqq\left(  1-\lambda\right)  X_{0}+\lambda X_{1},\\
W_{\lambda}  &  \coloneqq\left(  1-\lambda\right)  W_{0}+\lambda W_{1}.
\end{align}
Then%
\begin{multline}
\operatorname{Tr}[AX_{\lambda}]-c\left\Vert B-\Phi(X_{\lambda})-W_{\lambda
}\right\Vert _{p_{1}}\geq\left(  1-\lambda\right)  \left(  \operatorname{Tr}%
[AX_{0}]-c\left\Vert B-\Phi(X_{0})-W_{0}\right\Vert _{p_{1}}\right) \\
+\lambda\left(  \operatorname{Tr}[AX_{1}]-c\left\Vert B-\Phi(X_{1}%
)-W_{1}\right\Vert _{p_{1}}\right)  ,
\end{multline}
which follows from linearity of the superoperator $\Phi$ and convexity of the
norm $\left\Vert \cdot\right\Vert _{p_{1}}$. Furthermore, due to convexity of
the norm $\left\Vert \cdot\right\Vert _{q_{2}}$, it follows that $\left\Vert
X_{\lambda}\right\Vert _{q_{2}}\leq c$.\ Thus, the claim about $\alpha
_{p_{1},q_{2}}^{\prime}(c)$ follows. A similar argument implies that
$\beta_{p_{2},q_{1}}^{\prime}(c)$ is a convex optimization problem.
\end{proof}

\medskip

An additional feature of these optimization problems is that the weak-duality
inequality $\alpha_{p_{1},q_{2}}^{\prime}(c)\leq\beta_{p_{2},q_{1}}^{\prime
}(c)$ holds for every $c>0$. To establish this claim, we begin with the
following lemma:

\begin{lemma}
Let $X,W,Y,Z\geq0$, let $A$ and $B$ be Hermitian matrices, and let $\Phi$ be a
Hermiticity-preserving superoperator. Then%
\begin{equation}
\operatorname{Tr}[AX]-\left\Vert Y\right\Vert _{q_{1}}\left\Vert
B-\Phi(X)-W\right\Vert _{p_{1}}\leq\operatorname{Tr}\!\left[  BY\right]
+\left\Vert X\right\Vert _{q_{2}}\left\Vert \Phi^{\dag}(Y)-A-Z\right\Vert
_{p_{2}},
\end{equation}
where $p_{1},p_{2}\geq1$ and $q_{i}$ is chosen such that $\frac{1}{p_{i}%
}+\frac{1}{q_{i}}=1$, for $i\in\left\{  1,2\right\}  $.
\end{lemma}

\begin{proof}
The idea is to proceed by means of an approximate version of the standard
proof of the weak-duality inequality (see, e.g., \cite[Proposition~2.21]%
{khatri2020principles}). Set%
\begin{align}
\varepsilon_{1}(X,W)  &  \coloneqq\left\Vert B-\Phi(X)-W\right\Vert _{p_{1}%
},\\
\varepsilon_{2}(Y,Z)  &  \coloneqq\left\Vert \Phi^{\dag}(Y)-A-Z\right\Vert
_{p_{2}}.
\end{align}
Recall the following variational characterizations of the Schatten norms (see,
e.g., \cite[Proposition~2.11]{khatri2020principles}):%
\begin{align}
\varepsilon_{1}(X,W)  &  =\sup_{W^{\prime}}\left\{  \left\vert \left\langle
B-\Phi(X)-W,W^{\prime}\right\rangle \right\vert :\left\Vert W^{\prime
}\right\Vert _{q_{1}}\leq1\right\} \\
&  =\sup_{W^{\prime}}\left\{  \left\vert \left\langle B-W,W^{\prime
}\right\rangle -\left\langle \Phi(X),W^{\prime}\right\rangle \right\vert
:\left\Vert W^{\prime}\right\Vert _{q_{1}}\leq1\right\}
,\label{eq:eps1-var-char}\\
\varepsilon_{2}(Y,Z)  &  =\sup_{Z^{\prime}}\left\{  \left\vert \left\langle
\Phi^{\dag}(Y)-A-Z,Z^{\prime}\right\rangle \right\vert :\left\Vert Z^{\prime
}\right\Vert _{q_{2}}\leq1\right\} \\
&  =\sup_{Z^{\prime}}\left\{  \left\vert \left\langle \Phi^{\dag
}(Y)-Z,Z^{\prime}\right\rangle -\left\langle A,Z^{\prime}\right\rangle
\right\vert :\left\Vert Z^{\prime}\right\Vert _{q_{2}}\leq1\right\}  .
\label{eq:eps2-var-char2}%
\end{align}
Now consider that%
\begin{align}
\operatorname{Tr}[AX]  &  \leq\operatorname{Tr}\!\left[  \left(  \Phi^{\dag
}(Y)-Z\right)  X\right]  +\left\Vert X\right\Vert _{q_{2}}\varepsilon
_{2}(Y,Z)\\
&  \leq\operatorname{Tr}\!\left[  \Phi^{\dag}(Y)X\right]  +\left\Vert
X\right\Vert _{q_{2}}\varepsilon_{2}(Y,Z)\\
&  =\operatorname{Tr}\!\left[  Y\Phi(X)\right]  +\left\Vert X\right\Vert
_{q_{2}}\varepsilon_{2}(Y,Z)\\
&  \leq\operatorname{Tr}\!\left[  Y\left(  B-W\right)  \right]  +\left\Vert
Y\right\Vert _{q_{1}}\varepsilon_{1}(X,W)+\left\Vert X\right\Vert _{q_{2}%
}\varepsilon_{2}(Y,Z)\\
&  \leq\operatorname{Tr}\!\left[  BY\right]  +\left\Vert Y\right\Vert _{q_{1}%
}\varepsilon_{1}(X,W)+\left\Vert X\right\Vert _{q_{2}}\varepsilon_{2}(Y,Z).
\end{align}
The first inequality follows because%
\begin{equation}
\operatorname{Tr}[AZ^{\prime}]\leq\operatorname{Tr}\!\left[  \left(
\Phi^{\dag}(Y)-Z\right)  Z^{\prime}\right]  +\left\Vert Z^{\prime}\right\Vert
_{q_{2}}\varepsilon_{2}(Y,Z)
\end{equation}
for every Hermitian $Z^{\prime}$, as a consequence of
\eqref{eq:eps2-var-char2}. The second inequality follows because
$\operatorname{Tr}[ZX]\geq0$, given that $X,Z\geq0$. The third inequality
follows because%
\begin{equation}
\operatorname{Tr}\!\left[  W^{\prime}\Phi(X)\right]  \leq\operatorname{Tr}%
\!\left[  W^{\prime}\left(  B-W\right)  \right]  +\left\Vert W^{\prime
}\right\Vert _{q_{1}}\varepsilon_{1}(X,W)
\end{equation}
for every Hermitian $W^{\prime}$, as a consequence of
\eqref{eq:eps1-var-char}. The final inequality follows because
$\operatorname{Tr}[YW]\geq0$, since $Y,W\geq0$. This implies that the
following inequality holds for all $X,W,Y,Z\geq0$:%
\begin{equation}
\operatorname{Tr}[AX]-\left\Vert Y\right\Vert _{q_{1}}\varepsilon_{1}%
(X,W)\leq\operatorname{Tr}\!\left[  BY\right]  +\left\Vert X\right\Vert
_{q_{2}}\varepsilon_{2}(Y,Z),
\end{equation}
or equivalently, that the following holds for all $X,W,Y,Z\geq0$:%
\begin{equation}
\operatorname{Tr}[AX]-\left\Vert Y\right\Vert _{q_{1}}\left\Vert
B-\Phi(X)-W\right\Vert _{p_{1}}\leq\operatorname{Tr}\!\left[  BY\right]
+\left\Vert X\right\Vert _{q_{2}}\left\Vert \Phi^{\dag}(Y)-A-Z\right\Vert
_{p_{2}}.
\end{equation}
This concludes the proof.
\end{proof}

\bigskip

Thus, if there is a uniform upper bound $c\geq\max\{ \left\Vert X\right\Vert
_{q_{2}},\left\Vert Y\right\Vert _{q_{1}}\} $ holding for every $X\geq0$ and
$Y\geq0$ under consideration, then the following inequality holds%
\begin{equation}
\operatorname{Tr}[AX]-c\left\Vert B-\Phi(X)-W\right\Vert _{p_{1}}%
\leq\operatorname{Tr}\!\left[  BY\right]  +c\left\Vert \Phi^{\dag
}(Y)-A-Z\right\Vert _{p_{2}},
\end{equation}
That is, we have the following corollary:

\begin{corollary}
Let $A$ and $B$ be Hermitian matrices, and let $\Phi$ be a
Hermiticity-preserving superoperator. Let $p_{1},p_{2}\geq1$ and let $q_{i}$
be chosen such that $\frac{1}{p_{i}}+\frac{1}{q_{i}}=1$, for $i\in\left\{
1,2\right\}  $. Then%
\begin{equation}
\alpha_{p_{1},q_{2}}^{\prime}(c)\leq\beta_{p_{2},q_{1}}^{\prime}(c)
\end{equation}
for all $c>0$, where $\alpha_{p_{1},q_{2}}^{\prime}(c)$ and $\beta
_{p_{2},q_{1}}^{\prime}(c)$ are defined in \eqref{eq:alpha-prime-SDP} and
\eqref{eq:beta-prime-SDP}, respectively.
\end{corollary}

By the same reasoning, all of the statements above apply to the following
optimization problems related to linear programming, which are variants of the
optimizations in
\eqref{eq:primal-LP-unconstrained}--\eqref{eq:dual-LP-unconstrained}:%
\begin{align}
\alpha_{L,p_{1},q_{2}}^{\prime}(c)  &  \coloneqq\sup_{x,w\geq0}\left\{
a^{T}x-c\left\Vert b-\phi x-w\right\Vert _{p_{1}}:\left\Vert x\right\Vert
_{q_{2}}\leq c\right\}  ,\\
\beta_{L,p_{2},q_{1}}^{\prime}(c)  &  \coloneqq\inf_{y,z\geq0}\left\{
b^{T}y+c\left\Vert \phi^{T}y-a-z\right\Vert _{p_{2}}:\left\Vert y\right\Vert
_{q_{1}}\leq c\right\}  .
\end{align}
That is, $\alpha_{L,p_{1},q_{2}}^{\prime}(c)$ is a concave optimization
problem, $\beta_{L,p_{2},q_{1}}^{\prime}(c)$ is a convex optimization problem,
and the weak-duality inequality $\alpha_{L,p_{1},q_{2}}^{\prime}(c)\leq
\beta_{L,p_{2},q_{1}}^{\prime}(c)$ holds for every $c>0$.

\section{QSlack and CSlack methods for semi-definite and linear programs with
inequality and equality constraints}

\label{app:QCSlack-ineq-eq-constraints}

In the main text, we considered semi-definite and linear programs with just
inequality constraints. However, in practice, one might encounter
optimizations that feature both inequality and equality constraints. In this
appendix, we show how to handle this more general case.

Let $A$, $B_{1}$, and $B_{2}$ be Hermitian matrices, and let $\Phi_{1}$ and
$\Phi_{2}$ be Hermiticity-preserving superoperators. Let us suppose that the
primal SDP\ has the following form:%
\begin{equation}
\alpha\coloneqq\sup_{X\geq0}\left\{  \operatorname{Tr}[AX]:\Phi_{1}(X)\leq
B_{1},\, \Phi_{2}(X)=B_{2}\right\}  . \label{eq:primal-SDP-app-ineq-eq}%
\end{equation}
The dual SDP\ is then as follows:%
\begin{equation}
\beta\coloneqq\inf_{Y_{1}\geq0,Y_{2}\in\operatorname{Herm}}\left\{
\operatorname{Tr}[B_{1}Y_{1}]+\operatorname{Tr}[B_{2}Y_{2}]:\Phi_{1}^{\dag
}(Y_{1})+\Phi_{2}^{\dag}(Y_{2})\geq A\right\}  .
\label{eq:dual-SDP-app-ineq-eq}%
\end{equation}
Indeed, we can quickly derive the dual SDP\ as follows:%
\begin{align}
&  \sup_{X\geq0}\left\{  \operatorname{Tr}[AX]:\Phi_{1}(X)\leq B_{1},\,
\Phi_{2}(X)=B_{2}\right\} \nonumber\\
&  =\sup_{X\geq0}\inf_{\substack{Y_{1}\geq0,\\Y_{2}\in\operatorname{Herm}%
}}\left\{  \operatorname{Tr}[AX]+\operatorname{Tr}[Y_{1}\left(  B_{1}-\Phi
_{1}(X)\right)  ]+\operatorname{Tr}[Y_{2}\left(  B_{2}-\Phi_{2}(X)\right)
]\right\} \\
&  =\sup_{X\geq0}\inf_{\substack{Y_{1}\geq0,\\Y_{2}\in\operatorname{Herm}%
}}\left\{  \operatorname{Tr}[B_{1}Y_{1}]+\operatorname{Tr}[B_{2}%
Y_{2}]+\operatorname{Tr}\!\left[  \left(  A-\Phi_{1}^{\dag}(Y_{1})-\Phi
_{2}^{\dag}(Y_{2})\right)  X\right]  \right\} \\
&  \leq\inf_{\substack{Y_{1}\geq0,\\Y_{2}\in\operatorname{Herm}}}\sup_{X\geq
0}\left\{  \operatorname{Tr}[B_{1}Y_{1}]+\operatorname{Tr}[B_{2}%
Y_{2}]+\operatorname{Tr}\!\left[  \left(  A-\Phi_{1}^{\dag}(Y_{1})-\Phi
_{2}^{\dag}(Y_{2})\right)  X\right]  \right\} \\
&  =\inf_{Y_{1}\geq0,Y_{2}\in\operatorname{Herm}}\left\{  \operatorname{Tr}%
[B_{1}Y_{1}]+\operatorname{Tr}[B_{2}Y_{2}]:\Phi_{1}^{\dag}(Y_{1})+\Phi
_{2}^{\dag}(Y_{2})\geq A\right\}  .
\end{align}
The first equality holds by introducing Lagrange multipliers for the two
constraints: indeed, the constraint $\Phi_{1}(X)\leq B_{1}$ does not hold if
and only if $B_{1} - \Phi_{1}(X)$ has a negative eigenvalue, which implies
that $\inf_{Y_{1}\geq0}\operatorname{Tr}[Y_{1}\left(  B_{1}-\Phi
_{1}(X)\right)  ]=-\infty$, and the constraint $\Phi_{2}(X)=B_{2}$ does not
hold if and only if $B_{2}-\Phi_{2}(X) \neq0$, which implies that $\inf
_{Y_{2}\in\operatorname{Herm}}\operatorname{Tr}[Y_{2}\left(  B_{2}-\Phi
_{2}(X)\right)  ]=-\infty$. The second equality follows from algebra and using
the definition of the adjoint superoperator (see
\eqref{eq:adjoint-superop-def}). The inequality follows from the max-min
inequality, and the final equality holds for similar reasons as the first one:
one can think of $X\geq0$ as a Lagrange multiplier, so that the constraint
$\Phi_{1}^{\dag}(Y_{1})+\Phi_{2}^{\dag}(Y_{2})\geq A$ does not hold if and
only if $A-\Phi_{1}^{\dag}(Y_{1})-\Phi_{2}^{\dag}(Y_{2})$ has a positive
eigenvalue, which implies that $\sup_{X\geq0}\operatorname{Tr}\!\left[
\left(  A-\Phi_{1}^{\dag}(Y_{1})-\Phi_{2}^{\dag}(Y_{2})\right)  X\right]
=+\infty$. Let us note that the inequality above is saturated in the case that
strong duality holds.

Let us now discuss the QSlack method for this case. The main observation is
that there is a need to introduce a slack variable only for the inequality
constraint in \eqref{eq:primal-SDP-app-ineq-eq} and no need to do so for the
equality constraint therein. As such, the basic steps of QSlack are as
follows:%
\begin{align}
\alpha &  =\sup_{X\geq0}\left\{  \operatorname{Tr}[AX]:\Phi_{1}(X)\leq
B_{1},\,\Phi_{2}(X)=B_{2}\right\} \\
&  =\sup_{X,W\geq0}\left\{  \operatorname{Tr}[AX]:\Phi_{1}(X)+W=B_{1}%
,\,\Phi_{2}(X)=B_{2}\right\} \\
&  =\sup_{\substack{\lambda,\mu\geq0,\\\rho,\sigma\in\mathcal{D}}}\left\{
\lambda\operatorname{Tr}[A\rho]:B_{1}=\lambda\Phi_{1}(\rho)+\mu\sigma,\,
B_{2}=\lambda\Phi_{2}(\rho)\right\} \\
&  =\lim_{c\rightarrow\infty}\sup_{\substack{\lambda,\mu\geq0,\\\rho,\sigma
\in\mathcal{D}}}\left\{  \lambda\operatorname{Tr}[A\rho]-c\left\Vert
B_{1}-\lambda\Phi_{1}(\rho)-\mu\sigma\right\Vert _{2}^{2}-c\left\Vert
B_{2}-\lambda\Phi_{2}(\rho)\right\Vert _{2}^{2}\right\} \\
&  \geq\lim_{c\rightarrow\infty}\sup_{\substack{\lambda,\mu\geq0,\\\theta
_{1},\theta_{2}\in\Theta}}\left\{
\begin{array}
[c]{c}%
\lambda\operatorname{Tr}[A\rho(\theta_{1})]-c\left\Vert B_{1}-\lambda\Phi
_{1}(\rho(\theta_{1}))-\mu\sigma(\theta_{2})\right\Vert _{2}^{2}\\
-c\left\Vert B_{2}-\lambda\Phi_{2}(\rho(\theta_{1}))\right\Vert _{2}^{2}%
\end{array}
\right\}  .
\end{align}
For the dual SDP in \eqref{eq:dual-SDP-app-ineq-eq}, the main steps are as
follows:%
\begin{align}
\beta &  =\inf_{Y_{1}\geq0,Y_{2}\in\operatorname{Herm}}\left\{
\operatorname{Tr}[B_{1}Y_{1}]+\operatorname{Tr}[B_{2}Y_{2}]:\Phi_{1}^{\dag
}(Y_{1})+\Phi_{2}^{\dag}(Y_{2})\geq A\right\} \\
&  =\inf_{Y_{1},Z\geq0,Y_{2}\in\operatorname{Herm}}\left\{  \operatorname{Tr}%
[B_{1}Y_{1}]+\operatorname{Tr}[B_{2}Y_{2}]:\Phi_{1}^{\dag}(Y_{1})+\Phi
_{2}^{\dag}(Y_{2})=A+Z\right\} \\
&  =\inf_{\substack{\kappa_{1},\kappa_{2,1},\kappa_{2,2},\nu\geq0,\\\tau
_{1},\tau_{2,1},\tau_{2,2},\omega\in\mathcal{D}}}\left\{
\begin{array}
[c]{c}%
\kappa_{1}\operatorname{Tr}[B_{1}\tau_{1}]+\kappa_{2,1}\operatorname{Tr}%
[B_{2}\tau_{2,1}]-\kappa_{2,2}\operatorname{Tr}[B_{2}\tau_{2,2}]:\\
\kappa_{1}\Phi_{1}^{\dag}(\tau_{1})+\kappa_{2,1}\Phi_{2}^{\dag}(\tau
_{2,1})-\kappa_{2,2}\Phi_{2}^{\dag}(\tau_{2,2})=A+\nu\omega
\end{array}
\right\} \\
&  =\lim_{c\rightarrow\infty}\inf_{\substack{\kappa_{1},\kappa_{2,1}%
,\kappa_{2,2},\nu\geq0,\\\tau_{1},\tau_{2,1},\tau_{2,2},\omega\in\mathcal{D}%
}}\left\{
\begin{array}
[c]{c}%
\kappa_{1}\operatorname{Tr}[B_{1}\tau_{1}]+\kappa_{2,1}\operatorname{Tr}%
[B_{2}\tau_{2,1}]-\kappa_{2,2}\operatorname{Tr}[B_{2}\tau_{2,2}]\\
+c\left\Vert \kappa_{1}\Phi_{1}^{\dag}(\tau_{1})+\kappa_{2,1}\Phi_{2}^{\dag
}(\tau_{2,1})-\kappa_{2,2}\Phi_{2}^{\dag}(\tau_{2,2})-A-\nu\omega\right\Vert
_{2}^{2}%
\end{array}
\right\} \\
&  \leq\lim_{c\rightarrow\infty}\inf_{\substack{\kappa_{1},\kappa_{2,1}%
,\kappa_{2,2},\nu\geq0,\\\theta_{1},\theta_{2,1},\theta_{2,2},\varphi
\in\mathcal{D}}}\left\{
\begin{array}
[c]{c}%
\kappa_{1}\operatorname{Tr}[B_{1}\tau_{1}(\theta_{1})]+\kappa_{2,1}%
\operatorname{Tr}[B_{2}\tau_{2,1}(\theta_{2,1})]\\
-\kappa_{2,2}\operatorname{Tr}[B_{2}\tau_{2,2}(\theta_{2,2})]\\
+c\left\Vert
\begin{array}
[c]{c}%
\kappa_{1}\Phi_{1}^{\dag}(\tau_{1}(\theta_{1}))+\kappa_{2,1}\Phi_{2}^{\dag
}(\tau_{2,1}(\theta_{2,1}))\\
-\kappa_{2,2}\Phi_{2}^{\dag}(\tau_{2,2}(\theta_{2,1}))-A-\nu\omega(\varphi)
\end{array}
\right\Vert _{2}^{2}%
\end{array}
\right\}  .
\end{align}
In the above, we made use of the fact that $Y_{2}\in\operatorname{Herm}$ if
and only if there exist $\kappa_{2,1},\kappa_{2,2}\geq0$ and $\tau_{2,1}%
,\tau_{2,2}\in\mathcal{D}$ such that $Y_{2} = \kappa_{2,1}\tau_{2,1} -
\kappa_{2,2}\tau_{2,2}$.

Let us now briefly write down what happens for CSlack. Indeed, this is simply
the diagonal case of the above. For this case, $a$, $b_{1}$, and $b_{2}$ are
real vectors, and $\phi_{1}$ and $\phi_{2}$ are real matrices. Then the primal
and dual LPs are as follows:%
\begin{align}
\alpha_{L}  &  \coloneqq\sup_{x\geq0}\left\{  a^{T}x:\phi_{1}x\leq b_{1},\,
\phi_{2}x=b_{2}\right\}  ,\\
\beta_{L}  &  \coloneqq\inf_{y_{1}\geq0,y_{2}\in\mathbb{R}}\left\{  b_{1}%
^{T}y_{1}+b_{2}^{T}y_{2}:\phi_{1}^{T}y_{1}+\phi_{2}^{\dag}y_{2}\geq a\right\}
.
\end{align}
Then, for CSlack, we find the following for the primal LP:%
\begin{align}
\alpha_{L}  &  =\sup_{\substack{\lambda,\mu\geq0,\\r,s\in\mathcal{P}}}\left\{
\lambda a^{T}r:b_{1}=\lambda\phi_{1}r+\mu s,\ b_{2}=\lambda\phi_{2}r\right\}
\\
&  =\lim_{c\rightarrow\infty}\sup_{\substack{\lambda,\mu\geq0,\\\rho,\sigma
\in\mathcal{D}}}\left\{  \lambda a^{T}r-c\left\Vert b_{1}-\lambda\phi_{1}r-\mu
s\right\Vert _{2}^{2}-c\left\Vert b_{2}-\lambda\phi_{2}r\right\Vert _{2}%
^{2}\right\} \\
&  \geq\lim_{c\rightarrow\infty}\sup_{\substack{\lambda,\mu\geq0,\\\theta
_{1},\theta_{2}\in\Theta}}\left\{
\begin{array}
[c]{c}%
\lambda a^{T}r(\theta_{1})-c\left\Vert b_{1}-\lambda\phi_{1}r(\theta_{1})-\mu
s(\theta_{2})\right\Vert _{2}^{2}\\
-c\left\Vert b_{2}-\lambda\phi_{2}r(\theta_{1})\right\Vert _{2}^{2}%
\end{array}
\right\}  ,
\end{align}
and the following for the dual LP:%
\begin{align}
\beta_{L}  &  =\inf_{y_{1},z\geq0,y_{2}\in\mathbb{R}}\left\{  b_{1}^{T}%
y_{1}+b_{2}^{T}y_{2}:\phi_{1}^{T}y_{1}+\phi_{2}^{T}y_{2}=a+z\right\} \\
&  =\inf_{\substack{\kappa_{1},\kappa_{2,1},\kappa_{2,2},\nu\geq
0,\\t_{1},t_{2,1},t_{2,2},w\in\mathcal{P}}}\left\{
\begin{array}
[c]{c}%
\kappa_{1}b_{1}^{T}t_{1}+\kappa_{2,1}b_{2}^{T}t_{2,1}-\kappa_{2,2}b_{2}%
^{T}t_{2,2}:\\
\kappa_{1}\phi_{1}^{T}t_{1}+\kappa_{2,1}\phi_{2}^{T}y_{2,1}-\kappa_{2,2}%
\phi_{2}^{T}t_{2,2}=a+\nu w
\end{array}
\right\} \\
&  =\lim_{c\rightarrow\infty}\inf_{\substack{\kappa_{1},\kappa_{2,1}%
,\kappa_{2,2},\nu\geq0,\\t_{1},t_{2,1},t_{2,2},w\in\mathcal{P}}}\left\{
\begin{array}
[c]{c}%
\kappa_{1}b_{1}^{T}t_{1}+\kappa_{2,1}b_{2}^{T}t_{2,1}-\kappa_{2,2}b_{2}%
^{T}t_{2,2}\\
+c\left\Vert \kappa_{1}\phi_{1}^{T}t_{1}+\kappa_{2,1}\phi_{2}^{T}%
t_{2,1}-\kappa_{2,2}\phi_{2}^{T}t_{2,2}-a-\nu w\right\Vert _{2}^{2}%
\end{array}
\right\} \\
&  \leq\lim_{c\rightarrow\infty}\inf_{\substack{\kappa_{1},\kappa_{2,1}%
,\kappa_{2,2},\nu\geq0,\\\theta_{1},\theta_{2,1},\theta_{2,2},\varphi\in
\Theta}}\left\{
\begin{array}
[c]{c}%
\kappa_{1}b_{1}^{T}t_{1}(\theta_{1})+\kappa_{2,1}b_{2}^{T}t_{2,1}(\theta
_{2,1})\\
-\kappa_{2,2}b_{2}^{T}t_{2,2}(\theta_{2,2})\\
+c\left\Vert
\begin{array}
[c]{c}%
\kappa_{1}\phi_{1}^{T}t_{1}(\theta_{1})+\kappa_{2,1}\phi_{2}^{T}t_{2,1}%
(\theta_{2,1})\\
-\kappa_{2,2}\phi_{2}^{T}t_{2,2}(\theta_{2,1})-a-\nu w(\varphi)
\end{array}
\right\Vert _{2}^{2}%
\end{array}
\right\}  .
\end{align}

\section{Efficiently measurable observables and input models}

\label{sec:eff-meas-obs-input-models}Here we present several input models for
the Hermitian matrices $A$ and $B$ and the Hermiticity-preserving
superoperator $\Phi$ in \eqref{eq:primal-SDP} and \eqref{eq:dual-SDP}. We
first consider a general model and present some associated calculations. We
then present the linear combination of states model, in which $A$, $B$, and
$\Phi$ can be written as a linear combination of quantum states, to which we
have sample access. Third we present the Pauli input model, in which $A$, $B$,
and $\Phi$ can be represented as a linear combination of tensor products of
Pauli operators. We then discuss how it is possible to have hybrids of these
models, which include the possibility that the eigendecomposition of $A$, $B$,
and $\Phi$ can be written in terms of an efficient quantum circuit
(eigenvectors) and an efficiently computable function (eigenvalues). The
hybrid case also includes writing $X$ in the parameterized form $X(\theta,
\varphi) = \sum_{i} \lambda_{\varphi}(i) U(\theta) |i\rangle\!\langle i|
U(\theta)^{\dag}$, where we parameterize its eigenvalues by a neural network
with parameter vector $\varphi$ and we parameterize its eigenvectors by a
parameterized circuit $U(\theta)$ with parameter vector $\theta$. After
presenting these, we then present versions of them for the classical case,
which include the linear combination of distributions model and the
Walsh--Hadamard model, the latter being the classical reduction of the Pauli
input model.

The Pauli input model is commonly considered in the literature and was given
in the original proposal for the variational quantum eigensolver
\cite{Peruzzo2014}. The linear combination of states model has been considered
in the context of the density matrix exponentiation algorithm
\cite{Kimmel2017HamiltonianComplexity}. The hybrid approach involving
$X(\theta, \varphi)$ is related to an approach adopted in
\cite{verdon2019quantum,sbahi2022provably,goldfeld2023quantum}.

\subsection{General form}

\label{sec:gen-form-input-model}In what follows, we consider a general form
for input models, which helps to illuminate some common points of all the
input models that follow this section. Let us suppose that both $A$ and $B$
are representable as follows:%
\begin{align}
A  &  =\sum_{i}\alpha_{i}A_{i},\label{eq:A-rep}\\
B  &  =\sum_{j}\beta_{j}B_{j}, \label{eq:B-rep}%
\end{align}
where $\alpha_{i},\beta_{j}\in\mathbb{R}$ and $A_{i}$ and $B_{j}$ are
Hermitian for all $i$ and $j$. Let us also suppose that the
Hermiticity-preserving superoperator $\Phi$ can be written as%
\begin{equation}
\Phi(X)=\sum_{i,j}\phi_{i,j}F_{j}\operatorname{Tr}[E_{i}X], \label{eq:Phi-rep}%
\end{equation}
where $\phi_{i,j}\in\mathbb{R}$ and $E_{i}$ and $F_{j}$ are Hermitian for all
$i$ and $j$. This implies that%
\begin{equation}
\Phi^{\dag}(Y)=\sum_{i,j}\phi_{i,j}E_{i}\operatorname{Tr}[F_{j}Y].
\label{eq:Phi-dag-rep}%
\end{equation}
Note that the above expansions are always possible for all $A$, $B$, and
$\Phi$ if $\left\{  A_{i}\right\}  _{i}$, $\left\{  B_{j}\right\}  _{j}$,
$\left\{  E_{i}\right\}  _{i}$, and $\left\{  F_{j}\right\}  _{j}$ are bases
for the space of all Hermitian operators.

In order for the objective functions in
\eqref{eq:param-primal-SDP}--\eqref{eq:param-dual-SDP} to be estimated
efficiently, it is necessary for 1)\ the number of non-zero coefficients in
the tuples $(\alpha_{i})_{i}$, $(\beta_{j})_{j}$, and $(\phi_{i,j})_{i,j}$ to
be polynomial in $n$ and $m$, and 2)\ each $A_{i}$, $B_{j}$, $E_{i}$, and
$F_{j}$ should correspond to an observable that is efficiently measurable or a
state that is efficiently preparable. As such, this is a key assumption of QSlack.

The following lemma explicitly shows how all terms in
\eqref{eq:primal-q-states}--\eqref{eq:dual-q-states}\ can be expanded by using
the representations in \eqref{eq:A-rep}--\eqref{eq:Phi-dag-rep}. As such, the
resulting expressions represent the terms that need to be estimated by a
quantum computer.

\begin{lemma}
\label{lem:expansion-obj-funcs}Let $\lambda,\mu,\kappa,\nu\geq0$, let $\rho$,
$\sigma$, $\tau$, and $\omega$ be density matrices, suppose that $A$ and $B$
are Hermitian matrices with the representations in
\eqref{eq:A-rep}--\eqref{eq:B-rep}, and suppose that $\Phi$ is a
Hermiticity-preserving superoperator with the representation in
\eqref{eq:Phi-rep}. Then%
\begin{align}
\lambda\operatorname{Tr}[A\rho]  &  =\lambda\sum_{i}\alpha_{i}%
\operatorname{Tr}[A_{i}\rho],\\
\kappa\operatorname{Tr}[B\tau]  &  =\kappa\sum_{j}\beta_{j}\operatorname{Tr}%
[B_{j}\tau],
\end{align}%
\begin{multline}
\left\Vert B-\lambda\Phi(\rho)-\mu\sigma\right\Vert _{2}^{2}=\lambda^{2}%
\sum_{i_{1},j_{1},i_{2},j_{2}}\phi_{i_{1},j_{1}}\phi_{i_{2},j_{2}%
}\operatorname{Tr}[\left(  E_{i_{1}}\otimes E_{i_{2}}\otimes F_{j_{1}}\right)
\left(  \rho^{\otimes2}\otimes F_{j_{2}}\right)
]\label{eq:primal-constraint-expansion-full}\\
+\sum_{j,k}\beta_{j}\beta_{k}\operatorname{Tr}[B_{j}B_{k}]+\mu^{2}%
\operatorname{Tr}[\sigma^{2}]-2\lambda\sum_{j,i,j^{\prime}}\beta_{j}%
\phi_{i,j^{\prime}}\operatorname{Tr}[\left(  E_{i}\otimes B_{j}\right)
(\rho\otimes F_{j^{\prime}})]\\
-2\mu\sum_{j}\beta_{j}\operatorname{Tr}[B_{j}\sigma]+2\lambda\mu\sum_{i,j}%
\phi_{i,j}\operatorname{Tr}[\left(  E_{i}\otimes F_{j}\right)  \left(
\rho\otimes\sigma\right)  ],
\end{multline}
and%
\begin{multline}
\left\Vert \kappa\Phi^{\dag}(\tau)-A-\nu\omega\right\Vert _{2}^{2}=\kappa
^{2}\sum_{i_{1},j_{1},i_{2},j_{2}}\phi_{i_{1},j_{1}}\phi_{i_{2},j_{2}%
}\operatorname{Tr}[\left(  F_{j_{1}}\otimes F_{j_{2}}\otimes E_{i_{1}}\right)
\left(  \tau^{\otimes2}\otimes E_{i_{2}}\right)
]\label{eq:dual-constraint-expansion-full}\\
+\sum_{i,i^{\prime}}\alpha_{i}\alpha_{i^{\prime}}\operatorname{Tr}\!\left[
A_{i}A_{i^{\prime}}\right]  +\nu^{2}\operatorname{Tr}[\omega^{2}]-2\kappa
\sum_{i,j,i^{\prime}}\phi_{i,j}\alpha_{i^{\prime}}\operatorname{Tr}[\left(
F_{j}\otimes E_{i}\right)  \left(  \tau\otimes A_{i^{\prime}}\right)  ]\\
+2\nu\sum_{i}\alpha_{i}\operatorname{Tr}[A_{i}\omega]-2\kappa\nu\sum_{i,j}%
\phi_{i,j}\operatorname{Tr}[\left(  F_{j}\otimes E_{i}\right)  \left(
\tau\otimes\omega\right)  ].
\end{multline}

\end{lemma}

\begin{proof}
Let us begin by observing that%
\begin{multline}
\left\Vert B-\lambda\Phi(\rho)-\mu\sigma\right\Vert _{2}^{2}=\operatorname{Tr}%
[B^{2}]+\lambda^{2}\operatorname{Tr}[\left(  \Phi(\rho)\right)  ^{2}]+\mu
^{2}\operatorname{Tr}[\sigma^{2}]-2\lambda\operatorname{Tr}[B\Phi
(\rho)]\label{eq:expansion-HS-primal-const}\\
-2\mu\operatorname{Tr}[B\sigma]+2\lambda\mu\operatorname{Tr}[\Phi(\rho
)\sigma].
\end{multline}
As such, we need to evaluate six terms. Consider that%
\begin{align}
\operatorname{Tr}[B^{2}]  &  =\operatorname{Tr}\!\left[  \left(  \sum_{j}%
\beta_{j}B_{j}\right)  \left(  \sum_{k}\beta_{k}B_{k}\right)  \right] \\
&  =\sum_{j,k}\beta_{j}\beta_{k}\operatorname{Tr}[B_{j}B_{k}],\\
\operatorname{Tr}[\left(  \Phi(\rho)\right)  ^{2}]  &  =\operatorname{Tr}%
\!\left[  \left(  \sum_{i_{1},j_{1}}\phi_{i_{1},j_{1}}F_{j_{1}}%
\operatorname{Tr}[E_{i_{1}}\rho]\right)  \left(  \sum_{i_{2},j_{2}}\phi
_{i_{2},j_{2}}F_{j_{2}}\operatorname{Tr}[E_{i_{2}}\rho]\right)  \right] \\
&  =\sum_{i_{1},j_{1},i_{2},j_{2}}\phi_{i_{1},j_{1}}\phi_{i_{2},j_{2}%
}\operatorname{Tr}[E_{i_{1}}\rho]\operatorname{Tr}[E_{i_{2}}\rho
]\operatorname{Tr}\!\left[  F_{j_{1}}F_{j_{2}}\right] \\
&  =\sum_{i_{1},j_{1},i_{2},j_{2}}\phi_{i_{1},j_{1}}\phi_{i_{2},j_{2}%
}\operatorname{Tr}[\left(  E_{i_{1}}\otimes E_{i_{2}}\otimes F_{j_{1}}\right)
\left(  \rho^{\otimes2}\otimes F_{j_{2}}\right)  ],\\
\operatorname{Tr}[B\Phi(\rho)]  &  =\operatorname{Tr}\!\left[  \left(
\sum_{j}\beta_{j}B_{j}\right)  \left(  \sum_{i,j^{\prime}}\phi_{i,j^{\prime}%
}F_{j^{\prime}}\operatorname{Tr}[E_{i}\rho]\right)  \right] \\
&  =\sum_{j,i,j^{\prime}}\beta_{j}\phi_{i,j^{\prime}}\operatorname{Tr}%
[E_{i}\rho]\operatorname{Tr}[B_{j}F_{j^{\prime}}]\\
&  =\sum_{j,i,j^{\prime}}\beta_{j}\phi_{i,j^{\prime}}\operatorname{Tr}[\left(
E_{i}\otimes B_{j}\right)  (\rho\otimes F_{j^{\prime}})],\\
\operatorname{Tr}[B\sigma]  &  =\sum_{j}\beta_{j}\operatorname{Tr}[B_{j}%
\sigma],\\
\operatorname{Tr}[\Phi(\rho)\sigma]  &  =\operatorname{Tr}\!\left[  \left(
\sum_{i,j}\phi_{i,j}F_{j}\operatorname{Tr}[E_{i}\rho]\right)  \sigma\right] \\
&  =\sum_{i,j}\phi_{i,j}\operatorname{Tr}[E_{i}\rho]\operatorname{Tr}\!\left[
F_{j}\sigma\right] \\
&  =\sum_{i,j}\phi_{i,j}\operatorname{Tr}[\left(  E_{i}\otimes F_{j}\right)
\left(  \rho\otimes\sigma\right)  ].
\end{align}
Plugging the six terms above into \eqref{eq:expansion-HS-primal-const} then
leads to \eqref{eq:primal-constraint-expansion-full}.

Now observe that%
\begin{multline}
\left\Vert \kappa\Phi^{\dag}(\tau)-A-\nu\omega\right\Vert _{2}^{2}=\kappa
^{2}\operatorname{Tr}[(\Phi^{\dag}(\tau))^{2}]+\operatorname{Tr}%
[A^{2}]\label{eq:expansion-HS-dual-const-2}\\
+\nu^{2}\operatorname{Tr}[\omega^{2}]-2\kappa\operatorname{Tr}[\Phi^{\dag
}(\tau)A]\\
+2\nu\operatorname{Tr}[A\omega]-2\kappa\nu\operatorname{Tr}[\Phi^{\dag}%
(\tau)\omega].
\end{multline}
Then we evaluate the six terms above as follows:%
\begin{align}
\operatorname{Tr}[(\Phi^{\dag}(\tau))^{2}]  &  =\operatorname{Tr}\!\left[
\left(  \sum_{i_{1},j_{1}}\phi_{i_{1},j_{1}}E_{i_{1}}\operatorname{Tr}%
[F_{j_{1}}\tau]\right)  \left(  \sum_{i_{2},j_{2}}\phi_{i_{2},j_{2}}E_{i_{2}%
}\operatorname{Tr}[F_{j_{2}}\tau]\right)  \right] \\
&  =\sum_{i_{1},j_{1},i_{2},j_{2}}\phi_{i_{1},j_{1}}\phi_{i_{2},j_{2}%
}\operatorname{Tr}[F_{j_{1}}\tau]\operatorname{Tr}[F_{j_{2}}\tau
]\operatorname{Tr}\!\left[  E_{i_{1}}E_{i_{2}}\right] \\
&  =\sum_{i_{1},j_{1},i_{2},j_{2}}\phi_{i_{1},j_{1}}\phi_{i_{2},j_{2}%
}\operatorname{Tr}[\left(  F_{j_{1}}\otimes F_{j_{2}}\otimes E_{i_{1}}\right)
\left(  \tau^{\otimes2}\otimes E_{i_{2}}\right)  ],\\
\operatorname{Tr}[A^{2}]  &  =\operatorname{Tr}\!\left[  \left(  \sum
_{i}\alpha_{i}A_{i}\right)  \left(  \sum_{i^{\prime}}\alpha_{i^{\prime}%
}A_{i^{\prime}}\right)  \right] \\
&  =\sum_{i,i^{\prime}}\alpha_{i}\alpha_{i^{\prime}}\operatorname{Tr}\!\left[
A_{i}A_{i^{\prime}}\right]  ,\\
\operatorname{Tr}[\Phi^{\dag}(\tau)A]  &  =\operatorname{Tr}\!\left[  \left(
\sum_{i,j}\phi_{i,j}E_{i}\operatorname{Tr}[F_{j}\tau]\right)  \left(
\sum_{i^{\prime}}\alpha_{i^{\prime}}A_{i^{\prime}}\right)  \right] \\
&  =\sum_{i,j,i^{\prime}}\phi_{i,j}\alpha_{i^{\prime}}\operatorname{Tr}%
[F_{j}\tau]\operatorname{Tr}\!\left[  E_{i}A_{i^{\prime}}\right] \\
&  =\sum_{i,j,i^{\prime}}\phi_{i,j}\alpha_{i^{\prime}}\operatorname{Tr}%
[\left(  F_{j}\otimes E_{i}\right)  \left(  \tau\otimes A_{i^{\prime}}\right)
],\\
\operatorname{Tr}[A\omega]  &  =\sum_{i}\alpha_{i}\operatorname{Tr}%
[A_{i}\omega],\\
\operatorname{Tr}[\Phi^{\dag}(\tau)\omega]  &  =\operatorname{Tr}\!\left[
\left(  \sum_{i,j}\phi_{i,j}E_{i}\operatorname{Tr}[F_{j}\tau]\right)
\omega\right] \\
&  =\sum_{i,j}\phi_{i,j}\operatorname{Tr}[F_{j}\tau]\operatorname{Tr}\!\left[
E_{i}\omega\right] \\
&  =\sum_{i,j}\phi_{i,j}\operatorname{Tr}[\left(  F_{j}\otimes E_{i}\right)
\left(  \tau\otimes\omega\right)  ].
\end{align}
Plugging the six terms above into \eqref{eq:expansion-HS-dual-const-2} then
leads to \eqref{eq:dual-constraint-expansion-full}.
\end{proof}

\bigskip

When executing QSlack, it is of interest to know how many samples are required
in order to obtain a desired accuracy and success probability. A key tool in
this regard is the Hoeffding bound \cite{H63}, which indicates that the sample
mean is an $\varepsilon$-accurate estimate of the true mean, with success probability
not smaller than $1-\delta$, if the number $n$ of samples satisfies
$n\geq\frac{M^{2}}{\varepsilon^{2}}\ln\!\left(  \frac{1}{\delta}\right)  $,
where $M$ is the range of values that a finite-dimensional random variable
takes on (see, e.g., \cite[Theorem~1]{bandyopadhyay2023efficient} for the
precise statement that we need). For our case of interest, the objective
functions are%
\begin{align}
&  \lambda\operatorname{Tr}[A\rho]-c\left\Vert B-\lambda\Phi(\rho)-\mu
\sigma\right\Vert _{2}^{2},\label{eq:app-primal-obj-states}\\
&  \kappa\operatorname{Tr}[B\tau]+c\left\Vert \kappa\Phi^{\dag}(\tau
)-A-\nu\omega\right\Vert _{2}^{2}, \label{eq:app-dual-obj-states}%
\end{align}
as given in \eqref{eq:primal-q-states}--\eqref{eq:dual-q-states}. As such, if
we obtain an upper bound on these expectations, which is independent of the
states $\rho$, $\sigma$, $\tau$, and $\omega$, then that serves as an upper
bound on $M$. So our goal here is to obtain such an upper bound for both
objective functions in
\eqref{eq:app-primal-obj-states}--\eqref{eq:app-dual-obj-states}. We can do so
by repeatedly applying the H\"{o}lder inequality and the triangle inequality,
where the H\"{o}lder inequality is given by
\begin{equation}
\left\vert \operatorname{Tr}[C^{\dag}D]\right\vert \leq\left\Vert C\right\Vert
_{1}\left\Vert D\right\Vert _{\infty},
\end{equation}
for matrices $C$ and $D$. Doing so and employing the expansions in
Lemma~\ref{lem:expansion-obj-funcs}, we find the following upper bounds, which
are independent of the states $\rho$, $\sigma$, $\tau$, and $\omega$:%
\begin{multline}
\left\vert \lambda\operatorname{Tr}[A\rho]-c\left\Vert B-\lambda\Phi(\rho
)-\mu\sigma\right\Vert _{2}^{2}\right\vert \leq\lambda\sum_{i}\left\vert
\alpha_{i}\right\vert \left\Vert A_{i}\right\Vert _{\infty}%
\label{eq:sample-compl-up-bnd-primal}\\
+c\Bigg[\lambda^{2}\left(  \sum_{i_{1},j_{1}}\left\vert \phi_{i_{1},j_{1}%
}\right\vert \left\Vert E_{i_{1}}\right\Vert _{\infty}\left\Vert F_{j_{1}%
}\right\Vert _{\infty}\right)  \left(  \sum_{i_{2},j_{2}}\left\vert
\phi_{i_{2},j_{2}}\right\vert \left\Vert E_{i_{2}}\right\Vert _{\infty
}\left\Vert F_{j_{2}}\right\Vert _{1}\right) \\
+\left(  \sum_{j}\left\vert \beta_{j}\right\vert \left\Vert B_{j}\right\Vert
_{\infty}\right)  \left(  \sum_{k}|\beta_{k}| \left\Vert B_{k}\right\Vert
_{1}\right)  +\mu^{2}\\
+2\lambda\left(  \sum_{j}\left\vert \beta_{j}\right\vert \left\Vert
B_{j}\right\Vert _{\infty}\right)  \left(  \sum_{i,j^{\prime}}\left\vert
\phi_{i,j^{\prime}}\right\vert \left\Vert E_{i}\right\Vert _{\infty}\left\Vert
F_{j^{\prime}}\right\Vert _{1}\right) \\
+2\mu\sum_{j}\left\vert \beta_{j}\right\vert \left\Vert B_{j}\right\Vert
_{\infty}+2\lambda\mu\sum_{i,j}\left\vert \phi_{i,j}\right\vert \left\Vert
E_{i}\right\Vert _{\infty}\left\Vert F_{j}\right\Vert _{\infty}\Bigg],
\end{multline}%
\begin{multline}
\left\vert \kappa\operatorname{Tr}[B\tau]+c\left\Vert \kappa\Phi^{\dag}%
(\tau)-A-\nu\omega\right\Vert _{2}^{2}\right\vert \leq\kappa\sum_{j}\left\vert
\beta_{j}\right\vert \left\Vert B_{j}\right\Vert _{\infty}%
\label{eq:sample-compl-up-bnd-dual}\\
+c\Bigg[\kappa^{2}\left(  \sum_{i_{1},j_{1}}\left\vert \phi_{i_{1},j_{1}%
}\right\vert \left\Vert F_{j_{1}}\right\Vert _{\infty}\left\Vert E_{i_{1}%
}\right\Vert _{\infty}\right)  \left(  \sum_{i_{2},j_{2}}\left\vert
\phi_{i_{2},j_{2}}\right\vert \left\Vert F_{j_{2}}\right\Vert _{\infty
}\left\Vert E_{i_{2}}\right\Vert _{1}\right) \\
+\left(  \sum_{i}\left\vert \alpha_{i}\right\vert \left\Vert A_{i}\right\Vert
_{\infty}\right)  \left(  \sum_{i^{\prime}}\left\vert \alpha_{i^{\prime}%
}\right\vert \left\Vert A_{i^{\prime}}\right\Vert _{1}\right)  +\nu^{2}\\
+2\kappa\left(  \sum_{i,j}\left\vert \phi_{i,j}\right\vert \left\Vert
F_{j}\right\Vert _{\infty}\left\Vert E_{i}\right\Vert _{\infty}\right)
\left(  \sum_{i^{\prime}}\left\vert \alpha_{i^{\prime}}\right\vert \left\Vert
A_{i^{\prime}}\right\Vert _{1}\right) \\
+2\nu\sum_{i}\left\vert \alpha_{i}\right\vert \left\Vert A_{i}\right\Vert
_{\infty}+2\kappa\nu\sum_{i,j}\left\vert \phi_{i,j}\right\vert \left\Vert
F_{j}\right\Vert _{\infty}\left\Vert E_{i}\right\Vert _{\infty}\Bigg].
\end{multline}
As given, the squares of these upper bounds are directly proportional to upper
bounds on the number of samples needed to obtain a given accuracy and success
probability when evaluating the primal and dual objective functions. In the
example models that follow, these upper bounds or related ones simplify
significantly, due to particular properties of the input models. Furthermore,
the upper bounds above illustrate the need for norms such as $\left\Vert
F_{j}\right\Vert _{\infty}$, $\left\Vert E_{i_{2}}\right\Vert _{1}$, etc.~not
to grow more than polynomially in $n$ or $m$, so that the overall runtime for
estimating the objective function is polynomial in $n$ and $m$. We comment
further on this point in the forthcoming subsections for some particular input models.

Let us finally remark that the whole formalism above has a classical
reduction, applicable for CSlack, such that all of the matrices above reduce
to diagonal matrices. As such, they can be represented as vectors, the trace
overlaps above reduce to vector overlaps, and the Hilbert--Schmidt norms
reduce to Euclidean norms. All of the statements above then apply to this
classical case. This kind of reduction to the classical case is possible for
the models that we discuss in the forthcoming subsections, and we discuss the
reduction specifically in Appendices~\ref{app:linear-combo-dist-model} and
\ref{app:Walsh-Hadamard-input-model}.

\subsection{Linear combination of states input model}

\label{app:linear-combo-states}

The first specific input model that we consider is the linear combination of
states model, in which $A$, $B$, and $\Phi$ are expressed as follows:%
\begin{align}
A  &  =\sum_{i}\alpha_{i}\rho_{i},\label{eq:linear-combo-states-A}\\
B  &  =\sum_{j}\beta_{j}\tau_{j},\label{eq:linear-combo-states-B}\\
\Phi(\cdot)  &  =\sum_{k,\ell}\phi_{k,\ell}\operatorname{Tr}[\sigma_{k}%
(\cdot)]\omega_{\ell}, \label{eq:linear-combo-states-Phi}%
\end{align}
where $\alpha_{i},\beta_{j},\phi_{k,\ell}\in\mathbb{R}$ for all indices
$i,j,k,\ell$ and $\left(  \rho_{i}\right)  _{i}$, $\left(  \tau_{j}\right)
_{j}$, $\left(  \sigma_{k}\right)  _{k}$, and $\left(  \omega_{\ell}\right)
_{\ell}$ are tuples of states (i.e., density matrices). In general, all
Hermitian matrices $A$ and $B$ and every Hermiticity-preserving superoperator
$\Phi$ can be written as above. This also implies that the adjoint
superoperator $\Phi^{\dag}$ can be written as%
\begin{equation}
\Phi^{\dag}(\cdot)=\sum_{k,\ell}\phi_{k,\ell}\operatorname{Tr}[\omega_{\ell
}(\cdot)]\sigma_{k}.
\end{equation}
See \cite{vanapeldoorn_et_al:LIPIcs:2019:10675} for a related input model.

The triple $\left(  A,B,\Phi\right)  $ is efficiently representable in this
input model if the total number of indices is polynomial in $n$ and $m$ and if
all states in the tuples $\left(  \rho_{i}\right)  _{i}$, $\left(  \tau
_{j}\right)  _{j}$, $\left(  \sigma_{k}\right)  _{k}$, and $\left(
\omega_{\ell}\right)  _{\ell}$ are efficiently preparable on quantum
computers, via a quantum circuit or some related method of preparation.

By plugging directly into Lemma~\ref{lem:expansion-obj-funcs}, we find that%
\begin{align}
\lambda\operatorname{Tr}[A\rho]  &  =\lambda\sum_{i}\alpha_{i}%
\operatorname{Tr}[\rho_{i}\rho],\label{eq:linear-comb-states-dev-1}\\
\kappa\operatorname{Tr}[B\tau]  &  =\kappa\sum_{j}\beta_{j}\operatorname{Tr}%
\!\left[  \tau_{j}\tau\right]  ,
\end{align}%
\begin{multline}
\left\Vert B-\lambda\Phi(\rho)-\mu\sigma\right\Vert _{2}^{2}=\lambda^{2}%
\sum_{k_{1},\ell_{1},k_{2},\ell_{2}}\phi_{k_{1},\ell_{1}}\phi_{k_{2},\ell_{2}%
}\operatorname{Tr}[\left(  \sigma_{k_{1}}\otimes\sigma_{k_{2}}\otimes
\omega_{\ell_{1}}\right)  \left(  \rho^{\otimes2}\otimes\omega_{\ell_{2}%
}\right)  ]\\
+\mu^{2}\operatorname{Tr}[\sigma^{2}]+\sum_{j_{1},j_{2}}\beta_{j_{1}}%
\beta_{j_{2}}\operatorname{Tr}\!\left[  \tau_{j_{1}}\tau_{j_{2}}\right]
-2\lambda\sum_{j,k,\ell}\beta_{j}\phi_{k,\ell}\operatorname{Tr}[\left(
\sigma_{k}\otimes\tau_{j}\right)  \left(  \rho\otimes\omega_{\ell}\right)  ]\\
-2\mu\sum_{j}\beta_{j}\operatorname{Tr}[\tau_{j}\sigma]+2\mu\lambda
\sum_{k,\ell}\phi_{k,\ell}\operatorname{Tr}[\left(  \sigma_{k}\otimes
\omega_{\ell}\right)  \left(  \rho\otimes\sigma\right)  ],
\end{multline}
and%
\begin{multline}
\left\Vert \kappa\Phi^{\dag}(\tau)-A-\nu\omega\right\Vert _{2}^{2}=\kappa
^{2}\sum_{k_{1},\ell_{1},k_{2},\ell_{2}}\phi_{k_{1},\ell_{1}}\phi_{k_{2}%
,\ell_{2}}\operatorname{Tr}[\left(  \omega_{\ell_{1}}\otimes\omega_{\ell_{2}%
}\otimes\sigma_{k_{1}}\right)  \left(  \tau^{\otimes2}\otimes\sigma_{k_{2}%
}\right)  ]\\
+\sum_{i_{1},i_{2}}\alpha_{i_{1}}\alpha_{i_{2}}\operatorname{Tr}\!\left[
\rho_{i_{1}}\rho_{i_{2}}\right]  -2\kappa\sum_{k,\ell,i}\phi_{k,\ell}%
\alpha_{i}\operatorname{Tr}[\left(  \omega_{\ell}\otimes\sigma_{k}\right)
\left(  \tau\otimes\rho_{i}\right)  ]\\
+2\nu\sum_{i}\alpha_{i}\operatorname{Tr}[\rho_{i}\omega]-2\kappa\nu
\sum_{k,\ell}\phi_{k,\ell}\operatorname{Tr}[\left(  \omega_{\ell}\otimes
\sigma_{k}\right)  \left(  \tau\otimes\omega\right)  ].
\end{multline}
Inspecting above, it is clear that all terms can be estimated by means of the
destructive swap test (see \cite[Section~2.2]{bandyopadhyay2023efficient} for
a detailed review). If the convex-combination ansatz is used for one of the
states, then the trace overlap can be estimated by the mixed-state Loschmidt
echo test discussed around
\eqref{eq:LE-test-basic-1}--\eqref{eq:LE-test-basic-last}. Furthermore, the
following terms
\begin{align}
&  \sum_{j_{1},j_{2}}\beta_{j_{1}}\beta_{j_{2}}\operatorname{Tr}\!\left[
\tau_{j_{1}}\tau_{j_{2}}\right]  ,\\
&  \sum_{i_{1},i_{2}}\alpha_{i_{1}}\alpha_{i_{2}}\operatorname{Tr}\!\left[
\rho_{i_{1}}\rho_{i_{2}}\right]  , \label{eq:linear-comb-states-dev-last}%
\end{align}
can be estimated offline because they are needed only once and do not change
from iteration to iteration of the optimization.\ By applying the same
reasoning that leads to
\eqref{eq:sample-compl-up-bnd-primal}--\eqref{eq:sample-compl-up-bnd-dual}, as
well as the fact that $\operatorname{Tr}[\rho\sigma]\leq1$ for all states
$\rho$ and $\sigma$, we arrive at the following upper bounds for this input
model:%
\begin{multline}
\left\vert \lambda\operatorname{Tr}[A\rho]-c\left\Vert B-\lambda\Phi(\rho
)-\mu\sigma\right\Vert _{2}^{2}\right\vert \leq\lambda\left\Vert \vec{\alpha
}\right\Vert _{1}\\
+c\left(  \lambda^{2}\left\Vert \vec{\phi}\right\Vert _{1}^{2}+\mu
^{2}+\left\Vert \vec{\beta}\right\Vert _{1}^{2}+2\lambda\left\Vert \vec{\beta
}\right\Vert _{1}\left\Vert \vec{\phi}\right\Vert _{1}+2\mu\left\Vert
\vec{\beta}\right\Vert _{1}+2\mu\lambda\left\Vert \vec{\phi}\right\Vert
_{1}\right)  ,
\end{multline}%
\begin{multline}
\left\vert \kappa\operatorname{Tr}[B\tau]+c\left\Vert \kappa\Phi^{\dag}%
(\tau)-A-\nu\omega\right\Vert _{2}^{2}\right\vert \leq\kappa\left\Vert
\vec{\beta}\right\Vert _{1}\\
+c\left(  \kappa^{2}\left\Vert \vec{\phi}\right\Vert _{1}^{2}+\left\Vert
\vec{\alpha}\right\Vert _{1}^{2}+2\kappa\left\Vert \vec{\phi}\right\Vert
_{1}\left\Vert \vec{\alpha}\right\Vert _{1}+2\nu\left\Vert \vec{\alpha
}\right\Vert _{1}+2\kappa\nu\left\Vert \vec{\phi}\right\Vert _{1}\right)  ,
\end{multline}
where%
\begin{equation}
\left\Vert \vec{\alpha}\right\Vert _{1}=\sum_{i}\left\vert \alpha
_{i}\right\vert ,\qquad\left\Vert \vec{\beta}\right\Vert _{1}=\sum
_{j}\left\vert \beta_{j}\right\vert ,\qquad\left\Vert \vec{\phi}\right\Vert
_{1}=\sum_{k,\ell}\left\vert \phi_{k,\ell}\right\vert .
\end{equation}

\subsection{Pauli input model}

\label{app:Pauli-input-model}

Another model that we consider is the Pauli input model. Let us denote the
Pauli matrices as follows:
\begin{align}
\sigma_{0}  &  \equiv\sigma_{I}\coloneqq%
\begin{bmatrix}
1 & 0\\
0 & 1
\end{bmatrix}
,\qquad\sigma_{1}\equiv\sigma_{X}\coloneqq%
\begin{bmatrix}
0 & 1\\
1 & 0
\end{bmatrix}
,\label{eq:Pauli-def-1}\\
\sigma_{2}  &  \equiv\sigma_{Y}\coloneqq%
\begin{bmatrix}
0 & -i\\
i & 0
\end{bmatrix}
,\qquad\sigma_{3}\equiv\sigma_{Z}\coloneqq%
\begin{bmatrix}
1 & 0\\
0 & -1
\end{bmatrix}
. \label{eq:Pauli-def-2}%
\end{align}
In the Pauli input model, both $A$ and $B$ are representable in terms of the
Pauli basis as%
\begin{align}
A  &  =\sum_{\overrightarrow{x}}\alpha_{\overrightarrow{x}}\sigma
_{\overrightarrow{x}},\label{eq:Pauli-rep-A}\\
B  &  =\sum_{\overrightarrow{y}}\beta_{\overrightarrow{y}}\sigma
_{\overrightarrow{y}}, \label{eq:Pauli-rep-B}%
\end{align}
where $\overrightarrow{x}\in\left\{  0,1,2,3\right\}  ^{n}$, $\overrightarrow
{y}\in\left\{  0,1,2,3\right\}  ^{m}$, $\alpha_{\overrightarrow{x}}%
,\beta_{\overrightarrow{y}}\in\mathbb{R}$, and $\sigma_{\overrightarrow{x}%
}\equiv\sigma_{x_{1}}\otimes\cdots\otimes\sigma_{x_{n}}$ is a tensor product
of Pauli operators, Similarly, $\sigma_{\overrightarrow{y}}\equiv\sigma
_{y_{1}}\otimes\cdots\otimes\sigma_{y_{m}}$. Additionally, the
superoperator$~\Phi$ is representable in terms of the Pauli basis as%
\begin{equation}
\Phi(X)=\sum_{\overrightarrow{x},\overrightarrow{y}}\phi_{\overrightarrow
{x},\overrightarrow{y}}\sigma_{\overrightarrow{y}}\operatorname{Tr}\!\left[
\sigma_{\overrightarrow{x}}X\right]  , \label{eq:Pauli-rep-Phi}%
\end{equation}
where $\phi_{\overrightarrow{x},\overrightarrow{y}}\in\mathbb{R}$.\ The
Hilbert--Schmidt adjoint $\Phi^{\dag}$ is then given by%
\begin{equation}
\Phi^{\dag}(Y)=\sum_{\overrightarrow{x},\overrightarrow{y}}\phi
_{\overrightarrow{x},\overrightarrow{y}}\sigma_{\overrightarrow{x}%
}\operatorname{Tr}[\sigma_{\overrightarrow{y}}Y]. \label{eq:Pauli-rep-Phi-dag}%
\end{equation}
The main difference between this model and the generic model presented in
Appendix~\ref{sec:gen-form-input-model} is that the Pauli operators form an
orthogonal basis with respect to the Hilbert--Schmidt inner product, and as a
consequence, several of the expressions in Lemma~\ref{lem:expansion-obj-funcs}
simplify immensely. Later on, we comment on the values of the coefficients
$\alpha_{\overrightarrow{x}}$, $\beta_{\overrightarrow{y}}$, and
$\phi_{\overrightarrow{x},\overrightarrow{y}}$ that are needed in order to
ensure the objective functions in the primal and dual optimization problems
can be evaluated efficiently. Interestingly, all of the examples presented in
Sections~\ref{sec:fidelity-main-text}, \ref{sec:entanglement-negativity}, and
\ref{sec:constrained-Ham-opt} of the main text have objective functions that
can be evaluated efficiently.

We say that $A$ is efficiently representable in the Pauli basis if only
poly$(n)$ of the coefficients in the tuple $\left(  \alpha_{\overrightarrow
{x}}\right)  _{\overrightarrow{x}}$ are non-zero. Similarly, $B$ and $\Phi$
are efficiently representable if only poly$(m)$ of $\left(  \beta
_{\overrightarrow{y}}\right)  _{\overrightarrow{y}}$ and poly$(n,m)$ of
$\left(  \phi_{\overrightarrow{x},\overrightarrow{y}}\right)
_{\overrightarrow{x},\overrightarrow{y}}$ are non-zero, respectively. The
estimation part of the variational quantum algorithms that we propose here are
efficient if all of $A$, $B$, and $\Phi$ are efficiently representable in the
Pauli basis and if further constraints on $\left(  \alpha_{\overrightarrow{x}%
}\right)  _{\overrightarrow{x}}$, $\left(  \beta_{\overrightarrow{y}}\right)
_{\overrightarrow{y}}$, and $\left(  \phi_{\overrightarrow{x},\overrightarrow
{y}}\right)  _{\overrightarrow{x},\overrightarrow{y}}$ hold, as detailed in \eqref{eq:conditions-efficient-sampling-pauli}\ below.

\begin{proposition}
For the Pauli input model described in
\eqref{eq:Pauli-rep-A}--\eqref{eq:Pauli-rep-Phi}, the following equalities
hold%
\begin{equation}
\operatorname{Tr}[A\rho]=\sum_{\overrightarrow{x}}\alpha_{\overrightarrow{x}%
}\operatorname{Tr}[\sigma_{\overrightarrow{x}}\rho],\qquad\operatorname{Tr}%
[B\tau]=\sum_{\overrightarrow{y}}\beta_{\overrightarrow{y}}\operatorname{Tr}%
[\sigma_{\overrightarrow{y}}\tau], \label{eq:pauli-input-dev-1}%
\end{equation}%
\begin{multline}
\left\Vert B-\lambda\Phi(\rho)-\mu\sigma\right\Vert _{2}^{2}=\lambda^{2}%
2^{m}\sum_{\overrightarrow{x_{1}},\overrightarrow{x_{2}}}\left(
\sum_{\overrightarrow{y}}\phi_{\overrightarrow{x_{1}},\overrightarrow{y}}%
\phi_{\overrightarrow{x_{2}},\overrightarrow{y}}\right)  \operatorname{Tr}%
\!\left[  \left(  \sigma_{\overrightarrow{x_{1}}}\otimes\sigma
_{\overrightarrow{x_{2}}}\right)  \rho^{\otimes2}\right]
\label{eq:constraint-primal-pauli}\\
+2^{m}\left\Vert \overrightarrow{\beta}\right\Vert _{2}^{2}+\mu^{2}%
\operatorname{Tr}[\sigma^{2}]-\lambda2^{m+1}\sum_{\overrightarrow{x}}\left(
\sum_{\overrightarrow{y}}\beta_{\overrightarrow{y}}\phi_{\overrightarrow
{x},\overrightarrow{y}}\right)  \operatorname{Tr}[\sigma_{\overrightarrow{x}%
}\rho]\\
-2\mu\sum_{\overrightarrow{y}}\beta_{\overrightarrow{y}}\operatorname{Tr}%
[\sigma_{\overrightarrow{y}}\sigma]+2\lambda\mu\sum_{\overrightarrow
{x},\overrightarrow{y}}\phi_{\overrightarrow{x},\overrightarrow{y}%
}\operatorname{Tr}\!\left[  \left(  \sigma_{\overrightarrow{x}}\otimes
\sigma_{\overrightarrow{y}}\right)  \left(  \rho\otimes\sigma\right)  \right]
,
\end{multline}
and%
\begin{multline}
\left\Vert \kappa\Phi^{\dag}(\tau)-A-\nu\omega\right\Vert _{2}^{2}=\kappa
^{2}2^{n}\sum_{\overrightarrow{y_{1}},\overrightarrow{y_{2}}}\left(
\sum_{\overrightarrow{x}}\phi_{\overrightarrow{x},\overrightarrow{y_{1}}}%
\phi_{\overrightarrow{x},\overrightarrow{y_{2}}}\right)  \operatorname{Tr}%
\left[  \left(  \sigma_{\overrightarrow{y_{1}}}\otimes\sigma_{\overrightarrow
{y_{2}}}\right)  \tau^{\otimes2}\right] \label{eq:constraint-dual-pauli}\\
+2^{n}\left\Vert \overrightarrow{\alpha}\right\Vert _{2}^{2}+\nu
^{2}\operatorname{Tr}[\omega^{2}]-\kappa2^{n+1}\sum_{\overrightarrow{y}%
}\left(  \sum_{\overrightarrow{x}}\alpha_{\overrightarrow{x}}\phi
_{\overrightarrow{x},\overrightarrow{y}}\right)  \operatorname{Tr}%
[\sigma_{\overrightarrow{y}}\tau]\\
+2\nu\sum_{\overrightarrow{x}}\alpha_{\overrightarrow{x}}\operatorname{Tr}%
[\sigma_{\overrightarrow{x}}\omega]-2\kappa\nu\sum_{\overrightarrow
{x},\overrightarrow{y}}\phi_{\overrightarrow{x},\overrightarrow{y}%
}\operatorname{Tr}\!\left[  \left(  \sigma_{\overrightarrow{x}}\otimes
\sigma_{\overrightarrow{y}}\right)  \left(  \omega\otimes\tau\right)  \right]
.
\end{multline}

\end{proposition}

\begin{proof}
The equalities above follow by plugging
\eqref{eq:Pauli-rep-A}--\eqref{eq:Pauli-rep-Phi} into
\eqref{eq:primal-constraint-expansion-full} of
Lemma~\ref{lem:expansion-obj-funcs} and using the orthogonality relation
$\operatorname{Tr}[\sigma_{\overrightarrow{x_{1}}}\sigma_{\overrightarrow
{x_{2}}}]=2^{n}\delta_{\overrightarrow{x_{1}},\overrightarrow{x_{2}}}$, which
holds for tensor products of Pauli operators acting on $n$ qubits.
Specifically, consider that%
\begin{align}
\operatorname{Tr}[B^{2}]  &  =\sum_{\overrightarrow{y_{1}}}\sum
_{\overrightarrow{y_{2}}}\beta_{\overrightarrow{y_{1}}}\beta_{\overrightarrow
{y_{2}}}\operatorname{Tr}\!\left[  \sigma_{\overrightarrow{y_{1}}}%
\sigma_{\overrightarrow{y_{2}}}\right] \\
&  =2^{m}\sum_{\overrightarrow{y}}\beta_{\overrightarrow{y}}^{2}\\
&  =2^{m}\left\Vert \overrightarrow{\beta}\right\Vert _{2}^{2},\\
\operatorname{Tr}[\left(  \Phi(\rho)\right)  ^{2}]  &  =\sum_{\overrightarrow
{x_{1}},\overrightarrow{y_{1}},\overrightarrow{x_{2}},\overrightarrow{y_{2}}%
}\phi_{\overrightarrow{x_{1}},\overrightarrow{y_{1}}}\phi_{\overrightarrow
{x_{2}},\overrightarrow{y_{2}}}\operatorname{Tr}[\sigma_{\overrightarrow
{x_{1}}}\rho]\operatorname{Tr}[\sigma_{\overrightarrow{x_{2}}}\rho
]\operatorname{Tr}\!\left[  \sigma_{\overrightarrow{y_{1}}}\sigma
_{\overrightarrow{y_{2}}}\right] \\
&  =2^{m}\sum_{\overrightarrow{x_{1}},\overrightarrow{x_{2}}}\left(
\sum_{\overrightarrow{y}}\phi_{\overrightarrow{x_{1}},\overrightarrow{y}}%
\phi_{\overrightarrow{x_{2}},\overrightarrow{y}}\right)  \operatorname{Tr}%
\left[  \left(  \sigma_{\overrightarrow{x_{1}}}\otimes\sigma_{\overrightarrow
{x_{2}}}\right)  \left(  \rho\otimes\rho\right)  \right]  ,\\
\operatorname{Tr}[B\Phi(\rho)]  &  =\sum_{\overrightarrow{y_{1}}%
,\overrightarrow{x},\overrightarrow{y_{2}}}\beta_{\overrightarrow{y_{1}}}%
\phi_{\overrightarrow{x},\overrightarrow{y_{2}}}\operatorname{Tr}%
[\sigma_{\overrightarrow{x}}\rho]\operatorname{Tr}\!\left[  \sigma
_{\overrightarrow{y_{1}}}\sigma_{\overrightarrow{y_{2}}}\right] \\
&  =2^{m}\sum_{\overrightarrow{x}}\left(  \sum_{\overrightarrow{y}}%
\beta_{\overrightarrow{y}}\phi_{\overrightarrow{x},\overrightarrow{y}}\right)
\operatorname{Tr}[\sigma_{\overrightarrow{x}}\rho],\\
\operatorname{Tr}[B\sigma]  &  =\sum_{\overrightarrow{y}}\beta
_{\overrightarrow{y}}\operatorname{Tr}\!\left[  \sigma_{\overrightarrow{y}%
}\sigma\right]  ,\\
\operatorname{Tr}[\Phi(\rho)\sigma]  &  =\sum_{\overrightarrow{x}%
,\overrightarrow{y}}\phi_{\overrightarrow{x},\overrightarrow{y}}%
\operatorname{Tr}[\sigma_{\overrightarrow{x}}\rho]\operatorname{Tr}%
[\sigma_{\overrightarrow{y}}\sigma]\\
&  =\sum_{\overrightarrow{x},\overrightarrow{y}}\phi_{\overrightarrow
{x},\overrightarrow{y}}\operatorname{Tr}\!\left[  \left(  \sigma
_{\overrightarrow{x}}\otimes\sigma_{\overrightarrow{y}}\right)  \left(
\rho\otimes\sigma\right)  \right]  .
\end{align}
Combining these expressions according to
\eqref{eq:primal-constraint-expansion-full} of
Lemma~\ref{lem:expansion-obj-funcs} leads to \eqref{eq:constraint-primal-pauli}.

Furthermore, for the term $\left\Vert \kappa\Phi^{\dag}(\tau)-A-\nu
\omega\right\Vert _{2}^{2}$, we plug
\eqref{eq:Pauli-rep-A}--\eqref{eq:Pauli-rep-Phi-dag} into
\eqref{eq:dual-constraint-expansion-full} of
Lemma~\ref{lem:expansion-obj-funcs} and find that%
\begin{align}
\operatorname{Tr}[(\Phi^{\dag}(\tau))^{2}]  &  =\sum_{\overrightarrow{x_{1}%
},\overrightarrow{y_{1}}}\sum_{\overrightarrow{x_{2}},\overrightarrow{y_{2}}%
}\phi_{\overrightarrow{x_{1}},\overrightarrow{y_{1}}}\phi_{\overrightarrow
{x_{2}},\overrightarrow{y_{2}}}\operatorname{Tr}[\sigma_{\overrightarrow
{y_{1}}}\tau]\operatorname{Tr}[\sigma_{\overrightarrow{y_{2}}}\tau
]\operatorname{Tr}\!\left[  \sigma_{\overrightarrow{x_{1}}}\sigma
_{\overrightarrow{x_{2}}}\right] \nonumber\\
&  =2^{n}\sum_{\overrightarrow{y_{1}},\overrightarrow{y_{2}}}\left(
\sum_{\overrightarrow{x}}\phi_{\overrightarrow{x},\overrightarrow{y_{1}}}%
\phi_{\overrightarrow{x},\overrightarrow{y_{2}}}\right)  \operatorname{Tr}%
[\sigma_{\overrightarrow{y_{1}}}\tau]\operatorname{Tr}[\sigma_{\overrightarrow
{y_{2}}}\tau]\\
&  =2^{n}\sum_{\overrightarrow{y_{1}},\overrightarrow{y_{2}}}\left(
\sum_{\overrightarrow{x}}\phi_{\overrightarrow{x},\overrightarrow{y_{1}}}%
\phi_{\overrightarrow{x},\overrightarrow{y_{2}}}\right)  \operatorname{Tr}%
\left[  \left(  \sigma_{\overrightarrow{y_{1}}}\otimes\sigma_{\overrightarrow
{y_{2}}}\right)  \tau^{\otimes2}\right]  ,\\
\operatorname{Tr}[A^{2}]  &  =2^{n}\left\Vert \overrightarrow{\alpha
}\right\Vert _{2}^{2},\\
\operatorname{Tr}[A\omega]  &  =\sum_{\overrightarrow{x}}\alpha
_{\overrightarrow{x}}\operatorname{Tr}[\sigma_{\overrightarrow{x}}\omega],\\
\operatorname{Tr}[\Phi^{\dag}(\tau)\omega]  &  =\operatorname{Tr}[\omega
\Phi(\tau)]\\
&  =\sum_{\overrightarrow{x},\overrightarrow{y}}\phi_{\overrightarrow
{x},\overrightarrow{y}}\operatorname{Tr}\!\left[  \left(  \sigma
_{\overrightarrow{x}}\otimes\sigma_{\overrightarrow{y}}\right)  \left(
\omega\otimes\tau\right)  \right]  ,\\
\operatorname{Tr}[\Phi^{\dag}(\tau)A]  &  =\operatorname{Tr}\!\left[  \left(
\sum_{\overrightarrow{x_{1}},\overrightarrow{y}}\phi_{\overrightarrow{x_{1}%
},\overrightarrow{y}}\sigma_{\overrightarrow{x_{1}}}\operatorname{Tr}%
[\sigma_{\overrightarrow{y}}\tau]\right)  \sum_{\overrightarrow{x_{2}}}%
\alpha_{\overrightarrow{x_{2}}}\sigma_{\overrightarrow{x_{2}}}\right] \\
&  =\sum_{\overrightarrow{x_{1}},\overrightarrow{y},\overrightarrow{x_{2}}%
}\alpha_{\overrightarrow{x_{2}}}\phi_{\overrightarrow{x_{1}},\overrightarrow
{y}}\operatorname{Tr}[\sigma_{\overrightarrow{y}}\tau]\operatorname{Tr}%
\!\left[  \sigma_{\overrightarrow{x_{1}}}\sigma_{\overrightarrow{x_{2}}%
}\right] \\
&  =2^{n}\sum_{\overrightarrow{y}}\left(  \sum_{\overrightarrow{x}}%
\alpha_{\overrightarrow{x}}\phi_{\overrightarrow{x},\overrightarrow{y}%
}\right)  \operatorname{Tr}[\sigma_{\overrightarrow{y}}\tau].
\end{align}
Combining these expressions according to
\eqref{eq:dual-constraint-expansion-full} of
Lemma~\ref{lem:expansion-obj-funcs} leads to \eqref{eq:constraint-dual-pauli}.
\end{proof}

\bigskip

Let us recall here how to perform sampling in order to estimate an expectation
of the form $\operatorname{Tr}[\sigma_{\overrightarrow{x}}\rho]$. Let us
denote the eigendecomposition of the Pauli matrix $\sigma_{i}$ for
$i\in\left\{  0,1,2,3\right\}  $ by%
\begin{equation}
\sigma_{i}=\sum_{y\in\left\{  0,1\right\}  }\left(  -1\right)  ^{y\cdot
f(i)}|\phi_{y,i}\rangle\!\langle\phi_{y,i}|,
\end{equation}
where%
\begin{equation}
f(i)\coloneqq \left\{
\begin{array}
[c]{cc}%
1 & \text{if }i\in\left\{  1,2,3\right\} \\
0 & \text{if }i=0
\end{array}
\right.  ,
\end{equation}
and we set%
\begin{align}
|\phi_{0,0}\rangle &  =|\phi_{0,3}\rangle\equiv|0\rangle,\\
|\phi_{1,0}\rangle &  =|\phi_{1,3}\rangle\equiv|1\rangle,\\
|\phi_{0,1}\rangle &  \equiv|+\rangle,\\
|\phi_{1,1}\rangle &  \equiv|-\rangle,\\
|\phi_{0,2}\rangle &  \equiv|+_{Y}\rangle,\\
|\phi_{1,2}\rangle &  \equiv|-_{Y}\rangle.
\end{align}
Then, with this notation, we can write the spectral decomposition of
$\sigma_{\overrightarrow{x}}$ as follows:%
\begin{align}
\sigma_{\overrightarrow{x}}  &  =\sigma_{x_{1}}\otimes\cdots\otimes
\sigma_{x_{n}}\\
&  =\left(  \sum_{y_{1}\in\left\{  0,1\right\}  }\left(  -1\right)
^{y_{1}\cdot f(x_{1})}|\phi_{y_{1},x_{1}}\rangle\!\langle\phi_{y_{1},x_{1}%
}|\right)  \otimes\cdots\nonumber\\
&  \qquad\qquad\otimes\left(  \sum_{y_{n}\in\left\{  0,1\right\}  }\left(
-1\right)  ^{y_{n}\cdot f(x_{n})}|\phi_{y_{n},x_{n}}\rangle\!\langle
\phi_{y_{n},x_{n}}|\right) \\
&  =\sum_{y_{1},\ldots,y_{n}\in\left\{  0,1\right\}  }\left(  -1\right)
^{\sum_{j=1}^{n}y_{j}\cdot f(x_{i})}|\phi_{y_{1},x_{1}}\rangle\!\langle
\phi_{y_{1},x_{1}}|\otimes\cdots\otimes|\phi_{y_{n},x_{n}}\rangle\!\langle
\phi_{y_{n},x_{n}}|\\
&  =\sum_{y_{1},\ldots,y_{n}\in\left\{  0,1\right\}  }\left(  -1\right)
^{\overrightarrow{y}\cdot f(\overrightarrow{x})}|\phi_{y_{1},x_{1}}%
\rangle\!\langle\phi_{y_{1},x_{1}}|\otimes\cdots\otimes|\phi_{y_{n},x_{n}%
}\rangle\!\langle\phi_{y_{n},x_{n}}|,
\end{align}
where $f(\overrightarrow{x})\equiv\left(  f(x_{1}),\ldots,f(x_{n})\right)  $.
Then it follows that%
\begin{align}
\operatorname{Tr}[\sigma_{\overrightarrow{x}}\rho]  &  =\sum_{y_{1}%
,\ldots,y_{n}\in\left\{  0,1\right\}  }\left(  -1\right)  ^{\overrightarrow
{y}\cdot f(\overrightarrow{x})}p(\overrightarrow{y}|\overrightarrow{x}),\\
p(\overrightarrow{y}|\overrightarrow{x})  &  \coloneqq\left(  \langle
\phi_{y_{1},x_{1}}|\otimes\cdots\otimes\langle\phi_{y_{n},x_{n}}|\right)
\rho\left(  |\phi_{y_{1},x_{1}}\rangle\!\langle\phi_{y_{1},x_{1}}%
|\otimes\cdots\otimes|\phi_{y_{n},x_{n}}\rangle\right)  .
\end{align}

The following algorithm outputs an estimate of $\operatorname{Tr}%
[\sigma_{\overrightarrow{x}}\rho]$:

\begin{algorithm}
\label{alg:pauli-expectation-obs}Given is an $n$-qubit state $\rho$ and an
$n$-dimensional vector $\overrightarrow{x}$, with each entry taking values in
$\left\{  0,1,2,3\right\}  $, such that $\overrightarrow{x}$ specifies the
Pauli string $\sigma_{\overrightarrow{x}}$.

\begin{enumerate}
\item Fix $\varepsilon>0$ and $\delta\in\left(  0,1\right)  $. Set $T\geq
\frac{1}{2\varepsilon^{2}}\ln\!\left(  \frac{2}{\delta}\right)  $ and set
$t=1$.

\item For $j\in\left[  n\right]  $, measure qubit $j$ of $\rho$ in the Pauli
basis $\left\{  |\phi_{y_{j},x_{j}}\rangle\right\}  _{y_{j}\in\left\{
0,1\right\}  }$ and record the outcome $y_{j}$.

\item Set $Z_{t}=\left(  -1\right)  ^{\overrightarrow{y}\cdot
f(\overrightarrow{x})}$.

\item Increment $t$.

\item Repeat Steps 2.-4.~until $t>T$ and then output $\overline{Z}%
\coloneqq\frac{1}{T}\sum_{t=1}^{T}Z_{t}$ as an estimate of $\operatorname{Tr}%
[\sigma_{\overrightarrow{x}}\rho]$.
\end{enumerate}
\end{algorithm}

By the Hoeffding inequality \cite{H63} (see also \cite[Theorem~1]%
{bandyopadhyay2023efficient} for the precise statement that we need), we are
guaranteed that the output of
Algorithm~\ref{alg:pauli-expectation-obs}\ satisfies%
\begin{equation}
\Pr\!\left[  \left\vert \overline{Z}-\operatorname{Tr}[\sigma_{\overrightarrow
{x}}\rho]\right\vert \leq\varepsilon\right]  \geq1-\delta,
\end{equation}
due to the choice $T\geq\frac{1}{2\varepsilon^{2}}\ln\!\left(  \frac{2}%
{\delta}\right)  $.

It is important to determine the overhead of sampling when estimating the
objective functions in
\eqref{eq:app-primal-obj-states}--\eqref{eq:app-dual-obj-states}, for the
Pauli input model. This affects the overall runtime of estimating the
objective functions in
\eqref{eq:app-primal-obj-states}--\eqref{eq:app-dual-obj-states}, and thus the
overall runtime of QSlack. In the case of \eqref{eq:constraint-primal-pauli},
observe that the term $2^{m}\left\Vert \overrightarrow{\beta}\right\Vert
_{2}^{2}$ can be precomputed, and so it is not necessary to factor it in when
determining the overhead of sampling. The same goes for the term
$2^{n}\left\Vert \overrightarrow{\alpha}\right\Vert _{2}^{2}$ in
\eqref{eq:constraint-dual-pauli}. As such, our goal is to find upper bounds on
the absolute values of \eqref{eq:app-primal-obj-states} and
\eqref{eq:app-dual-obj-states} with the aforementioned terms subtracted out.
Then consider that%
\begin{multline}
\left\vert \lambda\operatorname{Tr}[A\rho]-c\left(  \left\Vert B-\lambda
\Phi(\rho)-\mu\sigma\right\Vert _{2}^{2}-2^{m}\left\Vert \overrightarrow
{\beta}\right\Vert _{2}^{2}\right)  \right\vert \leq\lambda\left\Vert
\overrightarrow{\alpha}\right\Vert _{1}\\
+c\left(
\begin{array}
[c]{c}%
\lambda^{2}2^{m}\sum_{\overrightarrow{x_{1}},\overrightarrow{x_{2}}}\left\vert
\sum_{\overrightarrow{y}}\phi_{\overrightarrow{x_{1}},\overrightarrow{y}}%
\phi_{\overrightarrow{x_{2}},\overrightarrow{y}}\right\vert +\mu^{2}%
+\lambda2^{m+1}\sum_{\overrightarrow{x}}\left\vert \sum_{\overrightarrow{y}%
}\beta_{\overrightarrow{y}}\phi_{\overrightarrow{x},\overrightarrow{y}%
}\right\vert \\
+2\mu\left\Vert \overrightarrow{\beta}\right\Vert _{1}+2\lambda\mu\left\Vert
\vec{\phi}\right\Vert _{1}%
\end{array}
\right)  ,
\end{multline}%
\begin{multline}
\left\vert \kappa\operatorname{Tr}[B\tau]+c\left(  \left\Vert \kappa\Phi
^{\dag}(\tau)-A-\nu\omega\right\Vert _{2}^{2}-2^{n}\left\Vert \overrightarrow
{\alpha}\right\Vert _{2}^{2}\right)  \right\vert \leq\kappa\left\Vert
\overrightarrow{\beta}\right\Vert _{1}\\
+c\left(
\begin{array}
[c]{c}%
\kappa^{2}2^{n}\sum_{\overrightarrow{y_{1}},\overrightarrow{y_{2}}}\left\vert
\sum_{\overrightarrow{x}}\phi_{\overrightarrow{x},\overrightarrow{y_{1}}}%
\phi_{\overrightarrow{x},\overrightarrow{y_{2}}}\right\vert +\nu^{2}%
+\kappa2^{n+1}\sum_{\overrightarrow{y}}\left\vert \sum_{\overrightarrow{x}%
}\alpha_{\overrightarrow{x}}\phi_{\overrightarrow{x},\overrightarrow{y}%
}\right\vert \\
+2\nu\left\Vert \overrightarrow{\alpha}\right\Vert _{1}+2\kappa\nu\left\Vert
\vec{\phi}\right\Vert _{1}%
\end{array}
\right)  .
\end{multline}
These inequalities follow from repeated application of the triangle inequality
and the fact that $\left\vert \operatorname{Tr}[\sigma_{\overrightarrow{x}%
}\rho]\right\vert \leq1$ for every Pauli string $\sigma_{\overrightarrow{x}}$
and state $\rho$. Thus, by inspecting the above upper bounds, we require that
\begin{equation}
\max\left\{
\begin{array}
[c]{c}%
2^{m}\sum_{\overrightarrow{x_{1}},\overrightarrow{x_{2}}}\left\vert
\sum_{\overrightarrow{y}}\phi_{\overrightarrow{x_{1}},\overrightarrow{y}}%
\phi_{\overrightarrow{x_{2}},\overrightarrow{y}}\right\vert ,2^{m}%
\sum_{\overrightarrow{x}}\left\vert \sum_{\overrightarrow{y}}\beta
_{\overrightarrow{y}}\phi_{\overrightarrow{x},\overrightarrow{y}}\right\vert
,\\
2^{n}\sum_{\overrightarrow{y_{1}},\overrightarrow{y_{2}}}\left\vert
\sum_{\overrightarrow{x}}\phi_{\overrightarrow{x},\overrightarrow{y_{1}}}%
\phi_{\overrightarrow{x},\overrightarrow{y_{2}}}\right\vert ,2^{n}%
\sum_{\overrightarrow{y}}\left\vert \sum_{\overrightarrow{x}}\alpha
_{\overrightarrow{x}}\phi_{\overrightarrow{x},\overrightarrow{y}}\right\vert
\end{array}
\right\}  =\text{poly}(n,m), \label{eq:conditions-efficient-sampling-pauli}%
\end{equation}
in order for the sampling complexity of estimating the objective functions to
be efficient. This demand thus places a constraint on the coefficients
$\alpha_{\overrightarrow{x}}$, $\beta_{\overrightarrow{y}}$, and~$\phi
_{\overrightarrow{x},\overrightarrow{y}}$.

In the above analysis, we did not consider reducing the sampling complexity by
exploiting commutation. Indeed, when estimating multiple terms of the form
$\operatorname{Tr}[\sigma_{\overrightarrow{x_{1}}}\rho]$ and
$\operatorname{Tr}[\sigma_{\overrightarrow{x_{2}}}\rho]$, there can be a
significant reduction in sampling complexity if the Pauli strings
$\sigma_{\overrightarrow{x_{1}}}$ and $\sigma_{\overrightarrow{x_{2}}}$
commute. In such a case, it is possible to estimate these observables
simultaneously and use the results of repetitions of a single measurement to
estimate both $\operatorname{Tr}[\sigma_{\overrightarrow{x_{1}}}\rho]$ and
$\operatorname{Tr}[\sigma_{\overrightarrow{x_{2}}}\rho]$, without having to do
repetitions of multiple different measurements. Let us note that there has
been significant work on finding maximally commutative subtuples of a tuple
$\left(  \sigma_{\overrightarrow{x}}\right)  _{\overrightarrow{x}}$ of Pauli
strings and corresponding post-processing schemes\ for estimation (see, e.g.,
\cite[Section~4.2.2]{McClean2016} and
\cite{jena2019pauli,Huang2021,Shlosberg2023adaptiveestimation}), and the
emphasis on this problem is precisely due to the goal of reducing the sampling
complexity of algorithms that make use of estimates of Pauli expectations.

\subsection{Hybrid input and optimization models}

We briefly remark here that some problems admit hybrids of the input models
and optimizations discussed in Appendices~\ref{app:linear-combo-states} and
\ref{app:Pauli-input-model}. For example, the matrix $A$ in~\eqref{eq:primal-SDP}--\eqref{eq:dual-SDP}\ could be a linear combination of
states, the matrix $B$ in \eqref{eq:primal-SDP}--\eqref{eq:dual-SDP} could be
a linear combination of Paulis, and the superoperator $\Phi$ could be a hybrid
of both, i.e., of the form%
\begin{equation}
\Phi(\cdot)=\sum_{i,\overrightarrow{y}}\phi_{i,\overrightarrow{y}}%
\sigma_{\overrightarrow{y}}\operatorname{Tr}[\rho_{i}(\cdot)],
\end{equation}
where $\left(  \rho_{i}\right)  _{i}$ is a tuple of states and $\left(
\sigma_{\overrightarrow{y}}\right)  _{\overrightarrow{y}}$ is a tuple of Pauli
strings. In this case, one could repeat the analysis given in the previous
appendices to determine how precisely to evaluate the objective functions in
\eqref{eq:primal-q-states}--\eqref{eq:dual-q-states}. We leave this as an
exercise for this example.

Additionally, the optimizations being performed could also be hybrid. While we
replace optimizations over positive semi-definite matrices with scaled,
parameterized states, if an optimization is over a general matrix or a
Hermitian matrix, we could replace these optimizations with linear
combinations of Pauli matrices with complex and real coefficients,
respectively. We mentioned this point explicitly in
Section~\ref{sec:qslack-alg},\ and it also came up in the examples of fidelity
and entanglement negativity in Sections~\ref{sec:fidelity-main-text}\ and
\ref{sec:entanglement-negativity}, respectively.

Let us also note that several of the examples we considered in the main text
established the usefulness of considering a variety of input models and
optimizations. The normalized trace distance
(Section~\ref{sec:normalized-TD-example})\ made use of the linear combination
of states input model, as well as optimizations over parameterized states. The
fidelity (Section~\ref{sec:fidelity-main-text}) made use of the linear
combination of states input model and included optimizations over
parameterized states and a linear combination of Paulis with complex
coefficients. The entanglement negativity
(Section~\ref{sec:entanglement-negativity}) made use of the linear combination
of states input model and included optimizations over parameterized states and
linear combinations of Paulis with real coefficients. Finally, constrained
Hamiltonian optimization (Section~\ref{sec:constrained-Ham-opt}) made use of
the Pauli input model and included optimizations over parameterized states.

\subsection{Linear combination of distributions input model}

\label{app:linear-combo-dist-model}

As remarked at the end of Appendix~\ref{sec:gen-form-input-model}, every input
model for the SDP / quantum case has a reduction for the LP / classical case.
This includes the linear combination of states model, which becomes the linear
combination of distributions model. In short, $A$ and $B$ in
\eqref{eq:linear-combo-states-A} and \eqref{eq:linear-combo-states-B},
respectively, become vectors that are linear combinations of probability
distributions, and the superoperator $\Phi$ in
\eqref{eq:linear-combo-states-Phi} becomes a matrix. In detail, $\left(
\rho_{i}\right)  _{i}$, $\left(  \tau_{j}\right)  _{j}$, $\left(  \sigma
_{k}\right)  _{k}$, and $\left(  \omega_{\ell}\right)  _{\ell}$ become tuples
of probability vectors, and the trace overlap in
\eqref{eq:linear-combo-states-Phi} becomes a vector overlap (i.e., standard
vector inner product). The same holds for all of the subsequent terms in \eqref{eq:linear-comb-states-dev-1}--\eqref{eq:linear-comb-states-dev-last}.

In order to estimate vector overlap terms like $p^{T}q$ for probability
distributions $\left(  p(x)\right)  _{x}$ and $\left(  q(y)\right)  _{y}$
defined over the same alphabet, we can use the standard collision test, which
we recall now.

\begin{algorithm}
\label{alg:collision-test}Given are probability distributions $p$ and $q$,
such that we can sample from $p$ and $q$.

\begin{enumerate}
\item Fix $\varepsilon>0$ and $\delta\in\left(  0,1\right)  $. Set $T\geq
\frac{1}{2\varepsilon^{2}}\ln\!\left(  \frac{2}{\delta}\right)  $ and set
$t=1$.

\item Sample $x$ from $p$, and sample $y$ from $q$.

\item Set $Z_{t}=1$ if $x=y$ and set $Z_{t}=0$ otherwise.

\item Increment $t$.

\item Repeat Steps 2.-4.~until $t>T$ and then output $\overline{Z}%
\coloneqq\frac{1}{T}\sum_{t=1}^{T}Z_{t}$ as an estimate of $p^{T}q$.
\end{enumerate}
\end{algorithm}

By the Hoeffding inequality \cite{H63} (see also \cite[Theorem~1]%
{bandyopadhyay2023efficient} for the precise statement that we need), we are
guaranteed that the output of Algorithm~\eqref{alg:collision-test}\ satisfies%
\begin{equation}
\Pr\!\left[  \left\vert \overline{Z}-p^{T}q\right\vert \leq\varepsilon\right]
\geq1-\delta,
\end{equation}
due to the choice $T\geq\frac{1}{2\varepsilon^{2}}\ln\!\left(  \frac{2}%
{\delta}\right)  $.

\subsection{Walsh--Hadamard input model}

\label{app:Walsh-Hadamard-input-model}

Another classical input model to consider is the Walsh--Hadamard\ input model,
which is the classical version of the Pauli input model, the latter considered
already in Appendix~\ref{app:Pauli-input-model}. The Walsh--Hadamard input
model retains only the diagonal Pauli matrices $\sigma_{I}$ and $\sigma_{Z}$
and represents their diagonal entries as the following vectors:%
\begin{equation}
s_{0}\coloneqq%
\begin{bmatrix}
1\\
1
\end{bmatrix}
,\qquad s_{1}\coloneqq%
\begin{bmatrix}
1\\
-1
\end{bmatrix}
.
\end{equation}
These lead to the following tensor-product vectors that are orthogonal basis
vectors of the Walsh--Hadamard\ basis:%
\begin{equation}
s_{\overrightarrow{x}}\coloneqq s_{x_{1}}\otimes\cdots\otimes s_{x_{n}}.
\end{equation}

The Walsh--Hadamard\ input model has real input vectors $A$ and $B$,
represented as in \eqref{eq:Pauli-rep-A}--\eqref{eq:Pauli-rep-B}, but with
$\sigma_{\overrightarrow{x}}$ and $\sigma_{\overrightarrow{y}}$ replaced with
$s_{\overrightarrow{x}}$ and $s_{\overrightarrow{y}}$, respectively. The
superoperator $\Phi$ in \eqref{eq:Pauli-rep-Phi} becomes a matrix, with
$\sigma_{\overrightarrow{x}}$ and $\sigma_{\overrightarrow{y}}$ replaced with
$s_{\overrightarrow{x}}$ and $s_{\overrightarrow{y}}$, and the trace overlap
therein is replaced with the standard vector inner product. Similarly, all of
the analysis in
\eqref{eq:pauli-input-dev-1}--\eqref{eq:conditions-efficient-sampling-pauli}
features similar replacements, using the fact that $s_{\overrightarrow{x}}%
^{T}s_{\overrightarrow{y}}=2^{n}\delta_{\overrightarrow{x},\overrightarrow{y}%
}$.

It is worthwhile to remark how inner products like $s_{\overrightarrow{x}}%
^{T}p$, where $p$ is a $2^{n}$-dimensional probability vector are the
classical reduction of the expectation of an observable $\operatorname{Tr}%
[\sigma_{\overrightarrow{x}}\rho]$, and how we can efficiently estimate such
overlaps by a classical sampling procedure. Consider the standard basis%
\begin{equation}
e_{0}\coloneqq
\begin{bmatrix}
1\\
0
\end{bmatrix}
,\qquad e_{1}\coloneqq
\begin{bmatrix}
0\\
1
\end{bmatrix}
.
\end{equation}
We can write a $2^{n}\times1$ probability vector $p$ in terms of the
standard basis as%
\begin{equation}
p=\sum_{\overrightarrow{i}}p_{\overrightarrow{i}}\ e_{i_{1}}\otimes e_{i_{2}%
}\otimes\cdots\otimes e_{i_{n}}.
\end{equation}
Observe that, for $i,j\in\left\{  0,1\right\}  $,%
\begin{equation}
s_{i}^{T}e_{j}=\left(  -1\right)  ^{i\cdot j}.
\end{equation}
This means that the inner product $s_{\overrightarrow{x}}^{T}p$ can be
evaluated as%
\begin{align}
s_{\overrightarrow{x}}^{T}p  &  =\left(  s_{x_{1}}\otimes s_{x_{2}}%
\otimes\cdots\otimes s_{x_{n}}\right)  ^{T}\left(  \sum_{\overrightarrow{i}%
}p_{\overrightarrow{i}}\ e_{i_{1}}\otimes\cdots\otimes e_{i_{n}}\right) \\
&  =\sum_{\overrightarrow{i}}p_{\overrightarrow{i}}\ s_{x_{1}}^{T}e_{i_{1}%
}\cdot s_{x_{2}}^{T}e_{i_{2}}\cdot\cdots\cdot s_{x_{n}}^{T}e_{i_{n}}\\
&  =\sum_{\overrightarrow{i}}p_{\overrightarrow{i}}\ \left(  -1\right)
^{x_{1}\cdot i_{1}}\left(  -1\right)  ^{x_{2}\cdot i_{2}}\cdots\left(
-1\right)  ^{x_{n}\cdot i_{n}}\\
&  =\sum_{\overrightarrow{i}}p_{\overrightarrow{i}}\ \left(  -1\right)
^{\overrightarrow{x}\cdot\overrightarrow{i}},
\end{align}
which demonstrates that%
\begin{equation}
s_{\overrightarrow{x}}^{T}p=\mathbb{E}_{p_{\overrightarrow{i}}}[\left(
-1\right)  ^{\overrightarrow{x}\cdot\overrightarrow{I}}].
\end{equation}
As such, the value $s_{\overrightarrow{x}}^{T}p$ can be estimated by a
classical sampling approach, according to the following procedure:

\begin{algorithm}
\label{alg:classical-expectation-obs}Given is a $2^{n}$-dimensional
probability distribution $p$, such that we can sample from $p$, and a
length-$n$ bitstring $\overrightarrow{x}$, which specifies the $2^{n}$-dimensional
Walsh--Hadamard\ vector $s_{\overrightarrow{x}}$.

\begin{enumerate}
\item Fix $\varepsilon>0$ and $\delta\in\left(  0,1\right)  $. Set $T\geq
\frac{1}{2\varepsilon^{2}}\ln\!\left(  \frac{2}{\delta}\right)  $ and set
$t=1$.

\item Sample $\overrightarrow{i}$ from $p$.

\item Set $Z_{t}=\left(  -1\right)  ^{\overrightarrow{x}\cdot\overrightarrow
{i}}$.

\item Increment $t$.

\item Repeat Steps 2.-4.~until $t>T$ and then output $\overline{Z}%
\coloneqq\frac{1}{T}\sum_{t=1}^{T}Z_{t}$ as an estimate of $s_{\overrightarrow{x}}^{T}p$. 
\end{enumerate}
\end{algorithm}

By the Hoeffding inequality \cite{H63} (see also \cite[Theorem~1]%
{bandyopadhyay2023efficient} for the precise statement that we need), we are
guaranteed that the output of Algorithm~\ref{alg:classical-expectation-obs}%
\ satisfies%
\begin{equation}
\Pr\!\left[  \left\vert \overline{Z}-s_{\overrightarrow{x}}^{T}p\right\vert
\leq\varepsilon\right]  \geq1-\delta,
\end{equation}
due to the choice $T\geq\frac{1}{2\varepsilon^{2}}\ln\!\left(  \frac{2}%
{\delta}\right)  $.

\section{Ans\"{a}tze for parameterizing mixed states}

\label{sec:purification-CC-ansatze}In this appendix, we describe two methods
for parameterizing the set of density matrices by means of parameterized
quantum circuits. The first is the purification ansatz, which has already been
used in \cite{CSZW20,patel2021variational,Ezzell2023}. The second is the convex
combination ansatz, used in
\cite{verdon2019quantum,Liu2021,sbahi2022provably,Ezzell2023} and \tn{so} called quantum
Hamiltonian-based models in \cite{verdon2019quantum,sbahi2022provably}. The
general concept of the purification ansatz goes back to
\cite{Uhlmann1976,Uhlmann1986}, and this concept has been used in quantum
state reconstruction \cite{Buzek1998}\ and extensively in the analysis of
many-body physics \cite{Verstraete2004} (therein called matrix product density
operators; see also \cite{GuthJarkovsky2020}).

Before we review these, let us briefly review the notion of a parameterized
unitary and a parameterized pure state. Let $\theta=\left(  \theta_{1}%
,\ldots,\theta_{L}\right)  $ be a parameter vector, where $\theta_{j}%
\in\mathbb{R}$ for all $j\in\left[  L\right]  \coloneqq\left\{  1,\ldots
,L\right\}  $. A parameterized unitary $U(\theta)$\ consists of an alternating
sequence of parameterized gates and unparameterized gates, defined as%
\begin{equation}
U(\theta)\coloneqq W_{L}V_{L}(\theta_{L})W_{L-1}V_{L-1}(\theta_{L-1})\cdots
W_{2}V_{2}(\theta_{2})W_{1}V_{1}(\theta_{1}), \label{eq:param-unitary}%
\end{equation}
where, for $j\in\left[  L\right]  $, the unitary $W_{j}$ is unparameterized
and the unitary $V_{j}(\theta_{j})\coloneqq e^{-iH_{j}\theta_{j}}$ is
parameterized, with $H_{j}$ a fixed Hamiltonian. The most common choices in
the literature are for $W_{j}$ to be a CNOT or a controlled-phase gate
coupling two qubits or to simply be the identity operator, and for $H_{j}$ to
be a single- or two-qubit Pauli operator, realizing a rotation. In such cases,
$W_{j}$ acts non-trivially on just two qubits and $V_{j}(\theta_{j})$ acts
non-trivially on just one or two qubits, so that the unitaries act as the
identity on all of the other qubits.

The form in \eqref{eq:param-unitary} is quite general and allows for layered
ans\"{a}tze as well. For example, if $L=4$, and there are two qubits, we could
have $V_{1}(\theta_{1})=e^{-iX_{1}\theta_{1}}$, $W_{1}=I^{\otimes2}$,
$V_{2}(\theta_{2})=e^{-iX_{2}\theta_{2}}$, $W_{2}=~$CNOT, $V_{3}(\theta
_{3})=e^{-iZ_{1}\theta_{3}}$, $W_{3}=I^{\otimes2}$, $V_{4}(\theta
_{4})=e^{-iZ_{2}\theta_{4}}$, $W_{4}=~$CNOT, which gives a first layer of two
single-qubit $X$ rotations acting on the two qubits, followed by a CNOT,
followed by a second layer of two single-qubit $Z$ rotations acting on the two
qubits, and finished by a CNOT. There are many other choices that have been
proposed in the literature \cite{Cerezo2021}.

A parameterized pure state $|\psi(\theta)\rangle$ simply consists of a
parameterized unitary $U(\theta)$ applied to a fixed initial state that is
easy to prepare, such as the tensor-product state $|0\rangle^{\otimes n}$:%
\begin{equation}
|\psi(\theta)\rangle\coloneqq U(\theta)|0\rangle^{\otimes n}.
\label{eq:pure-state-ansatz}%
\end{equation}
The idea behind these parameterizations is for them to be as expressive as
possible, with the hope being that the solution to a given problem, such as
the ground-state energy minimization problem, is contained within the set
$\left\{  |\psi(\theta)\rangle\right\}  _{\theta\in\mathbb{R}^{L}}$.

\subsection{Purification ansatz}

\label{app:purification-ansatz}A first approach to parameterizing the set of
mixed states is to exploit the idea of purification
\cite{bures1969extension,Uhlmann1976,Uhlmann1986}. That is, every density
matrix $\rho_{S}$ on a system $S$ can be understood as arising from lack of
access to a reference system $R$ of a pure state $\psi_{RS}\equiv|\psi
\rangle\!\langle\psi|_{RS}$. Thus, one obtains $\rho_{S}$ from $\psi_{RS}$ by
performing a partial trace:%
\begin{equation}
\rho_{S}=\operatorname{Tr}_{R}[\psi_{RS}].
\end{equation}
Every state $\rho_{S}$ can be expressed as a convex combination of pure
states, as follows:%
\begin{equation}
\rho_{S}=\sum_{x}p(x)\psi_{S}^{x},\label{eq:mixed-convex-combo-of-pure}%
\end{equation}
where $\left(  p(x)\right)  _{x}$ is a probability distribution and $\left(
\psi_{S}^{x}\right)  _{x}$ is a tuple of pure states, with $\psi_{S}^{x}%
\equiv|\psi^{x}\rangle\!\langle\psi^{x}|_{S}$. As such, a canonical
construction of a purification $|\psi\rangle_{RS}$ of $\rho_{S}$ is as
follows:%
\begin{equation}
|\psi\rangle_{RS}=\sum_{x}\sqrt{p(x)}|x\rangle_{R}|\psi^{x}\rangle_{S},
\end{equation}
where $\left\{  |x\rangle\right\}  _{x}$ is an orthonormal basis.

With the above being recalled, the basic idea behind the purification ansatz
is to parameterize the set of mixed states simply by parameterizing the set of
pure states on a larger number of qubits and then not performing any
operations on the reference system, so that they are effectively discarded or
traced out. More specifically, we define%
\begin{equation}
\rho_{S}(\theta)=\operatorname{Tr}_{R}[|\psi(\theta)\rangle\!\langle
\psi(\theta)|_{RS}], \label{eq:purification-ansatz}%
\end{equation}
where $|\psi(\theta)\rangle$ is defined in \eqref{eq:pure-state-ansatz}, but
with the understanding that the unitary $U(\theta)$ acts non-trivially on both
the registers $R$ and $S$. To access the most general set of mixed states on
system $S$, the reference system $R$ should be just as large as the system
$S$. In this case, it is highly likely that the rank of the resulting state
$\rho_{S}(\theta)$ is full, i.e., equal to $2^{n}$. However, if one would like
to prepare lower rank states, one could take the dimension of the reference
system $R$ to be less than that of system $S$. In general, if the reference
system~$R$ consists of $n_{R}$ qubits and the system $S$ consists of $n$ qubits,
then it is guaranteed that the rank of $\rho_{S}(\theta)$ does not exceed
$2^{n_{R}}$, which is a simple way to place a limitation on the rank. In the
case that the reference system is trivial (i.e., consisting of no qubits, so
that $n_{R}=0$), then the purification ansatz reduces to the usual ansatz for
generating pure states.

\subsection{Convex-combination ansatz}

\label{app:convex-combo-ansatz}

The convex-combination ansatz makes use of the representation of mixed states
in~\eqref{eq:mixed-convex-combo-of-pure}, but with the restriction that
$\left\{  |\psi^{x}\rangle_{S}\right\}  _{x}$ is an orthonormal basis. This
restriction is in fact general, given that every mixed state admits a spectral
decomposition of such a form, with $\left(  p(x)\right)  _{x}$ the tuple of
eigenvalues and $\left(  |\psi^{x}\rangle_{S}\right)  _{x}$ the tuple of
eigenvectors. We can then write%
\begin{equation}
\rho_{S}=\sum_{x}p(x)U|x\rangle\!\langle x|U^{\dag},
\end{equation}
where $\left\{  |x\rangle\right\}  _{x}$ is the canonical or standard basis
and $U$ is a unitary satisfying $U|x\rangle=|\psi^{x}\rangle_{S}$ for all $x$.

The idea behind the convex-combination ansatz is then to parameterize the
eigenvalues and eigenvectors separately, as $\left(  p_{\varphi}(x)\right)
_{x}$ and $\left(  U(\gamma)|x\rangle\right)  _{x}$, where $\varphi$ and
$\gamma$ are parameter vectors, so that the state is parameterized in terms of
them as follows:%
\begin{equation}
\rho_{S}(\varphi,\gamma)=\sum_{x}p_{\varphi}(x)U(\gamma)|x\rangle\!\langle
x|U(\gamma)^{\dag}.\label{eq:convex-comb-ansatz-basic}%
\end{equation}
The tuple $\left(  p_{\varphi}(x)\right)  _{x}$ is then a parameterized
probability distribution and $\left(  U(\gamma)|x\rangle\right)  _{x}$ is a
parameterized tuple of states. This effectively decomposes the
parameterization into a classical part and a quantum part.

The parameterized probability distribution $p_{\varphi}$ can be realized in
the standard way that generative models are realized classically
\cite{murphy2012machine}:\ first prepare some input bits uniformly at random
or real values randomly according to standard Gaussians, propagate them
through a parameterized neural network, which generates output bits by
thresholding. Alternatively, one could generate $p_{\varphi}$ by means of
a\ quantum circuit Born machine \cite{Benedetti2019}, which amounts to
applying a parameterized unitary $U^{\prime}(\varphi)$ of the form in
\eqref{eq:param-unitary} to a standard state $|0\rangle^{\otimes n}$ and
measuring the final state in the computational basis to obtain an outcome $x$,
i.e.,%
\begin{equation}
p_{\varphi}(x)=\left\vert \langle x|U^{\prime}(\varphi)|0\rangle^{\otimes
n}\right\vert ^{2}. \label{eq:qcbm-dist}%
\end{equation}
With this approach, there is the possibility of efficiently generating complex
probability distributions that are inaccessible to classical devices
\cite{Movassagh2023}\ (i.e., it is believed that they cannot efficiently
sample from such complex distributions).

To generate a state $\rho_{S}(\varphi,\gamma)$, one first samples $x$ from the
distribution $p_{\varphi}$, prepares the quantum computer in the computational
basis state $|x\rangle$, and then applies the unitary $U(\gamma)$. This is
summarized by the following flow diagram:%
\begin{equation}
\underrightarrow{\text{sample from }p_{\varphi}}\quad x\quad\underrightarrow
{\text{prepare}}\quad|x\rangle\quad\underrightarrow{\text{apply unitary}}\quad
U(\gamma)|x\rangle.\label{eq:flow-convex-combo}%
\end{equation}
See also \cite[Figure~1]{verdon2019quantum}. Repeating this process over many
runs or trials then leads to outcomes that are consistent with $\rho
_{S}(\varphi,\gamma)$ being the state of the system.

\subsubsection{Born convex-combination ansatz as a special case of
purification ansatz}

Let us note here that there is a simple connection between the convex
combination ansatz in \eqref{eq:convex-comb-ansatz-basic} and the purification
ansatz in \eqref{eq:purification-ansatz}, in the case that the
distribution~$p_{\varphi}$ is realized by a quantum circuit Born machine. Let
us call this special case the Born convex-combination ansatz, and let us note, by the observations that follow, its similarity to the state purification principal component analysis ansatz
\cite[Appendix~B.3]{Ezzell2023}\ and the state efficient ansatz \cite{liu2022mitigating}. 
In what follows, we show that the Born convex-combination ansatz is actually a
special case of the purification ansatz. To see this, consider the following
quantum circuit:%
\begin{equation}
\left(  I_{R}\otimes U(\gamma)_{S}\right)  \left(  \sum_{x\in\left\{
0,1\right\}  ^{n}}|x\rangle\!\langle x|\otimes X^{x}\right)  \left(
U^{\prime}(\varphi)_{R}\otimes I_{S}\right)  ,\label{eq:convex-combo-purified}%
\end{equation}
where $X^{x}$ is a shorthand for $X_{1}^{x_{1}}\otimes X_{2}^{x_{2}}%
\otimes\cdots\otimes X_{n}^{x_{n}}$. Thus, it follows that the controlled
unitary in the middle is equivalent to a tensor product of CNOT gates%
\begin{equation}
\sum_{x\in\left\{  0,1\right\}  ^{n}}|x\rangle\!\langle x|\otimes
X^{x}=\bigotimes\limits_{i=1}^{n}\text{CNOT}_{i,n+1},
\end{equation}
with the source and target qubits indexed by the subscripts given above. As
such, by adopting the shorthand $|\overline{0}\rangle\equiv|0\rangle^{\otimes
n}$, the state resulting from the circuit above acting on an initial state
$|\overline{0}\rangle_{R}\otimes|\overline{0}\rangle_{S}$ is as follows:%
\begin{align}
&  |\psi(\varphi,\gamma)\rangle_{RS}\nonumber\\
&  \coloneqq\left(  I_{R}\otimes U(\gamma)_{S}\right)  \left(  \sum
_{x\in\left\{  0,1\right\}  ^{n}}|x\rangle\!\langle x|\otimes X^{x}\right)
\left(  U^{\prime}(\varphi)_{R}\otimes I_{S}\right)  |\overline{0}\rangle
_{R}\otimes|\overline{0}\rangle_{S}\\
&  =\left(  I_{R}\otimes U(\gamma)_{S}\right)  \left(  \sum_{x\in\left\{
0,1\right\}  ^{n}}|x\rangle\!\langle x|\otimes X^{x}\right)  \left(  \left[
U^{\prime}(\varphi)|\overline{0}\rangle_{R}\right]  \otimes|\overline
{0}\rangle_{S}\right)  \\
&  =\sum_{x\in\left\{  0,1\right\}  ^{n}}|x\rangle\!\langle x|U^{\prime
}(\varphi)|\overline{0}\rangle_{R}\otimes U(\gamma)_{S}|x\rangle_{S}\\
&  =\sum_{x\in\left\{  0,1\right\}  ^{n}}\langle x|U^{\prime}(\varphi
)|\overline{0}\rangle_{R}|x\rangle_{R}U(\gamma)_{S}|x\rangle_{S}\\
&  =\sum_{x\in\left\{  0,1\right\}  ^{n}}\sqrt{p(x)}e^{i\xi_{x}}|x\rangle
_{R}U(\gamma)_{S}|x\rangle_{S},
\end{align}
where we have exploited \eqref{eq:qcbm-dist} to write $\langle x|U^{\prime
}(\varphi)|\overline{0}\rangle_{R}=\sqrt{p(x)}e^{i\xi_{x}}$, for some
tuple~$\left(  e^{i\xi_{x}}\right)  _{x}$ of phases. Thus, after tracing over
the reference system $R$, we find that%
\begin{equation}
\operatorname{Tr}_{R}[|\psi(\varphi,\gamma)\rangle\!\langle\psi(\varphi
,\gamma)|_{RS}]=\rho_{S}(\varphi,\gamma),
\end{equation}
where $\rho_{S}(\varphi,\gamma)$ is defined in \eqref{eq:convex-comb-ansatz-basic}.

Although is conceptually helpful to make this connection between the Born
convex-combination ansatz and the purification ansatz, let us emphasize that
it is much simpler to prepare $\rho_{S}(\varphi,\gamma)$ by the process
described in \eqref{eq:flow-convex-combo}. If one were instead to generate
$\rho_{S}(\varphi,\gamma)$ by means of the circuit in
\eqref{eq:convex-combo-purified} followed by tracing out the reference $R$, it
would require maintaining coherence across twice as many qubits, which is
difficult to do in near-term quantum computers, and at the same time, one is
simply throwing away the reference system. Regardless, the connection outlined
above can be helpful for analyzing theoretical aspects of the Born convex
combination ansatz, such as trainability and expressivity.

\subsubsection{Mixed-state Loschmidt echo algorithm for estimating trace
overlap}

\label{app:Loschmidt-echo-alg}In all of the QSlack problems considered in this
paper, an essential component is the ability to estimate trace overlap terms
like $\operatorname{Tr}[\rho\sigma]$ efficiently on quantum computers, where
$\rho$ and $\sigma$ are quantum states. One method for doing so, as emphasized
in the main text, is the destructive swap test (see \cite[Section~2.2]%
{bandyopadhyay2023efficient} for a detailed review of this method). For
$n$-qubit states $\rho$ and $\sigma$, the destructive swap requires $2n$
qubits in total and the total depth of the circuit required is the maximum of
the depths of the circuits needed to prepare $\rho$ and $\sigma$.

In this section, we outline another approach to estimating the trace overlap
$\operatorname{Tr}[\rho\sigma]$, when $\rho$ and $\sigma$ are prepared by the
convex-combination ansatz. This approach is a generalization of
\cite[Algorithm~1]{RASW23} and only requires $n$ qubits, at the cost of its
depth being the sum of the depths of the circuits that prepare $\rho$ and
$\sigma$. So it trades width for depth. Since the algorithm involves time
reversal of a unitary, we call it the \textquotedblleft mixed-state Loschmidt
echo\textquotedblright\ algorithm. To see how it works, let us begin with some
preliminary calculations. Let
\begin{align}
\rho &  =\sum_{x}p(x)U|x\rangle\!\langle x|U^{\dag}%
,\label{eq:rho-convex-combo}\\
\sigma &  =\sum_{y}q(y)V|y\rangle\!\langle y|V^{\dag},
\label{eq:sigma-convex-combo}%
\end{align}
where the distributions and unitaries are implicitly parameterized, as in
\eqref{eq:convex-comb-ansatz-basic}, but we suppress the dependence on the
parameters for notational convenience. Additionally, $\left\{  |x\rangle
\right\}  _{x}$ and $\left\{  |y\rangle\right\}  _{y}$ are both the standard,
computational bases, and we have used different indices $x$ and $y$ for
clarity in the calculation that follows. Defining the conditional probability
distribution%
\begin{equation}
r(x|y)\coloneqq\left\vert \langle x|U^{\dag}V|y\rangle\right\vert ^{2},
\end{equation}
consider that%
\begin{align}
\operatorname{Tr}[\rho\sigma]  &  =\operatorname{Tr}\!\left[  \left(  \sum
_{x}p(x)U|x\rangle\!\langle x|U^{\dag}\right)  \left(  \sum_{y}q(y)V|y\rangle
\!\langle y|V^{\dag}\right)  \right] \\
&  =\sum_{x,y}p(x)q(y)\operatorname{Tr}\!\left[  U|x\rangle\!\langle
x|U^{\dag}V|y\rangle\!\langle y|V^{\dag}\right] \\
&  =\sum_{x,y}p(x)q(y)\left\vert \langle x|U^{\dag}V|y\rangle\right\vert
^{2}\\
&  =\sum_{x}p(x)\sum_{y}r(x|y)q(y)\\
&  =\sum_{x}p(x)\left(  \sum_{x^{\prime},y}r(x^{\prime}|y)q(y)\right)
\delta_{x,x^{\prime}}.
\end{align}

The expression in the last line is what leads to our mixed-state Loschmidt
echo algorithm for estimating $\operatorname{Tr}[\rho\sigma]$. Defining the
independent random variables $X$ and $X^{\prime}$, with respective
distributions $p(x)$ and $t(x^{\prime})\coloneqq\sum_{x^{\prime},y}%
r(x^{\prime}|y)q(y)$, the last line indicates that%
\begin{equation}
\operatorname{Tr}[\rho\sigma]=\mathbb{E}_{p\otimes t}[\boldsymbol{1}%
_{X=X^{\prime}}]=\sum_{x}p(x)t(x),
\end{equation}
where $\boldsymbol{1}_{X=X^{\prime}}$ is an indicator random variable, equal
to one if $X=X^{\prime}$ and equal to zero otherwise. This leads to the
following algorithm for estimating $\operatorname{Tr}[\rho\sigma]$:

\begin{algorithm}
\label{alg:LS-echo-overlap-trace-convex-combo}Given are probability
distributions $p$ and $q$ and unitaries $U$ and $V$, as defined in
\eqref{eq:rho-convex-combo}--\eqref{eq:sigma-convex-combo}, such that we can
sample from $p$ and $q$ and we can act with the unitaries $U$ and $V$.

\begin{enumerate}
\item Fix $\varepsilon>0$ and $\delta\in\left(  0,1\right)  $. Set $T\geq
\frac{1}{2\varepsilon^{2}}\ln\!\left(  \frac{2}{\delta}\right)  $ and set
$t=1$.

\item Sample $x$ from $p$, and sample $y$ from $q$.

\item Prepare the state $U^{\dag}V|y\rangle$, where $y$ is a computational
basis state.

\item Perform a computational basis measurement of $U^{\dag}V|y\rangle$, which
leads to the measurement outcome $x^{\prime}$.

\item Set $Z_{t}=1$ if $x=x^{\prime}$ and set $Z_{t}=0$ otherwise.

\item Increment $t$.

\item Repeat Steps 2.-6.~until $t>T$ and then output $\overline{Z}%
\coloneqq \frac{1}{T}\sum_{t=1}^{T}Z_{t}$ as an estimate of $\operatorname{Tr}%
[\rho\sigma]$.
\end{enumerate}
\end{algorithm}

Observe that Steps 2.-5.~lead to a sample $x^{\prime}$ from $t$, and then we
are simply performing a collision test of $x$ and $x^{\prime}$ in order to
estimate the collision probability $\sum_{x} p(x) t(x)$.

By the Hoeffding inequality \cite{H63} (see also \cite[Theorem~1]%
{bandyopadhyay2023efficient} for the precise statement that we need), we are
guaranteed that the output of
Algorithm~\ref{alg:LS-echo-overlap-trace-convex-combo}\ satisfies%
\begin{equation}
\Pr\!\left[  \left\vert \overline{Z}-\operatorname{Tr}[\rho\sigma]\right\vert
\leq\varepsilon\right]  \geq1-\delta,
\end{equation}
due to the choice $T\geq\frac{1}{2\varepsilon^{2}}\ln\!\left(  \frac{2}%
{\delta}\right)  $.

As already observed in \cite[Eqs.~(24)--(29)]{Ezzell2023}, a similar
mixed-state Loschmidt echo algorithm can be used to estimate the trace overlap
$\operatorname{Tr}[\rho\sigma]$ if $\rho$ is prepared by the convex
combination ansatz but one instead only has sample access to $\sigma$. This
follows because%
\begin{align}
\operatorname{Tr}[\rho\sigma]  &  =\operatorname{Tr}\!\left[  \left(  \sum
_{x}p(x)U|x\rangle\!\langle x|U^{\dag}\right)  \sigma\right]
\label{eq:LE-test-basic-1}\\
&  =\sum_{x}p(x)\operatorname{Tr}\!\left[  U|x\rangle\!\langle x|U^{\dag
}\sigma\right] \\
&  =\sum_{x}p(x)\langle x|U^{\dag}\sigma U|x\rangle\\
&  =\sum_{x}p(x)q(x), \label{eq:LE-test-basic-last}%
\end{align}
where we have defined the probability distribution $q$ as
$q(x)\coloneqq\langle x|U^{\dag}\sigma U|x\rangle$. Observing that one can
sample from the distribution $q$ by 1)\ preparing $\sigma$, 2) applying
$U^{\dag}$, and 3)\ measuring in the standard basis $\left\{  |x\rangle
\right\}  _{x}$, we conclude that one can estimate $\operatorname{Tr}%
[\rho\sigma]$ by precisely the same collision test procedure recalled in
Algorithm~\ref{alg:collision-test}.

\subsubsection{Correlated convex-combination ansatz}

\label{app:gen-convex-combo-ansatz}In the above approach to the convex
combination ansatz, the parameters of the probability distribution
$p_{\varphi}$ and the unitary $U(\theta)$ have no dependence on each other.
Here we briefly outline an alternative approach in which they do have a
dependence, at the cost of a higher complexity. Let us again suppose that we
first generate an outcome $x$ according to the parameterized distribution
$p_{\varphi}$. However, rather than feed this outcome as a computational basis
state into a unitary, we instead feed the outcome $x$ into a parameterized,
deterministic neural network $f_{w}(x)$, the latter parameterized by a vector
$w$. This neural network then outputs $\theta$, which is a parameter vector
that is used to select a unitary $U(\theta)$ that acts on the state
$|\overline{0}\rangle$. In summary, the state output by this process is as
follows:%
\begin{equation}
\rho_{S}(\varphi,w)=\sum_{x}p_{\varphi}(x)U(f_{w}(x))|\overline{0}%
\rangle\!\langle\overline{0}|U(f_{w}(x))^{\dag}.
\end{equation}


\section{Estimating the objective functions in
Equations~\eqref{eq:param-primal-SDP} and \eqref{eq:param-dual-SDP}}

\label{sec:estimating-objective-functions}

This appendix features a brief discussion about estimating the objective
function. For more specific cases and details with certain input models, see
Appendix~\ref{sec:eff-meas-obs-input-models}. To start, let us define%
\begin{align}
& f(\lambda,\mu,\theta_{1},\theta_{2})\nonumber\\
& \coloneqq\lambda\operatorname{Tr}[A\rho(\theta_{1})]-c\left\Vert
B-\lambda\Phi(\rho(\theta_{1}))-\mu\sigma(\theta_{2})\right\Vert _{2}%
^{2}\label{eq:obj-function-lam-mu-1}\\
& =\lambda\operatorname{Tr}[A\rho(\theta_{1})]-c\left(
\begin{array}
[c]{c}%
\operatorname{Tr}[B^{2}]+\lambda^{2}\operatorname{Tr}[(\Phi(\rho(\theta
_{1})))^{2}]+\mu^{2}\operatorname{Tr}[(\sigma(\theta_{2}))^{2}]\\
-2\lambda\operatorname{Tr}[B\Phi(\rho(\theta_{1}))]-2\mu\operatorname{Tr}%
[B\sigma(\theta_{2})]\\
+2\lambda\mu\operatorname{Tr}[\Phi(\rho(\theta_{1}))\sigma(\theta_{2})]
\end{array}
\right)  .\label{eq:obj-function-lam-mu-2}%
\end{align}
Then it follows that $\operatorname{Tr}[A\rho(\theta_{1})]$ and
$\operatorname{Tr}[B\sigma(\theta_{2})]$ are expectations of observables and
can be estimated as such. Rewriting%
\begin{equation}
\operatorname{Tr}[B\Phi(\rho(\theta_{1}))]=\operatorname{Tr}[\Phi^{\dag
}(B)\rho(\theta_{1})],
\end{equation}
this term is also an expecation of an observable. The term $\operatorname{Tr}%
[(\sigma(\theta_{2}))^{2}]$ can be estimated by means of the destructive swap
test or by means of the mixed-state Loschmidt-echo test when using the convex
combination ansatz (see Algorithm~\ref{alg:LS-echo-overlap-trace-convex-combo}%
). However, for the convex-combination ansatz, this particular term is equal to the purity of the underlying distribution, and so it is more sensible to use the collision test in this case to estimate it (see also \cite[Eqs.~(20)--(23)]{Ezzell2023} in this context). The two other terms $\operatorname{Tr}[(\Phi(\rho(\theta_{1})))^{2}]$ and
$\operatorname{Tr}[\Phi(\rho(\theta_{1}))\sigma(\theta_{2})]$ can be expanded
for particular input models and estimated using a combination of observable
expectation estimation and the destructive swap test.

\section{Estimating gradients of the objective functions in
Equations~\eqref{eq:param-primal-SDP} and \eqref{eq:param-dual-SDP}}

\label{sec:estimating-gradients}

\subsection{General considerations}

Let us begin by discussing how to estimate the gradient of the objective
function in~\eqref{eq:param-primal-SDP}\ in general terms (we omit the
discussion for the dual objective function in~\eqref{eq:param-dual-SDP}, as it
follows similarly). Then we consider the purification ansatz discussed in
Appendix~\ref{app:purification-ansatz} and the convex-combination ansatz from
Appendix~\ref{app:convex-combo-ansatz}. Some of the observations made here are
similar to those from \cite{patel2021variational}. Recalling
\eqref{eq:obj-function-lam-mu-1}--\eqref{eq:obj-function-lam-mu-2}, we
conclude that%
\begin{align}
\frac{\partial}{\partial\lambda}f(\lambda,\mu,\theta_{1},\theta_{2}) &
=\operatorname{Tr}[A\rho(\theta_{1})]-2c\lambda\operatorname{Tr}[(\Phi
(\rho(\theta_{1})))^{2}]\nonumber\\
&  \qquad+2c\operatorname{Tr}[B\Phi(\rho(\theta_{1}))]-2c\mu\operatorname{Tr}%
[\Phi(\rho(\theta_{1}))\sigma(\theta_{2})],\\
\frac{\partial}{\partial\mu}f(\lambda,\mu,\theta_{1},\theta_{2}) &
=-2c\mu\operatorname{Tr}[(\sigma(\theta_{2}))^{2}]+2c\operatorname{Tr}%
[B\sigma(\theta_{2})]-2c\lambda\operatorname{Tr}[\Phi(\rho(\theta_{1}%
))\sigma(\theta_{2})].
\end{align}
Thus, these two terms can be estimated using estimates from the objective
function. Furthermore, consider that%
\begin{align}
\nabla_{\theta_{1}}f(\lambda,\mu,\theta_{1},\theta_{2})  & =\lambda
\operatorname{Tr}[A\left(  \nabla_{\theta_{1}}\rho(\theta_{1})\right)
]-c\lambda^{2}2\operatorname{Tr}[\Phi(\rho(\theta_{1}))\Phi(\left(
\nabla_{\theta_{1}}\rho(\theta_{1})\right)  )]\nonumber\\
& \qquad+2c\lambda\operatorname{Tr}[B\Phi(\nabla_{\theta_{1}}\rho(\theta
_{1}))]-2c\lambda\mu\operatorname{Tr}[\Phi(\nabla_{\theta_{1}}\rho(\theta
_{1}))\sigma(\theta_{2})],\label{eq:gradient-theta-1}\\
\nabla_{\theta_{2}}f(\lambda,\mu,\theta_{1},\theta_{2})  & =-2c\mu
^{2}\operatorname{Tr}[\sigma(\theta_{2})(\nabla_{\theta_{2}}\sigma(\theta
_{2}))]+2c\mu\operatorname{Tr}[B\left(  \nabla_{\theta_{2}}\sigma(\theta
_{2})\right)  ]\nonumber\\
& \qquad-2c\lambda\mu\operatorname{Tr}[\Phi(\rho(\theta_{1}))\nabla
_{\theta_{2}}\sigma(\theta_{2})],\label{eq:gradient-theta-2}%
\end{align}
which follows from linearity and because%
\begin{align}
\nabla_{\theta_{1}}\operatorname{Tr}[(\Phi(\rho(\theta_{1})))^{2}]  &
=\nabla_{\theta_{1}}\left(  \operatorname{Tr}[\Phi(\rho(\theta_{1}))\Phi
(\rho(\theta_{1}))]\right)  \\
& =\operatorname{Tr}[\Phi(\left(  \nabla_{\theta_{1}}\rho(\theta_{1})\right)
)\Phi(\rho(\theta_{1}))]\nonumber\\
& \qquad+\operatorname{Tr}[\Phi(\rho(\theta_{1}))\Phi(\left(  \nabla
_{\theta_{1}}\rho(\theta_{1})\right)  )]\\
& =2\operatorname{Tr}[\Phi(\rho(\theta_{1}))\Phi(\left(  \nabla_{\theta_{1}%
}\rho(\theta_{1})\right)  )],\\
\nabla_{\theta_{2}}\operatorname{Tr}[(\sigma(\theta_{2}))^{2}]  &
=2\operatorname{Tr}[\sigma(\theta_{2})(\nabla_{\theta_{2}}\sigma(\theta
_{2}))].
\end{align}

\subsection{Gradient estimation with purification or Born convex-combination
ansatz}

Related to the observations in \cite[Appendix~C]{Ezzell2023}, one can estimate
each of the terms in \eqref{eq:gradient-theta-1}--\eqref{eq:gradient-theta-2}
by means of the parameter-shift rule when using either the purification ansatz
or the Born convex-combination ansatz along with parameterized single-qubit
Pauli rotations. That is, setting $e_{k}$ to be the $k$th standard basis
vector, one makes the substitutions%
\begin{align}
\left[  \nabla_{\theta_{1}}\rho(\theta_{1})\right]  _{k}  & \rightarrow
\frac{1}{2}\left[  \rho(\theta_{1}+e_{k}\pi/2)-\rho(\theta_{1}-e_{k}%
\pi/2)\right]  ,\label{eq:parameter-shift-claim-app}\\
\left[  \nabla_{\theta_{2}}\sigma(\theta_{2})\right]  _{k}  & \rightarrow
\frac{1}{2}\left[  \sigma(\theta_{2}+e_{k}\pi/2)-\sigma(\theta_{2}-e_{k}%
\pi/2)\right]  ,
\end{align}
in \eqref{eq:gradient-theta-1}--\eqref{eq:gradient-theta-2} in order to
evaluate the $k$th component of the gradients $\nabla_{\theta_{1}}%
f(\lambda,\mu,\theta_{1},\theta_{2})$ and $\nabla_{\theta_{2}}f(\lambda
,\mu,\theta_{1},\theta_{2})$. Here we are applying the parameter-shift rule \cite{Li2017,Mitarai2018,Schuld2019}. To see this, consider that%
\begin{align}
2\operatorname{Tr}[\Phi(\rho(\theta_{1}))\Phi(\left(  \nabla_{\theta_{1}}%
\rho(\theta_{1})\right)  )]  & =2\operatorname{Tr}[\Phi^{\dag}(\Phi
(\rho(\theta_{1})))\left(  \nabla_{\theta_{1}}\rho(\theta_{1})\right)  ],\\
\operatorname{Tr}[B\Phi(\nabla_{\theta_{1}}\rho(\theta_{1}))]  &
=\operatorname{Tr}[\Phi^{\dag}(B)(\nabla_{\theta_{1}}\rho(\theta_{1}))],\\
\operatorname{Tr}[\Phi(\nabla_{\theta_{1}}\rho(\theta_{1}))\sigma(\theta
_{2})]  & =\operatorname{Tr}[\Phi^{\dag}(\sigma(\theta_{2}))(\nabla
_{\theta_{1}}\rho(\theta_{1}))].
\end{align}
Now consider using the purification ansatz (or its special case, the Born
convex-combination ansatz), so that $\psi_{RS}^{\rho}(\theta_{1})$ and
$\psi_{RS}^{\sigma}(\theta_{2})$ are bipartite pure states satisfying%
\begin{align}
\rho(\theta_{1})  & =\operatorname{Tr}_{R}[\psi_{RS}^{\rho}(\theta_{1})],\\
\sigma(\theta_{2})  & =\operatorname{Tr}_{R}[\psi_{RS}^{\sigma}(\theta_{2})].
\end{align}
Then it follows that%
\begin{align}
\operatorname{Tr}[A\left(  \nabla_{\theta_{1}}\rho(\theta_{1})\right)  ]  &
=\operatorname{Tr}[\left(  I_{R}\otimes A_{S}\right)  \left(  \nabla
_{\theta_{1}}\psi_{RS}^{\rho}(\theta_{1})\right)  ],\\
2\operatorname{Tr}[\Phi^{\dag}(\Phi(\rho(\theta_{1})))\left(  \nabla
_{\theta_{1}}\rho(\theta_{1})\right)  ]  & =2\operatorname{Tr}[\left(
I_{R}\otimes\Phi^{\dag}(\Phi(\rho(\theta_{1})))\right)  \left(  \nabla
_{\theta_{1}}\psi_{RS}^{\rho}(\theta_{1})\right)  ],\\
\operatorname{Tr}[\Phi^{\dag}(B)(\nabla_{\theta_{1}}\rho(\theta_{1}))]  &
=\operatorname{Tr}[\left(  I_{R}\otimes\Phi^{\dag}(B)\right)  (\nabla
_{\theta_{1}}\psi_{RS}^{\rho}(\theta_{1}))],\\
\operatorname{Tr}[\Phi^{\dag}(\sigma(\theta_{2}))(\nabla_{\theta_{1}}%
\rho(\theta_{1}))]  & =\operatorname{Tr}[\left(  I_{R}\otimes\Phi^{\dag
}(\sigma(\theta_{2}))\right)  (\nabla_{\theta_{1}}\psi_{RS}^{\rho}(\theta
_{1}))],\\
\operatorname{Tr}[\sigma(\theta_{2})(\nabla_{\theta_{2}}\sigma(\theta_{2}))]
& =\operatorname{Tr}[\left(  I_{R}\otimes\sigma(\theta_{2})\right)
(\nabla_{\theta_{2}}\psi_{RS}^{\sigma}(\theta_{2}))],\\
\operatorname{Tr}[B\left(  \nabla_{\theta_{2}}\sigma(\theta_{2})\right)  ]  &
=\operatorname{Tr}[\left(  I_{R}\otimes B\right)  \left(  \nabla_{\theta_{2}%
}\psi_{RS}^{\sigma}(\theta_{2})\right)  ],\\
\operatorname{Tr}[\Phi(\rho(\theta_{1}))\nabla_{\theta_{2}}\sigma(\theta
_{2})]  & =\operatorname{Tr}[\left(  I_{R}\otimes\Phi(\rho(\theta
_{1}))\right)  \left(  \nabla_{\theta_{2}}\psi_{RS}^{\sigma}(\theta
_{2})\right)  ].
\end{align}
As such, we have reduced all terms to be of the form $\operatorname{Tr}%
[H\left(  \nabla_{\theta}\psi(\theta)\right)  ]$, where $H$ is an observable
and $\psi(\theta)$ is a parameterized pure state. Thus the parameter-shift
rule applies to this case. Substituting it and propagating the equalities back
leads to the claim around \eqref{eq:parameter-shift-claim-app}.

\subsection{Gradient estimation with convex-combination ansatz}
\label{sec:gradient-CCA}
For the more general convex-combination ansatz discussed in
Appendix~\ref{app:convex-combo-ansatz}, the states $\rho$ and $\sigma$ have
the following form:%
\begin{align}
\rho(\varphi_{1},\gamma_{1})  & =\sum_{x}p_{\varphi_{1}}(x)U(\gamma
_{1})|x\rangle\!\langle x|U(\gamma_{1})^{\dag},\\
\sigma(\varphi_{2},\gamma_{2})  & =\sum_{y}p_{\varphi_{2}}(y)U(\gamma
_{2})|y\rangle\!\langle y|U(\gamma_{2})^{\dag}.
\end{align}
Setting $\theta_{1}=\left(  \varphi_{1},\gamma_{1}\right)  $ and $\theta
_{2}=\left(  \varphi_{2},\gamma_{2}\right)  $, it follows after substituting
into \eqref{eq:gradient-theta-1}--\eqref{eq:gradient-theta-2} that
$\nabla_{\gamma_{1}}f(\lambda,\mu,\theta_{1},\theta_{2})$ and $\nabla
_{\gamma_{2}}f(\lambda,\mu,\theta_{1},\theta_{2})$ (i.e., the gradients with
respect to $\gamma_{1}$ and $\gamma_{2}$ only) can be estimated by means of
the parameter-shift rule if the parameterized unitaries are single-qubit Pauli
rotations. Above we have discussed the case in which $p_{\varphi_{1}}$ and
$p_{\varphi_{2}}$ are realized by quantum circuit Born machines. In the case
that they are realized by neural-network-based generative models, perhaps the
simplest means for estimating the gradients is to make use of the simultaneous
perturbation stochastic approximation \cite{Spall1992}, in which a
perturbation vector $z$ is chosen at random, with each component an
independent Rademacher random variable (possibly scaled), and the $k$th
component of the gradient $\operatorname{Tr}[A\left(  \nabla_{\varphi_{1}}%
\rho(\theta_{1})\right)  ]$ estimated as follows:%
\begin{align}
\left[  \operatorname{Tr}[A\left(  \nabla_{\varphi_{1}}\rho(\varphi_{1}%
,\gamma_{1})\right)  ]\right]  _{k}  & \approx\frac{1}{2z_{k}}\left(
\operatorname{Tr}[A\rho(\varphi_{1}+z,\gamma_{1})]-\operatorname{Tr}%
[A\rho(\varphi_{1}-z,\gamma_{1})]\right)  \\
& =\frac{1}{2z_{k}}\left(  \sum_{x}\left(  p_{\varphi_{1}+z}(x)-p_{\varphi
_{1}-z}(x)\right)  \langle x|U(\gamma_{1})^{\dag}AU(\gamma_{1})|x\rangle
\right)  .\label{eq:SPSA-1}%
\end{align}
Estimating the gradient in this way is helpful for avoiding the difficulties
associated with applying differentiation and backpropagation to a discrete
probability distribution for a generative model (see
\cite{bengio2013estimating} for discussions and further ideas). To estimate
the gradient in this way, take a sample $x$ from the shifted distribution
$p_{\varphi_{1}+z}$, prepare the state $U(\gamma_{1})|x\rangle$, measure the
observable $A$, and record the outcome. Then repeat this procedure many times
and calculate the sample mean as an estimate of the term $\sum_{x}%
p_{\varphi_{1}+z}(x)\langle x|U(\gamma_{1})^{\dag}AU(\gamma_{1})|x\rangle$.
Similarly one can estimate the term $\sum_{x}p_{\varphi_{1}-z}(x)\langle
x|U(\gamma_{1})^{\dag}AU(\gamma_{1})|x\rangle$. For every $k$, one then
divides their difference by $2z_{k}$ in order to estimate the $k$th component
of the gradient $\operatorname{Tr}[A\left(  \nabla_{\varphi_{1}}\rho
(\theta_{1})\right)  ]$. A key advantage of this approach is that it only
involves estimating $\operatorname{Tr}[A\rho(\varphi_{1}+z,\gamma_{1})]$ and
$\operatorname{Tr}[A\rho(\varphi_{1}-z,\gamma_{1})]$ in order to estimate all
the components of the gradient. A similar approach can be taken to estimate
the other gradient terms in \eqref{eq:gradient-theta-1}--\eqref{eq:gradient-theta-2}.

\section{Details of examples: Theory}

\subsection{Normalized trace distance}

\label{app:norm-TD}

Here we expand on Section~\ref{sec:normalized-TD-example}
and establish a proof for \eqref{eq:TD-dual-rewrite-QSlack}. We also provide a
QSlack formulation of the primal problem recalled in
\eqref{eq:norm-TD-primal-app} below.

Let us first recall the semi-definite programming formulations of the
normalized trace distance for quantum states $\rho$ and $\sigma$:%
\begin{align}
\frac{1}{2}\left\Vert \rho-\sigma\right\Vert _{1}  &  =\sup_{\Lambda\geq
0}\left\{  \operatorname{Tr}[\Lambda(\rho-\sigma)]:\Lambda\leq I\right\}
\label{eq:norm-TD-primal-app}\\
&  =\inf_{Y\geq0}\left\{  \operatorname{Tr}[Y]:Y\geq\rho-\sigma\right\}  .
\end{align}

\begin{proposition}
[Normalized trace distance dual]\label{prop:TD-dual-QSlack}For states $\rho$
and $\sigma$, the following equality holds%
\begin{equation}
\inf_{Y\geq0}\left\{  \operatorname{Tr}[Y]:Y\geq\rho-\sigma\right\}
=\lim_{c\rightarrow\infty}\inf_{\substack{\lambda,\mu\geq0,\\\omega,\tau
\in\mathcal{D}}}\left\{  \lambda+c\left\Vert \lambda\omega-\rho+\sigma-\mu
\tau\right\Vert _{2}^{2}\right\}  , \label{eq:TD-dual-QSlack}%
\end{equation}
where $c>0$ is the penalty parameter and%
\begin{multline}
\left\Vert \lambda\omega-\rho+\sigma-\mu\tau\right\Vert _{2}^{2}=\lambda
^{2}\operatorname{Tr}[\omega^{2}]+\operatorname{Tr}[\rho^{2}%
]+\operatorname{Tr}[\sigma^{2}]+\mu^{2}\operatorname{Tr}[\tau^{2}%
]\label{eq:dual-TD-expansion}\\
-2\lambda\operatorname{Tr}[\omega\rho]+2\lambda\operatorname{Tr}[\omega
\sigma]-2\lambda\mu\operatorname{Tr}[\omega\tau]-2\operatorname{Tr}[\rho
\sigma]+2\mu\operatorname{Tr}[\rho\tau]-2\mu\operatorname{Tr}[\sigma\tau].
\end{multline}

\end{proposition}

\begin{proof}
Consider that%
\begin{align}
\inf_{Y\geq0}\left\{  \operatorname{Tr}[Y]:Y\geq\rho-\sigma\right\}   &
=\inf_{Y,Z\geq0}\left\{  \operatorname{Tr}[Y]:Y-\rho+\sigma=Z\right\} \\
&  =\inf_{\substack{\lambda,\mu\geq0,\\\omega,\tau\in\mathcal{D}}}\left\{
\operatorname{Tr}[\lambda\omega]:\lambda\omega-\rho+\sigma=\mu\tau\right\} \\
&  =\inf_{\substack{\lambda,\mu\geq0,\\\omega,\tau\in\mathcal{D}}}\left\{
\lambda:\lambda\omega-\rho+\sigma=\mu\tau\right\} \\
&  =\lim_{c\rightarrow\infty}\inf_{\substack{\lambda,\mu\geq0,\\\omega,\tau
\in\mathcal{D}}}\left\{  \lambda+c\left\Vert \lambda\omega-\rho+\sigma-\mu
\tau\right\Vert _{2}^{2}\right\}  ,
\end{align}
where we have invoked \cite[Proposition~5.2.1]{Bertsekas2016} for the last
equality. Eq.~\eqref{eq:dual-TD-expansion} is an expansion of the
Hilbert--Schmidt norm using its definition.
\end{proof}

\bigskip

For completeness, we also derive the QSlack formulation of the primal
optimization problem in \eqref{eq:norm-TD-primal-app}.

\begin{proposition}
[Normalized trace distance primal]
\label{prop:TD-primal-qslack}
For $n$-qubit states $\rho$ and~$\sigma$,
the following equality holds%
\begin{equation}
\sup_{\Lambda\geq0}\left\{  \operatorname{Tr}[\Lambda(\rho-\sigma
)]:\Lambda\leq I\right\}  =\lim_{c\rightarrow\infty}\sup_{\substack{\lambda
,\mu\geq0,\\\tau,\omega\in\mathcal{D}}}\left\{  \lambda\operatorname{Tr}%
[\tau\rho]-\lambda\operatorname{Tr}[\tau\sigma]-c\left\Vert I-\lambda\tau
-\mu\omega\right\Vert _{2}^{2}\right\}  , \label{eq:TD-primal-QSlack}%
\end{equation}
where $c>0$ is the penalty parameter and%
\begin{equation}
\left\Vert I-\lambda\tau-\mu\omega\right\Vert _{2}^{2}=2^{n}+\lambda
^{2}\operatorname{Tr}[\tau^{2}]+\mu^{2}\operatorname{Tr}[\omega^{2}%
]-2\lambda-2\mu+2\lambda\mu\operatorname{Tr}[\tau\omega].
\label{eq:primal-TD-HS-expansion}%
\end{equation}

\end{proposition}

\begin{proof}
Consider that%
\begin{align}
&  \sup_{\Lambda\geq0}\left\{  \operatorname{Tr}[\Lambda(\rho-\sigma
)]:\Lambda\leq I\right\} \nonumber\\
&  =\sup_{\Lambda,Z\geq0}\left\{  \operatorname{Tr}[\Lambda(\rho
-\sigma)]:I-\Lambda=Z\right\} \\
&  =\sup_{\substack{\lambda,\mu\geq0,\\\tau,\omega\in\mathcal{D}}}\left\{
\operatorname{Tr}[\lambda\tau(\rho-\sigma)]:I-\lambda\tau=\mu\omega\right\} \\
&  =\lim_{c\rightarrow\infty}\sup_{\substack{\lambda,\mu\geq0,\\\tau,\omega
\in\mathcal{D}}}\left\{  \lambda\operatorname{Tr}[\tau\rho]-\lambda
\operatorname{Tr}[\tau\sigma]-c\left\Vert I-\lambda\tau-\mu\omega\right\Vert
_{2}^{2}\right\}  ,
\end{align}
where we have invoked \cite[Proposition~5.2.1]{Bertsekas2016} for the last
equality. Eq.~\eqref{eq:dual-TD-expansion} is an expansion of the
Hilbert--Schmidt norm using its definition.
\end{proof}

\bigskip

As mentioned in Section~\ref{sec:normalized-TD-example}, we can employ
parameterized circuits and destructive swap tests or mixed-state Loschmidt
echo tests to estimate all of the terms in the objective functions in
\eqref{eq:TD-dual-QSlack} and \eqref{eq:TD-primal-QSlack}.

\subsection{Root fidelity}

\label{app:fidelity-qslack-details}As indicated in
Section~\ref{sec:fidelity-main-text}, the root fidelity can be written as the
following primal and dual SDPs \cite{Watrous2013}:%
\begin{align}
\sqrt{F}(\rho,\sigma)  &  =\sup_{X\in\mathcal{L}}\left\{  \operatorname{Re}%
[\operatorname{Tr}[X]]:%
\begin{bmatrix}
\rho & X^{\dag}\\
X & \sigma
\end{bmatrix}
\geq0\right\} \label{eq:fid-prim-app}\\
&  =\frac{1}{2}\inf_{Y,Z\geq0}\left\{  \operatorname{Tr}[Y\rho
]+\operatorname{Tr}[Z\sigma]:%
\begin{bmatrix}
Y & I\\
I & Z
\end{bmatrix}
\geq0\right\}  .
\end{align}
Let us recall from \cite{Watrous2013} that%
\begin{equation}
\frac{1}{2}\inf_{H>0}\left\{  \operatorname{Tr}[H\rho]+\operatorname{Tr}%
[H^{-1}\sigma]\right\}  =\frac{1}{2}\inf_{Y,Z\geq0}\left\{  \operatorname{Tr}%
[Y\rho]+\operatorname{Tr}[Z\sigma]:%
\begin{bmatrix}
Y & I\\
I & Z
\end{bmatrix}
\geq0\right\}  , \label{eq:fidelity-H-inv-H-rewrite}%
\end{equation}
which follows from an application of the Schur complement lemma. The formula
on the left-hand side of \eqref{eq:fidelity-H-inv-H-rewrite} was used recently
in \cite{goldfeld2023quantum} to estimate the fidelity of states $\rho$ and
$\sigma$, assuming that we have sample access to them.

Here we show how to rewrite the primal and dual SDPs such that they can be
estimated by means of the QSlack method.\ Let us begin with the primal SDP in~\eqref{eq:fid-prim-app}:

\begin{proposition}
[Root fidelity primal]\label{prop:fid-primal-SDP}For $n$-qubit states $\rho$
and $\sigma$, the following equality holds:%
\begin{multline}
\sup_{X\in\mathcal{L}}\left\{  \operatorname{Re}[\operatorname{Tr}[X]]:%
\begin{bmatrix}
\rho & X^{\dag}\\
X & \sigma
\end{bmatrix}
\geq0\right\}  =\\
\lim_{c\rightarrow\infty}\sup_{\substack{\left(  \alpha_{\overrightarrow{x}%
}\right)  _{\overrightarrow{x}},\\\lambda\geq0,\omega\in\mathcal{D}}}\left\{
2^{n}\operatorname{Re}[\alpha_{\overrightarrow{0}}]-c\cdot f(\rho
,\sigma,\lambda,\omega,\left(  \alpha_{\overrightarrow{x}}\right)
_{\overrightarrow{x}})\right\}  ,
\end{multline}
where $\alpha_{\overrightarrow{x}}\in\mathbb{C}$,%
\begin{multline}
f(\rho,\sigma,\lambda,\omega,\left(  \alpha_{\overrightarrow{x}}\right)
_{\overrightarrow{x}})\coloneqq\operatorname{Tr}[\rho^{2}]+\operatorname{Tr}%
[\sigma^{2}]+\lambda^{2}\operatorname{Tr}[\omega^{2}]+2^{n+1}\left\Vert
\overrightarrow{\alpha}\right\Vert _{2}^{2}\\
-2\lambda\operatorname{Tr}[\left(  |0\rangle\!\langle0|\otimes\rho\right)
\omega]-2\lambda\operatorname{Tr}[\left(  |1\rangle\!\langle1|\otimes
\sigma\right)  \omega]\\
-2\lambda\sum_{\overrightarrow{x}}\operatorname{Re}\!\left[  \alpha
_{\overrightarrow{x}}\operatorname{Tr}\!\left[  \left(  \left(  \sigma
_{X}-i\sigma_{Y}\right)  \otimes\sigma_{\overrightarrow{x}}\right)
\omega\right]  \right]  ,
\end{multline}
and $\sigma_{\overrightarrow{x}}\equiv\sigma_{x_{1}}\otimes\cdots\otimes
\sigma_{x_{n}}$ is a Pauli string.
\end{proposition}

\begin{proof}
Consider that%
\begin{align}
&  \sup_{X\in\mathcal{L}}\left\{  \operatorname{Re}[\operatorname{Tr}[X]]:%
\begin{bmatrix}
\rho & X^{\dag}\\
X & \sigma
\end{bmatrix}
\geq0\right\} \nonumber\\
&  =\sup_{X\in\mathcal{L},Z\geq0}\left\{  \operatorname{Re}[\operatorname{Tr}%
[X]]:%
\begin{bmatrix}
\rho & X^{\dag}\\
X & \sigma
\end{bmatrix}
=Z\right\} \\
&  =\sup_{X\in\mathcal{L},\lambda\geq0,\omega\in\mathcal{D}}\left\{
\operatorname{Re}[\operatorname{Tr}[X]]:%
\begin{bmatrix}
\rho & X^{\dag}\\
X & \sigma
\end{bmatrix}
=\lambda\omega\right\}  .
\end{align}
Then%
\begin{equation}%
\begin{bmatrix}
\rho & X^{\dag}\\
X & \sigma
\end{bmatrix}
=|0\rangle\!\langle0|\otimes\rho+|1\rangle\!\langle1|\otimes\sigma
+|0\rangle\!\langle1|\otimes X^{\dag}+|1\rangle\!\langle0|\otimes X,
\end{equation}
and writing $X$ as%
\begin{equation}
X=\sum_{\overrightarrow{x}}\alpha_{\overrightarrow{x}}\sigma_{\overrightarrow
{x}},
\end{equation}
we find that
\begin{align}
&  \sup_{\substack{X\in\mathcal{L},\\\lambda\geq0,\omega\in\mathcal{D}%
}}\left\{  \operatorname{Re}[\operatorname{Tr}[X]]:%
\begin{bmatrix}
\rho & X^{\dag}\\
X & \sigma
\end{bmatrix}
=\lambda\omega\right\} \nonumber\\
&  =\sup_{\substack{\left(  \alpha_{\overrightarrow{x}}\right)
_{\overrightarrow{x}},\\\lambda\geq0,\omega\in\mathcal{D}}}\left\{
2^{n}\operatorname{Re}\!\left[  \alpha_{\overrightarrow{0}}\right]
:|0\rangle\!\langle0|\otimes\rho+|1\rangle\!\langle1|\otimes\sigma
+|0\rangle\!\langle1|\otimes X^{\dag}+|1\rangle\!\langle0|\otimes
X=\lambda\omega\right\} \\
&  =\lim_{c\rightarrow\infty}\sup_{\substack{\left(  \alpha_{\overrightarrow
{x}}\right)  _{\overrightarrow{x}},\\\lambda\geq0,\\\omega\in\mathcal{D}%
}}\left\{  2^{n}\operatorname{Re}\!\left[  \alpha_{\overrightarrow{0}}\right]
-c\left\Vert
\begin{array}
[c]{c}%
|0\rangle\!\langle0|\otimes\rho+|1\rangle\!\langle1|\otimes\sigma\\
+|0\rangle\!\langle1|\otimes X^{\dag}+|1\rangle\!\langle0|\otimes
X-\lambda\omega
\end{array}
\right\Vert _{2}^{2}\right\}  , \label{eq:fid-prim-rewrite-unconstrained}%
\end{align}
where we have invoked \cite[Proposition~5.2.1]{Bertsekas2016} for the last
equality. Continuing,
\begin{multline}
\left\Vert |0\rangle\!\langle0|\otimes\rho+|1\rangle\!\langle1|\otimes
\sigma+|0\rangle\!\langle1|\otimes X^{\dag}+|1\rangle\!\langle0|\otimes
X-\lambda\omega\right\Vert _{2}^{2}\label{eq:fid-prim-constraint}\\
=\operatorname{Tr}[\rho^{2}]+\operatorname{Tr}[\sigma^{2}]+\lambda
^{2}\operatorname{Tr}[\omega^{2}]+2\operatorname{Tr}[X^{\dag}X]-2\lambda
\operatorname{Tr}[\left(  |0\rangle\!\langle0|\otimes\rho\right)  \omega]\\
-2\lambda\operatorname{Tr}[\left(  |1\rangle\!\langle1|\otimes\sigma\right)
\omega]-4\lambda\operatorname{Re}[\operatorname{Tr}[\left(  |1\rangle
\!\langle0|\otimes X\right)  \omega]].
\end{multline}
Now consider that
\begin{align}
\operatorname{Tr}[\left(  |1\rangle\!\langle0|\otimes X\right)  \omega]  &
=\operatorname{Tr}\!\left[  \left(  |1\rangle\!\langle0|\otimes\sum
_{\overrightarrow{x}}\alpha_{\overrightarrow{x}}\sigma_{\overrightarrow{x}%
}\right)  \omega\right] \\
&  =\sum_{\overrightarrow{x}}\alpha_{\overrightarrow{x}}\operatorname{Tr}%
\left[  \left(  \left(  \frac{\sigma_{X}-i\sigma_{Y}}{2}\right)  \otimes
\sigma_{\overrightarrow{x}}\right)  \omega\right]  ,\\
2\operatorname{Tr}[X^{\dag}X]  &  =2\operatorname{Tr}\!\left[  \left(
\sum_{\overrightarrow{x_{1}}}\overline{\alpha}_{\overrightarrow{x_{1}}}%
\sigma_{\overrightarrow{x_{1}}}\right)  \left(  \sum_{\overrightarrow{x_{2}}%
}\alpha_{\overrightarrow{x_{2}}}\sigma_{\overrightarrow{x_{2}}}\right)
\right] \\
&  =2\sum_{\overrightarrow{x_{1}},\overrightarrow{x_{2}}}\overline{\alpha
}_{\overrightarrow{x_{1}}}\alpha_{\overrightarrow{x_{2}}}\operatorname{Tr}%
\left[  \sigma_{\overrightarrow{x_{1}}}\sigma_{\overrightarrow{x_{2}}}\right]
\\
&  =2\cdot2^{n}\sum_{\overrightarrow{x}}\left\vert \alpha_{\overrightarrow{x}%
}\right\vert ^{2}\\
&  =2^{n+1}\left\Vert \overrightarrow{\alpha}\right\Vert _{2}^{2}.
\end{align}
Substituting these expressions into \eqref{eq:fid-prim-constraint} and then
\eqref{eq:fid-prim-rewrite-unconstrained}, we conclude the proof.
\end{proof}

\bigskip

\begin{proposition}
[Root fidelity dual]\label{prop:fid-dual-SDP}For $n$-qubit states $\rho$ and
$\sigma$, the following equality holds:%
\begin{multline}
\frac{1}{2}\inf_{Y,Z\geq0}\left\{  \operatorname{Tr}[Y\rho]+\operatorname{Tr}%
[Z\sigma]:%
\begin{bmatrix}
Y & I\\
I & Z
\end{bmatrix}
\geq0\right\} \\
=\lim_{c\rightarrow\infty}\inf_{\substack{\lambda,\mu,\nu
\geq0,\\\omega,\tau,\xi\in\mathcal{D}}}\left\{  \frac{1}{2}\lambda\operatorname{Tr}%
[\omega\rho]+\frac{1}{2}\mu\operatorname{Tr}[\tau\sigma]+c\cdot g(\lambda,\mu,\nu
,\omega,\tau,\xi)\right\}  ,
\end{multline}
where%
\begin{multline}
g(\lambda,\mu,\nu,\omega,\tau,\xi)\coloneqq\lambda^{2}\operatorname{Tr}%
[\omega^{2}]+\mu^{2}\operatorname{Tr}[\tau^{2}]+2^{n+1}+\nu^{2}%
\operatorname{Tr}[\xi^{2}]\\
-2\lambda\nu\operatorname{Tr}[\left(  |0\rangle\!\langle0|\otimes
\omega\right)  \xi]-2\mu\nu\operatorname{Tr}[\left(  |1\rangle\!\langle
1|\otimes\tau\right)  \xi]-2\nu\operatorname{Tr}[\left(  \sigma_{X}\otimes
I\right)  \xi].
\end{multline}

\end{proposition}

\begin{proof}
Consider that%
\begin{align}
&  \frac{1}{2}\inf_{Y,Z\geq0}\left\{  \operatorname{Tr}[Y\rho
]+\operatorname{Tr}[Z\sigma]:%
\begin{bmatrix}
Y & I\\
I & Z
\end{bmatrix}
\geq0\right\} \nonumber\\
&  =\frac{1}{2}\inf_{Y,Z,W\geq0}\left\{  \operatorname{Tr}[Y\rho
]+\operatorname{Tr}[Z\sigma]:%
\begin{bmatrix}
Y & I\\
I & Z
\end{bmatrix}
=W\right\} \\
&  =\frac{1}{2}\inf_{\substack{\lambda,\mu,\nu\geq0,\\\omega,\tau,\xi
\in\mathcal{D}}}\left\{
\begin{array}
[c]{c}%
\lambda\operatorname{Tr}[\omega\rho]+\mu\operatorname{Tr}[\tau\sigma]:%
\begin{bmatrix}
\lambda\omega & I\\
I & \mu\tau
\end{bmatrix}
=\nu\xi
\end{array}
\right\} \\
&  =\lim_{c\rightarrow\infty}\inf_{\substack{\lambda,\mu,\nu
\geq0,\\\omega,\tau,\xi\in\mathcal{D}}}\left\{
\begin{array}
[c]{c}%
\frac{1}{2}\lambda\operatorname{Tr}[\omega\rho]+\frac{1}{2}\mu\operatorname{Tr}[\tau\sigma]\\
+c\left\Vert |0\rangle\!\langle0|\otimes\lambda\omega+|1\rangle\!\langle
1|\otimes\mu\tau+\sigma_{X}\otimes I-\nu\xi\right\Vert _{2}^{2}%
\end{array}
\right\}  ,
\end{align}
where we have invoked \cite[Proposition~5.2.1]{Bertsekas2016} for the last
equality.
Then consider that%
\begin{multline}
\left\Vert |0\rangle\!\langle0|\otimes\lambda\omega+|1\rangle\!\langle
1|\otimes\mu\tau+\sigma_{X}\otimes I-\nu\xi\right\Vert _{2}^{2}=\lambda
^{2}\operatorname{Tr}[\omega^{2}]+\mu^{2}\operatorname{Tr}[\tau^{2}%
]+2^{n+1}+\nu^{2}\operatorname{Tr}[\xi^{2}]\\
-2\lambda\nu\operatorname{Tr}[\left(  |0\rangle\!\langle0|\otimes
\omega\right)  \xi]-2\mu\nu\operatorname{Tr}[\left(  |1\rangle\!\langle
1|\otimes\tau\right)  \xi]-2\nu\operatorname{Tr}[\left(  \sigma_{X}\otimes
I\right)  \xi].
\end{multline}
This concludes the proof.
\end{proof}

\subsection{Entanglement negativity}

\label{app:negativity}Recall that the negativity of a bipartite state
$\rho_{AB}$ is defined as follows and also has the following primal and dual
SDP formulations \cite[Eqs.~(5.1.101)--(5.1.102)]{khatri2020principles}:%
\begin{align}
\left\Vert T_{B}(\rho_{AB})\right\Vert _{1}  &  =\sup_{H_{AB}\in
\operatorname{Herm}}\left\{  \operatorname{Tr}[T_{B}(H_{AB})\rho_{AB}%
]:-I_{AB}\leq H_{AB}\leq I_{AB}\right\} \\
&  =\inf_{K_{AB},L_{AB}\geq0}\left\{  \operatorname{Tr}[K_{AB}+L_{AB}%
]:T_{B}(K_{AB}-L_{AB})=\rho_{AB}\right\}  .
\end{align}

\begin{proposition}
[Entanglement negativity primal]\label{prop:neg-primal-SDP}For an $n$-qubit
bipartite state $\rho_{AB}$, where $n=n_{A}+n_{B}$, the following equality
holds:%
\begin{multline}
\sup_{H_{AB}\in\operatorname{Herm}}\left\{  \operatorname{Tr}[T_{B}%
(H_{AB})\rho_{AB}]:-I_{AB}\leq H_{AB}\leq I_{AB}\right\} \\
=\lim_{c\rightarrow\infty}\sup_{\substack{\overrightarrow{\alpha},\lambda
,\mu\geq0,\\\sigma_{AB},\tau_{AB}\in\mathcal{D}}}\left\{  g_{1}\!\left(
\overrightarrow{\alpha},\rho_{AB}\right)  -c\cdot g_{2}\!\left(
\overrightarrow{\alpha},\lambda,\mu,\sigma_{AB},\tau_{AB}\right)  \right\}  ,
\end{multline}
where $\overrightarrow{\alpha}\equiv\left(  \alpha_{\overrightarrow{x_{A}%
},\overrightarrow{x_{B}}}\in\mathbb{R}\right)  _{\overrightarrow{x_{A}%
},\overrightarrow{x_{B}}}$,
\begin{equation}
g_{1}\!\left(  \overrightarrow{\alpha},\rho_{AB}\right)  \coloneqq\sum
_{\overrightarrow{x_{A}},\overrightarrow{x_{B}}}\left(  -1\right)
^{f(\overrightarrow{x_{B}})}\alpha_{\overrightarrow{x_{A}},\overrightarrow
{x_{B}}}\operatorname{Tr}\!\left[  \left(  \sigma_{\overrightarrow{x_{A}}%
}\otimes\sigma_{\overrightarrow{x_{B}}}\right)  \rho_{AB}\right]  ,
\end{equation}%
\begin{multline}
g_{2}\!\left(  \overrightarrow{\alpha},\lambda,\mu,\sigma_{AB},\tau
_{AB}\right)  \coloneqq2^{n_{A}+n_{B}+1}+2\left\Vert \overrightarrow{\alpha
}\right\Vert _{2}^{2}+\lambda^{2}\operatorname{Tr}[\sigma_{AB}^{2}]+\mu
^{2}\operatorname{Tr}[\tau_{AB}^{2}]\\
-2\lambda-2\mu+2\lambda\sum_{\overrightarrow{x_{A}},\overrightarrow{x_{B}}%
}\alpha_{\overrightarrow{x_{A}},\overrightarrow{x_{B}}}\operatorname{Tr}%
\!\left[  \left(  \sigma_{\overrightarrow{x_{A}}}\otimes\sigma
_{\overrightarrow{x_{B}}}\right)  \sigma_{AB}\right] \\
-2\mu\sum_{\overrightarrow{x_{A}},\overrightarrow{x_{B}}}\alpha
_{\overrightarrow{x_{A}},\overrightarrow{x_{B}}}\operatorname{Tr}\!\left[
\left(  \sigma_{\overrightarrow{x_{A}}}\otimes\sigma_{\overrightarrow{x_{B}}%
}\right)  \tau_{AB}\right]  ,
\end{multline}
and $f(\overrightarrow{x_{B}})\coloneqq\sum_{i=1}^{n_{B}}\delta_{x_{B}^{i},2}$
counts the number of $\sigma_{Y}$ terms in the sequence $\overrightarrow
{x_{B}}$, with $x_{B}^{i}$ denoting the $i$th entry in $\overrightarrow{x_{B}%
}$.
\end{proposition}

\begin{proof}
Consider that%
\begin{align}
&  \sup_{H_{AB}\in\operatorname{Herm}}\left\{  \operatorname{Tr}[T_{B}%
(H_{AB})\rho_{AB}]:-I_{AB}\leq H_{AB}\leq I_{AB}\right\} \nonumber\\
&  =\sup_{\substack{H_{AB}\in\operatorname{Herm},\\Y_{AB},Z_{AB}\geq
0}}\left\{
\begin{array}
[c]{c}%
\operatorname{Tr}[T_{B}(H_{AB})\rho_{AB}]:I_{AB}-H_{AB}=Y_{AB},\\
I_{AB}+H_{AB}=Z_{AB}%
\end{array}
\right\} \\
&  =\sup_{\substack{H_{AB}\in\operatorname{Herm},\\\lambda,\mu\geq
0,\\\sigma_{AB},\tau_{AB}\in\mathcal{D}}}\left\{
\begin{array}
[c]{c}%
\operatorname{Tr}[T_{B}(H_{AB})\rho_{AB}]:I_{AB}-H_{AB}=\lambda\sigma_{AB},\\
I_{AB}+H_{AB}=\mu\tau_{AB}%
\end{array}
\right\} \\
&  =\lim_{c\rightarrow\infty}\sup_{\substack{H_{AB}\in\operatorname{Herm}%
,\\\lambda,\mu\geq0,\\\sigma_{AB},\tau_{AB}\in\mathcal{D}}}\left\{
\begin{array}
[c]{c}%
\operatorname{Tr}[T_{B}(H_{AB})\rho_{AB}]-c\left\Vert I_{AB}-H_{AB}%
-\lambda\sigma_{AB}\right\Vert _{2}^{2}\\
-c\left\Vert I_{AB}+H_{AB}-\mu\tau_{AB}\right\Vert _{2}^{2}%
\end{array}
\right\}  ,
\end{align}
where we have invoked \cite[Proposition~5.2.1]{Bertsekas2016} for the last
equality.
Now consider that%
\begin{align}
&  \sup_{\substack{H_{AB}\in\operatorname{Herm},\\\lambda,\mu\geq
0,\\\sigma_{AB},\tau_{AB}\in\mathcal{D}}}\left\{
\begin{array}
[c]{c}%
\operatorname{Tr}[T_{B}(H_{AB})\rho_{AB}]-c\left\Vert I_{AB}-H_{AB}%
-\lambda\sigma_{AB}\right\Vert _{2}^{2}\\
-c\left\Vert I_{AB}+H_{AB}-\mu\tau_{AB}\right\Vert _{2}^{2}%
\end{array}
\right\} \nonumber\\
&  =\sup_{\substack{H_{AB}\in\operatorname{Herm},\\\lambda,\mu\geq
0,\\\sigma_{AB},\tau_{AB}\in\mathcal{D}}}\left\{
\begin{array}
[c]{c}%
\operatorname{Tr}[T_{B}(H_{AB})\rho_{AB}]-c\left(
\begin{array}
[c]{c}%
\operatorname{Tr}[I_{AB}]+\operatorname{Tr}[H_{AB}^{2}]+\lambda^{2}%
\operatorname{Tr}[\sigma_{AB}^{2}]\\
-2\operatorname{Tr}[H_{AB}]-2\lambda+2\lambda\operatorname{Tr}[H_{AB}%
\sigma_{AB}]
\end{array}
\right) \\
-c\left(
\begin{array}
[c]{c}%
\operatorname{Tr}[I_{AB}]+\operatorname{Tr}[H_{AB}^{2}]+\mu^{2}%
\operatorname{Tr}[\tau_{AB}^{2}]\\
+2\operatorname{Tr}[H_{AB}]-2\mu-2\mu\operatorname{Tr}[H_{AB}\tau_{AB}]
\end{array}
\right)
\end{array}
\right\} \\
&  =\sup_{\substack{H_{AB}\in\operatorname{Herm},\\\lambda,\mu\geq
0,\\\sigma_{AB},\tau_{AB}\in\mathcal{D}}}\left\{  \operatorname{Tr}%
[T_{B}(H_{AB})\rho_{AB}]-c\left(
\begin{array}
[c]{c}%
2\operatorname{Tr}[I_{AB}]+2\operatorname{Tr}[H_{AB}^{2}]+\lambda
^{2}\operatorname{Tr}[\sigma_{AB}^{2}]\\
+\mu^{2}\operatorname{Tr}[\tau_{AB}^{2}]-2\lambda+2\lambda\operatorname{Tr}%
[H_{AB}\sigma_{AB}]\\
-2\mu-2\mu\operatorname{Tr}[H_{AB}\tau_{AB}]
\end{array}
\right)  \right\}  . \label{eq:ent-neg-primal-proof-app-transition}%
\end{align}
Recalling that $n=n_{A}+n_{B}$, let us write%
\begin{equation}
H_{AB}=\sum_{\overrightarrow{x_{A}},\overrightarrow{x_{B}}}\alpha
_{\overrightarrow{x_{A}},\overrightarrow{x_{B}}}\sigma_{\overrightarrow{x_{A}%
}}\otimes\sigma_{\overrightarrow{x_{B}}},
\end{equation}
where $\alpha_{\overrightarrow{x_{A}},\overrightarrow{x_{B}}}\in\mathbb{R}$,
$\overrightarrow{x_{A}}\in\left\{  0,1,2,3\right\}  ^{n_{A}}$, and
$\overrightarrow{x_{B}}\in\left\{  0,1,2,3\right\}  ^{n_{B}}$. Now consider
that%
\begin{align}
\operatorname{Tr}[T_{B}(H_{AB})\rho_{AB}]  &  =\operatorname{Tr}\!\left[
\left(  \sum_{\overrightarrow{x_{A}},\overrightarrow{x_{B}}}\alpha
_{\overrightarrow{x_{A}},\overrightarrow{x_{B}}}\sigma_{\overrightarrow{x_{A}%
}}\otimes T_{B}(\sigma_{\overrightarrow{x_{B}}})\right)  \rho_{AB}\right] \\
&  =\sum_{\overrightarrow{x_{A}},\overrightarrow{x_{B}}}\alpha
_{\overrightarrow{x_{A}},\overrightarrow{x_{B}}}\operatorname{Tr}\!\left[
\left(  \sigma_{\overrightarrow{x_{A}}}\otimes T_{B}(\sigma_{\overrightarrow
{x_{B}}})\right)  \rho_{AB}\right]  .
\end{align}
Recalling \eqref{eq:Pauli-def-1}--\eqref{eq:Pauli-def-2}, it follows that%
\begin{equation}
T(\sigma_{0})=\sigma_{0},\qquad T(\sigma_{1})=\sigma_{1},\qquad T(\sigma
_{2})=-\sigma_{2},\qquad T(\sigma_{3})=\sigma_{3}.
\end{equation}
Let $f(\overrightarrow{x_{B}})$ be a function that counts the number of times
that $2$ appears in the sequence $\overrightarrow{x_{B}}$:%
\begin{equation}
f(\overrightarrow{x_{B}})=\sum_{i=1}^{n_{B}}\delta_{x_{i},2}.
\end{equation}
Then the objective function is given by%
\begin{align}
\operatorname{Tr}[T_{B}(H_{AB})\rho_{AB}]  &  =\sum_{\overrightarrow{x_{A}%
},\overrightarrow{x_{B}}}\alpha_{\overrightarrow{x_{A}},\overrightarrow{x_{B}%
}}\operatorname{Tr}\!\left[  \left(  \sigma_{\overrightarrow{x_{A}}}\otimes
T_{B}(\sigma_{\overrightarrow{x_{B}}})\right)  \rho_{AB}\right] \\
&  =\sum_{\overrightarrow{x_{A}},\overrightarrow{x_{B}}}\left(  -1\right)
^{f(\overrightarrow{x_{B}})}\alpha_{\overrightarrow{x_{A}},\overrightarrow
{x_{B}}}\operatorname{Tr}\!\left[  \left(  \sigma_{\overrightarrow{x_{A}}%
}\otimes\sigma_{\overrightarrow{x_{B}}}\right)  \rho_{AB}\right]  ,
\end{align}
We then find that \eqref{eq:ent-neg-primal-proof-app-transition} is equal to%
\begin{equation}
\sup_{\substack{\left\{  \alpha_{\overrightarrow{x_{A}},\overrightarrow{x_{B}%
}}\in\mathbb{R}\right\}  ,\\\lambda,\mu\geq0,\\\sigma_{AB},\tau_{AB}%
\in\mathcal{D}}}\left\{  \operatorname{Tr}[T_{B}(H_{AB})\rho_{AB}]-c\cdot
g\left(  \overrightarrow{\alpha},\lambda,\mu,\sigma_{AB},\tau_{AB}\right)
\right\}
\end{equation}
where%
\begin{multline}
g\left(  \overrightarrow{\alpha},\lambda,\mu,\sigma_{AB},\tau_{AB}\right)
\coloneqq2^{n_{A}+n_{B}+1}+2\left\Vert \overrightarrow{\alpha}\right\Vert
_{2}^{2}+\lambda^{2}\operatorname{Tr}[\sigma_{AB}^{2}]+\mu^{2}%
\operatorname{Tr}[\tau_{AB}^{2}]\\
-2\lambda-2\mu+2\lambda\sum_{\overrightarrow{x_{A}},\overrightarrow{x_{B}}%
}\alpha_{\overrightarrow{x_{A}},\overrightarrow{x_{B}}}\operatorname{Tr}%
\!\left[  \left(  \sigma_{\overrightarrow{x_{A}}}\otimes\sigma
_{\overrightarrow{x_{B}}}\right)  \sigma_{AB}\right] \\
-2\mu\sum_{\overrightarrow{x_{A}},\overrightarrow{x_{B}}}\alpha
_{\overrightarrow{x_{A}},\overrightarrow{x_{B}}}\operatorname{Tr}\!\left[
\left(  \sigma_{\overrightarrow{x_{A}}}\otimes\sigma_{\overrightarrow{x_{B}}%
}\right)  \tau_{AB}\right]  .
\end{multline}
This concludes the proof.
\end{proof}

\begin{proposition}
[Entanglement negativity dual]\label{prop:neg-dual-SDP}For an $n$-qubit
bipartite state$~\rho_{AB}$, where $n=n_{A}+n_{B}$, the following equality
holds:%
\begin{multline}
\inf_{K_{AB},L_{AB}\geq0}\left\{  \operatorname{Tr}[K_{AB}+L_{AB}%
]:T_{B}(K_{AB}-L_{AB})=\rho_{AB}\right\}  =\\
\lim_{c\rightarrow\infty}\inf_{\substack{\overrightarrow{\alpha}%
,\overrightarrow{\beta},\\\lambda,\mu\geq0,\\\sigma_{AB},\tau_{AB}%
\in\mathcal{D}}}\left\{  2^{n}\left(  \alpha_{\overrightarrow{0}%
,\overrightarrow{0}}+\beta_{\overrightarrow{0},\overrightarrow{0}}\right)
+c\cdot g_{3}\!\left(  \overrightarrow{\alpha},\overrightarrow{\beta}%
,\lambda,\mu,\sigma_{AB},\tau_{AB}\right)  \right\}  ,
\end{multline}
where%
\begin{equation}
\overrightarrow{\alpha}\equiv\left(  \alpha_{\overrightarrow{x_{A}%
},\overrightarrow{x_{B}}}\in\mathbb{R}\right)  _{\overrightarrow{x_{A}%
},\overrightarrow{x_{B}}},\qquad\overrightarrow{\beta}\equiv\left(
\beta_{\overrightarrow{x_{A}},\overrightarrow{x_{B}}}\in\mathbb{R}\right)
_{\overrightarrow{x_{A}},\overrightarrow{x_{B}}},
\end{equation}%
\begin{multline}
g_{3}\!\left(  \overrightarrow{\alpha},\overrightarrow{\beta},\lambda
,\mu,\sigma_{AB},\tau_{AB}\right)  \coloneqq2^{n+1}\left(  \left\Vert
\overrightarrow{\alpha}\right\Vert _{2}^{2}+\left\Vert \overrightarrow{\beta
}\right\Vert _{2}^{2}\right)  +\operatorname{Tr}[\rho_{AB}^{2}]+\mu
^{2}\operatorname{Tr}[\tau_{AB}^{2}]\\
-2\sum_{\overrightarrow{x_{A}},\overrightarrow{x_{B}}}\left(  -1\right)
^{f(\overrightarrow{x_{B}})}\left(  \alpha_{\overrightarrow{x_{A}%
},\overrightarrow{x_{B}}}-\beta_{\overrightarrow{x_{A}},\overrightarrow{x_{B}%
}}\right)  \operatorname{Tr}\!\left[  \left(  \sigma_{\overrightarrow{x_{A}}%
}\otimes\sigma_{\overrightarrow{x_{B}}}\right)  \rho_{AB}\right] \\
-2^{n+1}\sum_{\overrightarrow{x_{A}},\overrightarrow{x_{B}}}\alpha
_{\overrightarrow{x_{A}},\overrightarrow{x_{B}}}\beta_{\overrightarrow{x_{A}%
},\overrightarrow{x_{B}}}+\lambda^{2}\operatorname{Tr}[\sigma_{AB}^{2}]\\
-2\lambda\sum_{\overrightarrow{x_{A}},\overrightarrow{x_{B}}}\alpha
_{\overrightarrow{x_{A}},\overrightarrow{x_{B}}}\operatorname{Tr}\!\left[
\left(  \sigma_{\overrightarrow{x_{A}}}\otimes\sigma_{\overrightarrow{x_{B}}%
}\right)  \sigma_{AB}\right]  -2\mu\sum_{\overrightarrow{x_{A}}%
,\overrightarrow{x_{B}}}\beta_{\overrightarrow{x_{A}},\overrightarrow{x_{B}}%
}\operatorname{Tr}\!\left[  \left(  \sigma_{\overrightarrow{x_{A}}}%
\otimes\sigma_{\overrightarrow{x_{B}}}\right)  \tau_{AB}\right]  .
\end{multline}

\end{proposition}

\begin{proof}
Consider that%
\begin{align}
&  \inf_{K_{AB},L_{AB}\geq0}\left\{  \operatorname{Tr}[K_{AB}+L_{AB}%
]:T_{B}(K_{AB}-L_{AB})=\rho_{AB}\right\} \nonumber\\
&  =\inf_{\substack{K_{AB},L_{AB}\in\operatorname{Herm},\\Y_{AB},Z_{AB}\geq
0}}\left\{
\begin{array}
[c]{c}%
\operatorname{Tr}[K_{AB}+L_{AB}]:T_{B}(K_{AB}-L_{AB})=\rho_{AB},\\
K_{AB}=Y_{AB},\ L_{AB}=Z_{AB}%
\end{array}
\right\} \\
&  =\inf_{\substack{K_{AB},L_{AB}\in\operatorname{Herm},\\\lambda,\mu
\geq0,\\\sigma_{AB},\tau_{AB}\in\mathcal{D}}}\left\{
\begin{array}
[c]{c}%
\operatorname{Tr}[K_{AB}+L_{AB}]:T_{B}(K_{AB}-L_{AB})=\rho_{AB},\\
K_{AB}=\lambda\sigma_{AB},\ L_{AB}=\mu\tau_{AB}%
\end{array}
\right\} \\
&  =\lim_{c\rightarrow\infty}\inf_{\substack{K_{AB},L_{AB}\in
\operatorname{Herm},\\\lambda,\mu\geq0,\\\sigma_{AB},\tau_{AB}\in\mathcal{D}%
}}\left\{
\begin{array}
[c]{c}%
\operatorname{Tr}[K_{AB}+L_{AB}]+c\left\Vert T_{B}(K_{AB}-L_{AB})-\rho
_{AB}\right\Vert _{2}^{2}\\
+c\left\Vert K_{AB}-\lambda\sigma_{AB}\right\Vert _{2}^{2}+c\left\Vert
L_{AB}-\mu\tau_{AB}\right\Vert _{2}^{2}%
\end{array}
\right\}  ,
\end{align}
where we have invoked \cite[Proposition~5.2.1]{Bertsekas2016} for the last
equality.
Setting%
\begin{align}
K_{AB}  &  =\sum_{\overrightarrow{x_{A}},\overrightarrow{x_{B}}}%
\alpha_{\overrightarrow{x_{A}},\overrightarrow{x_{B}}}\sigma_{\overrightarrow
{x_{A}}}\otimes\sigma_{\overrightarrow{x_{B}}},\\
L_{AB}  &  =\sum_{\overrightarrow{x_{A}},\overrightarrow{x_{B}}}%
\beta_{\overrightarrow{x_{A}},\overrightarrow{x_{B}}}\sigma_{\overrightarrow
{x_{A}}}\otimes\sigma_{\overrightarrow{x_{B}}},
\end{align}
we find that%
\[
\operatorname{Tr}[K_{AB}+L_{AB}]=2^{n}\left(  \alpha_{\overrightarrow
{0},\overrightarrow{0}}+\beta_{\overrightarrow{0},\overrightarrow{0}}\right)
,
\]%
\begin{align}
&  \left\Vert T_{B}(K_{AB}-L_{AB})-\rho_{AB}\right\Vert _{2}^{2}\nonumber\\
&  =\operatorname{Tr}[\left(  T_{B}(K_{AB})\right)  ^{2}]+\operatorname{Tr}%
[\left(  T_{B}(L_{AB})\right)  ^{2}]\nonumber\\
&  \qquad+\operatorname{Tr}[\rho_{AB}^{2}]-2\operatorname{Tr}[T_{B}%
(K_{AB})\rho_{AB}]\nonumber\\
&  \qquad+2\operatorname{Tr}[T_{B}(L_{AB})\rho_{AB}]-2\operatorname{Tr}%
[T_{B}(K_{AB})T_{B}(L_{AB})]\\
&  =2^{n}\left(  \left\Vert \overrightarrow{\alpha}\right\Vert _{2}%
^{2}+\left\Vert \overrightarrow{\beta}\right\Vert _{2}^{2}\right)
+\operatorname{Tr}[\rho_{AB}^{2}]\nonumber\\
&  \qquad-2\sum_{\overrightarrow{x_{A}},\overrightarrow{x_{B}}}\left(
-1\right)  ^{f(\overrightarrow{x_{B}})}\left(  \alpha_{\overrightarrow{x_{A}%
},\overrightarrow{x_{B}}}-\beta_{\overrightarrow{x_{A}},\overrightarrow{x_{B}%
}}\right)  \operatorname{Tr}\!\left[  \left(  \sigma_{\overrightarrow{x_{A}}%
}\otimes\sigma_{\overrightarrow{x_{B}}}\right)  \rho_{AB}\right] \nonumber\\
&  \qquad-2\cdot2^{n}\sum_{\overrightarrow{x_{A}},\overrightarrow{x_{B}}%
}\alpha_{\overrightarrow{x_{A}},\overrightarrow{x_{B}}}\beta_{\overrightarrow
{x_{A}},\overrightarrow{x_{B}}},
\end{align}%
\begin{align}
\left\Vert K_{AB}-\lambda\sigma_{AB}\right\Vert _{2}^{2}  &  =2^{n}\left\Vert
\overrightarrow{\alpha}\right\Vert _{2}^{2}-2\lambda\sum_{\overrightarrow
{x_{A}},\overrightarrow{x_{B}}}\alpha_{\overrightarrow{x_{A}},\overrightarrow
{x_{B}}}\operatorname{Tr}\!\left[  \left(  \sigma_{\overrightarrow{x_{A}}%
}\otimes\sigma_{\overrightarrow{x_{B}}}\right)  \sigma_{AB}\right]
+\lambda^{2}\operatorname{Tr}[\sigma_{AB}^{2}],\\
\left\Vert L_{AB}-\mu\tau_{AB}\right\Vert _{2}^{2}  &  =2^{n}\left\Vert
\overrightarrow{\beta}\right\Vert _{2}^{2}-2\mu\sum_{\overrightarrow{x_{A}%
},\overrightarrow{x_{B}}}\beta_{\overrightarrow{x_{A}},\overrightarrow{x_{B}}%
}\operatorname{Tr}\!\left[  \left(  \sigma_{\overrightarrow{x_{A}}}%
\otimes\sigma_{\overrightarrow{x_{B}}}\right)  \tau_{AB}\right]  +\mu
^{2}\operatorname{Tr}[\tau_{AB}^{2}].
\end{align}
This concludes the proof.
\end{proof}

\subsection{Constrained Hamiltonian optimization}

\label{app:constrained-Ham-opt}

Recall that%
\begin{align}
\mathcal{L}(H,A_{1},\ldots,A_{\ell})  &  \coloneqq\inf_{\rho\in\mathcal{D}%
}\left\{  \operatorname{Tr}[H\rho]:\operatorname{Tr}[A_{i}\rho]\geq
b_{i}\text{ }\forall i\in\left[  \ell\right]  \right\}
\label{eq:constrained-Ham-opt-primal-app}\\
&  =\sup_{\substack{y_{1},\ldots,y_{\ell}\geq0,\\\mu\in\mathbb{R}}}\left\{
\sum_{i=1}^{\ell}b_{i}y_{i}+\mu:\sum_{i=1}^{\ell}y_{i}A_{i}+\mu I\leq
H\right\}  . \label{eq:constrained-Ham-opt-dual-app}%
\end{align}

Let us first briefly derive the dual SDP\ in
\eqref{eq:constrained-Ham-opt-dual-app}. We can rewrite the objective function
in \eqref{eq:constrained-Ham-opt-primal-app} as follows:%
\begin{align}
&  \inf_{\rho\geq0}\left\{  \operatorname{Tr}[H\rho]:\operatorname{Tr}%
[\rho]=1,\ \operatorname{Tr}[A_{i}\rho]\geq b_{i}\text{ }\forall i\in\left[
\ell\right]  \right\} \nonumber\\
&  =\inf_{\rho\geq0}\sup_{\substack{y_{1},\ldots,y_{\ell}\geq0,\\\mu
\in\mathbb{R}}}\left\{  \operatorname{Tr}[H\rho]+\mu\left(
1-\operatorname{Tr}[\rho]\right)  +\sum_{i=1}^{\ell}y_{i}\left(
b_{i}-\operatorname{Tr}[A_{i}\rho]\right)  \right\}
\label{eq:CHO-rewrite-dual-1}\\
&  =\inf_{\rho\geq0}\sup_{\substack{y_{1},\ldots,y_{\ell}\geq0,\\\mu
\in\mathbb{R}}}\left\{  \sum_{i=1}^{\ell}b_{i}y_{i}+\mu+\operatorname{Tr}%
\!\left[  \left(  H-\sum_{i=1}^{\ell}y_{i}A_{i}-\mu I\right)  \rho\right]
\right\} \\
&  \geq\sup_{\substack{y_{1},\ldots,y_{\ell}\geq0,\\\mu\in\mathbb{R}}%
}\inf_{\rho\geq0}\left\{  \sum_{i=1}^{\ell}b_{i}y_{i}+\mu+\operatorname{Tr}%
\!\left[  \left(  H-\sum_{i=1}^{\ell}y_{i}A_{i}-\mu I\right)  \rho\right]
\right\} \\
&  =\sup_{\substack{y_{1},\ldots,y_{\ell}\geq0,\\\mu\in\mathbb{R}}}\left\{
\sum_{i=1}^{\ell}b_{i}y_{i}+\mu:\sum_{i=1}^{\ell}y_{i}A_{i}+\mu I\leq
H\right\}  .
\end{align}
The first equality follows by introducing the Lagrange multipliers
$y_{1},\ldots,y_{\ell}\geq0$ and $\mu\in\mathbb{R}$, and noting that the
constraints $\operatorname{Tr}[\rho]=1$ and$\ \operatorname{Tr}[A_{i}\rho]\geq
b_{i}$ $\forall i\in\left[  \ell\right]  $ being violated implies that the
inner optimization in \eqref{eq:CHO-rewrite-dual-1} evaluates to $+\infty$, so
that the Lagrange multipliers enforce the constraints. Indeed, the constraint
$\operatorname{Tr}[\rho]=1$ does not hold if and only if $\sup_{\mu
\in\mathbb{R}}\mu\left(  1-\operatorname{Tr}[\rho]\right)  =+\infty$, and each
constraint $\operatorname{Tr}[A_{i}\rho]\geq b_{i}$ does not hold if and only
if $\sup_{y_{i}\geq0}\left\{  y_{i}\left(  b_{i}-\operatorname{Tr}[A_{i}%
\rho]\right)  \right\}  =+\infty$. The second equality follows from simple
algebra. The inequality follows from the max-min inequality and is an equality
in the case that strong duality holds. The final equality holds for reasons
similar to the first one:\ we can think of $\rho\geq0$ as a Lagrange
multiplier, enforcing the constraint $\sum_{i=1}^{\ell}y_{i}A_{i}+\mu I\leq
H$. Indeed, the constraint $\sum_{i=1}^{\ell}y_{i}A_{i}+\mu I\leq H$ does not
hold if and only if $\inf_{\rho\geq0}\operatorname{Tr}\!\left[  \left(
H-\sum_{i=1}^{\ell}y_{i}A_{i}-\mu I\right)  \rho\right]  =-\infty$.

For the propositions that follow, we set%
\begin{align}
H  &  =\sum_{\overrightarrow{x}}h_{\overrightarrow{x}}\sigma_{\overrightarrow
{x}},\label{eq:Pauli-rep-H-app}\\
A_{i}  &  =\sum_{\overrightarrow{x}}a_{\overrightarrow{x}}^{i}\sigma
_{\overrightarrow{x}}\qquad\forall i\in\left[  \ell\right]  .
\label{eq:Pauli-rep-Ai-app}%
\end{align}

\begin{proposition}
[Constrained Hamiltonian primal]\label{prop:constrained-Ham-opt-primal-SDP}For
$H,A_{1},\ldots,A_{\ell}$ as defined in
\eqref{eq:Pauli-rep-H-app}--\eqref{eq:Pauli-rep-Ai-app}, the following
equality holds:%
\begin{multline}
\inf_{\rho\in\mathcal{D}}\left\{  \operatorname{Tr}[H\rho]:\operatorname{Tr}%
[A_{i}\rho]\geq b_{i}\text{ }\forall i\in\left[  \ell\right]  \right\} \\
=\lim_{c\rightarrow\infty}\inf_{\substack{\rho\in\mathcal{D},\\z_{1}%
,\ldots,z_{\ell}\geq0}}\left\{  \sum_{\overrightarrow{x}}h_{\overrightarrow
{x}}\operatorname{Tr}[\sigma_{\overrightarrow{x}}\rho]+c\sum_{i=1}^{\ell
}\left(  \sum_{\overrightarrow{x}}a_{\overrightarrow{x}}^{i}\operatorname{Tr}%
[\sigma_{\overrightarrow{x}}\rho]-b_{i}-z_{i}\right)  ^{2}\right\}  .
\end{multline}

\end{proposition}

\begin{proof}
Consider that%
\begin{align}
&  \inf_{\rho\in\mathcal{D}}\left\{  \operatorname{Tr}[H\rho
]:\operatorname{Tr}[A_{i}\rho]\geq b_{i}\text{ }\forall i\in\left[
\ell\right]  \right\} \nonumber\\
&  =\inf_{\rho\in\mathcal{D},z_{1},\ldots,z_{\ell}\geq0}\left\{
\operatorname{Tr}[H\rho]:\operatorname{Tr}[A_{i}\rho]-b_{i}=z_{i}\text{
}\forall i\in\left[  \ell\right]  \right\} \\
&  =\lim_{c\rightarrow\infty}\inf_{\substack{\rho\in\mathcal{D},\\z_{1}%
,\ldots,z_{\ell}\geq0}}\left\{  \operatorname{Tr}[H\rho]+c\sum_{i=1}^{\ell
}\left(  \operatorname{Tr}[A_{i}\rho]-b_{i}-z_{i}\right)  ^{2}\right\} \\
&  =\lim_{c\rightarrow\infty}\inf_{\substack{\rho\in\mathcal{D},\\z_{1}%
,\ldots,z_{\ell}\geq0}}\left\{  \sum_{\overrightarrow{x}}h_{\overrightarrow
{x}}\operatorname{Tr}[\sigma_{\overrightarrow{x}}\rho]+c\sum_{i=1}^{\ell
}\left(  \sum_{\overrightarrow{x}}a_{\overrightarrow{x}}^{i}\operatorname{Tr}%
[\sigma_{\overrightarrow{x}}\rho]-b_{i}-z_{i}\right)  ^{2}\right\}  , \label{eq:primal_constrained_initializations}
\end{align}
where we have invoked \cite[Proposition~5.2.1]{Bertsekas2016} for the last
equality.
\end{proof}

\begin{proposition}
[Constrained Hamiltonian dual]\label{prop:constrained-Ham-opt-dual-SDP}For
$H,A_{1},\ldots,A_{\ell}$ as defined in
\eqref{eq:Pauli-rep-H-app}--\eqref{eq:Pauli-rep-Ai-app}, the following
equality holds:%
\begin{multline}
\sup_{\substack{y_{1},\ldots,y_{\ell}\geq0,\\\mu\in\mathbb{R}}}\left\{
\sum_{i=1}^{\ell}b_{i}y_{i}+\mu:\sum_{i=1}^{\ell}y_{i}A_{i}+\mu I\leq
H\right\} \\
=\lim_{c\rightarrow\infty}\sup_{\substack{y_{1},\ldots,y_{\ell}\geq0,\\\mu
\in\mathbb{R},\nu\geq0,\\\omega\in\mathcal{D}}}\left\{  \sum_{i=1}^{\ell}%
b_{i}y_{i}+\mu-c\cdot f\!\left(  \overrightarrow{h},\left(  \overrightarrow
{a}^{i}\right)  _{i=1}^{\ell},\overrightarrow{y},\overrightarrow{A},\mu
,\nu,\omega\right)  \right\}  ,
\end{multline}
where%
\begin{multline}
f\!\left(  \overrightarrow{h},\left(  \overrightarrow{a}^{i}\right)
_{i=1}^{\ell},\overrightarrow{y},\overrightarrow{A},\mu,\nu,\omega\right)
\coloneqq2^{n}\left\Vert \overrightarrow{h}\right\Vert _{2}^{2}+2^{n}%
\sum_{i,j=1}^{\ell}y_{i}y_{j}\left(  \overrightarrow{a}^{i}\cdot
\overrightarrow{a}^{j}\right)  +\mu^{2}2^{n}\\
+\nu^{2}\operatorname{Tr}[\omega^{2}]-2^{n+1}\sum_{i=1}^{\ell}y_{i}%
\overrightarrow{h}\cdot\overrightarrow{a}^{i}+2^{n+1}\mu h_{\overrightarrow
{0}}-2\nu\sum_{\overrightarrow{x}}h_{\overrightarrow{x}}\operatorname{Tr}%
[\sigma_{\overrightarrow{x}}\omega]\\
-2^{n+1}\mu\sum_{i=1}^{\ell}y_{i}a_{\overrightarrow{0}}^{i}+2\nu\sum
_{i=1}^{\ell}y_{i}\sum_{\overrightarrow{x}}a_{\overrightarrow{x}}%
^{i}\operatorname{Tr}\!\left[  \sigma_{\overrightarrow{x}}\omega\right]
-2\mu\nu.
\end{multline}

\end{proposition}

\begin{proof}
Consider that%
\begin{align}
&  \sup_{\substack{y_{1},\ldots,y_{\ell}\geq0,\\\mu\in\mathbb{R}}}\left\{
\sum_{i=1}^{\ell}b_{i}y_{i}+\mu:\sum_{i=1}^{\ell}y_{i}A_{i}+\mu I\leq
H\right\} \nonumber\\
&  =\sup_{\substack{y_{1},\ldots,y_{\ell}\geq0,\\\mu\in\mathbb{R},W\geq
0}}\left\{  \sum_{i=1}^{\ell}b_{i}y_{i}+\mu:H-\sum_{i=1}^{\ell}y_{i}A_{i}-\mu
I=W\right\} \\
&  =\sup_{\substack{y_{1},\ldots,y_{\ell}\geq0,\\\mu\in\mathbb{R},\nu
\geq0,\\\omega\in\mathcal{D}}}\left\{  \sum_{i=1}^{\ell}b_{i}y_{i}+\mu
:H-\sum_{i=1}^{\ell}y_{i}A_{i}-\mu I=\nu\omega\right\} \\
&  =\lim_{c\rightarrow\infty}\sup_{\substack{y_{1},\ldots,y_{\ell}\geq
0,\\\mu\in\mathbb{R},\nu\geq0,\\\omega\in\mathcal{D}}}\left\{  \sum
_{i=1}^{\ell}b_{i}y_{i}+\mu-c\left\Vert H-\sum_{i=1}^{\ell}y_{i}A_{i}-\mu
I-\nu\omega\right\Vert _{2}^{2}\right\}  ,
\end{align}
where we have invoked \cite[Proposition~5.2.1]{Bertsekas2016} for the last
equality.
Expanding the Hilbert--Schmidt norm, we find that%
\begin{multline}
\left\Vert H-\sum_{i=1}^{\ell}y_{i}A_{i}-\mu I-\nu\omega\right\Vert _{2}%
^{2}=\operatorname{Tr}[H^{2}]+\operatorname{Tr}\!\left[  \left(  \sum
_{i=1}^{\ell}y_{i}A_{i}\right)  \left(  \sum_{j=1}^{\ell}y_{j}A_{j}\right)
\right]  +\mu^{2}2^{n}\label{eq:constrained-Ham-opt-HS-norm}\\
+\nu^{2}\operatorname{Tr}[\omega^{2}]-2\sum_{i=1}^{\ell}y_{i}\operatorname{Tr}%
[HA_{i}]-2\mu\operatorname{Tr}[H]-2\nu\operatorname{Tr}[H\omega]\\
+2\mu\sum_{i=1}^{\ell}y_{i}\operatorname{Tr}[A_{i}]+2\nu\sum_{i=1}^{\ell}%
y_{i}\operatorname{Tr}[A_{i}\omega]+2\mu\nu. 
\end{multline}
Plugging \eqref{eq:Pauli-rep-H-app}--\eqref{eq:Pauli-rep-Ai-app} into the
expressions in \eqref{eq:constrained-Ham-opt-HS-norm}, we find that%
\begin{align}
\operatorname{Tr}[H^{2}]  &  =2^{n}\left\Vert \overrightarrow{h}\right\Vert
_{2}^{2},\\
\operatorname{Tr}\!\left[  \left(  \sum_{i=1}^{\ell}y_{i}A_{i}\right)  \left(
\sum_{j=1}^{\ell}y_{j}A_{j}\right)  \right]   &  =\sum_{i,j=1}^{\ell}%
y_{i}y_{j}\operatorname{Tr}\!\left[  \left(  \sum_{\overrightarrow{x_{1}}%
}a_{\overrightarrow{x_{1}}}^{i}\sigma_{\overrightarrow{x_{1}}}\right)  \left(
\sum_{\overrightarrow{x_{2}}}a_{\overrightarrow{x_{2}}}^{j}\sigma
_{\overrightarrow{x_{2}}}\right)  \right] \\
&  =\sum_{i,j=1}^{\ell}y_{i}y_{j}\sum_{\overrightarrow{x_{1}},\overrightarrow
{x_{2}}}a_{\overrightarrow{x_{1}}}^{i}a_{\overrightarrow{x_{2}}}%
^{j}\operatorname{Tr}\!\left[  \sigma_{\overrightarrow{x_{1}}}\sigma
_{\overrightarrow{x_{2}}}\right] \\
&  =2^{n}\sum_{i,j=1}^{\ell}y_{i}y_{j}\sum_{\overrightarrow{x}}%
a_{\overrightarrow{x}}^{i}a_{\overrightarrow{x}}^{j}\\
&  =2^{n}\sum_{i,j=1}^{\ell}y_{i}y_{j}\left(  \overrightarrow{a}^{i}%
\cdot\overrightarrow{a}^{j}\right)  ,\\
\operatorname{Tr}[HA_{i}]  &  =\operatorname{Tr}\!\left[  \left(
\sum_{\overrightarrow{x_{1}}}h_{\overrightarrow{x_{1}}}\sigma_{\overrightarrow
{x_{1}}}\right)  \left(  \sum_{\overrightarrow{x_{2}}}a_{\overrightarrow
{x_{2}}}^{i}\sigma_{\overrightarrow{x_{2}}}\right)  \right] \\
&  =\sum_{\overrightarrow{x_{1}}}h_{\overrightarrow{x_{1}}}\sum
_{\overrightarrow{x_{2}}}a_{\overrightarrow{x_{2}}}^{i}\operatorname{Tr}%
\left[  \sigma_{\overrightarrow{x_{1}}}\sigma_{\overrightarrow{x_{2}}}\right]
\\
&  =2^{n}\sum_{\overrightarrow{x}}h_{\overrightarrow{x}}a_{\overrightarrow{x}%
}^{i}\\
&  =2^{n}\overrightarrow{h}\cdot\overrightarrow{a}^{i},\\
\operatorname{Tr}[H]  &  =2^{n}h_{\overrightarrow{0}},\\
\operatorname{Tr}[H\omega]  &  =\sum_{\overrightarrow{x}}h_{\overrightarrow
{x}}\operatorname{Tr}[\sigma_{\overrightarrow{x}}\omega],\\
\sum_{i=1}^{\ell}y_{i}\operatorname{Tr}[A_{i}]  &  =2^{n}\sum_{i=1}^{\ell
}y_{i}a_{\overrightarrow{0}}^{i},\\
\sum_{i=1}^{\ell}y_{i}\operatorname{Tr}[A_{i}\omega]  &  =\sum_{i=1}^{\ell
}y_{i}\operatorname{Tr}\!\left[  \left(  \sum_{\overrightarrow{x}%
}a_{\overrightarrow{x}}^{i}\sigma_{\overrightarrow{x}}\right)  \omega\right]
\\
&  =\sum_{i=1}^{\ell}y_{i}\sum_{\overrightarrow{x}}a_{\overrightarrow{x}}%
^{i}\operatorname{Tr}\!\left[  \sigma_{\overrightarrow{x}}\omega\right]  .
\end{align}
Plugging these values into \eqref{eq:constrained-Ham-opt-HS-norm}, we find
that%
\begin{multline}
\left\Vert H-\sum_{i=1}^{\ell}y_{i}A_{i}-\mu I-\nu\omega\right\Vert _{2}%
^{2}=2^{n}\left\Vert \overrightarrow{h}\right\Vert _{2}^{2}+2^{n}\sum
_{i,j=1}^{\ell}y_{i}y_{j}\left(  \overrightarrow{a}^{i}\cdot\overrightarrow
{a}^{j}\right)  +\mu^{2}2^{n}\\
+\nu^{2}\operatorname{Tr}[\omega^{2}]-2^{n+1}\sum_{i=1}^{\ell}y_{i}%
\overrightarrow{h}\cdot\overrightarrow{a}^{i}-2^{n+1}\mu h_{\overrightarrow
{0}}-2\nu\sum_{\overrightarrow{x}}h_{\overrightarrow{x}}\operatorname{Tr}%
[\sigma_{\overrightarrow{x}}\omega]\\
+2^{n+1}\mu\sum_{i=1}^{\ell}y_{i}a_{\overrightarrow{0}}^{i}+2\nu\sum
_{i=1}^{\ell}y_{i}\sum_{\overrightarrow{x}}a_{\overrightarrow{x}}%
^{i}\operatorname{Tr}\!\left[  \sigma_{\overrightarrow{x}}\omega\right]
+2\mu\nu.
\end{multline}
This concludes the proof.
\end{proof}

\section{Details of examples: Simulations}

\label{sec:details-simulations}

In this section, we discuss some important features and specific details of our simulations from the main text (see Section~\ref{sec:qslack-sims} for QSlack simulations and Section~\ref{sec:cslack-sims} for CSlack simulations). The number of qubits, number of layers, penalty parameter, learning rate, and gradient method for the different simulations for the purification and convex-combination ansatz are given in Tables~\ref{tab:qslack-pur-sim-details} and \ref{tab:qslack-cc-sim-details}, respectively.

In many of the simulations, we use a decreasing learning rate scheme. Every 100 iterations, we use a least-squares regression model to fit a line to the previous $L$~iterations of objective function values, where the size of the window $L$ is the hyperparameter chosen beforehand. The value of $L$ is specified as part of the scheme. For example, in Table~\ref{tab:qslack-pur-sim-details}, for the normalized trace distance, we use the decreasing learning rate scheme with a window of size $500$. If the slope of this line was negative (in the case of the primal optimization) or positive (in the case of the dual optimization), we divided the learning rate by a factor of two. We repeated this process until some minimum learning rate was met. In this case, for both the primal and the dual we initialized the learning rate to a value of 0.1, and we allowed it to decrease to a minimum value of 0.001.

In some simulations, we used a slight variation of the scheme above, which additionally allows the learning rate to increase in the case where the objective function is training in the correct direction. This is true when the slope of the line is positive for the case of the primal, or negative in the case of the dual. In this modified scheme, we define a pair $(L, r)$, where $L$ is the window size (same as before), and the new hyperparameter $r$ is the multiplier of the learning rate when the condition is met to increase the learning rate.

In addition to the decreasing learning rate scheme, we often normalized the gradient estimate to have unit norm. We found that this prevents immediate divergence due to a high initial penalty value, and it also allowed us to initialize the learning rate to a higher value than would otherwise be possible.

\renewcommand{\arraystretch}{2}
\begin{table*}
\begin{tabular}{|P{2.5cm}|P{1.5cm}|P{1.5cm}|P{1.5cm}|P{1.6cm}|P{2.75cm}|P{2cm}|}
\hline
\multicolumn{2}{|c|}{Problem} & Number of qubits & Number of layers & Penalty parameter & Learning rate & Gradient \\ \hline\hline

\multirow[c]{2}{2.5cm}{\centering Normalized trace distance} & Primal & $2$ & $3$ & $10$ & \multirow[c]{2}{2.75cm}{\centering Decreasing rate scheme (500)} & \multirow[c]{2}{2cm}{\centering Normalized SPSA} \\ \cline{2-5}
\multicolumn{1}{|c|}{} & Dual & $2$ & $3$ & $100$ &  &  \\ \hline

\multirow[c]{2}{2.5cm}{\centering Root fidelity} & Primal & $2$ & 4  & $45$ &{\centering Decreasing rate scheme (500) } & {\centering Normalized SPSA} \\ \cline{2-7}
\multicolumn{1}{|c|}{} & Dual & $2$ & 3  & $5$ & Decreasing rate scheme (300)&  Normalized SPSA\\ \hline

\multirow[c]{2}{2.5cm}{\centering Entanglement negativity} & Primal & $2$ & $3$  & $5$ & \multirow[c]{2}{2.75cm}{\centering Decreasing rate scheme (500)} & \multirow[c]{2}{2cm}{\centering Normalized SPSA} \\ \cline{2-5}
\multicolumn{1}{|c|}{} & Dual & $2$ & $3$  & $100$ &  & \\ \hline

\multirow[c]{2}{2.5cm}{\centering Constrained Hamiltonian} & Primal & $2$ & 2  & $100$ &{\centering Decreasing rate by half every 10000 iterations } & {\centering Normalized SPSA} \\ \cline{2-7}
\multicolumn{1}{|c|}{} & Dual & $2$ & 2  & $100$ & Decreasing rate by half every 1000 iterations&  SPSA\\ \hline
\end{tabular}
\caption{Details of QSlack simulations using the purification ansatz.}
\label{tab:qslack-pur-sim-details}
\end{table*}

\begin{table*}
\begin{tabular}{|P{2.5cm}|P{1.5cm}|P{1.5cm}|P{1.5cm}|P{1.6cm}|P{2.75cm}|P{2cm}|}
\hline
\multicolumn{2}{|c|}{Problem} & Number of qubits & Number of layers & Penalty parameter & Learning rate & Gradient \\ \hline\hline

\multirow[c]{2}{2.5cm}{\centering Normalized trace distance} & Primal & $2$ & 4+2  & $10$ &{\centering Fixed at 0.005 } & {\centering Normalized SPSA} \\ \cline{2-7}
\multicolumn{1}{|c|}{} & Dual & $2$ & 3+2  & $100$ & Decreasing rate by half every 1000 iterations&  Normalized SPSA\\ \hline

\multirow[c]{2}{2.5cm}{\centering Root fidelity} & Primal & 2 & 8 + 3  & 50 & \multirow[c]{2}{2.75cm}{\centering Decreasing rate scheme (500, 1.1)} & \multirow[c]{2}{2cm}{\centering Normalized SPSA} \\ \cline{2-5}
\multicolumn{1}{|c|}{} & Dual & 2 & 4 + 3  & 5 & & \\ \hline

\multirow[c]{2}{2.5cm}{\centering Entanglement negativity} & Primal & 2 & 2+1  & 5 & \multirow[c]{2}{2.75cm}{\centering Decreasing rate scheme (500)} & \multirow[c]{2}{2cm}{\centering Normalized SPSA} \\ \cline{2-5}
\multicolumn{1}{|c|}{} & Dual & 2 & 3+2  & 100 &  & \\ \hline

\multirow[c]{2}{2.5cm}{\centering Constrained Hamiltonian} & Primal & 2 & 15+2  & 100 & \multirow[c]{2}{2.75cm}{\centering Decreasing rate by half every 1000 iterations} & \multirow[c]{2}{2cm}{\centering Normalized SPSA} \\ \cline{2-5}
\multicolumn{1}{|c|}{} & Dual & 2 & 15+2  & 100 &  & \\ \hline
\end{tabular}
\caption{Details of QSlack simulations using the convex-combination ansatz. The number of layers is defined as the number of layers used in the parameterized unitaries plus the number of layers used in the quantum circuit Born machines.}
\label{tab:qslack-cc-sim-details}
\end{table*}

\subsection{Normalized trace distance}
\label{app:trace-distance-sim}

Figure~\ref{fig:td-error-penalty} shows the total error and the penalty value during the optimization as a function of the number of iterations. 

\begin{figure}[H]
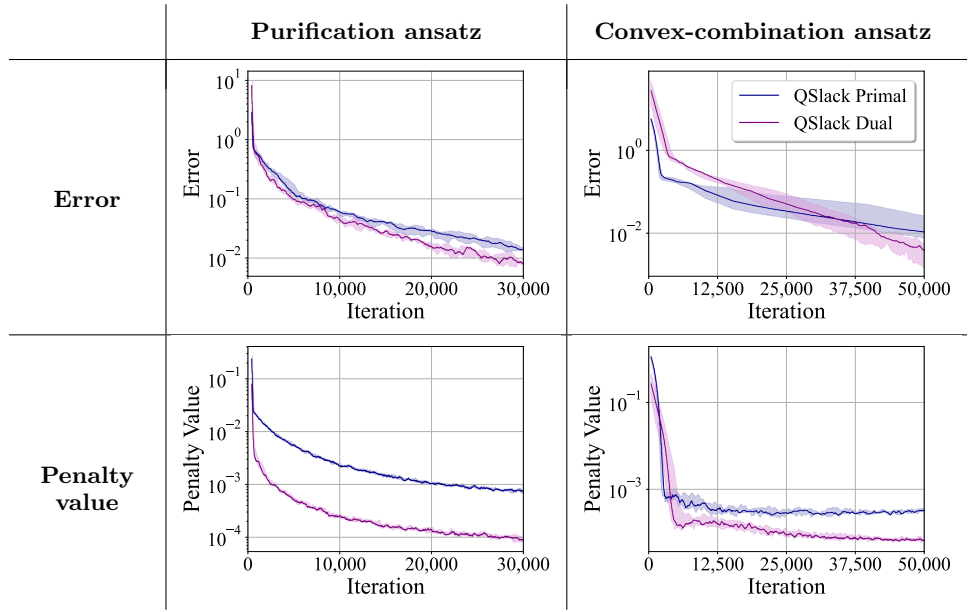

\renewcommand{\arraystretch}{1.8}
\centering
\resizebox{\linewidth}{!}{

\begin{tabular}{P{0.15\linewidth}|c|c}
& \textbf{Purification ansatz} & \textbf{Convex-combination ansatz}
\vspace{-0.4cm} \\
\hline
\centering \textbf{Error}  & \tdperror & \tdccaerror \\
\hline
\centering\textbf{Penalty value}  & \tdppenalty & \tdccapenalty \\
\end{tabular}
}
\caption{Normalized trace distance error and penalty terms.}
\label{fig:td-error-penalty}
\end{figure}

\subsection{Root fidelity}
\label{app:fidelity-sim}

Figure~\ref{fig:f-error-penalty} shows the total error and the penalty value during the optimization as a function of the number of iterations. 

\begin{figure}[H]
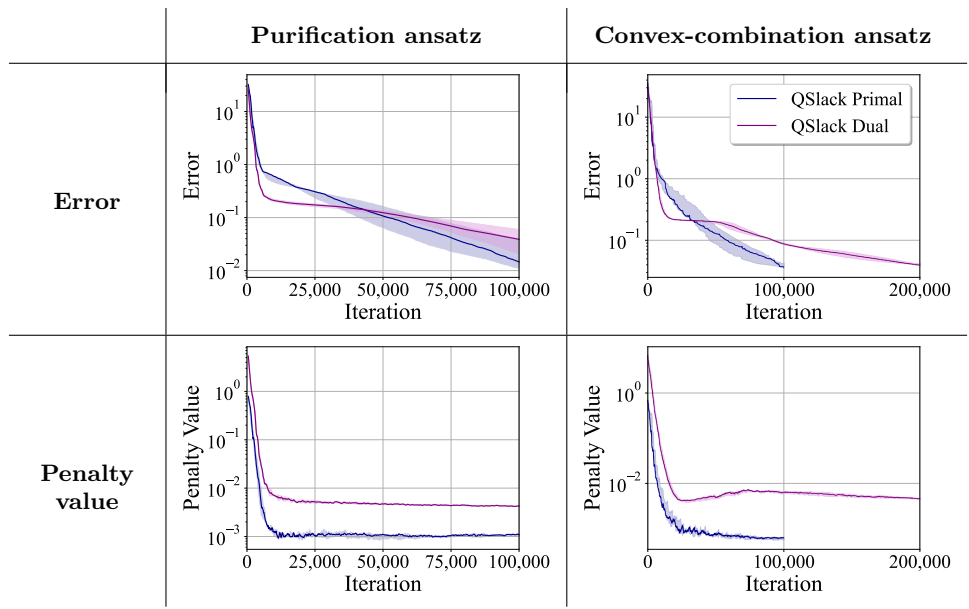

\renewcommand{\arraystretch}{1.8}
\centering
\resizebox{\linewidth}{!}{

\begin{tabular}{P{0.15\linewidth}|c|c}
& \textbf{Purification ansatz} & \textbf{Convex-combination ansatz}
\vspace{-0.4cm} \\
\hline
\centering \textbf{Error}  & \fperror & \fccaerror \\
\hline
\centering\textbf{Penalty value}  & \fppenalty & \fccapenalty \\
\end{tabular}
}
\caption{Root fidelity error and penalty terms.}
\label{fig:f-error-penalty}
\end{figure}

\subsection{Entanglement negativity}
\label{app:negativity-sim}

Figure~\ref{fig:en-error-penalty} shows the total error and the penalty value during the optimization as a function of the number of iterations. 

\begin{figure}[H]
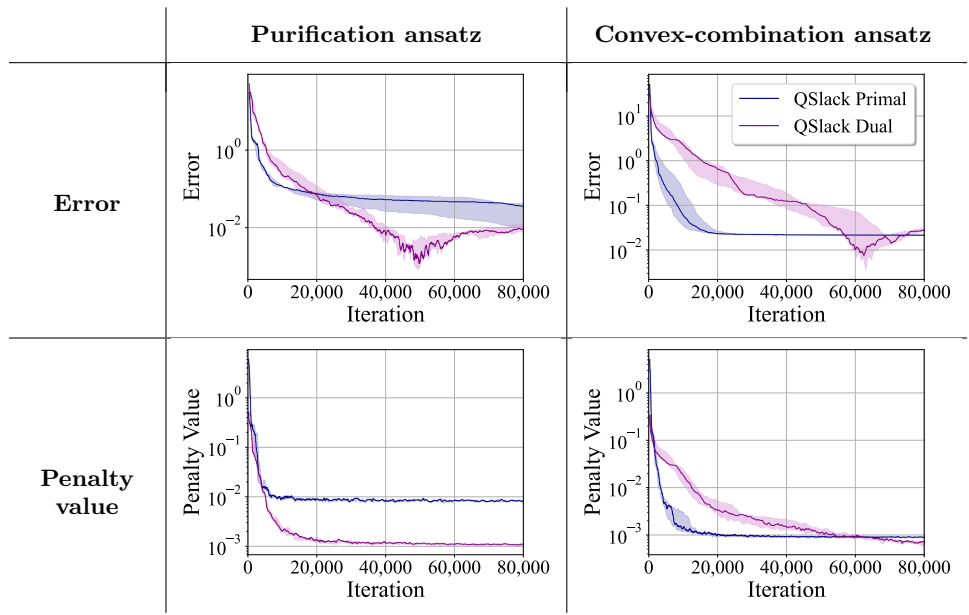

\renewcommand{\arraystretch}{1.8}
\centering
\resizebox{\linewidth}{!}{

\begin{tabular}{P{0.15\linewidth}|c|c}
& \textbf{Purification ansatz} & \textbf{Convex-combination ansatz}
\vspace{-0.4cm} \\
\hline
\centering \textbf{Error}  & \enperror & \enccaerror \\
\hline
\centering\textbf{Penalty value}  & \enppenalty & \enccapenalty \\
\end{tabular}
}
\caption{Entanglement negativity error and penalty terms.}
\label{fig:en-error-penalty}
\end{figure}

\subsection{Constrained Hamiltonian optimization}
\label{app:constrained-Ham-sim}

To showcase our algorithm, we considered the following example two-qubit Hamiltonian: 
\begin{equation}
   H= \sigma_{Z}^{1}\otimes\sigma_{Z}^{2}
+\sigma_{X}^{1}\otimes \sigma_{I}^2 + \sigma_{I}^1 \otimes \sigma_{X}^{2},
\end{equation}
and the following constraints:
\begin{align}
A_1 & = \sigma_Y \otimes \sigma_I, \qquad b_1  =0.2, \\
A_2 & = \sigma_I \otimes \sigma_Z, \qquad b_2=0.1.
\end{align}

Figure~\ref{fig:ConstrainedHamil-error-penalty} shows the total error and the penalty value during the optimization as a function of the number of iterations. 

\begin{figure}[H]
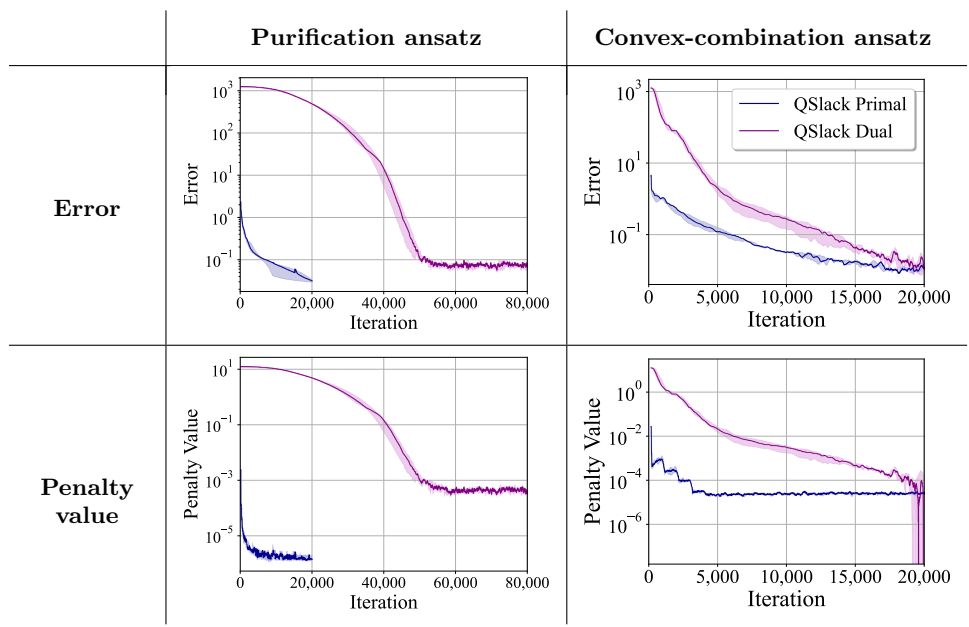

\renewcommand{\arraystretch}{1.8}
\centering
\resizebox{\linewidth}{!}{

\begin{tabular}{P{0.15\linewidth}|c|c}
& \textbf{Purification ansatz} & \textbf{Convex-combination ansatz}
\vspace{-0.4cm} \\
\hline
\centering \textbf{Error}  &  \hperror & \hccaerror \\
\hline
\centering\textbf{Penalty value}  & \hppenalty & \hccapenalty \\
\end{tabular}
}
\caption{Error and penalty values for the constrained Hamiltonian optimization problem with the true value $-2.2097$.}
\label{fig:ConstrainedHamil-error-penalty}
\end{figure}

\begin{table*}
\begin{tabular}{|P{2.5cm}|P{1.5cm}|P{1.5cm}|P{1.5cm}|P{1.6cm}|P{2.75cm}|P{2cm}|}
\hline
\multicolumn{2}{|c|}{Problem} & Number of Qubits & Number of Layers & Penalty parameter & Learning Rate & Gradient \\ \hline\hline

\multirow[c]{2}{2.5cm}{\centering Total variation distance} & Primal & $2$ & $2$ & $10$ & \multirow[c]{2}{2.75cm}{\centering Decreasing rate scheme (300)} & \multirow[c]{2}{2cm}{\centering Normalized SPSA} \\ \cline{2-5}
\multicolumn{1}{|c|}{} & Dual & $2$ & $2$ & $100$ &  &  \\ \hline

\multirow[c]{2}{2.5cm}{\centering Constrained classical Hamiltonian} & Primal & $2$ & $3$  & $10$ & \multirow[c]{2}{2.75cm}{\centering Decreasing rate scheme (300)} & \multirow[c]{2}{2cm}{\centering Normalized SPSA} \\ \cline{2-5}
\multicolumn{1}{|c|}{} & Dual & $2$ & $3$  & $10$ & & \\ \hline
\end{tabular}
\caption{Details of CSlack simulations.}
\label{tab:cslack-sim-details}
\end{table*}

\subsection{Total variation distance}
\label{app:total-variation-distance-sim}

The simulation results for both the primal and dual optimizations of the total variation distance on two probability distributions resulting from two-qubit quantum circuit Born machines are shown in Figure~\ref{fig:tvd-sandwich}. Plots showing decreasing error and penalty values across training are given in Figure \ref{fig:TotalVariationDistance-error-penalty}.

\begin{figure}[H]
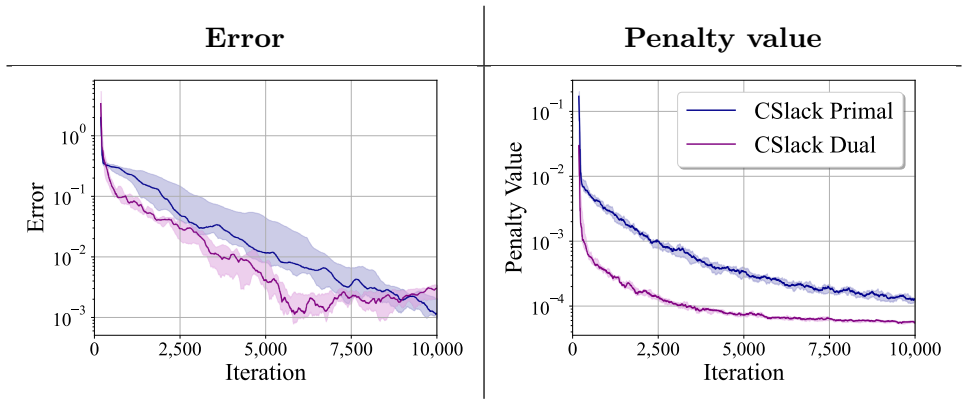

\renewcommand{\arraystretch}{1.8}
\centering
\resizebox{\linewidth}{!}{

\begin{tabular}{c|c}
\textbf{Error} & \textbf{Penalty value}\\
\hline
\tvderror & \tvdpenalty \\
\end{tabular}
}
\caption{Error and penalty values for the total variation distance problem.}
\label{fig:TotalVariationDistance-error-penalty}
\end{figure}

\subsection{Constrained classical Hamiltonian optimization}
\label{app:classical-constrained-Ham-sim}

In our simulations, the inputs were selected to be the following Hamiltonian vector, constraint vectors, and constraint values:
\begin{align}
 h & = s_1 \otimes s_1,\\
 a_1 & = 0.5 (s_1 \otimes s_0), \qquad b_1 = 0.1,\\
 a_2 & = 0.7 (s_0 \otimes s_1), \qquad b_2 = 0.3.
\end{align}
The simulation results for both the primal and dual of this problem instance are shown in Figure~\ref{fig:tvd-sandwich}. Plots showing decreasing error and penalty values across training are given in Figure \ref{fig:ClassicalConstrainedHamil-error-penalty}.

\begin{figure}[H]
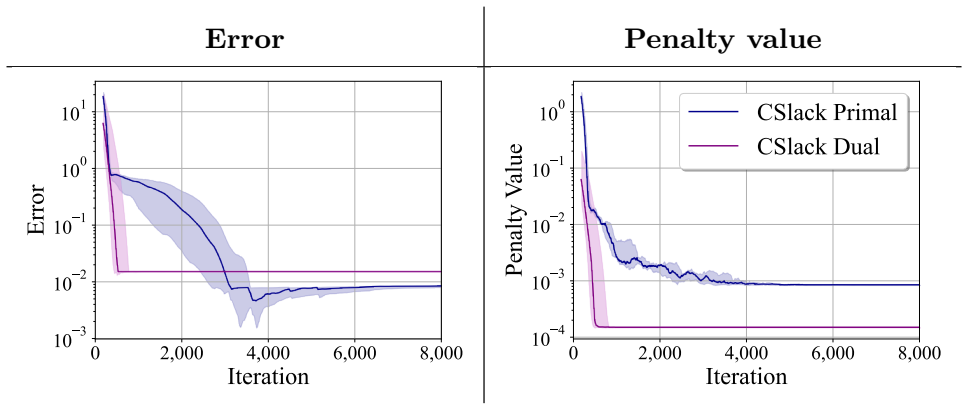

\renewcommand{\arraystretch}{1.8}
\centering
\resizebox{\linewidth}{!}{

\begin{tabular}{c|c}
\textbf{Error} & \textbf{Penalty value}\\
\hline
\cherror & \chpenalty \\
\end{tabular}
}
\caption{Error and penalty values for the constrained classical Hamiltonian optimization problem.}
\label{fig:ClassicalConstrainedHamil-error-penalty}
\end{figure}

\section{Barren plateau analysis}

\label{appendix:BarrenPlateaus}

Given the broad applicability of QSlack and CSlack, a complete discussion of all possible ways in which barren plateaus could manifest or be avoided in this framework would be formidable. As such, we only aim to provide an overview of some of the key considerations here. 

\medskip

\paragraph*{Background.}

A barren plateau is a cost landscape for which the magnitudes
of gradients vanish exponentially with problem size~\cite{McClean_2018}. This phenomenon has been shown to be equivalent to exponential concentration, in which the loss  concentrates with high probability to a single fixed value~\cite{arrasmith2022equivalence}. On such landscapes, training with a polynomial number of measurement shots results in a poorly trained model, regardless of the optimization method employed~\cite{arrasmith2020effect}. More precisely, exponential concentration can be formally defined as follows. 

\begin{definition} [Exponential concentration]\label{def:exp-concentration}
Consider a quantity $X(\alpha)$ that depends on a set of variables $\alpha$ and can be measured from a quantum computer as the expectation of some observable. $X(\alpha)$ is said to be deterministically exponentially concentrated in the number~$n$ of qubits  towards a certain fixed value $\mu$ if
\begin{align}
    |X(\alpha) - \mu |\leq \beta \in O(1/b^n) \;,
\end{align}
for some $b>1$ and all $\alpha$. Analogously, $X(\alpha)$ is probabilistically exponentially concentrated if
\begin{align} \label{eq:def-prob-concentration}
    {\rm Pr}_{\alpha}[|X(\alpha) - \mu| \geq \delta] \leq \frac{\beta}{\delta^2} \;\; , \; \beta \in O(1/b^n) \;,
\end{align}
for $b> 1$. That is, the probability that $X(\alpha)$ deviates from $\mu$ by a small amount $\delta > 0$ is exponentially small for all $\vec{\alpha}$.
\end{definition}

The barren-plateau phenomenon has predominantly been studied in the context of losses of the form
\begin{equation}\label{eq:VQEstylecost}
    C(\theta) = \text{Tr}[ O U(\theta) \rho U(\theta)^\dagger] \;, 
\end{equation}
where $\rho$ is an $n$-qubit input state and $O$ is a Hermitian operator. 

While a complete analysis of the barren plateau phenomenon requires the interplay between the choice of parameterized quantum circuit and measurement operations to be considered in conjunction~\cite{ragone2023unified, fontana2023adjoint}, in the case of problem-agnostic ans\"{a}tze, a good first step to assessing whether a problem will exhibit a barren plateau is to determine whether it is local or global. It was shown in \cite{cerezo2021cost} that, for random hardware efficient ans\"{a}tze, global costs (i.e., ones for which $O$ acts non-trivially on $\mathcal{O}(n)$ qubits), exhibit barren plateaus at all depths. In contrast local costs, that is losses where the measurement acts non-trivially on at most $\log(n)$ qubits, can exhibit gradients that vanish at worst polynomially (i.e. do not have a barren plateau). For this guarantee to hold, the initial state needs to not be too entangled or too mixed. In particular, the $\log(n)$ marginals of the initial state must not be exponentially close to maximally mixed. 

\medskip

\paragraph*{QSlack.}

The purification ansatz naturally fits into this framework and it can be shown that the convex combination ans\"{a}tze (both standard and correlated) can also be understood through this lens with a little thought.
We first note that the standard (uncorrelated) convex combination ansatz can be written in the form of \eqref{eq:VQEstylecost} by taking the initial state as $\rho = \sum_x p_{\varphi}(x) |x\rangle \!\langle x|$. To study the case of the correlated convex combination ansatz, we then note that whether the cost is optimized via gradient descent directly or using a tensor network, ultimately what matters in both cases are cost gradients~\cite{arrasmith2020effect}. Hence we can consider a simplified version of the correlated convex combination ansatz that takes the form 
\begin{equation}
    \rho(\varphi, \gamma) = \sum_x p_{\varphi}(x) U(\gamma
_{x})|x\rangle\!\langle x|U(\gamma_{x})^{\dag}  \, ,
\end{equation}
and so the cost under consideration is 
\begin{equation}\label{eq:VQEstylecost-2}
    C(\varphi, \gamma) = \sum_x p_{\varphi}(x) \text{Tr}[ O U(\gamma
_{x})|x\rangle\!\langle x|U(\gamma_{x})^{\dag} ] \;. 
\end{equation}
While this is not of precisely the same form as \eqref{eq:VQEstylecost}, given that the parameters for each circuit in the convex combination are independent of those in the other terms, the other terms vanish upon differentiation. 
For example, consider the derivative with respect to the $j$th  component of the $\gamma_x$ parameter vector. In this case the partial derivative takes the form 
\begin{equation}
     p_{\varphi_{1}}(x) \partial_{\gamma_x^j} \text{Tr}\left[ U(\gamma
_{x})|x\rangle\!\langle x|U(\gamma_{x})^{\dag} O \right] \, ,
\end{equation}
which is of precisely the same structure as taking the partial derivative of a single component of $\theta$ in \eqref{eq:VQEstylecost}. Hence the standard barren plateau analysis carries over to the correlated convex combination ansatz. 

Whether or not QSlack utilises global or local costs depends on the precise application to which it is applied. For example, for the case of constrained Hamiltonian optimization, as applied to local Hamiltonians (as it typical in many physically motivated cases), QSlack can involve only local costs. In this case the purification ansatz and correlated convex combination ansatz provably avoid barren plateaus for shallow hardware efficient circuits. The case of the standard convex combination ansatz is a little more subtle with the trainability also depending on the initial state. If the initial state from the generative model, i.e. $\rho = \sum_x p_{\varphi}(x) |x\rangle \!\langle x|$ is too close to maximally mixed, then even for shallow circuits and local Hamiltonians the landscape can exhibit a barren plateau. Hence careful consideration will need to be taken for the choice in initial distribution for $p_{\varphi}(x)$. 

In the case of the examples of trace distance, root fidelity, and entanglement negativity, global terms do appear in the costs, which are likely to lead to trainability difficulties for unstructured ans\"{a}tze. However, it is worth stressing at this point that barren plateaus are an average phenomenon, and it is possible that these difficulties could potentially be side stepped if one can develop clever initialization strategies. Whether such strategies could be developed in the context of QSlack remains an open question. Furthermore, whether or not they can may well depend on properties of the solution state, with highly mixed and highly random targets likely to prove more challenging to learn~\cite{holmes2021barren}.

\medskip

\paragraph*{CSlack.}

When a quantum circuit Born machine is used with CSlack one also needs to consider trainability concerns. The overlap terms $\sum_{x} p(x) t(x)$, as measured using the collision test, correspond to `explicit' costs and hence are subject to barren plateaus~\cite{rudolph2023trainability}. 
To see this note that in the collision test one is computing the expectation value of a Kronecker delta function, i.e., $\sum_{x} p(x) t(x) = \langle \delta_{x,y} \rangle_{x \sim p, y \sim t }$. 
The core problem is that one only obtains a signal in the collision test in the case that one draws two identical bit strings. However, if the distributions have exponential support, the probability of this occurring is exponentially small in general. Hence, with finite shots, the overlap will be precisely zero for nearly all parameter settings, making meaningful training practically impossible. 

One way one could attempt to resolve this would be to replace the Kronecker delta with a Gaussian kernel, as is done in the Maximum Mean Discrepancy loss used in generative modeling. That is, one could use the fact that $\sum_{x} p(x) t(x) = \lim_{\sigma \rightarrow 0} \langle e^{-|x - y | / \sigma^2} \rangle_{x \sim p, y \sim t }$, train initially with a finite $\sigma$ value (for which a larger signal should be observed), and then slowly decrease the $\sigma$ value during training to regain the original cost value. While we believe this suggestion is original in this context, a similar approach has been discussed in the context of generative modeling and would be a good reference to understand this suggestion further~\cite{rudolph2023trainability}.


\end{document}